\documentclass[11pt,letterpaper]{article}
\usepackage[letterpaper,left=1in,right=1in,top=1in,bottom=1.0in,footskip=0.4in]{geometry}

\usepackage{enumitem} 
\setlist{noitemsep,topsep=0pt,parsep=0pt} 

\usepackage[margin=1cm]{caption} 
\usepackage{todonotes}
\usepackage{subcaption}
\usepackage{lineno}
\usepackage{amsmath,amsthm}
\usepackage{graphicx}
\usepackage{amsfonts}
\usepackage{enumitem}
\usepackage{mathtools}
\usepackage{empheq}
\usepackage{enumitem}
\usepackage{float}
\usepackage{tikz}
\usetikzlibrary{positioning, shapes.misc}
\usepackage{mathtools}
\usepackage{nicematrix}
\usepackage{comment}
\usepackage{tikz}
\usetikzlibrary{arrows,backgrounds,calc,fit,decorations.pathreplacing,decorations.markings,shapes.geometric}

\tikzstyle{internal}=[circle,draw,fill=gray!30,minimum size=6pt,inner sep=0pt]

\tikzstyle{external} = [shape=circle]

\usepackage{subcaption}
\tikzset{every fit/.append style=text badly centered}

\usepackage{tyson}
\usepackage{braket}
\usepackage{bm}

\usepackage[bookmarks=true,hypertexnames=false]{hyperref}
\hypersetup{colorlinks=true, citecolor=blue, linkcolor=red, urlcolor=blue}
\newcommand{\Holant}{\operatorname{Holant}}
\newcommand{\PlHolant}{\operatorname{Pl-Holant}}
\newcommand{\holant}[2]{\ensuremath{\Holant\left(#1\mid #2\right)}}
\newcommand{\plholant}[2]{\ensuremath{\PlHolant\left(#1\mid #2\right)}}

\newcommand{\GH}{\operatorname{\#GH}}
\newcommand{\CSP}
{\operatorname{\#CSP}}
\newcommand{\PlCSP}{\operatorname{Pl-\#CSP}}

\newcommand{\BiGH}{\operatorname{Bi-\#GH}}

\newcommand{\numP}{{\rm \#P}}

\newcommand{\M}{\mathscr{M}}
\newcommand{\ii}{\mathfrak{i}}

\newenvironment{remark}{\medskip{\bfseries \noindent Remark:}}{\par\medskip}{\par\medskip}

\renewcommand{\thefigure}{\thesection.\arabic{figure}}

\newenvironment{claimproof}[1]{%
  \begin{proof}[\textbf{Proof of Claim #1}]%
  \def\customqed{End of Proof Claim #1}%
}{%
  \end{proof}%
}

\def\partIPlCSP2SecNum{5}
\def\partIPl2CSPSecPageNum{24}

\def\borderColor{blue!60}
\def\scale{0.6}
\def\nodeDist{1.4cm}

\title{New Planar Algorithms and a Full Complexity Classification of the Eight-Vertex Model}

\author{
Jin-Yi Cai\thanks{Department of Computer Sciences,
University of Wisconsin-Madison.}\\
\texttt{jyc@cs.wics.edu}
\and
Austen Fan\footnotemark[1]\\
\texttt{afan@cs.wisc.edu}
\and
Shuai Shao\thanks{School of Computer Science and Technology \& Hefei National Laboratory, University of Science and Technology of China. Supported by the Quantum Science and Technology $-$ National Science and Technology Major Project (QNMP), 2021ZD0302901.}\\
\texttt{shao10@ustc.edu.cn}
\and
Zhuxiao Tang\footnotemark[1]\\
\texttt{zztang@wisc.edu}
}

\date{} 

\begin{document}
\maketitle



\pagenumbering{arabic}

\thispagestyle{empty}
\begin{abstract}
We prove a complete complexity classification theorem for the planar eight-vertex model.
For \emph{every} parameter setting 
in ${\mathbb C}$ for the eight-vertex model,
the partition function is
either (1) computable in 
P-time for every graph, or
(2) \#P-hard for general graphs but computable in 
P-time for
planar graphs, or
(3) \#P-hard  even for planar graphs.
The classification has an explicit criterion.
In (2), we discover new P-time computable eight-vertex models
on planar graphs beyond 
Kasteleyn's algorithm 
for counting planar perfect matchings.\footnote{This is also known as the FKT algorithm.
Fisher and  Temperley~\cite{TF61}, and Kasteleyn~\cite{Kasteleyn1961} independently discovered
this algorithm on the grid graph. Subsequently Kasteleyn~\cite{Kasteleyn1967} 
generalized this to all
planar graphs.}  
They are obtained by a combinatorial transformation to the planar {\sc Even Coloring} problem followed by a holographic transformation to the tractable cases in the planar six-vertex model.
%
In the process, we also encounter non-local connections between the planar eight vertex model and the bipartite Ising model, conformal lattice interpolation and M\"{o}bius transformation
from complex analysis. The proof also makes use of cyclotomic fields.
\end{abstract}

\newpage
\setcounter{page}{1}
\tableofcontents
\addtocontents{toc}{\protect\setcounter{tocdepth}{2}}

\newpage

\setcounter{figure}{0}

\section{Introduction}

\renewcommand{\thefigure}{\arabic{figure}}

The study of partition functions in statistical mechanics for various physical models has a long history.
One considers a set of particles connected by bonds, with physical laws imposing various local constraints on  valid local configurations, associated with suitable weights. 
A global configuration $\sigma$ is valid if it satisfies
all local constraints, and the weight $w(\sigma)$ is the product of local weights.
The partition function is then the sum of the weights $w(\sigma)$ over all valid global configurations
$\sigma$.

Well-known examples of partition functions from physics
include the Ising model, the Potts model, the hardcore model, the
{six-vertex model}, and the eight-vertex model.
Most of these are \emph{spin systems}~\cite{jerrum-sinclair,
goldberg-jerrum-patterson,goldberg-jerrum-potts,lu-ying-li}, where each particle takes a spin value such as $+$ or $-$, modeled as a Boolean variable, and local interactions are expressed via binary (edge) constraint functions. 
These models are nicely captured by the \emph{counting constraint satisfaction problem} (\#CSP) framework~\cite{bulatov,dyer-richerby, bulatov2012csp, cai-chen-lu-nonnegative-csp, caichen}.
Some physical systems are more naturally formulated as orientation problems,
and these can be modeled by the more general \emph{Holant problems}~\cite{Cai-Lu-Xia}, of which
\#CSP is a special case.
In this paper, we study the eight vertex model, which is an orientation problem.

\begin{figure}[h!]
\centering
\begin{subfigure}[b]{0.115\linewidth}
\centering\includegraphics[width=\linewidth]{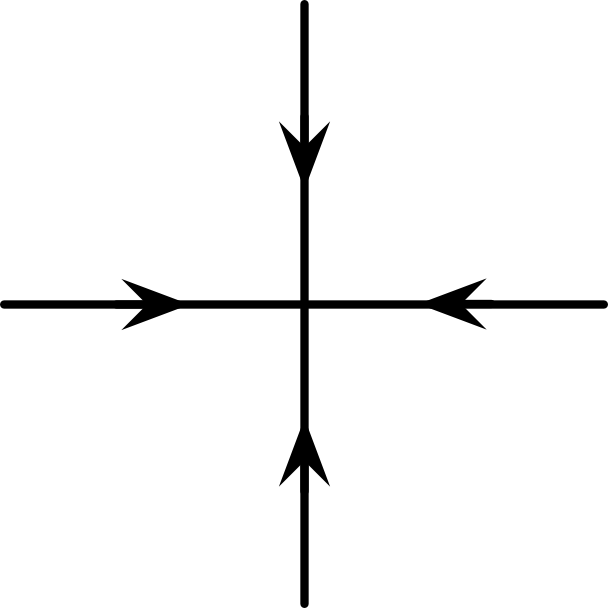}\caption{$w_1$}
\label{fig:orientations_1}
\end{subfigure}
\begin{subfigure}[b]{0.115\linewidth}
\centering\includegraphics[width=\linewidth]{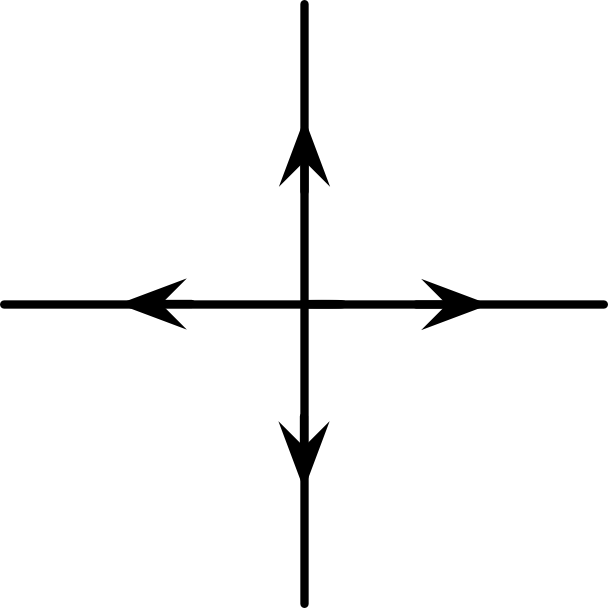}\caption{$w_2$}
\label{fig:orientations_2}
\end{subfigure}
\begin{subfigure}[b]{0.115\linewidth}
\centering\includegraphics[width=\linewidth]{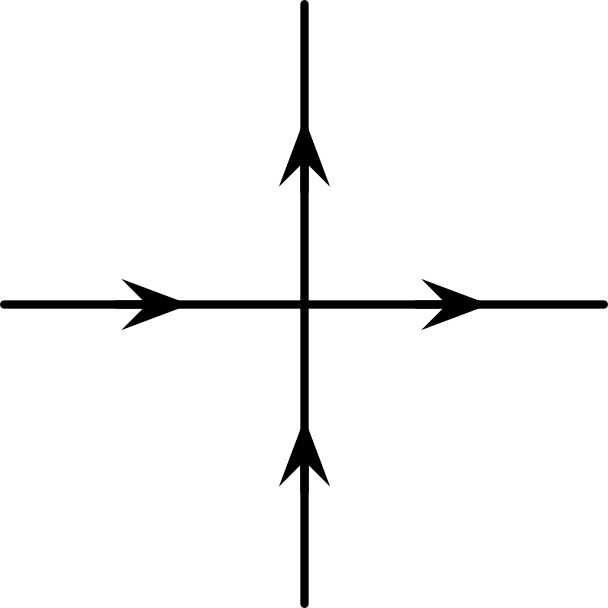}\caption{$w_3$}
\label{fig:orientations_3}
\end{subfigure}
\begin{subfigure}[b]{0.115\linewidth}
\centering\includegraphics[width=\linewidth]{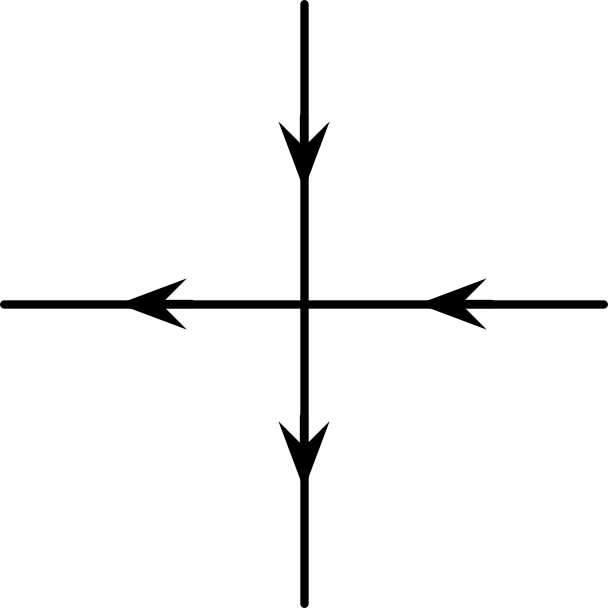}\caption{$w_4$}
\label{fig:orientations_4}
\end{subfigure}
\begin{subfigure}[b]{0.115\linewidth}
\centering\includegraphics[width=\linewidth]{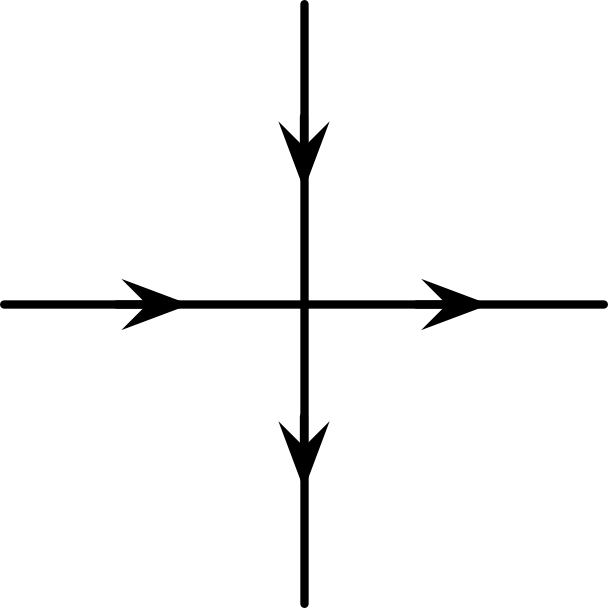}\caption{$w_5$}
\label{fig:orientations_5}
\end{subfigure}
\begin{subfigure}[b]{0.115\linewidth}
\centering\includegraphics[width=\linewidth]{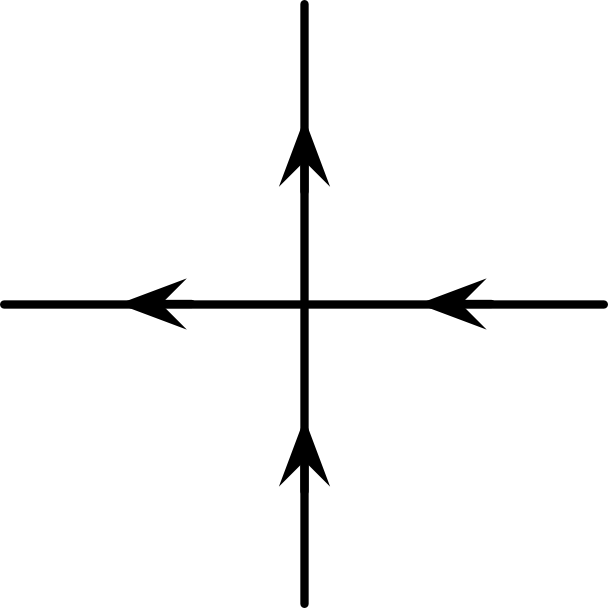}\caption{$w_6$}
\label{fig:orientations_6}
\end{subfigure}
\begin{subfigure}[b]{0.115\linewidth}
\centering\includegraphics[width=\linewidth]{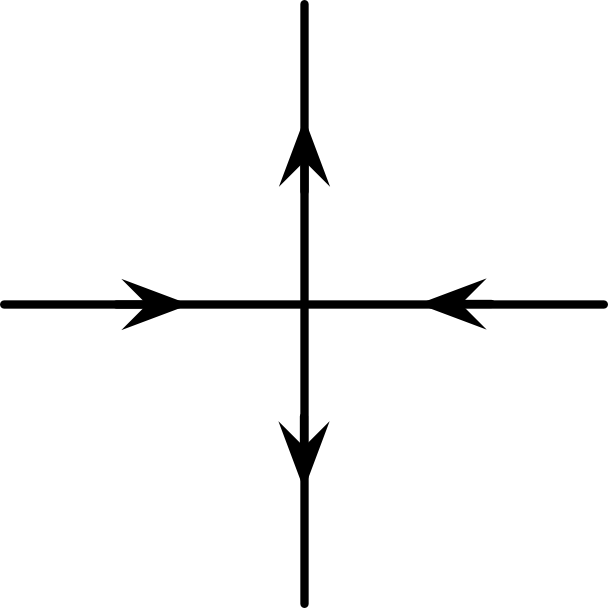}\caption{$w_7$}
\label{fig:orientations_7}
\end{subfigure}
\begin{subfigure}[b]{0.115\linewidth}
\centering\includegraphics[width=\linewidth]{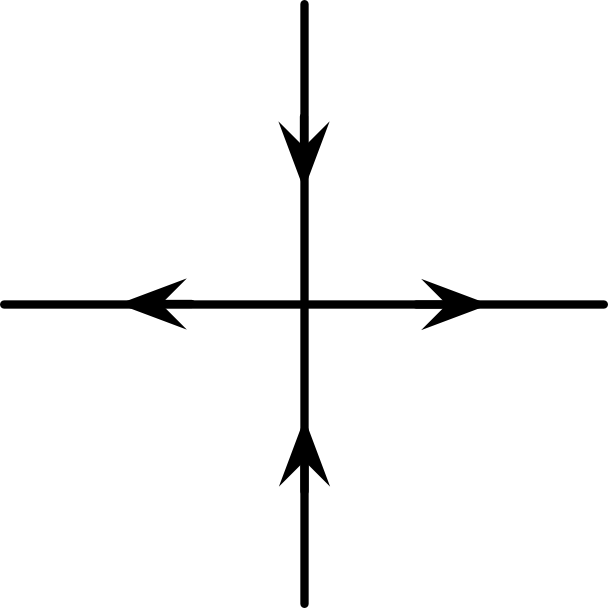}\caption{$w_8$}
\label{fig:orientations_8}
\end{subfigure}
\caption{Valid configurations of the eight-vertex model.}\label{fig: eight vertex orientations}
\end{figure}

The eight-vertex model was first introduced by Sutherland~\cite{sutherland1970two} and by Fan and  Wu~\cite{fan1970general}. 
It generalizes the six-vertex model~\cite{Pauling}, encompassing notable  physical systems such as the water-ice model, the KDP model, and the Rys F-model (for antiferroelectrics).
Traditionally, the model is defined on a square lattice with specific boundary conditions. 
More generally we consider the model on
4-regular graphs.
Given a 4-regular graph $G=(V, E)$, an orientation $\sigma$ assigns a direction to each edge.  An orientation is valid for the eight-vertex model if every vertex has even in-degree (equivalently, even out-degree). 
Such orientations are called even orientations.
The valid local configurations at each vertex, {\bf eight} in total, are illustrated in Figure~\ref{fig: eight vertex orientations}.
Each configuration type $i$ is associated with a (possibly complex) parameter $w_i$, called its Boltzmann weight.\footnote{In classical statistical mechanics, the weights $w_i$ are typically non-negative real numbers. However, complex weights arise naturally in quantum settings. For example, the partition function with complex parameters is closely related to quantum circuit output amplitudes~\cite{de2011quantum}.}
Given an even orientation $\sigma$, let $n_i$ denote the number of vertices in configuration type $i$. The weight of $\sigma$ is defined as $w(\sigma) = \prod_{1 \leq i \leq 8} w_i^{n_i}$. The partition function of the eight-vertex model on $G$ is then
\[ Z_{\rm 8V}(G) = \sum_{\sigma \in {\mathscr{O}}(G)} w(\sigma), \]
where ${\mathscr{O}}(G)$ denotes the set of all even orientations of $G$.
The six-vertex model arises as a special case where $w_1 = w_2 = 0$, i.e., only configurations with in-degree (and out-degree) exactly 2 are allowed; this corresponds to Eulerian orientations.

Computing the partition function is, in general, a \#P-hard problem due to the exponential number of configurations. However, for certain special parameter settings 
(and these are among the most interesting cases),
the partition function can be computed in polynomial time; such cases are called tractable.
In particular, over planar structures physicists
had discovered some remarkable algorithms, such as
the FKT algorithm~\cite{TF61, Kasteleyn1961, Kasteleyn1967}.
From the perspective of computational complexity, the goal is to classify all parameter regimes into either tractable or \#P-hard cases (with nothing in between!).
Without taking into account for planarity,
a complexity dichotomy theorem
for the eight-vertex model on general graphs has been established~\cite{CaiF17}. 
In this classification, some tractable cases for planar structures, including
the FKT cases, are \#P-hard (on general graphs). 
%
Previously, a complete complexity classification~\cite{CaiFS21} has been proved for the six-vertex model including for the planar graphs: the problem is
either (1) tractable for general graphs, or
(2)  tractable for
planar graphs (but \#P-hard for general graphs), or 
(3) \#P-hard  even for planar graphs.
The planar tractable case (2) includes those solvable via holographic transformations and the FKT algorithm.
The central question is this:
\emph{Are there further tractable planar cases of the eight-vertex model yet to be discovered, and
can we classify all eight-vertex models?}


In this paper, we answer this question affirmatively.
We identify a previously unknown tractable type for the planar eight-vertex model. 
This signifies that to achieve a
complete complexity classification will be challenging.
Nonetheless, we establish a complete classification: the eight-vertex model is either   (1) tractable for general graphs, or
(2)  tractable for
planar graphs (but \#P-hard for general graphs), or 
(3) \#P-hard  even for planar graphs.
The planar tractable case (2) includes those solvable by holographic transformations and the FKT algorithm, all the known tractable cases of the planar six-vertex model, and a newly discovered tractable type.

The new tractable type is discovered in a highly symmetric setting, namely when $w_1=w_2=a$, $w_3=w_4=b$, $w_5=w_6=c$, and $w_7=w_8=d$ (Lemma~\ref{lm: b=pm c, d=pm 1}).
This case also has a physical interpretation, as it corresponds to the ``zero field case" of the eight-vertex model~\cite{baxter1971eightvertex}.
In this setting, by a local ``combinatorial transformation", the planar eight-vertex model is equivalent to the planar {\sc even coloring} problem, which asks for the (weighted) number of green-red edge colorings of a planar graph such that every vertex has an even number of incident green edges.
After that, we apply a holographic transformation by $HZ=\frac{1}{2}\left[\begin{smallmatrix}   
1+\mathfrak{i} & 1-\mathfrak{i}\\
   1- \mathfrak{i} &1+\mathfrak{i}
\end{smallmatrix}\right]$, transforming the {\sc even coloring} problem back into an orientation problem.
For suitable parameters $a,b,c,d$, this orientation problem turns out to be a tractable case of the planar six-vertex model (which is provably not subsumed by holographic transformation plus FKT algorithm).
The above ``combinatorial transformation" plus holographic transformation paradigm also has the potential to find new tractable cases in general counting even orientation problems on planar graphs.

The classification of the six-vertex and eight-vertex models is not only  interesting
in its own right, but more importantly, it serves as foundational building blocks in  the
broader classification program for Holant problems.
The dichotomy for the eight-vertex model without the planarity constraint~\cite{CaiF17}
was already used as a crucial basic case for the  dichotomy of real-valued Holant problems on general graphs~\cite{real-holant-focs2021}.
 However, from the very beginning of the Holant framework~\cite{Val08}, {\bf planar} tractability, enabled by tools such as the FKT algorithm, has played a central role.
For symmetric constraint functions (a.k.a. signatures),
a complete classification has been achieved, including the planar setting~\cite{cai2015holant}.
However, most signatures are not symmetric, and a full understanding of Holant problems must account for asymmetric signatures.
Holographic algorithms ultimately rely on the FKT algorithm, which is for a Holant problem (namely, counting perfect matchings), and yet it has been shown to serve as a “universal” algorithmic
 paradigm for planar \#CSP~\cite{caifu16}.
But \cite{cai2015holant} shows an intriguing fact:
For Holant problem themselves, FKT is  \emph{not}  universal.
The previous work~\cite{CaiFS21} and our current work reinforce this by showing that FKT is not universal for Holant problems with asymmetric signatures.
The results proved in this paper will provide further building blocks for classifying the complexity of planar Holant problems with asymmetric signatures.

\section{Preliminaries}
Let $\mathfrak{i}$ denote a square root of $-1$.
Let $\zeta_m$ denote $e^{\frac{2\pi \mathfrak{i}}{m}}$ for $m\in \Z_+$.
Let 
$H = \frac{1}{\sqrt{2}}\left[\begin{smallmatrix} 1 & 1 \\ 1 & -1 \end{smallmatrix} \right]$, $Z = \frac{1}{\sqrt{2}}\left[\begin{smallmatrix} 1 & 1 \\ \mathfrak{i} & -\mathfrak{i} \end{smallmatrix} \right]$.
We denote by $\le_T^p$ and $\equiv_T^p$  polynomial time Turing reduction and equivalence, respectively.

\subsection{Definitions and notations}

A {\it constraint function} $f$, or a {\it signature}, of arity $k$
is a map $\{0,1\}^k  \rightarrow \mathbb{C}$.
Fix a set $\mathcal{F}$ of constraint functions. 
A {\it signature grid}
$\Omega=(G, \pi)$
 is a tuple, where $G = (V,E)$
is a graph called the {\it underlying graph} of $\Omega$, $\pi$ labels each $v\in V$ with a function
$f_v\in\mathcal{F}$ of arity ${\operatorname{deg}(v)}$,
and labels the incident edges
$E(v)$ at $v$ with input variables of $f_v$.
We consider all 0-1 edge assignments $\sigma$ on $E$,
each gives an evaluation
$\prod_{v\in V}f_v(\sigma|_{E(v)})$, where $\sigma|_{E(v)}$
denotes the restriction of $\sigma$ to $E(v)$. 
The counting problem on an instance $\Omega$ is to compute the {\it partition function}
\begin{equation}\label{eq: def of Holant value}
\Holant_{\Omega}=\Holant({\Omega};\mathcal{F})=\sum_{\sigma:E\rightarrow\{0, 1\}}\prod_{v\in V}f_v(\sigma|_{E_{(v)}}).
\end{equation}
The Holant problem $\Omega \mapsto \Holant_{\Omega}$ is a computational problem parameterized by a set $\mathcal{F}$ and is denoted by Holant$(\mathcal{F})$.
If $\mathcal{F}=\{f\}$ is a single set, for simplicity, we write $\{f\}$ as $f$ directly and also write $\{\mathcal{F}\} \cup \{g\}$ as $\mathcal{F}\cup g$ or $\mathcal{F},g$.
We denote
$\PlHolant(\mathcal{F})$ by the restriction of $\PlHolant(\mathcal{F})$ to planar graphs $G$.
We use $\holant{\mathcal{F}}{\mathcal{G}}$ to denote the Holant problem over signature grids with a bipartite graph $H = (U,V,E)$,
where each vertex in $U$ or $V$ is assigned a signature in $\mathcal{F}$ or $\mathcal{G}$
respectively.
$\plholant{\mathcal{F}}{\mathcal{G}}$ denotes its planar restriction.

Counting constraint satisfaction problems (\#CSP)
can be defined as a special case of Holant problems.
An instance of $\CSP(\mathcal{F})$ is presented
as a bipartite graph.
There is one node for each variable and for each occurrence
of constraint functions respectively.
Connect a constraint node to  a variable node if the
variable appears in that occurrence
of constraint, with a labeling on the edges
for the order of these variables.
This bipartite graph is also known as the \emph{constraint graph}.
If we attach each variable node with an \textsc{Equality} function,
and consider every edge as a variable, then
the \#CSP is just the Holant problem on this bipartite graph.
Thus
$\CSP(\mathcal{F}) \equiv_T^p \holant{\mathcal{EQ}}{\mathcal{F}}$,
where $\mathcal{EQ} = \{{=}_1, {=}_2, {=}_3, \dotsc\}$ is the set of \textsc{Equality} signatures of all arities.
Similarly, $\CSP^2(\mathcal{F})$ can be defined as $\holant{\mathcal{EQ}_2}{\mathcal{F}}$, where $\mathcal{EQ}_2=\{=_2,=_4,\ldots\}$ is the set of {\sc Equality} signatures of even arities, i.e., every variable appears an even number of times.

A spin system on $G = (V, E)$ has a variable for every $v \in V$
and a binary function $g$ for every edge $e \in E$.
The partition function is $$\sum_{\sigma: V \rightarrow\{0, 1\}}
\prod_{(u, v) \in E} g(\sigma(u), \sigma(v)).$$
Spin systems are special cases of $\CSP(\mathcal{F})$ where $\mathcal{F}$ consists of a single binary function.
They are also called graph homomorphisms.
We denote by $\GH(g)$ the spin system with the constraint function $g$, and $\BiGH(g)$ by its restriction on bipartite graphs.


\begin{figure}[h!]
\centering
\includegraphics[width=0.3\linewidth]{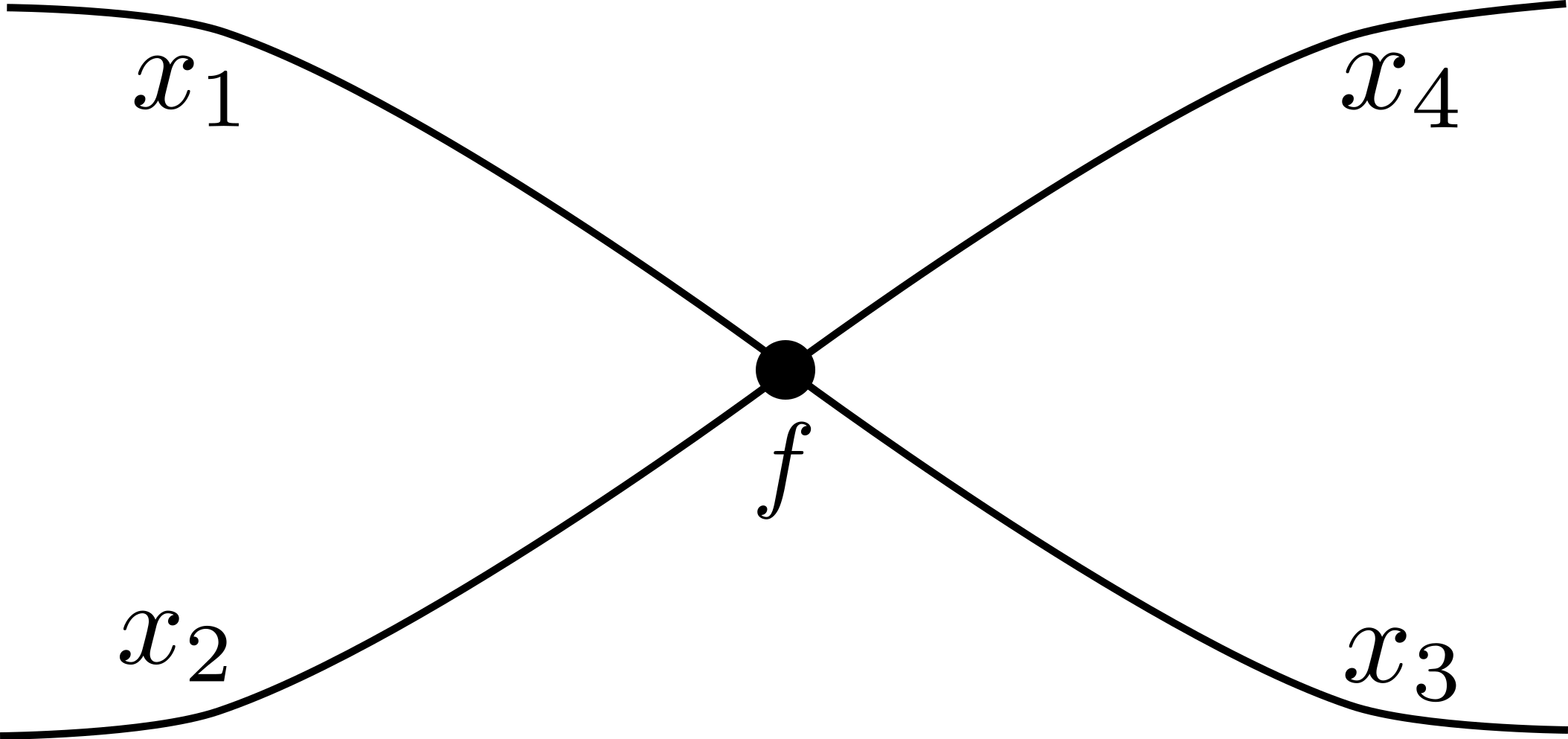}
\caption{A 4-ary signature $f$ with four edge inputs labeled counterclockwise as $x_1, x_2, x_3, x_4$.}
\label{fig:4ary-signature}
\end{figure}

For a 4-ary signature $f$, we consistently label its four edge inputs as $x_1,x_2,x_3,x_4$ in counterclockwise order, as shown in Figure~\ref{fig:4ary-signature}.
A signature $f$ of arity 4 has the signature matrix $M(f)=M_{x_{1} x_{2}, x_{4} x_{3}}(f)=\left[\begin{smallmatrix}f_{0000} & f_{0010} & f_{0001} & f_{0011} \\ f_{0100} & f_{0110} & f_{0101} & f_{0111} \\ f_{1000} & f_{1010} & f_{1001} & f_{1011} \\ f_{1100} & f_{1110} & f_{1101} & f_{1111}\end{smallmatrix}\right]$. 
Notice the order reversal $x_4x_3$. This is for the convenience of
connecting two signatures $f$ and $g$ in a planar fashion: the variables $x_4, x_3$ of $f$ are connected to $x_1, x_2$ of $g$  respectively. If $(i,j,k,\ell)$ is a 
permutation of $(1,2,3,4)$,
then the $4 \times 4$ matrix $M_{x_ix_j, x_\ell x_{k}}(f)$ lists the 16 values
 with row  index $x_ix_j \in\{0, 1\}^2$
and column index
$x_{\ell}x_{k} \in\{0, 1\}^2$ in lexicographic order.
We say that $f$ satisfies \emph{arrow reversal equality}, if $f(x_1,x_2,x_3,x_4)=f(\overline{x_1},\overline{x_2},\overline{x_3},\overline{x_4})$ for any $x_1,x_2,x_3,x_4\in \{0,1\}$, where $\overline{x_i}$ is the complement of $x_i$.
A binary signature $g$ has the signature vector
$(g_{00}, g_{01}, g_{10}, g_{11})^T$ (as a truth table)
and the signature matrix $M(g)=M_{x_1, x_2}(g)=\left[
\begin{smallmatrix}
g_{00} & g_{01}\\
g_{10} & g_{11}\\
\end{smallmatrix}
\right].$
We use $\neq_2$ to denote the binary {\sc Disequality} signature $(0,1, 1, 0)^{T}$; it
has the signature matrix 
$\left[\begin{smallmatrix}
0 & 1 \\
 1 & 0
\end{smallmatrix}\right]$.
Let
$N=
\left[\begin{smallmatrix}
0 & 1 \\
 1 & 0
\end{smallmatrix}\right]
\otimes
\left[\begin{smallmatrix}
0 & 1 \\
 1 & 0
\end{smallmatrix}\right]
=
\left[\begin{smallmatrix}
 0 & 0 & 0 & 1 \\
 0 & 0 & 1 & 0 \\
 0 & 1 & 0 & 0 \\
 1 & 0 & 0 & 0 \\
\end{smallmatrix}\right]$; it is 
the signature matrix of the double {\sc Disequality}, which is the function of connecting two pairs of edges by $\neq_2$. 
The \textbf{planar eight-vertex model} is the problem $\plholant{\neq_2}{f}$,
where $M(f)=
\left[\begin{smallmatrix}
a & 0 & 0 & b \\
0 & c & d & 0 \\
0 & w & z & 0 \\
y & 0 & 0 & x
\end{smallmatrix}\right]$.
The \emph{outer matrix}  of $M(f)$ 
is the submatrix
$
\left[\begin{smallmatrix}
a & b\\
y & x
\end{smallmatrix}\right]$, and is 
 denoted by $M_{\rm{Out}}(f)$.
The \emph{inner matrix} of $M(f)$ is
$
\left[\begin{smallmatrix}
c & d\\
w & z
\end{smallmatrix}\right]$, and is
 denoted by $M_{\rm{In}}(f)$.
The \emph{support} of a signature $f$ is the set of inputs on which $f$ is nonzero.

Let $\mathcal{S}$ denote the \emph{crossover} signature, i.e., the 4-ary signature with the signature matrix $M(\mathcal{S}) = \left[\begin{smallmatrix} 
1 & 0 & 0 & 0 \\ 
0 & 0 & 1 & 0 \\
0 & 1 & 0 & 0 \\ 
0 & 0 & 0 & 1  \end{smallmatrix}\right]$,
which is the indicator function of $(x_1=x_3)\land(x_2=x_4)$.
Getting $\mathcal{S}$ in a signature set $\mathcal{F}$ makes the problems
$\holant{=_2}{\mathcal{F}}$ and $\plholant{=_2}{\mathcal{F}}$
equivalent.
Also, let $\mathcal{S'}$ be the \emph{disequality crossover} signature, i.e., the 4-ary signature with the signature matrix $M(\mathcal{S}') = \left[\begin{smallmatrix} 
0 & 0 & 0 & 1 \\ 
0 & 1 & 0 & 0 \\
0 & 0 & 1 & 0 \\ 
1 & 0 & 0 & 0  \end{smallmatrix}\right]$, which is the indicator function of $(x_1 \ne x_3)\land(x_2 \ne x_4)$.
Note that $M(\mathcal{S})=M(\mathcal{S'})N$.

\begin{lemma}\label{lm: swap gate}
Let $\mathcal{F},\mathcal{G}$ be two arbitrary signature sets, then
$$
\holant{\mathcal{F}}{\mathcal{G}} \leq_T^p \plholant{\mathcal{F},(=_2)^{\otimes2}}{\mathcal{G,S}}.
$$
\end{lemma}
\begin{proof}
Let $\Omega$ be an instance of $\holant{\mathcal{F}}{\mathcal{G}}$.
We place those vertices in $\Omega$ labeled by signatures in $\mathcal{F}$ on the left and those labeled by signatures in $\mathcal{G}$ on the right. 
Without loss of generality, we may assume that no more than two edges in $\Omega$ intersect at one point. 
We then construct an instance
 $\Omega'$ of $\plholant{\mathcal{F},(=_2)^{\otimes2}}{\mathcal{G,S}}$, by replacing every intersection point with a crossover signature $\mathcal{S}$, plus a $(=_2)^{\otimes2}$ connected on the right.
This doesn't change the Holant value and preserves bipartiteness.
\end{proof}

\begin{lemma}\label{lm:disequality crossover}
    Let $\mathcal{F}$ be an arbitrary signature set, then 
    $\holant{\neq_2}{\mathcal{F}}\le_T^p\plholant{\neq_2}{\mathcal{F\cup S'}}$.
\end{lemma}

\begin{proof}
    Let $\Omega$ be an instance of $\holant{\neq_2}{\mathcal{F}}$.
    As in Lemma~\ref{lm: swap gate}, we place those vertices in $\Omega$ labeled by $\neq_2$ on the left and those labeled by signatures in $\mathcal{F}$ on the right, assuming that no more than two edges in $\Omega$ intersect at one point.
    For any point where two edges in $\Omega$ intersect, replace it with a crossover signature $\mathcal{S}$.
    This doesn't change the Holant value.
    Since $M(\mathcal{S})=M(\mathcal{S'})N$, if we further replace every $\mathcal{S}$ with a disequality crossover $\mathcal{S}'$ connected with a $(\neq_2)^{\otimes2}$ on the right, the Holant value remains the same.
    Thus, we obtain an instance $\Omega'$ of $\plholant{\neq_2}{\mathcal{F\cup S'}}$, and $\holant{\Omega;\neq_2}{\mathcal{F}}=\plholant{\Omega';\neq_2}{\mathcal{F\cup S'}}.$
    The lemma follows.
\end{proof}

\begin{remark}
    By Lemma~\ref{lm: swap gate} and Lemma~\ref{lm:disequality crossover}, once we realize $\mathcal{S}$ on RHS and $(=_2)^{\otimes2}$ on LHS, or we realize $\mathcal{S'}$ on RHS, then we can reduce the $\Holant$ problem to the corresponding $\PlHolant$ problem.
    In particular, the \numP-hardness for the $\Holant$ problem applies to the corresponding $\PlHolant$ problem.
\end{remark}

\begin{lemma}\label{lm:let a=x}
Let $f$ be a 4-ary signature with the signature matrix $M(f) = \left[\begin{smallmatrix} a & 0 & 0 & b \\ 0 & c & d & 0 \\ 0 & w & z & 0 \\ y & 0 & 0 &x \end{smallmatrix} \right]$, then 
$$
\plholant{\neq_2}{f} \equiv^p_{T} \plholant{\neq_2}{\Tilde{f}},
$$
where $M(\Tilde{f}) = \left[\begin{smallmatrix} \Tilde{a} & 0 & 0 & b \\ 0 & c & d & 0 \\ 0 & w & z & 0 \\ y & 0 & 0 & \Tilde{a} \end{smallmatrix} \right]$ with $\Tilde{a} = \sqrt{ax}$.
\end{lemma}

\begin{proof}
For any signature grid $\Omega$
with a 4-regular graph $G$ for $\plholant{\neq_2}{f}$, any valid orientation on $G$ must have an equal number of sources and sinks.
Hence the value $\plholant{\Omega;\neq_2}{f}$
as a polynomial in $a$ and $x$ is in fact a polynomial in the product $ax$.
So we can replace $(a,x)$ by any $(\tilde{a}, \tilde{x})$
such that $\tilde{a}\tilde{x} = ax$.
In particular,
let $\tilde{f}$  be a 4-ary signature with the signature matrix
$M(\tilde{f})=\left[\begin{smallmatrix}
\tilde{a} & 0 & 0 & b\\
0 & c & d & 0\\
0 & w & z & 0\\
y & 0 & 0 & \tilde{a}
\end{smallmatrix}\right]$ where $\tilde{a}=\sqrt{ax}$,
then
$\plholant{\Omega;\neq_2}{f}=\plholant{\Omega;\neq_2}{\tilde{f}}$.
Thus, $\plholant{\neq_2}{f}\equiv^p_{T} \plholant{\neq_2}{\tilde{f}}.
$
\end{proof}

\begin{definition}\label{def: redundant}
 A 4-ary signature is non-singular redundant iff in one of its  four $4\times4$ signature matrices,
  the middle two rows are identical and the
 middle two columns are identical, and the determinant
\[ \det \left[\begin{matrix}
f_{0000} & f_{0010} & f_{0011}\\
f_{0100} & f_{0110} & f_{0111}\\
f_{1100} & f_{1110} & f_{1111}
\end{matrix}\right] \not = 0. \]
 \end{definition}
\begin{theorem}\cite{caiguowilliams13}
\label{thm:redundant}
If $f$ is a non-singular redundant signature, 
 then $\operatorname{Pl-Holant}(\neq_2|f)$ is $\SHARPP$-hard. 
\end{theorem}

\subsection{Gadget construction}
One basic notion used throughout the paper is gadget construction.
We say a signature $f$ is \emph{constructible} or \emph{realizable} from a signature set $\mathcal{F}$
if there is a gadget with some dangling edges such that each vertex is assigned a signature from $\mathcal{F}$,
and the resulting graph,
when viewed as a black-box signature with inputs on the dangling edges,
is exactly $f$.
If $f$ is realizable from a set $\mathcal{F}$,
then we can freely add $f$ into $\mathcal{F}$ while preserving the complexity.
\begin{figure}[!htbp]
 \centering
 \begin{tikzpicture}[scale=\scale,transform shape,node distance=\nodeDist,semithick]
  \node[external]  (0)                     {};
  \node[internal]  (1) [below right of=0]  {};
  \node[external]  (2) [below left  of=1]  {};
  \node[internal]  (3) [above       of=1]  {};
  \node[internal]  (4) [right       of=3]  {};
  \node[internal]  (5) [below       of=4]  {};
  \node[internal]  (6) [below right of=5]  {};
  \node[internal]  (7) [right       of=6]  {};
  \node[internal]  (8) [below       of=6]  {};
  \node[internal]  (9) [below       of=8]  {};
  \node[internal] (10) [right       of=9]  {};
  \node[internal] (11) [above right of=6]  {};
  \node[internal] (12) [below left  of=8]  {};
  \node[internal] (13) [left        of=8]  {};
  \node[internal] (14) [below left  of=13] {};
  \node[external] (15) [left        of=14] {};
  \node[internal] (16) [below left  of=5]  {};
  \path let
         \p1 = (15),
         \p2 = (0)
        in
         node[external] (17) at (\x1, \y2) {};
  \path let
         \p1 = (15),
         \p2 = (2)
        in
         node[external] (18) at (\x1, \y2) {};
  \node[external] (19) [right of=7]  {};
  \node[external] (20) [right of=10] {};
  \path (1) edge                             (5)
            edge[bend left]                 (11)
            edge[bend right]                (13)
            edge node[near start] (e1) {}   (17)
            edge node[near start] (e2) {}   (18)
        (3) edge                             (4)
        (4) edge[out=-45,in=45]              (8)
        (5) edge[bend right, looseness=0.5] (13)
            edge[bend right, looseness=0.5]  (6)
        (6) edge[bend left]                  (8)
            edge[bend left]                  (7)
            edge[bend left]                 (14)
        (7) edge node[near start] (e3) {}   (19)
       (10) edge[bend right, looseness=0.5] (12)
            edge[bend left,  looseness=0.5] (11)
            edge node[near start] (e4) {}   (20)
       (12) edge[bend left]                 (16)
       (14) edge node[near start] (e5) {}   (15)
            edge[bend right]                (12)
       (16) edge[bend left,  looseness=0.5]  (9)
            edge[bend right, looseness=0.5]  (3);
  \begin{pgfonlayer}{background}
   \node[draw=\borderColor,thick,rounded corners,fit = (3) (4) (9) (e1) (e2) (e3) (e4) (e5),inner sep=0pt,transform shape=false] {};
  \end{pgfonlayer}
 \end{tikzpicture}
 \caption{An $\mathcal{F}$-gate with 5 dangling edges.}
 \label{fig:Fgate}
\end{figure}
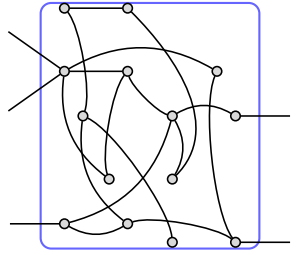

Formally,
this notion is defined by an $\mathcal{F}$-gate.
An $\mathcal{F}$-gate is similar to a signature grid $(G, \pi)$ for $\Holant(\mathcal{F})$ except that $G = (V,E,D)$ is a graph with some dangling edges $D$.
The dangling edges define external variables for the $\mathcal{F}$-gate
(see Figure~\ref{fig:Fgate} for an example).
We denote the regular edges in $E$ by $1, 2, \dotsc, m$ and the dangling edges in $D$ by $m+1, \dotsc, m+n$.
Then we can define a function $f$ for this $\mathcal{F}$-gate as
\[
f(y_1, \dotsc, y_n) = \sum_{x_1, \dotsc, x_m \in \{0, 1\}} H(x_1, \dotsc, x_m, y_1, \dotsc, y_n),
\]
where $(y_1, \dotsc, y_n) \in \{0, 1\}^n$ is an assignment on the dangling edges
and $H(x_1, \dotsc, x_m, y_1, \dotsc, y_n)$ is the value of the signature grid on an assignment of all edges in $G$,
which is the product of evaluations at all vertices in $V$.
We also call this function $f$ the signature of the $\mathcal{F}$-gate.

An $\mathcal{F}$-gate is planar if the underlying graph $G$ is a planar graph,
and the dangling edges,
ordered counterclockwise corresponding to the order of the input variables,
are in the outer face in a planar embedding.
A planar $\mathcal{F}$-gate can be used in a planar signature grid as if it is just a single vertex with the particular signature.

Using  planar $\mathcal{F}$-gates,
we can reduce one planar Holant problem to another.
Suppose $g$ is the signature of some planar $\mathcal{F}$-gate.
Then $\PlHolant(\mathcal{F}, g) \leq_T \PlHolant(\mathcal{F})$.
The reduction is simple.
Given an instance of $\PlHolant(\mathcal{F}, g)$,
by replacing every occurrence of $g$ by the $\mathcal{F}$-gate,
we get an instance of $\PlHolant(\mathcal{F})$.
Since the signature of the $\mathcal{F}$-gate is $g$,
the Holant values for these two signature grids are identical.

The simplest gadget construction is just \textbf{rotation}, i.e., cyclically renaming the variables of $f$.
The signature $f$ with matrix
$M(f)=M_{x_{1} x_{2}, x_{4} x_{3}}(f)=\left[\begin{smallmatrix}
a & 0 & 0 & b \\
0 & c & d & 0 \\
0 & w & z & 0 \\
y & 0 & 0 & x
\end{smallmatrix}\right]$
has other three rotated forms $f^{\frac{\pi}{2}},f^{\pi},f^{\frac{3\pi}{2}}$,
having signature matrices 
$M(f^{\frac{\pi}{2}}):=M_{x_2x_3, x_1x_4}(f) = \left[\begin{smallmatrix} a & 0 & 0 & z \\ 0 & b & w & 0 \\ 0 & d & y & 0 \\ c & 0 & 0 & x  \end{smallmatrix}\right]$, 
$M(f^{{\pi}}):=M_{x_3x_4, x_2x_1}(f) = \left[\begin{smallmatrix} a & 0 & 0 & y \\ 0 & z & d & 0 \\ 0 & w & c & 0 \\ b & 0 & 0 & x  \end{smallmatrix}\right]$, 
and $M(f^{\frac{3\pi}{2}}):=M_{x_4x_1, x_3x_2}(f) = \left[\begin{smallmatrix} a & 0 & 0 & c \\ 0 & y & w & 0 \\ 0 & d & b & 0 \\ z & 0 & 0 & x  \end{smallmatrix}\right]$, respectively.
One we have one of them, from it we can realize the other three.
In the proof, we denote this fact by \emph{rotational symmetry}.
The movement of signature entries under rotation by $\frac{\pi}{2}$ is illustrated
in Figure~\ref{fig:rotate_asymmetric_signature}.
Note that no matter in which signature matrix, the pair $(d, w)$ (and only
$(d, w)$) is always in the inner matrix. 
We call $(d, w)$ the inner pair, and $(b, y)$, $(c, z)$ 
the outer pairs. 

\begin{figure}[h]
 \centering
 \def\capWidth{6cm}
 \captionsetup[subfigure]{width=\capWidth}
 \tikzstyle{entry} = [internal, inner sep=2pt]
   \begin{tikzpicture}[scale=\scale,transform shape,>=stealth,node distance=\nodeDist,semithick]
    \node[entry] (11)               {};
    \node[entry] (12) [right of=11] {};
    \node[entry] (13) [right of=12] {};
    \node[entry] (14) [right of=13] {};
    \node[entry] (21) [below of=11] {};
    \node[entry] (22) [right of=21] {};
    \node[entry] (23) [right of=22] {};
    \node[entry] (24) [right of=23] {};
    \node[entry] (31) [below of=21] {};
    \node[entry] (32) [right of=31] {};
    \node[entry] (33) [right of=32] {};
    \node[entry] (34) [right of=33] {};
    \node[entry] (41) [below of=31] {};
    \node[entry] (42) [right of=41] {};
    \node[entry] (43) [right of=42] {};
    \node[entry] (44) [right of=43] {};
    \node[external] (nw) [above left  of=11] {};
    \node[external] (ne) [above right of=14] {};
    \node[external] (sw) [below left  of=41] {};
    \node[external] (se) [below right of=44] {};
    \path (13) edge[->, dotted]                (12)
          (12) edge[->, dotted]                (21)
          (21) edge[->, dotted]                (31)
          (31) edge[->, dotted,out=65,in=-155] (13)
          (42) edge[->, dashed]                (43)
          (43) edge[->, dashed]                (34)
          (34) edge[->, dashed]                (24)
          (24) edge[->, dashed,out=-115,in=25] (42)
          (14) edge[->, very thick]            (22)
          (22) edge[->, very thick]            (41)
          (41) edge[->, very thick]            (33)
          (33) edge[->, very thick]            (14)
          (23) edge[<->]                      (32);
    \path (nw.west) edge (sw.west)
          (ne.east) edge (se.east)
          (nw.west) edge (nw.east)
          (sw.west) edge (sw.east)
          (ne.west) edge (ne.east)
          (se.west) edge (se.east);
   \end{tikzpicture}
 \caption{The movement of the entries in the signature matrix of an arity~4 signature under a clockwise rotation of the input edges.
  The Hamming weight two entries are in the two solid cycles (one has length 4 and the other one is a swap).}
 \label{fig:rotate_asymmetric_signature}
\end{figure}
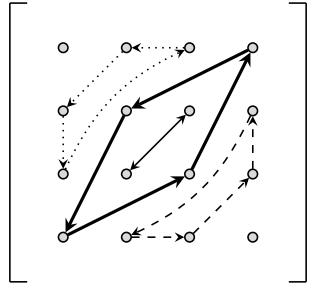

The most commonly used gadget in this paper is \textbf{connection} between two 4-ary signatures via $(\neq_2)^{\otimes2}$.
Suppose $f_1$ and $f_2$ have signature matrices
 $ M_{x_ix_j, x_\ell x_{k}}(f_1)$ and $
M_{x_{s}x_{t}, x_{v}x_{u}}(f_2)$, where $(i, j, k, \ell)$ and $(s, t, u, v)$ are rotations of $(1, 2, 3, 4)$.
By connecting $x_\ell$ with $x_{s}$, $x_{k}$ with $x_t$,
both using  {\sc Disequality}  $(\not =_2)$, we get a signature $f_3$
of arity 4 with the signature matrix
$$M(f_3)=M_{x_ix_j, x_\ell x_{k}}(f_1) N M_{x_{s}x_{t}, x_{v}x_{u}}(f_2)$$
by matrix product
with row index $x_ix_j$ and column index $x_{v}x_{u}$ (see Figure \ref{111}).

\begin{figure}[!htbp]
\centering
                \includegraphics[height=1.4in]{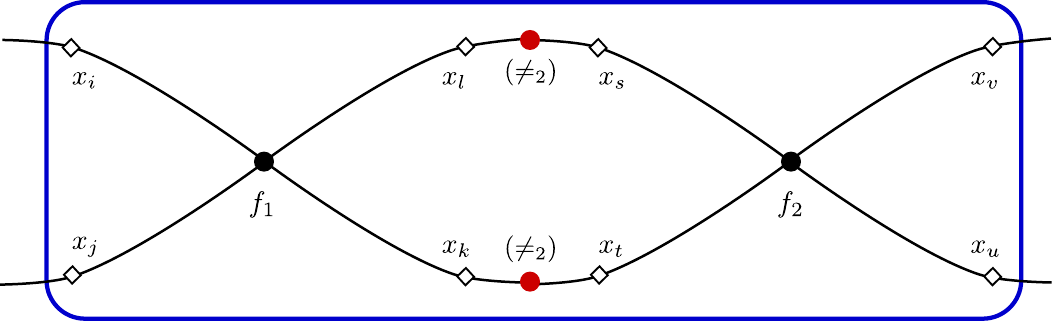}
        \caption{Connect variables $x_\ell$, $x_k$ of $f_1$ with variables $x_s$, $x_t$ of $f_2$ both using $\neq_2$.}
        \label{111}
        \end{figure}
Since this connection preserves planarity and bipartiteness,
for any signature set 
$\mathcal{F}$ containing $f_1$ and $f_2$, we have $\plholant{\neq_2}{\mathcal{F}\cup f_3}\le_T^p\plholant{\neq_2}{\mathcal{F}}$.
In the following, when we say connecting signatures $f$ and $g$ via $(\neq_2)^{\otimes2}$, we mean the above gadget where variables $x_4,x_3$ of $f$ are connected with variables $x_1,x_2$ of $g$ using $\neq_2$, respectively.

Another commonly used gadget is \textbf{looping} by a binary signature. 
Let $g$ be a binary signature with the signature vector $g(x_1, x_2)=(g_{00}, g_{01}, g_{10}, g_{11})^T$, 
and also $g(x_2, x_1)=(g_{00}, g_{10}, g_{01}, g_{11})^T$.  
Without other specification, $g$  denotes $g(x_1, x_2)$. 
Let $f$  be a signature of arity $4$ with the signature matrix
 $M_{x_ix_j, x_\ell x_k}(f)$ and $(s, t)$ be a permutation of $(1, 2)$. 
By connecting $x_\ell$ with $x_s$ and $x_{k}$ with $x_t$, both using {\sc Disequality} $(\neq_{2})$, 
 we get a binary signature with the signature matrix $M_{x_i x_j, x_k x_{\ell}}Ng{(x_s, x_t)}$ as a matrix product with index $x_i x_j$ (see Figure \ref{222}). 
 \begin{figure}[!htbp]
 \centering
		\includegraphics[height=1.4in]{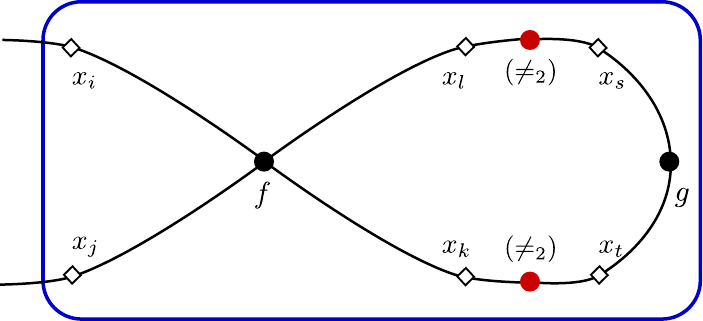}
	\caption{Connect variables $x_\ell$, $x_k$ of $f$ with variables $x_s$, $x_t$ of $g$ both using $\neq_2$.}\label{222}
	\end{figure}
 If $g_{00}=g_{11}$, 
 then $N(g_{00}, g_{01}, g_{10}, g_{11})^T=(g_{11}, g_{10}, g_{01}, g_{00})^T=(g_{00}, g_{10}, g_{01}, g_{11})^T$, and similarly, $N(g_{00}, g_{10}, g_{01}, g_{11})^T=(g_{00}, g_{01}, g_{10}, g_{11})^T$. 
 Therefore, $M_{x_i x_j, x_\ell x_k}Ng{(x_s, x_t)}=M_{x_i x_j, x_\ell x_{k}}g{(x_t, x_s)}$, 
 which means that
 connecting variables $x_\ell$, $x_k$ of $f$ with, respectively,
 variables $x_s$, $x_t$ of $g$
 using $N$ is equivalent to connecting them directly without $N$. 
Hence, in the setting
Pl-Holant$( \not =_2 \mid f, g)$ we can form
$M_{x_i x_j, x_\ell x_{k}}(f)g{(x_t, x_s)}$,
which is technically $M_{x_i x_j, x_\ell x_{k}}Ng{(x_s, x_t)}$,
provided that $g_{00}=g_{11}$.
Note that for a  binary signature $g$,
we can rotate it by $\pi$ without violating
planarity, and so  both $g(x_s, x_t)$ and $g(x_t, x_s)$ can be freely used once we get one of them.

The last kind of gadget construction we want to emphasize is \textbf{binary modification}.
A binary modification to the variable $x_i$ of $f(x_i,x_j,x_k,x_l)$ using the binary signature $g(x_s,x_t)$ means connecting the variable $x_i$ of $f$ to the variable $x_t$ of $g$ by {\sc Disequality} $(\neq_2)$ (see Figure~\ref{figure: binary_modification}).

 \begin{figure}[!htbp]
 \centering
		\includegraphics[height=1.4in]{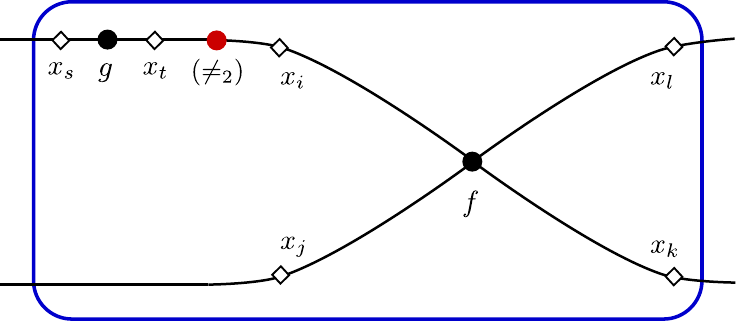}
	\caption{Connect variable $x_t$ of $g$ with variable $x_i$ of $f$ using $\neq_2$.}\label{figure: binary_modification}
	\end{figure}
Note that if $g(x_1,x_2) = (0,1,t,0)^T$ is a weighted binary \textsc{Disequality}; a modification to $x_1$ of $f$ using $g$ amounts to multiplying $t$ to every entry of $f$ where the index $x_i=1$. 
Thus, for $M(f) = \left[\begin{smallmatrix} a & 0 & 0 & b \\ 0 & c & d & 0 \\ 0 & w & z & 0 \\ y & 0 & 0 &x \end{smallmatrix} \right] = \left[\begin{smallmatrix}f_{0000} & 0 & 0 & f_{0011} \\ 0 & f_{0110} & f_{0101} & 0 \\ 0 & f_{1010} & f_{1001} & 0 \\ f_{1100} & 0 & 0 & f_{1111}\end{smallmatrix}\right]$, binary modification to the variable $x_1$ of $f$ using the binary signature $(0,1,t,0)^T$ we get a signature $f'$ with the signature matrix $M(f') = \left[\begin{smallmatrix}f_{0000} & 0 & 0 & f_{0011} \\ 0 & f_{0110} & f_{0101} & 0 \\ 0 & tf_{1010} & tf_{1001} & 0 \\ tf_{1100} & 0 & 0 & tf_{1111}\end{smallmatrix}\right]$. Similarly, modifications to $x_2, x_3, x_4$ of $f$ using $g$ give the following signatures respectively, 

$$ 
\left[\begin{smallmatrix}f_{0000} & 0 & 0 & f_{0011} \\ 0 & tf_{0110} & tf_{0101} & 0 \\ 0 & f_{1010} & f_{1001} & 0 \\ tf_{1100} & 0 & 0 & tf_{1111}\end{smallmatrix}\right];
\left[\begin{smallmatrix}f_{0000} & 0 & 0 & tf_{0011} \\ 0 & tf_{0110} & f_{0101} & 0 \\ 0 & tf_{1010} & f_{1001} & 0 \\ f_{1100} & 0 & 0 & tf_{1111}\end{smallmatrix}\right];
\left[\begin{smallmatrix}f_{0000} & 0 & 0 & tf_{0011} \\ 0 & f_{0110} & tf_{0101} & 0 \\ 0 & f_{1010} & tf_{1001} & 0 \\ f_{1100} & 0 & 0 & tf_{1111}\end{smallmatrix}\right].
$$

Binary modification is widely used in our proofs throughout this section and subsequent sections, especially in Lemma~\ref{lm: outer full inner degenerate}, Lemma~\ref{lm: with-any-binary}, Theorem~\ref{thm: with (0,1,0,0)} and many lemmas in Section~\ref{sec: At least a zero entry and at most a zero pair}.

\begin{lemma}\label{lm:inverse binary}
Let $g=(0,1,t,0)^T$ be a binary signature, where $t\neq 0$. Then for any signature set $\mathcal{F}$, we have 
$$
\plholant{\neq_2}{\mathcal{F}, g, (0,1,t^{-1},0)^T} \leq_T^p\plholant{\neq_2}{\mathcal{F},g}
$$
\end{lemma}

\begin{proof}
    Rotate $g$ by $\pi$ we get the signature $(0,t,1,0)^T$.
    Dividing $g$ by $t$ doesn't affect complexity.
    The lemma follows.
\end{proof}

\begin{lemma}\label{lm: outer full inner degenerate}
Let $f$ be a 4-ary signature with the signature matrix $M(f) = \left[\begin{smallmatrix} a & 0 & 0 & b \\ 0 & c & d & 0 \\ 0 & w & z & 0 \\ y & 0 & 0 &a \end{smallmatrix} \right]$. where $czdw \neq 0$. If $f$ satisfies the following conditions,
\begin{itemize}
    \item 
    $M_{\rm{Out}}(f)=\left[\begin{smallmatrix} a & b \\ y & a  \end{smallmatrix}\right]$ has full rank;
    \item $M_{\rm{In}}(f)=\left[\begin{smallmatrix} c & d \\ w & z  \end{smallmatrix}\right]$ is degenerate;
    \item $c+d \neq 0$ or $z+d \neq 0$;
\end{itemize}
then \plholant{\neq_2}{f} is \numP-hard.
\end{lemma}

\begin{proof}
We prove the lemma for $c+d \neq 0$. The case $z+d \neq 0$ can be obtained from $M(f^\pi)=M_{x_3x_4, x_2x_1}(f) = \left[\begin{smallmatrix} a & 0 & 0 & y \\ 0 & z & d & 0 \\ 0 & w & c & 0 \\ b & 0 & 0 &a \end{smallmatrix} \right]$.

Since $\left[\begin{smallmatrix} c & d \\ w & z  \end{smallmatrix}\right]$ is degenerate and $czdw \neq 0$, there exists $k\neq 0$ such that $\left[\begin{smallmatrix} c & d \\ w & z  \end{smallmatrix}\right] = \left[\begin{smallmatrix} c & d \\ ck & dk  \end{smallmatrix}\right]$.

By connecting the variables $x_3$ and $x_4$ of $f$ using $\neq_2$, we get a binary signature $(c+d)(0,1,k,0)^T$. By Lemma~\ref{lm:inverse binary}, we have $(0,1,k^{-1},0)^T$. 
By doing a binary modification to the variable $x_1$ of $f$ using $(0,1,k^{-1},0)^T$, we get a signature $f_1$ with the signature matrix $M(f_1) = \left[\begin{smallmatrix} a & 0 & 0 & b \\ 0 & c & d & 0 \\ 0 & c & d & 0 \\ \frac{y}{k} & 0 & 0 & \frac{a}{k} \end{smallmatrix} \right]$. Then by connecting the variables $x_1$ and $x_2$ of $f_1$ using $\neq_2$, we get the binary signature $2c(0,1,\frac{d}{c},0)^T$. By doing a binary modification to the variable $x_3$ of $f_1$ using $(0,1,\frac{d}{c},0)^T$, we get $f_2$ with the signature matrix $M(f_2) = \left[\begin{smallmatrix} a & 0 & 0 & \frac{bd}{c} \\ 0 & d & d & 0 \\ 0 & d & d & 0 \\ \frac{y}{k} & 0 & 0 & \frac{ad}{kc} \end{smallmatrix} \right]$, a signature in redundant form. Since $\det \left[\begin{smallmatrix} a & b \\ y & a  \end{smallmatrix}\right] \neq 0$, the compressed signature matrix of $f_2$ has full rank. Thus \plholant{\neq_2}{f_2} is \numP-hard by Theorem~\ref{thm:redundant}. 
So \plholant{\neq_2}{f} is \numP-hard.
\end{proof}

The above lemma was applied in Lemma~\ref{lm: degenerate inner matrices} to obtain \numP-hardness.

\subsection{Holographic transformation}\label{sec: holographic transformation}

Holographic transformation is used throughout this paper as a fundamental technique for complexity analysis and algorithm design. 
To introduce the idea of holographic transformation,
it is convenient to consider bipartite graphs.
For a general graph,
we can always transform it into a bipartite graph while preserving the Holant value,
as follows.
For each edge in the graph,
we replace it by a path of length two.
(This operation is called the \emph{2-stretch} of the graph and yields the edge-vertex incidence graph.)
Each new vertex is assigned the binary \textsc{Equality} signature $(=_2) = [1,0,1]$.

For an invertible $2$-by-$2$ matrix $T \in {\rm GL}_2({\mathbb{C}})$
 and a signature $f$ of arity $n$, written as
a column vector (contravariant tensor) $f \in \mathbb{C}^{2^n}$, we denote by
$T^{-1}f = (T^{-1})^{\otimes n} f$ the transformed signature.
  For a signature set $\mathcal{F}$,
define $T^{-1} \mathcal{F} = \{T^{-1}f \mid  f \in \mathcal{F}\}$ the set of
transformed signatures.
For signatures written as
 row vectors (covariant tensors) we define
$f T$ and  $\mathcal{F} T$ similarly.
Whenever we write $T^{-1} f$ or $T^{-1} \mathcal{F}$,
we view the signatures as column vectors;
similarly for $f T$ or $\mathcal{F} T$ as row vectors.

The holographic transformation defined by $T$ is the following operation:
given a signature grid $\Omega = (H, \pi)$ of $\holant{\mathcal{F}}{\mathcal{G}}$,
for the same bipartite graph $H$,
we get a new signature grid $\Omega' = (H, \pi')$ of $\holant{\mathcal{F} T}{T^{-1} \mathcal{G}}$ by replacing each signature in
$\mathcal{F}$ or $\mathcal{G}$ with the corresponding signature in $\mathcal{F} T$ or $T^{-1} \mathcal{G}$.

\begin{theorem}[Valiant's Holant Theorem~\cite{Val08}]\label{thm: holographic transformation}
 For any $T \in {\rm GL}_2({\mathbb{C}})$,
  \[\Holant(\Omega; \mathcal{F} \mid \mathcal{G}) = \Holant(\Omega'; \mathcal{F} T \mid T^{-1} \mathcal{G}).\]
\end{theorem}

Therefore,
a holographic transformation does not change the complexity of the Holant problem in the bipartite setting. This theorem also holds for planar instances.

\paragraph{Transformable}

\begin{definition} \label{def: prelim: pair transformable}
 We say a pair of signature set $(\mathcal{F},\mathcal{G})$ is $\mathscr{C}$-transformable
 if there exists a $T \in \rm{GL}_2(\mathbb{C})$ such that
 $\mathcal{F} T \in \mathscr{C}$ and $T^{-1}\mathcal{G} \subseteq \mathscr{C}$.
\end{definition}

\begin{definition} \label{def:prelim:trans}
 We say a signature set $\mathcal{F}$ is $\mathscr{C}$-transformable if the pair $(\neq_2,\mathcal{F})$ is $\mathscr{C}$-transformable.
\end{definition}

Definition~\ref{def:prelim:trans} is important because if $\PlHolant(\mathscr{C})$ is tractable,
then $\PlHolant(\neq_2 \mid\mathcal{F})$ is tractable for any $\mathscr{C}$-transformable set $\mathcal{F}$.

\subsection{Polynomial interpolation}
Polynomial interpolation is
a powerful reduction technique pioneered in \cite{Valiant79} to show \numP-completeness of various problems, and is widely used in our proof.
In Section~\ref{subsubsec: regular interpolation}, we give a simple illustration of this method.
In Section~\ref{subsuebsec: mobius interpolation}, we briefly introduce the use of M\"{o}bius transformation in complex analysis for interpolation.
In Section~\ref{subsubsec: conformal lattice interpolation}, we introduce conformal lattice interpolation, which captures the essence of interpolation technique  
that we will employ in this paper under many guises.

\subsubsection{Regular interpolation}\label{subsubsec: regular interpolation}
\begin{lemma}\label{lm: interpolate S for N=2}
    Let $f$ be a 4-ary signature with the signature matrix
    $M(f) = \left[\begin{smallmatrix}
        1 & 0 & 0 & k\\
        0 & k & 1 & 0\\
        0 & 1 & 0 & 0\\
        0 & 0 & 0 & 1
    \end{smallmatrix}\right]$ with $k\neq 0$.
    Then $$\plholant{\neq_2}{f,\mathcal{S}}\le_T^p \plholant{\neq_2}{f}.$$
\end{lemma}

\begin{figure}[!htbp]
\centering
\includegraphics[height=1.0in]{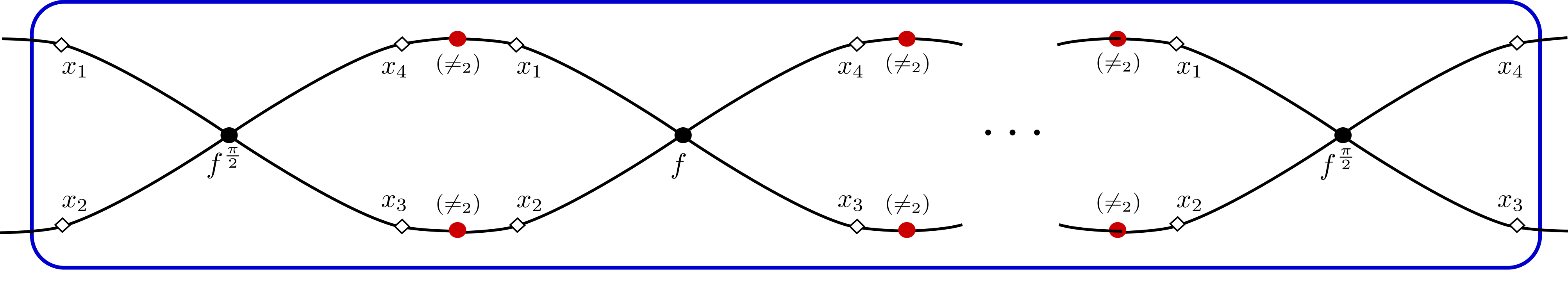}
\caption{A chain of $2s+1$ copies of signature $f$, where odd-indexed copies are rotated by $\frac{\pi}{2}$, connected by double {\sc Disequality} $N$}
\label{2s+1}
\end{figure}

\begin{proof}
     Rotate $f$ by $\frac{\pi}{2}$ we get $f^{\frac{\pi}{2}}$ with the signature matrix $M(f^\frac{\pi}{2}) = \left[\begin{smallmatrix}
        1 & 0 & 0 & 0\\
        0 & k & 1 & 0\\
        0 & 1 & 0 & 0\\
        k & 0 & 0 & 1
    \end{smallmatrix}\right]$.
     We construct a series of gadgets ${f}_{2s+1}$ by a chain of $2s+1$ many copies of $f$, with odd-indexed copies rotated by $\frac{\pi}{2}$, linked by the double {\sc Disequality} $N$ (see Figure~\ref{2s+1}).
     Clearly $f_{2s+1}$ has the signature matrix 
     $$
     M(f_{2s+1})=(M(f^\frac{\pi}{2})NM(f)N)^sM(f^\frac{\pi}{2}) = \left[\begin{matrix}
        1 & 0 & 0 & 0\\
        0 & (2s+1)k & 1 & 0\\
        0 & 1 & 0 & 0\\
        (2s+1)k & 0 & 0 & 1
    \end{matrix}\right].$$ 
    Given any signature grid $\Omega$ of $\plholant{\neq_2}{f,\mathcal{S}}$ where $\mathcal{S}$ appears $m$ times, we construct the signature grid $\Omega_s$ of $\plholant{\neq_2}{f}$ by replacing every $\mathcal{S}$ with a copy of the gadget $f_{2s+1}$.
    We divide $\Omega_{s}$ into two parts.
    One part consists of $m$ signatures $f_{2s+1}$ and its signature is represented by $(f_{2s+1})^{\otimes{m}}$.
    Here we rewrite $(f_{2s+1})^{\otimes{m}}$ as a column vector.
    The other part is the rest of $\Omega_{s}$, and its signature is represented by $A$ which is a tensor expressed as a row vector. 
    Then, the Holant value of $\Omega_{s}$ is the dot product $\langle A, (f_{2s+1})^{\otimes{m}}\rangle$, which is a summation over $4m$ bits. 
    That is, a sum over all $0, 1$ values for  the $4m$ edges connecting the two parts. 
    We can stratify all $0, 1$ assignments of these $4m$ bits having a nonzero evaluation of a term in Pl-Holant$_{\Omega_{s}}$ into the following categories:
    \begin{itemize}
    \item
    There are $i$ many copies of $f_{2s+1}$ receiving inputs $0000, 0101, 1010$, or $1111$;
    
    \item
    There are $j$ many copies of $f_{2s+1}$ receiving inputs $0110$ or $1100$;
    \end{itemize}
    where $i+j= m$.

    For any assignment in the category with parameter $(i,j)$, the evaluation of $f_{2s+1}$ is $((2s+1)k)^j$.
    Let $a_{j}$ be the summation of values of the part $A$ over all assignments in the category $(i,j)$.
    Note that $c_{j}$ is independent from the value of $s$ since we view the gadget $f_{2s+1}$ as a block. 
    Then, we rewrite the dot product summation and get 
    \begin{equation}\label{eq: linear system}
    \PlHolant_{\Omega_s} =\langle A, (f_{2s+1})^{\otimes{m}}\rangle
    = \sum\limits_{0 \leq j \leq m} ((2s+1)k)^ja_j
    \end{equation}
    Under this stratification, we have $\PlHolant_{\Omega}= \PlHolant(\Omega; \neq_2 \mid f, \mathcal{S})=a_0$.
    Viewing $a_j$'s as variables, we can then solve the system of linear equations~\ref{eq: linear system} given $\PlHolant_{\Omega_s}$ for $0\le s\le m$, since the coefficient matrix is a Vandermonde matrix of full rank ($k\neq0$). 
    Thus, we can compute $\PlHolant_{\Omega}$ in polynomial time.
    Therefore, $\plholant{\neq_2}{f,\mathcal{S}}\le_T^p\plholant{\neq_2}{f}$.
\end{proof}

The above lemma will be used in Lemma~\ref{lm: 1bcdw1} where there are two zero entries in $M(f)$.
Using similar argument, we can prove the following lemma which will be used in Lemma~\ref{lm:case d=w} where there are two zero pairs in $M(f)$.

\begin{lemma}\label{lm:interpolate S}
Let $f$ be a 4-ary signature with the signature matrix $M(f) = \left[\begin{smallmatrix}1 & 0 & 0 & 0 \\ 0 & 0 & d & 0 \\ 0 & d & 0 & 0 \\ 0 & 0 & 0 & 1 \end{smallmatrix} \right]$, where $d\neq 0$ is not a root of unity. Then for any signature $\mathcal{F}$ containing $f$, we have
$$
\plholant{\neq_2}{\mathcal{F} \cup \mathcal{S}} \leq_T^p \plholant{\neq_2}{\mathcal{F}}.
$$
\end{lemma}
\begin{proof}
We interpolate the crossover $\mathcal{S}$ from $f$. 
Given any signature grid $\Omega$ of $\plholant{\neq_2}{\mathcal{F} \cup \mathcal{S}}$ where $\mathcal{S}$ appears $n$ times, we construct signature grids $\Omega_s$ of $\plholant{\neq_2}{\mathcal{F}}$, $0 \leq s \leq n$, by replacing every occurrence of $\mathcal{S}$ with the gadget $f_s$, which consists of $2s+1$ copies of $f$ and $2s$ copies $N$ linked up in a path, whose signature matrix is
$$M(f_s) = M(f)(NM(f))^{2s} = \left[\begin{smallmatrix}
    1 & 0 & 0 & 0 \\
    0 & 0 & d^{2s+1} & 0 \\
    0 & d^{2s+1} & 0 & 0 \\
    0 & 0 & 0 & 1
\end{smallmatrix}\right].$$


If we stratify the assignments in $\Omega$ according to
the number of  occurrences of $\mathcal{S}$  assigned
0000 or 1111 being  exactly $i$,  and the number of occurrences of $\mathcal{S}$  assigned
0101 or 1010 being  exactly  $j$, then we have
$$
\Holant(\Omega) = \sum\limits_{i+j=n} a_j,
$$
where $a_j$ includes all contributions from $\mathcal{F}  \cup \mathcal{S}$ other
than $\mathcal{S}$.
Similarly,
\begin{equation}\label{eq: Vandermonde system for S}
   \Holant(\Omega_s) = \sum\limits_{i+j=n} (d^j)^{2s+1}a_j .
\end{equation}

We view $a_j$ as variables and Equation~(\ref{eq: Vandermonde system for S}) as a linear system. Since $d\neq 0$ is not a root of unity, this  is a Vandermonde system of full rank and thus we can solve for $a_j$ and calculate $\Holant(\Omega)$ in polynomial time. This proves $
\plholant{\neq_2}{\mathcal{F} \cup \mathcal{S}} \leq_T^p \plholant{\neq_2}{\mathcal{F}}.$
\end{proof}

\subsubsection{Interpolation via M\"{o}bius transformation}\label{subsuebsec: mobius interpolation}
A M\"{o}bius transformation~\cite{ahlfors} is
a mapping of the form ${\mathfrak z}
\mapsto \frac{a\mathfrak z+b}{c\mathfrak z+d}$,
where $\det \left[\begin{smallmatrix}
a & b \\
c & d
\end{smallmatrix}\right] \not =0$.
It is a bijective map of the extended complex plane $\widehat{\mathbb{C}}
= \mathbb{C} \cup \{ \infty \}$ to itself.
The projective order of a M\"obius transformation is defined as the projective order of the matrix 
$\left[\begin{smallmatrix}
a & b \\
c & d
\end{smallmatrix}\right] \in \mathrm{PGL_2}(\mathbb{C})$ that represents the transformation.
A M\"{o}bius transformation maps the unit circle $S^1=\{\mathfrak{z}
\in \mathbb{C} \mid |\mathfrak{z}|=1\}$ to itself iff it is of the form
$\varphi(\mathfrak z)=e^{i\theta}\frac{\mathfrak{z}+\lambda}{1+\overline{\lambda} \mathfrak{z}}$ denoted by $\mathcal{M}(\lambda, e^{i\theta})$, where $|\lambda|\neq 1$. 
A M\"{o}bius transformation is determined by its values on any 3
distinct points. 

The technique of using M\"{o}bius transformation to do polynomial interpolation was developed in the complexity classification of the (non-planar) eight-vertex model \cite{CaiF17} and the planar six-vertex model \cite{CaiFS21}.
The idea is to use M\"obius transformation (of infinite projective order) to generate polynomially many binary signatures $g_i=(0,1,t_i,0)^T$ as ``raw materials" for interpolation.

We generalize the corresponding results in \cite{CaiF17} and \cite{CaiFS21} to the planar eight-vertex model setting (Lemma~\ref{lm: get any binary from a binary not fifth root}, Corollaries~\ref{cor: get (0,1,0,0) from a binary third or fourth root},~\ref{cor:get any binary when exactly three nonzero}).
This generalization is nontrivial. 
In particular, to handle a M\"obius transformation of projective order 2 in Lemma~\ref{lm: get any binary from a binary not fifth root}, we use limit analysis and properties of cyclotomic fields.

These results (Lemma~\ref{lm: get any binary from a binary not fifth root}, Corollaries~\ref{cor: get (0,1,0,0) from a binary third or fourth root},~\ref{cor:get any binary when exactly three nonzero}) show that under certain conditions, we can realize arbitrary or specific binary signatures, allowing us to freely perform binary modifications. 
Corollary~\ref{cor:get any binary when exactly three nonzero} is first applied in Lemma~\ref{lm: 1bcdwy1}. Lemma~\ref{lm: get any binary from a binary not fifth root} and Corollary~\ref{cor: get (0,1,0,0) from a binary third or fourth root} are used later in Lemmas~\ref{lm: with-any-binary} and~\ref{thm: with (0,1,0,0)}.

\subsubsection{Conformal lattice interpolation}\label{subsubsec: conformal lattice interpolation}
\begin{definition}[Lattice]
    Let $k\in \Z_+$ and $\mathbf{x}=(x_1,x_2,\ldots,x_k)\in \C^k$.
    Define the lattice corresponding to $\mathbf{x}$ as $L_{\mathbf{x}}:=\{(j_1,j_2,\dots,j_k) \in \mathbb{Z}^k \mid x_1^{j_1}x_2^{j_2}\cdots x_k^{j_k} = 1\}$.
    We assume $0^0=1$ by convention.
\end{definition}

Note that $L_{\mathbf{x}}$ is a subgroup of the free group $\Z^k$.
So $L_{\mathbf{x}}$ is also free.
Define $\mathrm{rank}(L_{\mathbf{x}})$ as the number of generators of $L_{\mathbf{x}}$.
Clearly, $0\le \mathrm{rank}(L_{\mathbf{x}})\le k$.
If $L_{\mathbf{x}}$ is generated by $\{\mathbf{b}_1,\mathbf{b}_2,\ldots,\mathbf{b}_r\}$, where $r=\mathrm{rank}(L_{\mathbf{x}})$ and $\mathbf{b}_i\in \C^k$ for $1\leq i\leq r$, we denote $L_{\mathbf{x}}$ as $\braket{\mathbf{b}_1,\mathbf{b}_2,\ldots,\mathbf{b}_r}$ and say that $\{\mathbf{b}_1,\mathbf{b}_2,\ldots,\mathbf{b}_r\}$ is a lattice basis for $L_{\mathbf{x}}$.

The following lemma characterizes the situation where the lattice has full rank. 
Although it is not directly related to polynomial interpolation, we include it here as it will be used in Section~\ref{sec: three or opposite pairs}.

\begin{lemma}\label{lm: lattice full rank equiv roots of unity}
     Let $k\in \Z_+$ and $\mathbf{x}=(x_1,x_2,\ldots,x_k)\in \C^k$. Then $\mathrm{rank}(L_{\mathbf{x}})=k$ is equivalent to that $x_1,x_2,\ldots,x_k$ are all roots of unity.
\end{lemma}
\begin{proof}
    First, suppose $\mathbf{x}=(x_1,x_2,\ldots,x_k)=(\zeta_{p_1},\zeta_{p_2},\ldots,\zeta_{p_k})$.
    Then $\mathbf{b}_i=(0,\ldots,p_i,\ldots,0)\in L_{\mathbf{x}}$ for $1\le i \le k$, where $\mathbf{b}_i$ has $p_i$ at the $i$-th position and 0 at other positions.
    Since $\mathbf{b}_1,\mathbf{b}_2,\ldots,\mathbf{b}_k$ are linearly independent, we have $\mathrm{rank}(L_{\mathbf{x}})\ge k$.
    So $\mathrm{rank}(L_{\mathbf{x}})= k$.
    Conversely, suppose $\mathrm{rank}(L_\mathbf{x})=k$.
    Define the group homomorphism $\phi:\Z^k\to \C^*,(j_1,j_2,\ldots,j_k)\mapsto x_1^{j_1} x_2^{j_2}\cdots x_k^{j_k}$, where $\C^*$ is the multiplicative group of $\C$.
    Then $\mathrm{Ker}(\phi)=L_{\mathbf{x}}$ and so $\mathrm{Im}(\phi)\simeq \Z^k/L_{\mathbf{x}}$.
    Since $\mathrm{rank}(L_\mathbf{x})=k$, $\mathrm{Im}(\phi)$ is finite.
    This implies that $x_1,x_2,\ldots,x_k$ are all roots of unity.
\end{proof}

The following interpolation lemma is what we call ``conformal lattice interpolation'', which is a generalization of Lemma 5.1 in~\cite{CaiFX18}. 
It will be used in Lemma~\ref{lm:interpolation xyz}, which serves as a foundation for establishing \numP-hardness.
Informally and intuitively, the lemma states that we can evaluate a multivariate polynomial at $\mathbf{y}=(y_1,y_2,\ldots,y_k)$ by evaluating it at $\mathbf{x}=(x_1,x_2,\ldots,x_k)$ as long as the lattice of multiplicative relations of $\mathbf{y}$ is ``finer'' than that of $\mathbf{x}$.

\begin{lemma}[Conformal lattice interpolation]\label{lm:general interpolation}
    Let $k\in \Z_+$, $\mathbf{x}=(x_1,x_2,\ldots,x_k)\in \C^k$ and  $\mathbf{y}=(y_1,y_2,\ldots,y_k)\in \C^k$.
    Suppose 
    $L_{\mathbf{x}} \subseteq L_{\mathbf{y}}$.
    Given the numbers $$N_l(x_1, x_2, \ldots, x_k) := \sum\limits_{\substack{j_1,j_2,\dots, j_k \geq 0 \\ j_1+j_2+\dots+j_k\leq m}} (x_1^{j_1}x_2^{j_2}\cdots x_k^{j_k})^l z_{j_1,j_2,\dots,j_k}$$ for $l = 1,2,\dots,  \binom{m+k}{k}$,  we can compute $N_1(y_1, y_2, \ldots, y_k)$ in time polynomial in $m$.
\end{lemma}

\begin{proof}
We treat
$N_l = \sum\limits_{\substack{j_1,j_2,\dots, j_k \geq 0 \\ j_1+j_2+\dots+j_k\leq m}} (x_1^{j_1}x_2^{j_2}\cdots x_k^{j_k})^l z_{j_1,j_2,\dots,j_k}$
(where $1 \le {\ell} \le {m+k \choose k}$)
as a system of linear equations with unknowns $z_{j_1,...,j_k}$.
The coefficient vector of the first equation is $(x_1^{j_1} x_2^{j_2}\cdots x_k^{j_k})$,
indexed by the tuple $(j_1,j_2,\ldots,j_k)$, where $0 \le j_1,j_2,\ldots,j_k \le m$ and $j_1+j_2+\ldots+j_k\leq m$.
 The coefficient matrix of the linear system is a Vandermonde matrix,
with row index ${\ell}$ and column index $(j_1,j_2,...,j_k)$.
However, this Vandermonde matrix is in general rank deficient.
If $(j_1,j_2,\ldots,j_k)-(j_1',j_2',\ldots,j_k') \in L_{\mathbf{x}}$, then columns $(j_1,j_2,\ldots,j_k)$ and $(j_1',j_2',\ldots,j_k')$ have the same value.

We can combine the identical columns $(j_1,...,j_k)$ and $(j_1',...,j_k')$
if $(j_1,...,j_k)-(j_1',...,j_k') \in L_{\mathbf{x}}$, since for each coset $T$ of $L_{\mathbf{x}}$,
the value $x_1^{j_1} x_2^{j_2}\cdots x_k^{j_k}$ is the same.  
Thus, the sum
$$\sum\limits_{\substack{j_1,j_2,\dots, j_k \geq 0 \\ j_1+j_2+\dots+j_k\leq m}} (x_1^{j_1}x_2^{j_2}\cdots x_k^{j_k})^l z_{j_1,j_2,\dots,j_k}$$ can be
written as
$$\sum\limits_{T} x_T^l 
\sum\limits_{(j_1,j_2,\ldots,j_k)\in T\cap C}{z_{j_1,j_2,\dots,j_k}},
$$
where $x_T$ is the value of  $x_1^{j_1}x_2^{j_2}\cdots x_k^{j_k}$ 
for $(j_1, j_2, \ldots, j_k) \in T$, and the sum over $T$ is for all cosets $T$ of $L_{\mathbf{x}}$
having a non-empty intersection
with 
$C = \{ (j_1,j_2,\ldots, j_k) \mid 0 \le j_1,j_2,\ldots,j_k \le m, ~j_1+j_2+\ldots+j_k\leq m\}$.
Now  the coefficient matrix, indexed by $1 \le {\ell} \le {m+k \choose k}$
for the  rows and the cosets $T$ with $T \cap C \not = \emptyset$
for the columns,
has full rank.
And so we can solve 
$\left(\sum_{ (j_1,j_2,\ldots,j_k)\in T\cap C}{z_{j_1,j_2,\dots,j_k}}\right)$
for each coset $T$ with
$T \cap C \not = \emptyset$.
Notice that for the sum
$$\sum\limits_{\substack{j_1,j_2,\cdots, j_k \geq 0 \\ j_1+j_2+\dots+j_k\leq m}} y_1^{j_1}y_2^{j_2}\cdots y_k^{j_k} z_{j_1,j_2,\dots,j_k}$$ we also have a similar
expression
 $$\sum\limits_{T} y_T 
 \sum\limits_{(j_1,j_2,\ldots,j_k)\in T\cap C}{z_{j_1,j_2,\dots,j_k}},
 $$
 since $y_1^{j_1}y_2^{j_2}...y_k^{j_k}$ on each coset $T$ of $L_{\mathbf{y}}$ is also constant by $L_{\mathbf{x}}\subseteq L_{\mathbf{y}}$.
The lemma follows.
\end{proof}

\subsection{Tractable signature sets}\label{subsec: tractable signature sets}
We list some sets of signatures that are known to be tractable for various counting problems \cite{jcbook}.

\paragraph{Affine signatures}

\begin{definition}
For a signature $f$ of arity $n$, the support of $f$ is
$$
\operatorname{supp}(f) = \{ (x_1,x_2,\dots,x_n) \in \mathbb{Z}_2^n \mid f(x_1,x_2,\dots,x_n) \neq 0\}.
$$
\end{definition}

\begin{definition}\label{def: affine signatures}
A signature $f(x_1,\dots,x_n)$ of arity $n$ is \textbf{affine} if it has the form 
$$
\lambda\cdot \chi_{AX=0} \cdot \mathfrak{i}^{Q(X)},
$$
where $\lambda \in \mathbb{C}$, $X=(x_1,x_2,\dots,x_n,1)$, $A$ is a matrix over $\mathbb{Z}_2$, $Q(x_1,x_2, \dots,x_n) \in \mathbb{Z}_4[x_1,x_2,\dots,x_n]$ is a quadratic (total degree of most 2) multilinear polynomial with the additional requirement that the coefficients of all cross terms are even, i.e., $Q$ has the form
$$
Q(x_1,x_2, \dots,x_n) = a_0 +\sum\limits_{k=1}^{n}a_kx_k + \sum\limits_{1\leq i < j \leq n} 2b_{ij}x_ix_j,
$$
and $\chi$ is a 0-1 indicator function such that $\chi_{AX=0}$ is 1 iff $AX=0$. We use $\mathscr{A}$ to denote the set of all affine signatures.
\end{definition}

The following lemma follows directly from the definition.

\begin{lemma}\label{lm:is binary affine?}
Let $g$ be a binary signature with support of size 4. Then, $g\in \mathscr{A}$ iff $g$ has the signature matrix $M(g) = \lambda\left[\begin{smallmatrix} \mathfrak{i}^a & \mathfrak{i}^b \\ \mathfrak{i}^c & \mathfrak{i}^d\end{smallmatrix} \right]$ for some nonzero $\lambda \in \mathbb{C}, a,b,c,d\in \mathbb{N}$ and $a+b+c+d \equiv 0 \mod 2$.
\end{lemma}

\begin{lemma}\label{lm:is 1cz1 affine?}
Let $f$ be a 4-ary signature with the signature matrix $M(f) = \left[\begin{smallmatrix} a & 0 & 0 & 0\\ 0 & c & 0 & 0 \\ 0 & 0 & z & 0 \\ 0 & 0 & 0 & x \end{smallmatrix}\right]$ with $aczx \neq 0$, then $f \in \mathscr{A}$ iff the binary signature $g = (a,c,z,x)^T$ is affine.
\end{lemma}
\begin{proof}
    Note that $f(x_1,x_2,x_3,x_4)=g(x_1,x_2)\cdot \chi_{x_2=x_3}\cdot \chi_{x_1=x_4}$, the lemma follows.
\end{proof}

\paragraph{Product-type signatures}
\begin{definition}\label{def: product type}
A signature on a set of variables $X$ is of \textbf{product type} if it can be expressed as a product of unary functions, binary equality functions $([1,0,1])$, and binary disequality functions $([0,1,0])$, each on one or two variables of $X$. We use $\mathscr{P}$ to denote the set of product-type functions.
\end{definition}

The following lemma follows directly from the definition.
\begin{lemma}\label{lm: is binary product?}
    Let $f=\left[\begin{smallmatrix}
    f_{00}& f_{01}\\
    f_{10}& f_{11}
\end{smallmatrix}\right]
    =\left[\begin{smallmatrix}
    \alpha & \beta\\
    \gamma & \delta
\end{smallmatrix}\right]$
be a binary signature. Then $f\in \mathscr{P}$ iff $\alpha\delta-\beta\gamma=0$, or $\alpha=\delta=0$, or $\beta=\gamma=0$.
\end{lemma}

\begin{lemma}\label{lm:is product type?}
Let $f$ be a 4-ary signature. 
If $M(f) = \left[\begin{smallmatrix}1 & 0 & 0 & 0 \\ 0 & c & 0 & 0 \\ 0 & 0 & z & 0 \\ 0 & 0 & 0 & 1\end{smallmatrix}\right]$ with $cz \neq 0$, then $f \in \mathscr{P}$ iff $cz = 1.$ 
Similarly, if
$M(f) = \left[\begin{smallmatrix} 
1 & 0 & 0 & 0 \\ 
0 & 0 & d & 0 \\ 
0 & w & 0 & 0 \\ 
0 & 0 & 0 & 1
\end{smallmatrix} \right]$ with $dw \neq 0$, then $f \in \mathscr{P}$ iff $dw=1$.
\end{lemma}
\begin{proof}
Lemma 2.10 in~\cite{CaiF17} proved the first case.
The same proof easily generalizes for the second case. 
\end{proof}

\begin{remark}\label{rmk: affine or product imply support power of 2}
    By Definition~\ref{def: affine signatures} and Definition~\ref{def: product type}, if $f\in \mathscr{A}$ or $f\in \mathscr{P}$, then $|\mathrm{supp}(f)|$ is power of 2.
\end{remark}

\paragraph{Local affine signatures $\mathscr{L}$}
\begin{definition}
A function $f$ of arity $n$ is in $\mathscr{L}$, if for any $\sigma = s_1s_2\dots s_n$ in the support of $f$, the transformed function
$$
R_{\sigma}f: (x_1,x_2, \dots, x_n) \mapsto \sqrt{\mathfrak{i}}^{\sum_{i=1}^{n} s_ix_i}f(x_1,x_2,\dots,x_n)
$$
is in $\mathscr{A}$. Here each $s_i$ is a 0-1 valued integer, and the sum $\sum_{i=1}^{n} s_ix_i$ is evaluated as an integer (or an integer mod 8).
\end{definition}

\paragraph{Matchgate signatures}

\hspace{.3in}

Matchgates were introduced by Valiant~\cite{val02a, val02b} to give polynomial-time algorithms for a collection of counting problems over planar graphs.
As the name suggests,
problems expressible by matchgates can be reduced to computing a weighted sum of perfect matchings.
The latter problem is tractable over planar graphs by Kasteleyn's algorithm~\cite{Kasteleyn1967},
a.k.a.~the FKT algorithm~\cite{TF61,Kasteleyn1961}.
These counting problems are naturally expressed in the Holant framework using \emph{matchgate signatures}.
We use $\mathscr{M}$ to denote the set of all matchgate signatures;
thus $\PlHolant(\mathscr{M})$ is tractable, as well as
$\PlHolant(\neq_2 \mid \mathscr{M})$.

The parity of a signature is even (resp.~odd) if its support is on entries of even (resp.~odd) Hamming weight.
We say a signature satisfies the Even (resp. Odd) Parity Condition
if all entries of odd  (resp. even) weight are zero. For signatures of arity 
at most $4$, the matchgate signatures are characterized by the following lemma.

\begin{lemma}[\cite{val02b}\cite{jcbook}]\label{lm:is matchgate}
If $f$ has arity $\leq 3$, then $f \in \mathscr{M}$ iff $f$ satisfies the Parity Condition. 
If $f$ has arity 4, then $f\in\M$ iff $f$ satisfies the Even Parity Condition, i.e.,
$$
M_{x_1x_2,x_4x_3}(f) = \left[\begin{matrix} f_{0000} & 0 & 0 & f_{0011} \\ 0 & f_{0110} & f_{0101} & 0 \\ 0 & f_{1010} & f_{1001} & 0 \\ f_{1100} & 0 & 0 & f_{1111} \end{matrix} \right],
$$
and
$
\det M_{\rm{Out}}(f) = \det M_{\rm{In}}(f)
$;
or $f$ satisfies the Odd Parity Condition, i.e.,
$$
M_{x_1x_2,x_4x_3}(f) = \left[\begin{matrix} 0 & f_{0010} & f_{0001} & 0 \\ f_{0100} & 0 & 0 & f_{0111} \\ f_{1000} & 0 & 0 & f_{1011} \\ 0 & f_{1110} & f_{1101} & 0 \end{matrix} \right],
$$
and $f_{0010} f_{1101}-f_{0001} f_{1110} = f_{0100} f_{1011} - f_{0111} f_{1000}$.
\end{lemma}

For signatures $f$ with specific forms of signature matrices, we have explicit criteria for determining when $f$ is $\M$-transformable.
In the following, we denote $\widehat{\M} = H\M$, where 
$H=\frac{1}{\sqrt{2}}\left[\begin{smallmatrix}
    1& 1\\
    1& -1\\
\end{smallmatrix}\right].$

\begin{lemma}\cite{CaiFS21}\label{lm: six-vertex M transformable}
A signature $f$ with the signature matrix $M(f)=
\left[\begin{smallmatrix}
0 & 0 & 0 & a \\
0 & b & c & 0 \\
0 & z & y & 0 \\
x & 0 & 0 & 0
\end{smallmatrix}\right]$ is $\mathscr{M}$-transformable iff $f \in \mathscr{{M}}$ or $f \in \mathscr{\widehat{M}}$.
\end{lemma}

\begin{lemma}\label{lm: is matchgate hat}
A signature $f$ with the signature matrix $M(f) = \left[\begin{smallmatrix} 0 & 0 & 0 & 0 \\ 0 & c & d & 0 \\ 0 & w & z & 0 \\ 0 & 0& 0& 0 \end{smallmatrix} \right]$ is in $\widehat{\mathscr{M}}$ iff $c = \epsilon z$ and $d = \epsilon w$, where $\epsilon = \pm 1$.
\end{lemma}

\begin{proof}
Let $\hat{f}=H^{\otimes4}f$ be the transformed signature.
A direct calculation shows that $\hat{f}$ has the signature matrix
$$
M(\hat{f})=H^{\otimes2} M(f)H^{\otimes 2} = \frac{1}{4}\left[\begin{smallmatrix} c+d+w+z & -c+d-w+z & c-d+w-z & -c-d-w-z \\ 
-c-d+w+z & c-d-w+z & -c+d+w-z & c+d-w-z \\ 
c+d-w-z & -c+d+w-z & c-d-w+z & -c-d+w+z \\ 
-c-d-w-z & c-d+w-z & -c+d-w+z & c+d+w+z \end{smallmatrix} \right].$$
If $f$ satisfies the Even Parity Condition, then $-c+d-w+z=-c-d+w+z=0$, which implies that $c=z$ and $d=w$.
Plug in we have 
$M(\hat{f})=\frac{1}{2}\left[
\begin{smallmatrix}
    w+z & 0 & 0 & -w-z\\
    0 & -w+z & w-z & 0\\
    0 & w-z & -w+z & 0\\
    -w-z & 0 & 0 & w+z\\
\end{smallmatrix}
\right].$
Thus, $\det M_{\rm Out}(\hat{f})=\det M_{\rm In}(\hat{f})=0$.
Therefore, $\hat{f}\in \M$ by Lemma~\ref{lm:is matchgate}, and thus $f\in \widehat{\M}$.
Similarly, $f$ satisfies the Odd Parity Condition iff $c=-z$ and $d=-w$, in which case $f\in \widehat{\M}$.
\end{proof}

Similarly, by a direct calculation we can prove the following three lemmas:

\begin{lemma}\label{lm: is matchgate hat2}
A signature $f$ with the signature matrix $M(f) = \left[\begin{smallmatrix} 1 & 0 & 0 & 0 \\ 0 & c & 0 & 0 \\ 0 & 0 & z & 0 \\ 0 & 0& 0& 1 \end{smallmatrix} \right]$ is in $\widehat{\mathscr{M}}$ if $c = z$.
\end{lemma}

\begin{lemma}\label{lm: is matchgate hat3}
A signature $f$ with the signature matrix $M(f) = \left[\begin{smallmatrix} 
0 & 0 & 0 & 1 \\ 
0 & 0 & d & 0 \\ 
0 & w & 0 & 0 \\ 
1 & 0 & 0 & 0\end{smallmatrix} \right]$ is $\M$-transformable if and only if $d = w$.
\end{lemma}

\begin{lemma}\label{lm: M transformable by diagnol holographic transformation}
    Let $M(f) = \left[\begin{smallmatrix} 
    1 & 0 & 0 & b \\ 
    0 & c & d & 0 \\ 
    0 & \epsilon d & \epsilon c & 0 \\ 
    \epsilon b & 0 & 0 & 1 \end{smallmatrix}\right]$
    for some $\epsilon=\pm1$, and $T\in \mathrm{GL_2}(\C)$.
    If $T=\left[\begin{smallmatrix}
        w & 0\\
        0 & z
    \end{smallmatrix}\right]$
    or 
    $T=\left[\begin{smallmatrix}
        0 & x\\
        y & 0
    \end{smallmatrix}\right]$,
    then $T^{\otimes4}f\in \M\Leftrightarrow f\in \M$.
\end{lemma}

For the arrow symmetric equality case, we have an explicit criterion for determining whether $f$ is $\M$-transformable.

\begin{lemma}\label{lm:symmetric equality_matchgate-transformable}
    Let $M(f) = \left[\begin{smallmatrix} 1 & 0 & 0 & b \\ 0 & c & d & 0 \\ 0 & d & c & 0 \\ b & 0 & 0 & 1 \end{smallmatrix}\right]$. If $f$ is $\mathscr{M}$-transformable, then either $f \in \mathscr{M}$, or $d = \pm bc$.
    Conversely, if $d=\pm bc$, then $f$ is $\M$-transformable.
\end{lemma}
\begin{proof}
    First, suppose $f$ is $\M$-transformable by 
     $T = \left[\begin{smallmatrix}
        w & x\\
        y & z
    \end{smallmatrix} \right] \in \operatorname{GL}_2(\mathbb{C})$. We have $\det(T) = wz-xy \neq 0$. The condition that $f$ is $\mathscr{M}$-transformable is equivalent to  $(0,1,1,0)(T^{-1})^{\otimes 2} \in \mathscr{M}$ and $T^{\otimes 4}f \in \mathscr{M}$. 
    Note that $(0,1,1,0)(T^{-1})^{\otimes 2} = \frac{1}{\det{T}^2}(-2yz,xy+wz,xy+wz,-2wx)$. 
    By Lemma~\ref{lm:is matchgate}, $(0,1,1,0)(T^{-1})^{\otimes 2} \in \mathscr{M}$ iff (1) $(yz=0) \wedge (wx = 0)$ or (2) $xy+wz = 0$. 
    If it 
    is case (1),
    then $w = z = 0$ or $x=y=0$ since $T \in \operatorname{GL}_2(\mathbb{C})$.
    Then $T^{\otimes 4}f \in \mathscr{M}$ implies $f \in \mathscr{M}$ by Lemma~\ref{lm: M transformable by diagnol holographic transformation}.  
Assume that (1) does not hold but case (2) holds $xy+wz = 0$.
    Since  $T \in \operatorname{GL}_2(\mathbb{C})$ we have  $xyzw \neq 0$.
    A direct computation shows that $g=T^{\otimes 4}f$ has the following signature matrix $M(g)=T^{\otimes2}M(f)(T^T)^{\otimes 2}$:
    \begin{center}$
        \left[\begin{matrix}
        w^4+2(b+c+d) w^2 x^2+x^4 & w^3 y+x^3 z & w^3y +x^3 z & p(w,x,y,z)\\
        w^3 y+x^3 z & q(w,x,y,z) & r(w,x,y,z) & wy^3+xz^3\\
        w^3 y+x^3 z & r(w,x,y,z) & q(w,x,y,z) & wy^3+xz^3\\
        p(w,x,y,z) & wy^3+xz^3 & wy^3+xz^3 & y^4+2(b+c+d)y^2z^2+z^4
    \end{matrix}\right]$\end{center}
    
    where 
    \begin{equation*}
    \begin{split}
        p(w,x,y,z) = 2(c+d) w x y z+x^2(b y^2+z^2)+w^2(y^2+b z^2)\\
        q(w,x,y,z) = 2(b+d) w x y z+x^2(c y^2+z^2)+w^2(y^2+c z^2)\\
        r(w,x,y,z) = 2(b+c) w x y z+x^2(d y^2+z^2)+w^2(y^2+d z^2)
    \end{split}
    \end{equation*}

    By Lemma~\ref{lm:is matchgate}, if $g\in\M$, then either $g$ satisfies the Even Parity Condition and $\det M_{\rm{Out}}(g) = \det M_{\rm{In}}(g)$, or $g$ satisfies the Odd Parity Condition and $g_{0010} g_{1101}-g_{0001} g_{1110} = g_{0100} g_{1011} - g_{0111} g_{1000}$.
    If $g$ satisfies the Even Parity Condition, then $w^3 y+x^3 z = wy^3+xz^3=0$. Combined with $xy+wz = 0$ and $xyzw \neq 0$, we have $(x = \pm w) \wedge (y = \mp z)$ or $(x = \pm \mathfrak{i}w) \wedge (y = \pm \mathfrak{i}z)$. 
    Combined with $\det M_{\rm{Out}}(g) = \det M_{\rm{In}}(g)$, if it were the first two cases, we have $d = -bc$;
    if it were the second two cases, we have $d = bc$.
    If $g$ satisfies the Odd Parity Condition, then $ p(w,x,y,z) = q(w,x,y,z) = r(w,x,y,z) = 0$. 
    This implies $b=c=d$.
    Plug into $p(w,z,y,z)=0$, we have $x^2z^2+w^2y^2+bx^2y^2+bw^2z^2+4bxywz=0$.
    Combined with $xy+wz=0$, we have 
    \begin{equation}\label{sym equality m-trans eq1}
        x^2z^2+w^2y^2+2bxywz=0.
    \end{equation}
    Also, $g_{0000}=w^4+2(b+c+d)w^2x^2+x^4=0$.
    Plug in $b=c=d$, we have
    \begin{equation}\label{sym equality m-trans eq2}
        w^4+6bw^2x^2+x^4=0.
    \end{equation}
    Let $\alpha=\frac{w}{x}$, $\beta=\frac{z}{y}$, then $\alpha\beta=-1$ by $xy+wz=0$.
    Divide Equation~\ref{sym equality m-trans eq1} by $x^2y^2$, we have $\alpha^2+\beta^2+2b\alpha\beta = 0$.
    Plug in $\beta=-\frac{1}{\alpha}$, we have
    \begin{equation}\label{sym equality m-trans eq3}
        \alpha^4-2b\alpha^2+1=0.
    \end{equation}
    Divide Equation~\ref{sym equality m-trans eq2} by $x^4$, we have
    \begin{equation}\label{sym equality m-trans eq4}
         \alpha^4+6b\alpha^2+1 = 0.
    \end{equation}
    Since $\alpha\neq 0$, we conclude $b=0$ by Equations~\ref{sym equality m-trans eq3} and \ref{sym equality m-trans eq4}.
    Thus, $d=bc$ holds.

    Conversely, if $d=bc$, then $f$ is $\M$-transformable by $T = \left[\begin{smallmatrix}
        1 & \ii\\
        1 & -\ii
    \end{smallmatrix} \right]$; If $d=-bc$, then $f$ is $\M$-transformable by $T = \left[\begin{smallmatrix}
        1 & 1\\
        1 & -1
    \end{smallmatrix} \right]$.
\end{proof}

For the ``arrow symmetric disequality'' case where $(b=-y)\land (c=-z)\land(d=-w)$, we also have an explicit criterion for determining whether $f$ is $\M$-transformable.

\begin{lemma}\label{lm:symmetric disequality_matchgate-transformable}
    Let $M(f) = \left[\begin{smallmatrix} 
    1 & 0 & 0 & b \\ 
    0 & c & d & 0 \\ 
    0 & -d & -c & 0 \\ 
    -b & 0 & 0 & 1 \end{smallmatrix}\right]$. If $f$ is $\mathscr{M}$-transformable, then either $f \in \mathscr{M}$, or $d = \pm \ii bc$.
    Conversely, if $d=\pm \ii bc$, then $f$ is $\M$-transformable.
\end{lemma}
\begin{proof}
    Let $T = \left[\begin{smallmatrix}
        w & x\\
        y & z
    \end{smallmatrix} \right] \in \operatorname{GL}_2(\mathbb{C})$.  
    We have $\det(T) = wz-xy \neq 0$.
    Same as the proof of Lemma~\ref{lm:symmetric equality_matchgate-transformable}, the condition that $f$ is $\mathscr{M}$-transformable is equivalent to that $(yz=0) \wedge (wx = 0)$ or $xy+wz = 0$, and $T^{\otimes 4}f\in \M$.
    
    If $(yz=0) \wedge (wx = 0)$, then $w = z = 0$ or $x=y=0$ since $T \in \operatorname{GL}_2(\mathbb{C})$.
    Then $T^{\otimes 4}f \in \mathscr{M}$ implies $f \in \mathscr{M}$ by Lemma~\ref{lm: M transformable by diagnol holographic transformation}. 
    We therefore assume $xy+wz = 0$ and (using $T \in \operatorname{GL}_2(\mathbb{C})$) $xyzw \neq 0$.
     A direct computation shows that $g=T^{\otimes 4}f$ has the following signature matrix $M(g)=T^{\otimes2}M(f)(T^T)^{\otimes 2}$:
    \begin{center}$
        \left[\begin{matrix}
        
        w^4+x^4 & g_{0001} & g_{0010} & w^2 y^2 + x^2 z^2\\
       g_{0100} & w^2 y^2 + x^2 z^2 & w^2 y^2 + x^2 z^2 & g_{0111}\\
        g_{1000} & w^2 y^2 + x^2 z^2 & w^2 y^2 + x^2 z^2 & g_{1011}\\
        w^2 y^2 + x^2 z^2 & g_{1101} & g_{1110} & y^4 + z^4
\end{matrix}\right]$\end{center}
    
    where 
    \begin{equation*}
    \begin{aligned}
       &g_{0001} =  2 (b + c - d) w^2 x z + x^3 z+w^3 y\\
       &g_{0010} = 2 (b - c + d) w^2 x z + x^3 z +  w^3 y \\
       &g_{0100} =  2 (-b + c + d) w^2 x z + x^3 z+w^3 y\\
       &g_{1000} = - 2 (b + c + d) w^2 x z + x^3 z + w^3 y\\
       &g_{0111} = 2 (b + c + d) wyz^2 + x z^3 +w y^3 \\
       &g_{1011} = 2 (b - c - d) wyz^2 + x z^3 +w y^3\\
       &g_{1101} = - 2 (b - c + d) wyz^2 + x z^3 +  wy^3\\
       &g_{1110} =  - 2 (b + c - d) wyz^2 +x z^3 + w y^3
    \end{aligned}
    \end{equation*}
    
 By Lemma~\ref{lm:is matchgate}, if $g\in\M$, then either $g$ satisfies the Even Parity Condition and $\det M_{\rm{Out}}(g) = \det M_{\rm{In}}(g)$, or $g$ satisfies the Odd Parity Condition and $g_{0010} g_{1101}-g_{0001} g_{1110} = g_{0100} g_{1011} - g_{0111} g_{1000}$.
    If $g$ satisfies the Even Parity Condition, then $g_{0001}=g_{0010}=g_{0100}=g_{1000}=g_{0111}=g_{1011}=g_{1101}=g_{1110}=0$. 
    This implies $b=c=d=0$.
    Thus, $d=\ii bc$ holds.
    If $g$ satisfies the Odd Parity Condition, then $w^4+x^4=y^4+z^4=w^2y^2+x^2z^2=0$.
    Combined with $xy+wz = 0$ and $xyzw \neq 0$, we have $(x = \pm \sqrt{\ii}w) \wedge (z = \mp \sqrt{\ii}y)$ or $(x = \pm \mathfrak{i}\sqrt{\mathfrak{i}}w) \wedge (z = \mp \mathfrak{i}\sqrt{\mathfrak{i}}y)$. 
   Combined with $g_{0010} g_{1101}-g_{0001} g_{1110} = g_{0100} g_{1011} - g_{0111} g_{1000}$,  if it were the first two cases, we have $d = -\mathfrak{i}bc$;
    if it were the second two cases, we have $d = \mathfrak{i}bc$.

     Conversely, if $d=\mathfrak{i}bc$, then $f$ is $\M$-transformable by $T = \left[\begin{smallmatrix}
        1 & \mathfrak{i}\sqrt{\mathfrak{i}}\\
        1 & -\mathfrak{i}\sqrt{\mathfrak{i}}
    \end{smallmatrix} \right]$; If $d=-\mathfrak{i}bc$, then $f$ is $\M$-transformable by $T = \left[\begin{smallmatrix}
        1 & \sqrt{\mathfrak{i}}\\
        1 & -\sqrt{\mathfrak{i}}
    \end{smallmatrix} \right]$.
\end{proof}

The above two criterion will be  used extensively in Section~\ref{sec: three or opposite pairs}.

\subsection{Known dichotomies}
\begin{theorem}[Constraint satisfaction problems~\cite{CaiLX14}]\label{thm: CSP dichotomy}
Let $\mathcal{F}$ be any set of complex-valued signatures in Boolean variables. Then $\operatorname{\#CSP}(\mathcal{F})$
is \#P-hard unless
$\mathcal{F}\subseteq\mathscr{A}$ or
$\mathcal{F}\subseteq\mathscr{P}$,
  in which case the problem is computable in polynomial time.
\end{theorem}

\begin{theorem}[General eight-vertex model \cite{CaiF17}]\label{thm: general eight-vertex}
Let $f$ be a 4-ary signature with the signature matrix $M(f) = \left[\begin{smallmatrix} a & 0 & 0 & b \\ 0 & c & d & 0 \\ 0 &w &z & 0 \\ y & 0 & 0 &x \end{smallmatrix}\right]$. 
If $ax=0$, then then  $\holant{\neq_2}{f}$ is
equivalent to the six-vertex model $\holant{\neq_2}{f'}$ where 
$M(f') = 
\left[\begin{smallmatrix} 
0 & 0 & 0 & b \\ 
0 & c & d & 0 \\ 
0 & w & z & 0 \\ 
y & 0 &0 & 0 
\end{smallmatrix}\right]$. 
Explicitly, $\holant{\neq_2}{f}$ is \numP-hard except in the following cases:
\begin{itemize}
    \item $f' \in \mathscr{P}$,
    \item $f' \in \mathscr{A}$,
    \item there is at least one zero in each pair $(b,y), (c,z), (d,w)$.
\end{itemize}
If $ax \neq 0$, then \holant{\neq_2}{f} is \numP-hard except in the following cases:
\begin{itemize}
    \item $f$ is $\mathscr{P}$-transformable;
    \item $f$ is $\mathscr{A}$-transformable;
    \item $f$ is $\mathscr{L}$-transformable.
\end{itemize}
In all listed cases, \holant{\neq_2}{f} is computatble in polynomial time.
\end{theorem}

\begin{theorem}[Planar six-vertex model \cite{CaiFS21}]\label{thm:planar six-vertex}
Let $f$ be a signature with the signature matrix $M(f)=\left[\begin{smallmatrix}0 & 0 & 0 & b \\ 0 & c & d & 0 \\ 0 & w & z & 0 \\ y & 0 & 0 & 0\end{smallmatrix}\right]$, where $b, c, d, w, y, z \in$ $\mathbb{C}$. Then \plholant{\neq_2}{f} is polynomial time computable in the following cases, and \numP-hard otherwise:
\begin{enumerate}
\item $f \in \mathscr{P}$ or $\mathscr{A}$;
\item There is a zero in each pair $(b, y),(c, z), (d,w)$;
\item $f \in \mathscr{M}$ or $\widehat{\mathscr{M}}$;
\item $d=w=0$ and
\begin{enumerate}[label=(\roman*).]
    \item $(b y)^{2} = (cz)^2$, or
    \item $y=b \mathfrak{i}^{\alpha}, c=b \sqrt{\mathfrak{i}}^{\beta}$, and $z=b \sqrt{\mathfrak{i}}^{\gamma}$, where $\alpha, \beta, \gamma \in \mathbb{N}$, and $\beta \equiv \gamma (\bmod 2)$;
\end{enumerate}
\end{enumerate}
If $f$ satisfies condition 1 or 2 , then \holant{\neq_{2}}{f} is computable in polynomial time without the planarity restriction; otherwise (the non-planar) \holant{\neq_2}{f} is \numP-hard.
\end{theorem}

We have the following corollary.

\begin{corollary}\label{cor:exactly 3 nonzero entries}
Let $f$ be a 4-ary signature with the signature matrix $M(f) = \left[\begin{smallmatrix}0 & 0 & 0 & b \\ 0 & c & d & 0 \\ 0 & w & z & 0 \\ y & 0 & 0 & 0\end{smallmatrix}\right]$. If there are exactly three nonzero entries in $\left[\begin{smallmatrix} c & d \\ w & z \end{smallmatrix}\right]$, then $\holant{\neq_2}{f}$ is \numP-hard and $\plholant{\neq_2}{f}$ is \numP-hard unless $f \in \mathscr{M}$ or $\widehat{\mathscr{M}}$. 
\end{corollary}
\begin{proof}
Follow the argument in the proof of Lemma~\ref{lm:is product type?}, one can easily show that $f \notin \mathscr{P}$. Also, $f \notin \mathscr{A}$ because $0101 \oplus 0110 \oplus 1001=1010$, and if $\operatorname{supp}(f)$ were affine, then three entries in $\operatorname{supp}(f)$ would imply the fourth entry also belongs to $\operatorname{supp}(f)$. It cannot be the case that $d=w=0$ neither since there are three nonzero entries in $\left[\begin{smallmatrix} c & d \\ w & z \end{smallmatrix}\right]$. The proof is now complete by Theorem~\ref{thm:planar six-vertex}.
\end{proof}

\begin{theorem}[Spin systems on $k$-regular graphs \cite{k-regular-spin}]\label{thm:spin system}
Fix any integer $k\geq 3$. Suppose $f = (w,x,y,z)$ with $w,x,y,z \in \mathbb{C}$. Then $\holant{=_k}{f}$  is \numP-hard except in the following cases, where the problem is computable in polynomial time.
\begin{itemize}
    \item $f \in \mathscr{P}$: $wz=xy$, or $w=z=0$, or $x=y=0$;
    \item $(=_k\mid f)$ is $\mathscr{A}$-transformable: $wz=-xy$,
    \begin{itemize}
        \item if $k$ is odd, $x^2=\epsilon y^2$ and $z^{2k}=\epsilon w^{2k}$, where $\epsilon=\pm1$.
        \item if $k$ is even, $x^4=y^4$ and $z^{2k}=w^{2k}$.
    \end{itemize}
\end{itemize}
    If the input is restricted to planar graphs, then another case becomes polynomial time computable but everything else remains \numP-hard.
\begin{itemize}
    \item $(=_k \mid f)$ is $\mathscr{M}$-transformable:
    \begin{itemize}
        \item if $k$ is odd, $x=\epsilon y$ and $z^{k}= w^{k}$, where $\epsilon=\pm1$.
        \item if $k$ is even, $x^2=y^2$ and $z^{k} = w^{k}$.
    \end{itemize}
\end{itemize}
\end{theorem}

\section{Main Theorem and Proof Outline}
\begin{theorem}\label{thm: main theorem}
Let $f$ be a 4-ary signature with the signature matrix
$M(f)=\left[\begin{smallmatrix}
a & 0 & 0 & b\\
0 & c & d & 0\\
0 & w & z & 0\\
y & 0 & 0 & x
\end{smallmatrix}\right]$.
If  $ax=0$, then  $\plholant
{\neq_2}{f}$ is
equivalent to the planar six-vertex model $\plholant{\neq_2}{f'}$
where $f'$ is obtained from $f$ by setting $a=x=0$,  i.e.,
$M(f')=\left[\begin{smallmatrix}
0 & 0 & 0 & b\\
0 & c & d & 0\\
0 & w & z & 0\\
y & 0 & 0 & 0
\end{smallmatrix}\right]$.
The complexity of the problem $\plholant{\neq_2}{f'}$ is characterized by Theorem~\ref{thm:planar six-vertex}.
If $ax\neq 0$, then $\plholant{\neq_2}{f}$
is \numP-hard except in the following cases:
\begin{enumerate}
\item $f$ is $\mathscr{P}$-transformable;
\item $f$ is $\mathscr{A}$-transformable;
\item $f$ is $\mathscr{L}$-transformable;
\item $f$ is $\M$-transformable;
\item $d=w=\pm \sqrt{ax},b=y=-c=-z=\mathfrak{i}\tan(\frac{\beta\pi}{8})\sqrt{ax}$, where $\beta\in \mathbb{N}$ is odd.
\end{enumerate}
In all listed cases, $\plholant{\neq_2}{f}$ is computable in polynomial time.
\end{theorem}

If $ax=0$, then by Lemma~\ref{lm:let a=x}, we can normalize $a=x=0$.
Then this corresponds to the planar six-vertex model, which has already been classified in Theorem~\ref{thm:planar six-vertex}. 
If $ax\neq 0$, a holographic transformation by $T=\left[\begin{smallmatrix}
    1 & 0\\
    0 & \sqrt[4]{\frac{a}{x}}
\end{smallmatrix}\right]$ normalizes $a=x$.
Since multiplying $f$ by a nonzero scalar does not affect the complexity of the problem $\plholant{\neq_2}{f}$, we may normalize by setting $a = x = 1$.

We divide the proof of our main theorem into three disjoint cases, based on the number and arrangement of zero entries in $\{b, y, c, z, d, w\}$, and on how many of the pairs $(b, y), (c, z), (d, w)$ are zero pairs $(0, 0)$:

\begin{enumerate}\renewcommand{\labelenumi}{\Roman{enumi}.}
    \item
    At least two of the pairs in $\{(b, y), (c, z), (d, w)\}$ are zero pairs $(0, 0)$.

    \item
There is at least one zero entry in $\{b, y, c, z, d, w\}$, but at most one zero pair among $\{(b, y), (c, z), (d, w)\}$.

    \item
All six entries $\{b, y, c, z, d, w\}$ are nonzero.
\end{enumerate}

These three cases are clearly disjoint and collectively exhaustive.

The following is an outline of how Cases I, II, and III are handled.

\noindent{\bf Case I:}
 At least two of the pairs in $\{(b, y), (c, z), (d, w)\}$ are zero pairs $(0, 0)$.
  
    \begin{enumerate}
    \item[1.] \label{Case I.1}
    The inner pair $(d,w)$ is a zero pair and at least one of the outer pairs $(b, y)$,  $(c, z)$ is a zero pair. 
    We prove that $\plholant{\neq_2}{f}$ is tractable if $f \in \mathscr{P}$, or $f$ is $\mathscr{A}$-transformable, or $f$ is $\M$-transformable, and is \#P-hard  otherwise.  
    We reduce spin systems on $k$-regular graph (where $k$ is an even integer) to $\plholant{\neq_2}{f}$ in Lemma~\ref{lm: 1cz1}, and apply the dichotomy of spin systems on $k$-regular graph (Theorem~\ref{thm:spin system}) to get \#P-hardness.
    Tractability of  $\plholant{\neq_2}{f}$ follows from known tractable signatures.  
            
    \item[2.]
    The two outer pairs $(b,y),(c,z)$ are zero pairs.
    We prove that $\plholant{\neq_2}{f}$ is tractable if $f\in \mathscr{P}$, $f\in \mathscr{A}$, or $f\in \mathscr{M}$.
  To address this case, we first introduce a delicate non-local reduction from a (not necessarily planar) $\BiGH$ problem to the planar eight-vertex model, under the special condition that $d = w$ (Lemma~\ref{lm:case d=w}).
    Then we 
    handle this case by realizing a signature satisfying  $(b,y)=(c,z)=(0, 0)$ and  $d =  w$ (Lemma~\ref{lm: 1dw1}).
    \end{enumerate}

\noindent {\bf Case II:} There is at least one zero entry in $\{b, y, c, z, d, w\}$, but at most one zero pair among $\{(b, y), (c, z), (d, w)\}$.

We show that $\plholant{\neq_2}{f}$ is \numP-hard unless $f\in \mathscr{M}$ (Theorem~\ref{thm: at least one zero at most one zero pair}).
The proof of Theorem~\ref{thm: at least one zero at most one zero pair} is divided into subcases based on the number $\sf N$ of zero entries in $M(f)$ and their positions. 
Note that $\sf N \leq 4$, since there can be at most one zero pair in $f$.
\begin{enumerate}
   \item $\sf N = 1$: 
   There is exactly one zero entry, which may appear in either an outer pair (Lemma~\ref{lm: 1bcdwy1}) or an inner pair (Lemma~\ref{lm: 1bcwzy1}).
\item $\sf N = 2$: 
There are two zero entries, which may lie in two distinct pairs (Lemmas~\ref{lm: 1bcdw1},~\ref{lm: 1bcdz1} and~\ref{lm: 1bcwz1}) or in the same pair (Lemma~\ref{lm:1cdwz1} and~\ref{lm:1bycz1}).
\item $\sf N = 3$: 
There are three zero entries. This may include one zero pair (Lemmas~\ref{lm: 1cdw1}, \ref{lm: 1cwz1}, \ref{lm: 1cdz1}, and~\ref{lm: 1bcz1}), or the zeros may occur in three different pairs (Lemma~\ref{lm: 1bcdy1, by=0}).
\item $\sf N = 4$: 
There are four zero entries. 
In this case, there must be a zero pair, which could be either an outer pair (Lemma~\ref{lm:1bd1}) or an inner pair (Lemma~\ref{lm: 1bc1}).
\end{enumerate}

Although the proof involves a case-by-case analysis, the overall strategy for establishing \#P-hardness is to use gadget construction to realize a signature in the planar six-vertex model or Case~I.
In a few cases we also employ polynomial interpolation.

\noindent{\bf Case III:} All six entries $\{b, y, c, z, d, w\}$ are nonzero.

We prove that $\plholant{\neq_2}{f}$ is tractable if $f$ is $\mathscr{P}$-transformable, or $\mathscr{A}$-transformable, or $\mathscr{L}$-transformable, or $\mathscr{M}$-transformable, or $(b=y=-c=-z=\mathfrak{i}\tan(\frac{\beta\pi}{8}))\land (d=w=\pm1)$ for odd integer $\beta$, and is \numP-hard otherwise.

In this case, interpolation using M\"obius transformation is a powerful method. 
However, certain settings prevent the direct application of this technique, most notably the case where $f$ satisfies arrow reversal equality or disequality, i.e., $(a = \epsilon x)\land(b = \epsilon y)\land (c = \epsilon z)$ with $\epsilon = \pm 1$.
\footnote{In physics, the case $\epsilon = 1$ corresponds to the zero-field case of the eight-vertex model, representing the absence of an external electric field~\cite{baxter1971eightvertex}.}
We first address the case $\epsilon = 1$, i.e., when ($b, y), (c, z), (d, w)$ are three identical pairs. This is the most intricate setting, where we employ a combinatorial transformation to {\sc Even Coloring} to identify new tractable classes, and introduce a novel technique—conformal lattice interpolation—to establish \numP-hardness.
We define a lattice of multiplicative relations among the normalized eigenvalues of $M(f)$, after applying a holographic transformation.

    \begin{enumerate}
        \item[1.]
        If $b=-1$, the lattice is not well defined.
        By applying the combinatorial transformation to {\sc Even Coloring}, we get a {redundant} signature (Lemma~\ref{lm: ars b=-1 or c=-1}). 
        Then we appeal to Theorem~\ref{thm:redundant} and show that either $f\in \M$, or the redundant signature is non-singular, and thus $\plholant{\neq_2}{f}$ is \numP-hard.
        
        \item[2.]
        If $b\neq -1$, we apply conformal lattice interpolation (Lemma~\ref{lm:general interpolation}, Lemma~\ref{lm:interpolation xyz}).
        In generic cases, by carefully analyzing the lattice structure, we can either interpolate the disequality crossover $\mathcal{S'}$ (Lemma~\ref{lm:interpolate disequality crossover}), and thus $\plholant{\neq_2}{f}\equiv_T \holant{\neq_2}{f}$ by Lemma~\ref{lm:disequality crossover}; or we get \numP-hardness by interpolating a signature which is known to be \numP-hard in the planar six-vertex model (condition~\ref{cd:six vertex model} in Lemma~\ref{lm:interpolation something hard}), or is known to be \numP-hard in Case I (condition~\ref{cd:y=-1 two 0 pairs} in Lemma~\ref{lm:interpolation something hard}) and Case II (condition~\ref{cd:y=-1 not matchgate},\ref{cd:x=z not matchgate},\ref{cd:x+z=0 not matchgate} in Lemma~\ref{lm:interpolation something hard}).
        
        \item[3.]
        For some exceptional lattices, none of the conditions~\ref{cd:six vertex model},\ref{cd:y=-1 not matchgate},\ref{cd:y=-1 two 0 pairs},\ref{cd:x=z not matchgate},\ref{cd:x+z=0 not matchgate} in Lemma~\ref{lm:interpolation something hard} can be satisfied, and we have to appeal to the reduction from non-singular redundant signatures (condition~\ref{cd: redundent} in Lemma~\ref{lm:interpolation something hard}).

        It's in this case that we use  geometric properties  of polynomials.
        To apply the conformal lattice interpolation, the lattice parameters $x,y,z$ need to be Laurent monomials in terms of some $t\in \C$.
        The redundant signature property imposes an algebraic equation on $(x,y,z)$, which can be seen as a Laurent polynomial $p(t)=0$.
        In other words, we can interpolate a redundant signature if we can find a suitable root of a Laurent polynomial $p(t)$.
        To ensure that the redundant signature is non-singular, we need to find a ``good root'', meaning a root that lies outside the unit circle.
        Lemma~\ref{lm: roots on unit circle} provides a necessary condition for all the roots of a polynomial to lie on the unit circle. By showing that $p(t)$ does not satisfy this condition, we conclude that $p(t)$ has a ``good root''.
        Thus, we can interpolate a non-singular redundant signature.
        It follows from Theorem~\ref{thm:redundant} that $\plholant{\neq_2}{f}$ is \numP-hard.

        \item[4.]
        If the conformal lattice interpolation does not yield \numP-hardness or the disequality crossover $\mathcal{S'}$, even after rotating $f$, we show that $b,c$ and $d$ must satisfy certain relations.
        This either leads to $f$ is $\M$-transformable, or we can appeal to the combinatorial transformation to {\sc Even Coloring} again, which transforms the problem to Case I or to the planar six-vertex model (Lemma~\ref{lm: b=pm c, d=pm 1}).
        It's in the final case (via combinatorial transformation to the planar six-vertex model) where the new tractable type $(b=-c=\mathfrak{i}\tan(\frac{\beta\pi}{8}))\land (d=\pm1)$ is found.
    \end{enumerate}
    The above discussion culminates at Lemma~\ref{lm: symmetric equality dichotomy} and Theorem~\ref{thm: three equal pairs}, which gives a complexity classification on $\plholant{\neq_2}{f}$, where $\{(b,y),(c,z),(d,w)\}$ are three equal pairs.
    We then handle the case where $a = -x$, $b = -y$, and $c = -z$ by applying holographic transformation and binary modification to transform $f$ into the form corresponding to the three equal pairs case (Lemma~\ref{lm: symmetric disequality dichotomy}).
 
Before we can apply interpolation using M\"{o}bius transformations to complete Case III, we must handle a special degenerate setting. This occurs when the matrix that would define the M\"{o}bius transformation is singular; specifically, when $\det\left[\begin{smallmatrix}
    c & d \\
    w & z
    \end{smallmatrix}\right] =0$
    and
    $\det\left[\begin{smallmatrix}
    b & w \\
    d & y
    \end{smallmatrix}\right] =0$. 
    In this case (Lemma~\ref{lm: degenerate inner matrices}), we show that either $f\in \M$, or $\plholant{\neq_2}{f}$ is \numP-hard, or $f$ satisfies arrow symmetric equality $(a=x)\land(b= y)\land(c= z)$, which falls under the subclass already addressed above.

Having resolved this final exception, we are now fully equipped to complete the proof of Case III by Theorem~\ref{thm: all nonzero entries, not three equal or opposite pairs}. 
We show that either $f\in \M$, or $f$ is $\mathscr{A}$-transformable, or $\plholant{\neq_2}{f}$ is \numP-hard, assuming that $(b,y),(c,z),(d,w)$ are not three equal or opposite pairs. 
M\"{o}bius transformation will be used to interpolate arbitrary or specific binary signatures,
allowing us to transform $M(f)$ into a form with zero entries, reducing to previously handled cases and thereby establishing the dichotomy.

The rest of the paper is organized as follows.
In Sections~\ref{sec: at least two zero pairs} and~\ref{sec: At least a zero entry and at most a zero pair}, we handle Cases I and II respectively. 
In Section~\ref{sec: three or opposite pairs}, we handle the special setting where $(b,y),(c,z),(d,w)$ are three equal or opposite pairs of Case III.
In this section we present the new tractable class, and introduce the conformal lattice interpolation as a powerful tool to prove \numP-hardness. 
In Section~\ref{sec: dichotomy for all nonzero}, we finally complete the proof of Case III. 
We summarize the three cases and prove our main theorem in Section~\ref{sec: proof of main theorem}. 

\section{Case I: At Least Two Zero Pairs}\label{sec: at least two zero pairs}
In this section, we consider the case where at least two of the pairs in ${(b, y), (c, z), (d, w)}$ are zero pairs.
There are two subcases based on whether the inner pair $(d, w)$ is a zero pair.

\subsection{The inner pair $(d, w)$ is a zero pair}

We first consider the case that $(d, w)=(0, 0)$.  Then, by rotational symmetry,  we may assume that $f$ has a signature matrix of  the following form $M(f) = \left[\begin{smallmatrix} 1 & 0 & 0 & 0 \\ 0 & c & 0 & 0 \\ 0 & 0 & z & 0 \\ 0 & 0 & 0 & 1 \end{smallmatrix} \right].$

\begin{lemma}\label{lm: 1cz1}
Let $f$ be a 4-ary signature with the signature matrix $M(f) = \left[\begin{smallmatrix} 1 & 0 & 0 & 0 \\ 0 & c & 0 & 0 \\ 0 & 0 & z & 0 \\ 0 & 0 & 0 & 1 \end{smallmatrix} \right]$. Then $\plholant{\neq_2}{f}$ is \numP-hard unless $f \in \mathscr{P}, \mathscr{M}$-transformable, or $\mathscr{A}$-transformable, in which cases it is tractable. Explicitly, $f$ is \numP-hard unless (1) $cz=1$, (2) $(cz=-1) \wedge (c^4=z^4)$ or (3) $c = \pm z$.
\end{lemma}
\begin{proof}
Lemma 4.1 in~\cite{CaiF17} proves the reduction
$$\CSP^2(g) \leq_T^p \holant{\neq_2}{f}$$
where $g$ is the binary signature $(1,c,z,1)$. We observe their reduction is planar. Thus we have
$$\PlCSP^2(g) \leq_T^p \plholant{\neq_2}{f}.$$
In particular, we have
$$
\plholant{=_k}{g} \leq_T^p \plholant{\neq_2}{f}
$$
for any even integer $k$.
By Theorem~\ref{thm:spin system}, $\plholant{=_k}{g}$ is \numP-hard and thus $\plholant{\neq_2}{f}$ is \numP-hard unless $g \in \mathscr{P}$, or $(=_k|g)$ is $\mathscr{A}$-transformable, or  $(=_k| g)$ is $\mathscr{M}$-transformable.
If $g \in \mathscr{P}$, we have $cz = 1$ or $c=z=0$ by Lemma~\ref{lm: is binary product?}, in which cases $f \in \mathscr{P}$ by Lemma~\ref{lm:is product type?}.

If $(=_k|g)$ is $\mathscr{A}$-transformable, then $c^4=z^4$ and $cz=-1$ by Theorem~\ref{thm:spin system}.  By $c^4=z^4$ we know $c = \pm z$ or $c = \pm \mathfrak{i} z$. If $c = \pm z$, then $z^2 = \pm 1$ and thus $z = \pm 1$ or $z = \pm \mathfrak{i}$. Then by Lemma~\ref{lm:is 1cz1 affine?} and Lemma~\ref{lm:is binary affine?}, the signature $f$ is affine. If $c = \pm \mathfrak{i} z$, then $z^2 = \pm \mathfrak{i}$ and thus $z = \pm \alpha$ or $z = \pm \mathfrak{i} \alpha$. In all these cases, by a holographic transformation with the 2-by-2 invertible matrix $\left[\begin{smallmatrix} 1 & 0 \\ 0 & \sqrt{\alpha}\end{smallmatrix} \right]$, the signature $(\neq_2)$ is invariant under constant normalization and the signature $f$
becomes an affine signature
$f'$ with the signature matrix $M(f') = \left[\begin{smallmatrix} 1 & 0 & 0 & 0 \\ 0 & \alpha c & 0 & 0 \\ 0 & 0 & \alpha z & 0 \\ 0 & 0 & 0 & \alpha^2\end{smallmatrix} \right]$. By Lemma~\ref{lm:is 1cz1 affine?} and Lemma~\ref{lm:is binary affine?}, we know $f' \in \mathscr{A}$ since $(\alpha c) (\alpha z) (\alpha^2) = -\alpha^4 = 1 = \mathfrak{i}^4$.

If the pair $(=_k| g)$ is $\mathscr{M}$-transformable, then $c=\pm z$ by Theorem~\ref{thm:spin system}. If $c = z$, then $f \in \widehat{\mathscr{M}}$ by Lemma~\ref{lm: is matchgate hat2}. If $c=-z$, then $f$ is $\mathscr{M}$-transformable by $T = \left[\begin{smallmatrix}
    \alpha & 1 \\
    1 & \mathfrak{i}\alpha
\end{smallmatrix}\right] \in \operatorname{GL}_2(\mathbb{C})$. Indeed, since $T$ is symmetric, we have (up to a factor of 2)
$$M(T^{\otimes 4} f) = T^{\otimes 2} M(f) T^{\otimes 2} = \left[\begin{smallmatrix}
    0 & \alpha^3-\alpha c & \alpha^3+\alpha c & 0 \\
    \alpha^3-\alpha c & 0 & 0 & \alpha-\alpha^3 c \\
    \alpha^3+\alpha c & 0 & 0 & \alpha+\alpha^3 c \\
    0 & \alpha-\alpha^3 c & \alpha+\alpha^3 c & 0 \\
\end{smallmatrix} \right]$$
and

$$
M((\neq_2) (T^{-1})^{\otimes 2}) = T^{-1} \left[\begin{smallmatrix}
    0 & 1\\
    1 & 0
\end{smallmatrix} \right] T^{-1} = \left[\begin{smallmatrix}
    -\mathfrak{i}\alpha & 0 \\
    0 & -\alpha
\end{smallmatrix}\right]
$$
where both signatures have been transformed to matchgates by Lemma~\ref{lm:is matchgate}.
\end{proof}

\subsection{Two outer pairs $(b,y)$ and $(c, z)$ are zero pairs}

If the inner pair $(d, w)\neq (0, 0)$, then the two outer pairs must be zero pairs, i.e., $(b, y) = (c, z) = (0, 0)$, since $f$ has at least two zero pairs. 
By rotational symmetry, we may assume that
$M(f) = \left[\begin{smallmatrix} 1 & 0 & 0 & 0 \\ 0 & 0 & d & 0 \\ 0 & w & 0 & 0 \\ 0 & 0 & 0 & 1 \end{smallmatrix} \right]$.
To address this case, we first introduce a non-local reduction from a (not necessarily planar) bipartite $\GH$ problem (more specifically, a bipartite Ising model) to the planar eight-vertex model, under the special condition that $d = w$.

\begin{lemma}\label{lm:case d=w}
Let $f$ be a 4-ary signature with the signature matrix 
$M(f) = \left[\begin{smallmatrix} 1 & 0 & 0 & 0 \\ 0 & 0 & d & 0 \\ 0 & d & 0 & 0\\ 0 & 0 & 0 & 1 \end{smallmatrix} \right]$, then $\plholant{\neq_2}{f}$ is \numP-hard unless $d=0$, $d = \pm1$ or $d = \pm \mathfrak{i}$, in which cases $f\in \mathscr{A}$ and thus the problem is polynomial time solvable.
\end{lemma}
\begin{proof}
It can be directly verified that when $d=0$, $d = \pm1$ or $d = \pm \mathfrak{i}$, the signature $f \in \mathscr{A}$,
and thus the problem $\plholant{\neq_2}{f}$ is tractable.
In the following we assume that $d\neq 0,\pm 1,\pm\mathfrak{i}$.

This proof needs to go through a bipartite version of Ising model, which is a special case of the graph homomorphism problem. 
We use $\BiGH(g)$ to denote the problem, where the underlying graph is a bipartite graph $G=(V,E)$, and every vertex $v\in V$ corresponds to a Boolean variable, every $e\in E$ corresponds to a symmetric binary constraint $g=[1,d,1]$ with matrix 
$\left[\begin{smallmatrix}
    1 & d\\
    d & 1
\end{smallmatrix}\right]$.
We want to show the reduction
\begin{equation}\label{eq: BiGH reduction}
    \BiGH(g) \leq_T^p \plholant{\neq_2}{f}
\end{equation}
We prove  (\ref{eq: BiGH reduction}) in two steps:
First reduce $\BiGH(g)$ to $\plholant{\neq_2}{f,\mathcal{S}}$, and then interpolate $\mathcal{S}$ from $f$ on the RHS.

Let $\Omega$ be an instance for $\BiGH(g)$ with the underlying bipartite (multi)graph $G=(V,E)$, where $|V|=n$, $V= X\dot\cup Y$,  every  $v \in V$ is a Boolean variable, and every edge in the (multi)set $E$ is assigned  the binary constraint $g$.
Without loss of generality, we may assume that $G$ is connected.
We  list the vertices of $X= \{x_1, \ldots, x_{n_1}\}$
and $Y = \{y_1, \ldots, y_{n_2}\}$ where $n_1+n_2=n$, from top to bottom,
on the Left and on the Right respectively,
with all edges between Left and Right. Edge intersections are isolated and in the given order.
Multi-edges between the same $(x, y)$  are drawn as parallel edges.
Being bipartite, there are no self-loops.
The partition function value of $\Omega$ is


\begin{equation} \label{eq: sum of product of BiGH}
\BiGH(\Omega;g) = \sum\limits_{\sigma : V \rightarrow \{0,1\}} \prod\limits_{\{x, y\} \in E} g(\sigma(x),\sigma(y)).
\end{equation}

Now we construct a signature grid $\Omega'$ for the problem $\plholant{\neq_2}{f,\mathcal{S}}$ as follows (See Figure~\ref{fig: BiGH}).
First, replace every vertex $v\in V$ with a small circle $C_v$, such that no two circles overlap or intersect.
Second, for every edge connecting $x \in X$ and $y \in Y$,
we make a protrusion of $C_x$, called an  ``arm'',
along the edge $(x,y)$,
which is a ribbon-like band that makes a local extension of $C_{x}$ to intersect the circle $C_{y}$. This creates two intersection points, one {\it entry vertex} when the
arm enters $C_{y}$ on the
upper side of the arms, and one {\it exit vertex}  as it
leaves $C_{y}$.
As we traverse $C_y$ counterclockwise, the intersection point pairs on $C_y$
occur top to bottom sequentially, for all $x\in X$ such that $(x, y) \in E$, according to the order of $X$ (the same order as they were drawn as edges in  $G$).
Two arms cross if and only if their corresponding edges intersect in the original graph $G$.  Each such crossing creates 4 intersection points.
All such crossings occur at isolated points in the plane, corresponding to
the intersections in $G$,  outside all circles $C_y$.
Now, for each $x \in X$, we rename $C_x$  the cycle which is
the original circle modified by  all these protruding arms.

The drawing of these  $C_x$ and $C_y$  defines a 4-regular plane graph $G_1$, whose vertices are all the intersection points between the cycles $C_x$ and the circles $C_y$, as well as those between (protruding arms of) different cycles $C_x$.
We will define an instance  $\Omega'$ for $\plholant{\neq_2}{f,\mathcal{S}}$ as
follows.
Replace each edge of  $G_1$ by a path of length 2 (this is called a 2-stretch)
to obtain a bipartite $(2, 4)$-biregular plane graph $G'$.
Place a $(\neq_2)$ constraint at each vertex of degree 2.
Assign the signature $f$ to each entry vertex of $G'$. Assign $\mathcal{S}$ to all other
vertices of $G'$; these include all exit vertices and intersection points between
different arms.
Since $f$ and $\mathcal{S}$ are both rotationally symmetric, the
input order of variables of $f$ and $\mathcal{S}$ is immaterial.
We denote by $F$ the set of vertices labeled with $f$, by $S$ those labeled with $\mathcal{S}$, and by $W$ those labeled with $(\neq_2)$.
This completes the construction of the signature grid $\Omega'$,
with underlying graph $G'=(V',E')$.

\begin{figure}[!htbp]
\centering
\includegraphics[width=0.8\textwidth]{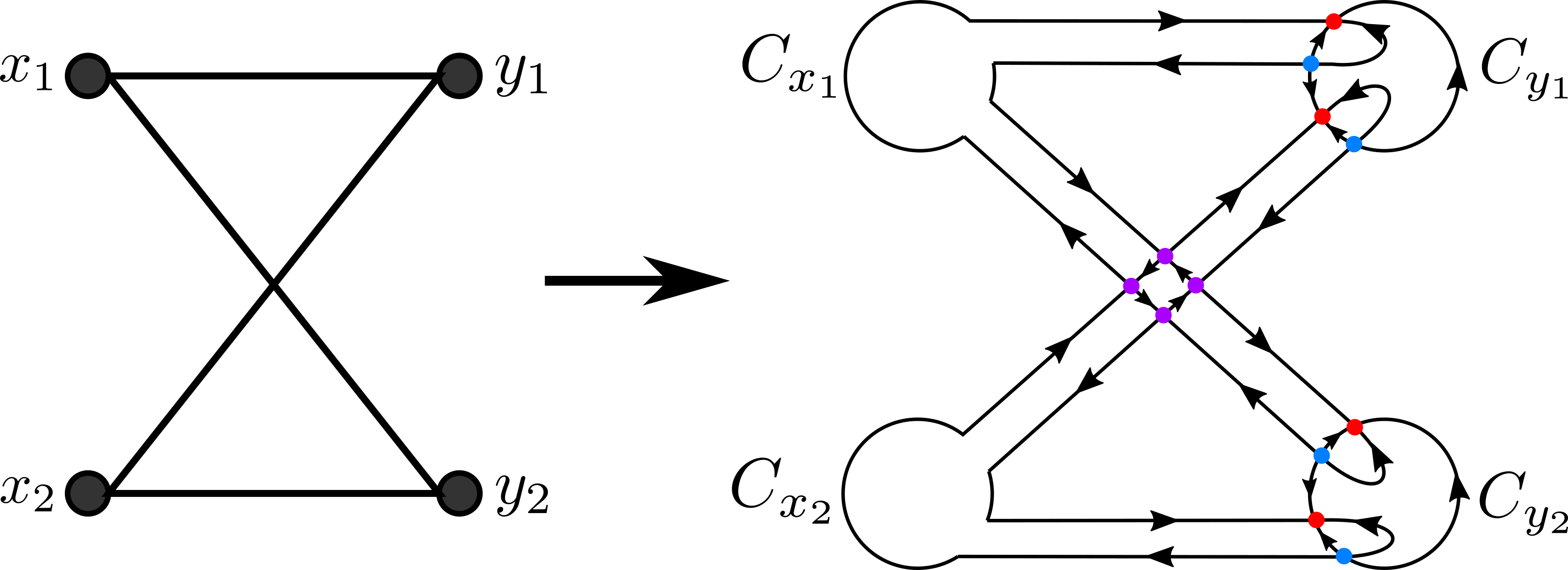}
\caption{Reduction from $\BiGH(g)$ to $\plholant{\neq_2}{f,\mathcal{S}}$.
LHS is an instance $\Omega$ for $\BiGH(g)$, where $x_1,x_2,y_1,y_2$ are Boolean variables and every edge is a binary constraint $g$.
RHS is an instance $\Omega'$ for $\plholant{\neq_2}{f}$. 
There is an implicit $\neq_2$ constraint on every edge.
The red nodes are entry vertices labeled by $f$, the blue nodes are exit vertices labeled by $\mathcal{S}$, and the purple nodes represent intersections between arms, also labeled by $\mathcal{S}$.
An incoming edge corresponds to an input of 0, while an outgoing edge corresponds to an input of 1.
The canonical assignment is illustrated by the orientation shown in the figure.}
\label{fig: BiGH}
\end{figure}

We call the pairs $(x_1, x_3)$ and $(x_2, x_4)$ {\it edge twins},
for $f$ and  $\mathcal{S}$ with input edges $x_1, x_2, x_3, x_4$ arranged counterclockwise.
Note that both $f$ and $\mathcal{S}$ have  support  $(x_1 = x_3) \land (x_2 = x_4)$. 
Thus, for any nonzero term in the sum
\begin{equation}\label{eq: Pl-Holant f,S}
\plholant{\Omega';\neq_2}{f,\mathcal{S}}=\sum_{\sigma:E'\rightarrow \{0, 1\}}\prod_{u\in F}f(\sigma \mid_{E_{(u)}})\prod_{v\in S}\mathcal{S}(\sigma\mid_{E_{(v)}})\prod_{w\in W}(\neq_2)(\sigma\mid_{E_{(w)}}),
\end{equation}
each edge twin must be assigned an equal value.
In $G'$, both $C_x$ and $C_y$  are cycles; for each
 we arbitrarily pick an edge, $e_x$ or $e_y$, to be its \emph{leader edge}.
Then there is a 1-1 correspondence between nonzero terms in (\ref{eq: Pl-Holant f,S})
and  assignments of
 leader edges $\sigma':\{e_v: v \in V \}\rightarrow \{0, 1\}$.
 This is because as we traverse each $C_x$ or $C_y$, we encounter (always  in pairs)
 an even number of
 degree 4 vertices from  $F$ or $S$, interspersed with $(\ne_2)$.
Thus, if $e = (u, w)$ is the leader edge  of a cycle ($u \in F \cup S$, $w \in W$)
taking value $0$ (resp. $1$), then all edges on the cycle must take values $(0, 0, 1, 1, \ldots, 0, 0, 1, 1)$ (resp. $(1, 1, 0, 0, \ldots, 1, 1, 0, 0)$) successively, starting with $e$ at $u$ in the direction away from $w$.
Then,
\begin{equation}\label{eq: Pl-Holant}
\plholant{\Omega';\neq_2}{f,\mathcal{S}}=\sum_{\sigma':\{e_v: v \in V \}\rightarrow \{0, 1\}}\prod_{u\in F}f(\widehat{\sigma'} \mid_{E_{(u)}}),
\end{equation}
where $\widehat{\sigma'}$ is the unique extension of $\sigma'$. We name
the leader edges $\{e_1, e_2, \ldots, e_n\}$.

Next, we show that $\BiGH(\Omega; g) = \plholant{\Omega'; \neq_2}{f, \mathcal{S}}$ by establishing a 1-1 correspondence between the terms in
(\ref{eq: sum of product of BiGH}) and (\ref{eq: Pl-Holant}).
We first show that there exists a canonical assignment to
the leader edges $\{e_1, e_2, \ldots, e_n\}$ such that its evaluation
in (\ref{eq: Pl-Holant}) equals that in (\ref{eq: sum of product of BiGH}) when $x_1 = x_2 = \cdots = x_n = 0$.
Then, we show that this equality is maintained after flipping an
arbitrary set of variables with the corresponding changes to the assignment of leader edges.
\begin{enumerate}
    \item
    Existence of canonical assignment:
    For any $x\in X$, start at any protruding arm of $C_x$ which enters  $C_y$ (for some $y \in Y$),  we assign the value $0$ to the edge on $C_x$ that is incident to its entry vertex; after passing through a $(\ne_2)$, the incident edge
    to  its exit vertex has  value $1$.
    On its way back if it intersects another arm, its incident edge
    has value 0 since the cycle $C_x$
    passed through another $(\ne_2)$, and leaves this arm-intersection (after two
    intersection points) again with
    value 1. There can be multiple arm-intersections and the same reasoning applies. If $C_x$ has another
    protruding arm that enters another $C_{y'}$, the same reasoning applies.
    This gives a valid assignment to all edges on $C_x$
    including its leader edge.
    Similarly, when we traverse $C_y$ counterclockwise,
     we assign the value $0$ to the edge on $C_y$ incident to
    all entry vertices (for various $x\in X$  with $(x,y) \in E$)
  and the value 1 at all exit vertices. There are an even number of these in total
  on $C_y$. This also defines a value to the leader edge of $C_y$.

    The above defines  the canonical assignment, which contributes a term $1$ in the sum of product (\ref{eq: Pl-Holant}), corresponding to the term where $x_1=x_2=\cdots=x_n=0$ in the sum of product (\ref{eq: sum of product of BiGH}).

    \item
    Flipping variables:
    Now, we flip the assignments to an arbitrary set of variables $\{e_{i_1}, e_{i_2}, \ldots, e_{i_k}\}$ in the canonical assignment, along with the corresponding variables $\{x_{i_1}, x_{i_2}, \ldots, x_{i_k}\}$ in the all-zero assignment.
    Then all the variables on $C_{i_1},C_{i_2},\ldots,C_{i_k}$ are flipped.
    Evaluations of $\mathcal{S}$  are unchanged.
    Consider two arbitrary cycles $C_{x}$ and $C_{y}$ for $x\in X$ and $y\in Y$:
    \begin{itemize}
        \item
        If neither variables on  $C_{x}$ nor on $C_{y}$ are flipped, then the signature $f$ at their intersection $C_{x} \cap C_{y}$ remains unaffected and evaluates to $f_{0000} = 1$.
        \item
        If both variables on  $C_{x}$ and on $C_{y}$ are flipped, then all inputs to $f$ are flipped, and it evaluates to $f_{1111} = 1$.
        \item
        If only one of $C_{x}$ or $C_{y}$ has its variables flipped, then $f$ has exactly one pair of edge twin inputs flipped and evaluates to either $f_{0110} = d$ or $f_{1001} = d$.
    \end{itemize}
    In the above three cases, $g(\sigma(x),\sigma(y))$ evaluates to $g_{00}=1$, $g_{11}=1$, and $g_{01}=d$ or $g_{10}=d$, respectively.
    Thus, after flipping variables, the evaluation for the assignment of leader edges $\{e_1, e_2, \ldots, e_n\}$ in (\ref{eq: Pl-Holant}) remains the same for the assignment of $\{x_1, x_2, \ldots, x_n\}$ in (\ref{eq: sum of product of BiGH}).
\end{enumerate}
Therefore, we have shown that
$
\BiGH(\Omega;g)=\plholant{\Omega';\neq_2}{f,\mathcal{S}},
$
and it follows that
$$ \BiGH(g) \leq_T^p \plholant{\neq_2}{f,\mathcal{S}}.$$

Next, we interpolate $\mathcal{S}$ from $f$ on RHS.
If $d\neq 0 $ is not a root of unity, then by Lemma~\ref{lm:interpolate S} we are done. 
Otherwise, consider the ``boxtimes" gadget \tikz[baseline=0.2ex]{
  \draw[thick] (0,0) rectangle (0.3,0.3);
  \draw[thick] (-0.05,-0.05) -- (0.35,0.35);
  \draw[thick] (0.35,-0.05) -- (-0.05,0.35);
}, where every vertex is labeled by $f$, and in the middle of every internal edge is assigned a $(\neq_2)$ constraint. 
Note that the input order of $f$ is immaterial since $f$ is rotationally symmetric.
This gives the 4-ary signature 
$f_\boxtimes$ with $M(f_\boxtimes) = \left[\begin{smallmatrix} 2d^2 & 0 & 0 & 0 \\ 0 & 0 & d+d^5 & 0 \\ 0 & d+ d^5 & 0 & 0 \\ 0 & 0 & 0 & 2d^2 \end{smallmatrix} \right]$.
When $d$ is a root of unity, we have $\frac{d+d^5}{2d^2}$ is not a root of unity unless $d = \pm 1 $ or $d = \pm \mathfrak{i}$. Indeed,  $\left| \frac{d+d^5}{2d^2} \right| =1$ iff $|1+d^4| = 2$. Since $d$ is a root of unity, we have $|d| = 1$ and thus $|1+d^4| = 2$ iff $d^4 = 1$.
This finishes the proof of the reduction~(\ref{eq: BiGH reduction}).

Next, we get \numP-hardness by reducing $\CSP$ problems to $\BiGH(g)$, and then appeal to Theorem~\ref{thm: CSP dichotomy}.
By a 2-stretching, we have the
 reduction $\CSP(g^2)\le_T^p\BiGH(g)$, where
$M(g^2)=\left[\begin{smallmatrix}
    1+d^2 & 2d\\
    2d & 1+d^2
\end{smallmatrix}\right]
=\left[\begin{smallmatrix}
    1 & d\\
    d & 1
\end{smallmatrix}\right]
\cdot\left[\begin{smallmatrix}
    1 & d\\
    d & 1
\end{smallmatrix}\right]
=M(g)^2.
$
By Lemma~\ref{lm: is binary product?}, if $g^2\in \mathscr{P}$, then $(1+d^2)^2-4d^2=0$, or $d=0$, or $1+d^2=0$, contradicting $d\neq 0,\pm1,\pm \mathfrak{i}$.
So $g^2\notin\mathscr{P}$.
By Lemma~\ref{lm:is binary affine?}, if $g^2\in \mathscr{A}$ then $1+d^2=\mathfrak{i}^\beta\cdot2d$ for some $\beta\in\{0,1,2,3\}$.
We either get a contradiction, or
$d=\mathfrak{i}(\epsilon_1+\epsilon_2{\sqrt{2}})$, where $\epsilon_1,\epsilon_2=\pm1$.
Then by Theorem~\ref{thm: CSP dichotomy}, $\CSP(g^2)$ is \numP-hard, and thus $\BiGH(g)$ is \numP-hard, unless $d=\mathfrak{i}(\epsilon_1+\epsilon_2{\sqrt{2}})$, where $\epsilon_1,\epsilon_2=\pm1$.

Now assume that $d=\mathfrak{i}(\epsilon_1+\epsilon_2{\sqrt{2}})$, for some $\epsilon_1,\epsilon_2=\pm1$.
 We prove the reduction $\CSP(h)\leq_T^p$ $\BiGH(g)$, where
$M(h)=\left[\begin{smallmatrix}
    1+d^4 & 2d^2\\
    2d^2 & 1+d^4
\end{smallmatrix}\right]$.
Start from an instance $\Omega_3$ for $\CSP(h)$ with the underlying graph $H$.
We stretch every edge in $H$ by 2 and get the graph $H'$.
Then thicken the edges in $H'$ by 2, i.e., add a multiple edge to every adjacent vertices in $H'$.
Finally, label every edge by $g$.
This gives us an instance $\Omega_3'$ for $\BiGH(g)$.
Note that  $M(h)=\left[\begin{smallmatrix}
    1 & d^2\\
    d^2 & 1
\end{smallmatrix}\right]
\left[\begin{smallmatrix}
    1 & d^2\\
    d^2 & 1
\end{smallmatrix}\right]
$.
Then we have $\CSP(\Omega_3;h)=\BiGH(\Omega_3';g)$, and
the reduction $\CSP(h)\leq_T^p$ $\BiGH(g)$ follows.
Again by Lemma~\ref{lm: is binary product?} and Lemma~\ref{lm:is binary affine?}, we can show that
$h\notin \mathscr{A}\cup\mathscr{P}$  unless $d^2=0,\pm1, \pm \mathfrak{i}$ or $d^2=\mathfrak{i}(\epsilon'_1+\epsilon'_2{\sqrt{2}})$, where $\epsilon'_1,\epsilon'_2\in \{0,1\}$.
This contradicts $d=\mathfrak{i}(\epsilon_1+\epsilon_2{\sqrt{2}})$ for some $\epsilon_1,\epsilon_2=\pm1$.
Therefore, $\CSP(h)$ is \numP-hard, and thus $\BiGH(g)$ is \numP-hard by Theorem~\ref{thm: CSP dichotomy}.
\end{proof}

\begin{remark}\label{remark:sec4-after-lm}
    In the above proof, the bipartiteness of the $\GH$ problem guarantees the existence of a canonical assignment.
    We can show that 
    if the underlying graph of $\GH(g)$ is not bipartite,
    then no canonical assignment exists, regardless of how the cycles are drawn or the $f$’s are labeled.
\end{remark}

\begin{lemma} \label{lm: 1dw1}
Let $f$ be a 4-ary signature with the signature matrix $M(f) = \left[\begin{smallmatrix} 1 & 0 & 0 & 0 \\ 0 & 0 & d & 0 \\ 0 & w & 0 & 0 \\ 0 & 0 & 0 & 1 \end{smallmatrix} \right]$, then $\plholant{\neq_2}{f}$ is \numP-hard unless $f\in \mathscr{P}$, $f\in \mathscr{A}$ or $f \in \mathscr{M}$, in which cases $\plholant{\neq_2}{f}$ is polynomial time solvable. More explicitly, $\plholant{\neq_2}{f}$ is \numP-hard unless $d=w=0$ or $dw = \pm 1$.
\end{lemma}
\begin{proof}
If $d=w=0$, then $f\in \mathscr{A}$.
If $dw=1$, then $f\in \mathscr{P}$ by Lemma~\ref{lm:is product type?}.
If $dw=-1$, then $f\in \mathscr{M}$ by Lemma~\ref{lm:is matchgate}.
In above three cases, the problem $\plholant{\neq_2}{f}$ is tractable.
In the following, we assume it is not the case that $d=w=0$ or $dw=\pm1$.

Now suppose there is exactly one zero in $(d,w)$.
We can realize 
$M(g) = M(f) \cdot N \cdot M(f) = \left[\begin{smallmatrix} 0 & 0 & 0 & 1 \\ 0 & 0 & d^2 & 0 \\ 0 & w^2 & 0 & 0 \\ 1 & 0 & 0 & 0 \end{smallmatrix} \right]$.
After rotation we have
$M(g^{\frac{\pi}{2}}) = \left[\begin{smallmatrix} 0 & 0 & 0 & 0 \\ 0 & 1 & w^2 & 0 \\ 0 & d^2 & 1 & 0 \\ 0 & 0 & 0 & 0 \end{smallmatrix} \right]$.
By Corollary~\ref{cor:exactly 3 nonzero entries}, $\plholant{\neq_2}{g^\frac{\pi}{2}}$ is \numP-hard unless $g^\frac{\pi}{2} \in \mathscr{M} \cup \widehat{\mathscr{M}}$. If $g^\frac{\pi}{2} \in \mathscr{M}$, then by Lemma~\ref{lm:is matchgate} we know $d^2w^2 = 1$, contradiction.
If $g^\frac{\pi}{2} \in \widehat{\mathscr{M}}$, by Lemma~\ref{lm: is matchgate hat} we know $d^2=w^2$, contradiction. In summary, we can now assume $dw\neq 0$.

By adding a self loop using $\neq_2$ to variables $x_1,x_2$ of $f$, we can get $(0,1,\frac{w}{d},0)^T$ on RHS.
Doing a binary modification using $(0,1,\frac{w}{d},0)^T$ to the variable $x_4$ of $f$, we get the signature $\left[\begin{smallmatrix} 1 & 0 & 0 & 0 \\ 0 & 0 & w & 0 \\ 0 & w & 0 & 0 \\ 0 & 0 & 0 & \frac{w}{d} \end{smallmatrix} \right]$.
By Lemma~\ref{lm:let a=x} and normalization, we have $$\plholant{\neq_2}{\hat{f}} \leq_T^p \plholant{\neq_2}{f}$$ where $M(f) = \left[\begin{smallmatrix} 1 & 0 & 0 & 0 \\ 0 & 0 & \sqrt{dw} & 0 \\ 0 & \sqrt{dw} & 0 & 0 \\ 0 & 0 & 0 & 1 \end{smallmatrix} \right]$. By Lemma~\ref{lm:case d=w} the problem $\plholant{\neq_2}{\hat{f}}$ is \numP-hard and thus the problem $\plholant{\neq_2}{f}$ is \numP-hard unless $\sqrt{dw} = \pm 1$ or $\sqrt{dw} = \pm \mathfrak{i}$, in which cases $dw = \pm 1$ and the problem $\plholant{\neq_2}{f}$ is tractable.
\end{proof}
\section{Case II: At Least One Zero Entry and at Most One Zero Pair}
\label{sec: At least a zero entry and at most a zero pair}
In this section, we consider the case where some entries among ${b, y, c, z, d, w}$ are zero, with the restriction that at most one of the pairs $(b, y)$, $(c, z)$, and $(d, w)$ is a zero pair. 
We show that $\plholant{\neq_2}{f}$ is \numP-hard unless $f\in \mathscr{M}$ (Theorem~\ref{thm: at least one zero at most one zero pair}).
The proof of Theorem~\ref{thm: at least one zero at most one zero pair} is divided into subcases based on the number $\sf N$ of zero entries in $M(f)$ and their positions. 
Note that $\sf N \leq 4$, since there can be at most one zero pair in $f$.
\begin{enumerate}
   \item $\sf N = 1$: 
   There is exactly one zero entry, which may appear in either an outer pair (Lemma~\ref{lm: 1bcdwy1}) or an inner pair (Lemma~\ref{lm: 1bcwzy1}).
\item $\sf N = 2$: 
There are two zero entries, which may lie in two distinct pairs (Lemmas~\ref{lm: 1bcdw1},~\ref{lm: 1bcdz1} and~\ref{lm: 1bcwz1}) or in the same pair (Lemma~\ref{lm:1cdwz1} and~\ref{lm:1bycz1}).
\item $\sf N = 3$: 
There are three zero entries. This may include one zero pair (Lemmas~\ref{lm: 1cdw1}, \ref{lm: 1cwz1}, \ref{lm: 1cdz1}, and~\ref{lm: 1bcz1}), or the zeros may occur in three different pairs (Lemma~\ref{lm: 1bcdy1, by=0}).
\item $\sf N = 4$: 
There are four zero entries. 
In this case, there must be a zero pair, which could be either an outer pair (Lemma~\ref{lm:1bd1}) or an inner pair (Lemma~\ref{lm: 1bc1}).
\end{enumerate}

Although the proof involves a case-by-case analysis, the overall strategy for establishing \#P-hardness follows a common set of ideas. 
In most cases (e.g., Lemmas~\ref{lm: 1bcdwy1}, \ref{lm: 1bcwzy1}), we prove \#P-hardness by doing gadget construction to realize a signature corresponding to the planar six-vertex model and then applying its classification result (Theorem~\ref{thm:planar six-vertex}).
In a few other cases (e.g., Lemmas~\ref{lm: 1bcdz1}, \ref{lm: 1bcz1}), we realize a signature of the eight-vertex model with at least two zero pairs, which falls under Case I and is handled accordingly.

In some cases, gadget construction alone is not sufficient to establish a complexity dichotomy, and we must appeal to polynomial interpolation.
In Lemma~\ref{lm: 1bcdw1}, we interpolate the crossover $\mathcal{S}$. 
This enables a reduction from the eight-vertex model on general graphs to its planar counterpart. We then establish \#P-hardness by invoking the classification result for the eight-vertex model on general graphs (Theorem~\ref{thm: general eight-vertex}).
In Lemmas~\ref{lm: 1bcdwy1}, \ref{lm: 1bcwzy1}, \ref{lm: 1bcdw1}, we also employ interpolation techniques based on M\"obius transformation (Corollary~\ref{cor:get any binary when exactly three nonzero}) to realize arbitrary binary signatures, which enable us to do binary modification.

\subsection{${\sf N}=1$: one zero entry}\label{sec: N=1}

If ${\sf N}=1$, by rotation symmetry, it suffices to consider the following two cases:
(i) $M(f) = \left[\begin{smallmatrix} 1 & 0 & 0 & b \\ 0 & c & d & 0 \\ 0 & w & 0 & 0 \\ y & 0 & 0 & 1\end{smallmatrix}\right]$ with $bcdyw \neq0$, 
(ii) $M(f) = \left[\begin{smallmatrix} 1 & 0 & 0 & b \\ 0 & c & 0 & 0 \\ 0 & w & z & 0 \\ y & 0 & 0 & 1\end{smallmatrix}\right]$ with $bcyzw \neq0$.

\begin{lemma}\label{lm: 1bcdwy1}
    If $M(f) = \left[\begin{smallmatrix} 1 & 0 & 0 & b \\ 0 & c & d & 0 \\ 0 & w & 0 & 0 \\ y & 0 & 0 & 1\end{smallmatrix}\right]$ with $bcdyw \neq0$, then \plholant{\neq_2}{f} is \numP-hard unless $f\in \mathscr{M}$, or more explicitly $1-by=-dw$, in which case $\plholant{\neq_2}{f}$ can be computed in polynomial time.
\end{lemma}
\begin{proof}
    By Corollary~\ref{cor:get any binary when exactly three nonzero}, we can interpolate a binary signature $(0,1,t,0)^T$ for any $t \in \mathbb{C}$ on RHS. Rotate $f$ by $\frac{\pi}{2}$, $\pi$ and $\frac{3\pi}{2}$ to get signatures $f^\frac{\pi}{2}$, $f^\pi$ and $f^\frac{3\pi}{2}$ respectively, whose signature matrices are $M(f^\frac{\pi}{2}) = \left[\begin{smallmatrix} 
    1 & 0 & 0 & 0 \\ 
    0 & b & w & 0 \\ 
    0 & d & y & 0 \\ 
    c & 0 & 0 & 1
    \end{smallmatrix}\right]$, 
    $M(f^\pi) = \left[\begin{smallmatrix} 
    1 & 0 & 0 & y \\ 
    0 & 0 & d & 0 \\ 
    0 & w & c & 0 \\ 
    b & 0 & 0 & 1
    \end{smallmatrix}\right]$ and 
    $M(f^\frac{3\pi}{2}) = \left[\begin{smallmatrix} 
    1 & 0 & 0 & c \\ 
    0 & y & w & 0 \\ 
    0 & d & b & 0 \\ 
    0 & 0 & 0 & 1
    \end{smallmatrix}\right]$. For a binary signature $(0,1,t,0)^T$, we can do the binary modification to $f^\frac{\pi}{2}$ on variable $x_4$ to get signature $f^{\frac{\pi}{2}}_t$ whose signature matrix is $M(f^{\frac{\pi}{2}}_t) = \left[\begin{smallmatrix} 
    1 & 0 & 0 & 0 \\ 
    0 & b & wt & 0 \\ 
    0 & d & yt & 0 \\ 
    c & 0 & 0 & t\end{smallmatrix}\right]$. 
    We are going to specify a particular $t \in\mathbb{C}$. Connect $f^{\frac{\pi}{2}}_t$ with $f^{\frac{3\pi}{2}}$ via $(\neq_2)^{\otimes2}$ to get a signature $\Tilde{f}$ on RHS whose signature matrix is $M(\Tilde{f}) = \left[\begin{smallmatrix} 
    0 & 0 & 0 & 1 \\ 
    0 & bd+twy & b^2+tw^2 & 0 \\ 
    0 & d^2+ty^2 & bd+twy & 0 \\ 
    t & 0 & 0 & c+ct\end{smallmatrix}\right]$. By Lemma~\ref{lm:let a=x}, the complexity of \plholant{\neq_2}{\Tilde{f}} is the same as the planar six-vertex model where the bottom-right entry is 0, which we denote to be $\Tilde{f}_0$. 
    We can choose $t\in \mathbb{C}$ such that $t\neq 0$, $bd+twy \neq 0$, $b^2+tw^2\neq 0$, $d^2+ty^2\neq 0$ and $t\neq -1$.
    Such $t$ always exists, since every equation only precludes one point in the complex plane (recall $bcdyw \neq 0$). 
    Furthermore, for such $t$, we have $\Tilde{f}_0 \notin \mathscr{P} \cup \mathscr{A} \cup \mathscr{\widehat{M}}$ 
    (a direct calculation shows that $\Tilde{f}_0 \in \widehat{\mathscr{M}}$  only if $t=-1$). 
    Therefore by Theorem~\ref{thm:planar six-vertex}, \plholant{\neq_2}{\Tilde{f}_0} is \numP-hard and therefore \plholant{\neq_2}{f} is \numP-hard unless $\Tilde{f}_0 \in \mathscr{M}$,
    a condition in this case is equivalent to
    $(bd+twy)^2-(b^2+tw^2)(d^2+ty^2)=-t$, which simplifies to  $(by-dw)^2=1$. 
    If $by-dw=1$, then $f \in \mathscr{M}$ and therefore \plholant{\neq_2}{f} can be computed in polynomial time. 
    Thus it remains to consider $dw-by=1$. 

    We now do the binary modification to $f^{\pi}$ on variable $x_4$ to get the signature $f^\pi_t$ on RHS whose signature matrix is $M(f^\pi_t) = \left[\begin{smallmatrix} 
    1 & 0 & 0 & yt \\ 
    0 & 0 & dt & 0 \\ 
    0 & w & ct & 0 \\ 
    b & 0 & 0 & t\end{smallmatrix}\right]$. 
    Again, we are going to specify some $t\in \mathbb{C}$. 
    Connect $f^\pi_t$ with $f$ via $(\neq_2)^{\otimes2}$ to get a signature $\hat{f}$ on RHS whose signature matrix is 
    $M(\hat{f}) = \left[\begin{smallmatrix} 
    y+ty & 0 & 0 & 1+bty \\ 
    0 & cdt & d^2t & 0 \\ 
    0 & c^2t+w^2 & cdt & 0 \\ 
    t+by & 0 & 0 & b+bt\end{smallmatrix}\right]$. 
    Let $t=-1$, we have $M(\hat{f}_{-1}) = \left[\begin{smallmatrix} 
    0 & 0 & 0 & 1-by \\ 
    0 & -cd & -d^2 & 0 \\ 
    0 & -c^2+w^2 & -cd & 0 \\ 
    by-1 & 0 & 0 & 0\end{smallmatrix}\right]$ which defines a planar six-vertex model. 
    Assume for now that $by\neq 1$. 
    Then, by the size of support we know $\hat{f}_{-1} \notin \mathscr{P} \cup \mathscr{A}$ and by a direct calculation we know $\hat{f}_{-1} \notin \widehat{\mathscr{M}}$. 
    Therefore, by Theorem~\ref{thm:planar six-vertex}, the problem is \numP-hard unless $\hat{f}_{-1} \in \mathscr{M}$,
    a condition in this case is equivalent to
    $c^2d^2+d^2(w^2-c^2) = (1-by)^2$, which simplifies to $d^2w^2 = (1-by)^2$.
    Since $dw = 1+by$, this will force $by=0$, a contradiction. 
    Thus it remains to consider the case when $by=1$ and $dw=2$.

    If we set $t=0$ in $\hat{f}$, we get a signature $\hat{f}_{0}$ whose signature matrix is $M(\hat{f}_{0}) = \left[\begin{smallmatrix} 
    y & 0 & 0 & 1 \\ 
    0 & 0 & 0 & 0 \\ 
    0 & w^2 & 0 & 0 \\ 
    1 & 0 & 0 & b\end{smallmatrix}\right]$.
    We are going to use $\hat{f}_{0}$ to interpolate $(=_2)^{\otimes2}$ on RHS which will be useful later. 
    By Lemma~\ref{lm:let a=x}, $\plholant{\neq_2}{\hat{f}_0}\equiv_T \plholant{\neq_2}{h}$, where $h$ has the signature matrix
    $M(h) = \left[\begin{smallmatrix} 
    1 & 0 & 0 & 1 \\ 
    0 & 0 & 0 & 0 \\ 
    0 & w^2 & 0 & 0 \\ 
    1 & 0 & 0 & 1\end{smallmatrix}\right].$
    Rotate this signature by $\frac{\pi}{2}$ we have $h^{\frac{\pi}{2}}$ on RHS whose signature matrix is 
    $M(h^{\frac{\pi}{2}}) = \left[\begin{smallmatrix} 
    1 & 0 & 0 & 0 \\ 
    0 & 1 & w^2 & 0 \\ 
    0 & 0 & 1 & 0 \\ 
    0 & 0 & 0 & 1\end{smallmatrix}\right]$. 
    Compute that $M(h)NM(h^{\frac{\pi}{2}})NM(h) = \left[\begin{smallmatrix} 
    2 & 0 & 0 & 2 \\ 
    0 & 0 & 0 & 0 \\ 
    0 & 0 & 0 & 0 \\ 
    2 & 0 & 0 & 2\end{smallmatrix}\right]$ and therefore after rotation (ignoring the global factor 2) we have $(=_2)^{\otimes2}= \left[\begin{smallmatrix} 
    1 & 0 & 0 & 0 \\ 
    0 & 1 & 0 & 0 \\ 
    0 & 0 & 1 & 0 \\ 
    0 & 0 & 0 & 1\end{smallmatrix}\right]$ on RHS. 
    Connecting $f$ and $(=_2)^{\otimes2}$ via $(\neq_2)^{\otimes2}$, we get the signature $g$ with the signature matrix $M(g)=M(f)N\left[\begin{smallmatrix} 
    1 & 0 & 0 & 0 \\ 
    0 & 1 & 0 & 0 \\ 
    0 & 0 & 1 & 0 \\ 
    0 & 0 & 0 & 1\end{smallmatrix}\right]
    =
    \left[\begin{smallmatrix} 
    b & 0 & 0 & 1 \\ 
    0 & d & c & 0 \\ 
    0 & 0 & w & 0 \\ 
    1 & 0 & 0 & y\end{smallmatrix}\right]$. 
    Rotate $g$ by $\frac{\pi}{2}$ we get the signature $g^{\frac{\pi}{2}}$, connecting $g^{\frac{\pi}{2}}$ and $(=_2)^{\otimes2}$ via $(\neq_2)^{\otimes2}$, we get a signature on RHS with signature matrix $\left[\begin{smallmatrix} 
    w & 0 & 0 & b \\ 
    0 & 0 & 1 & 0 \\ 
    0 & 1 & c & 0 \\ 
    y & 0 & 0 & d\end{smallmatrix}\right]$. Finally, rotate this signature by $\pi$ and then according to Lemma~\ref{lm:let a=x}, we have a planar eight-vertex model whose RHS is given by $\left[\begin{smallmatrix} 
    1 & 0 & 0 & \frac{y}{\sqrt{2}} \\ 
    0 & \frac{c}{\sqrt{2}} & \frac{1}{\sqrt{2}} & 0 \\ 
    0 & \frac{1}{\sqrt{2}} & 0 & 0 \\ 
    \frac{b}{\sqrt{2}} & 0 & 0 & 1
    \end{smallmatrix}\right]$ which is of the form $\left[\begin{smallmatrix} 1 & 0 & 0 & b' \\ 0 & c' & d' & 0 \\ 0 & w' & 0 & 0 \\ y' & 0 & 0 & 1\end{smallmatrix}\right]$. However, the $b',y',d',w'$ satisfy $b'y'-d'w' = 0$ instead of $\pm 1$. Therefore by the previous argument we know it is \numP-hard. This shows the hardness for the case $dw-by=1$ and thus completes the proof.
\end{proof}

\begin{lemma}\label{lm: 1bcwzy1}
    If $M(f) = \left[\begin{smallmatrix} 1 & 0 & 0 & b \\ 0 & c & 0 & 0 \\ 0 & w & z & 0 \\ y & 0 & 0 & 1\end{smallmatrix}\right]$ with $bcwyz \neq0$, then \plholant{\neq_2}{f} is \numP-hard unless $f\in \mathscr{M}$, or more explicitly $1-by=cz$, in which cases $\plholant{\neq_2}{f}$ can be computed in polynomial time.
\end{lemma}
\begin{proof}
    By Corollary~\ref{cor:get any binary when exactly three nonzero}, we can interpolate the binary signature $(0,1,t,0)^T$ for any $t \in \mathbb{C}$ on RHS. Rotate $f$ by $\frac{\pi}{2}$, $\pi$ and $\frac{3\pi}{2}$ to get signatures $f^{\frac{\pi}{2}}$, $f^\pi$ and $f^{\frac{3\pi}{2}}$ respectively, whose signature matrices are $M(f^{\frac{\pi}{2}}) = \left[\begin{smallmatrix} 
    1 & 0 & 0 & z \\ 
    0 & b & w & 0 \\ 
    0 & 0 & y & 0 \\ 
    c & 0 & 0 & 1\end{smallmatrix}\right]$, 
    $M(f^\pi) = \left[\begin{smallmatrix} 
    1 & 0 & 0 & y \\ 
    0 & z & 0 & 0 \\ 
    0 & w & c & 0 \\ 
    b & 0 & 0 & 1
    \end{smallmatrix}\right]$ and 
    $M(f^{\frac{3\pi}{2}}) = \left[\begin{smallmatrix} 
    1 & 0 & 0 & c \\ 
    0 & y & w & 0 \\ 
    0 & 0 & b & 0 \\ 
    z & 0 & 0 & 1
    \end{smallmatrix}\right]$.
    For a binary signature $(0,1,t,0)^T$, we can do the binary modification to $f^{\frac{\pi}{2}}$ on variable $x_4$ to get signature $f^\frac{\pi}{2}_t$ whose signature matrix is $M(f^\frac{\pi}{2}_t) = \left[\begin{smallmatrix} 
    1 & 0 & 0 & zt \\ 
    0 & b & wt & 0 \\ 
    0 & 0 & yt & 0 \\ 
    c & 0 & 0 & t\end{smallmatrix}\right]$. 
    Connect $f^\frac{\pi}{2}_t$ with $f^\frac{3\pi}{2}$ via $(\neq_2)^{\otimes2}$ we get a signature $\Tilde{f}$ on RHS whose signature matrix is $M(\Tilde{f}) 
    = M(f^\frac{\pi}{2}_t) N M(f^\frac{3\pi}{2})= \left[\begin{smallmatrix} 
    z(t+1) & 0 & 0 & 1+czt\\ 
    0 & twy & b^2+w^2t & 0 \\ 
    0 & ty^2 & twy & 0 \\ 
    t+cz & 0 & 0 & c(t+1)
    \end{smallmatrix}\right]$.
    We first consider the case when $cz=1$. 
    Set $t=0$ in $\Tilde{f}$ and apply Lemma~\ref{lm:let a=x} to get the signature $g$ on RHS whose signature matrix is  $M(\tilde{f}_0)=\left[\begin{smallmatrix} 
    1 & 0 & 0 & 1\\ 
    0 & 0 & b^2 & 0 \\ 
    0 & 0 & 0 & 0 \\ 
    1 & 0 & 0 & 1\end{smallmatrix}\right]$. 
    Connect $\tilde{f}_0,\tilde{f}_0^{\frac{\pi}{2}}$ and $\tilde{f}_0$ via $(\neq_2)^{\otimes 2}$, 
    we get a signature with the signature matrix 
    $ 
    M(\tilde{f}_0)NM(\tilde{f}_0^\frac{\pi}{2})NM(\tilde{f}_0)=
    \left[\begin{smallmatrix} 
    2 & 0 & 0 & 2\\ 
    0 & 0 & 0 & 0 \\ 
    0 & 0 & 0 & 0 \\ 
    2 & 0 & 0 & 2\end{smallmatrix}\right]$, 
    which is $(=_2)^{\otimes2}$ ignoring the factor 2 on RHS. 
    Now, if we connect $f$ and $(=_2)^{\otimes2}$ via $(\neq_2)^{\otimes2}$ and rotate by $\pi$, we will get a signature on RHS with the signature matrix $\left[\begin{smallmatrix}
    b & 0 & 0 & 1 \\
    0 & w & c & 0 \\
    0 & z & 0 & 0 \\
    1 & 0 & 0 & y\end{smallmatrix}\right]$. Under Lemma~\ref{lm:let a=x} this signature is of the form in Lemma~\ref{lm: 1bcdwy1} and it is \numP-hard unless $by-1=-cz$, which is impossible since $cz = 1$ and $by \neq 0$. We conclude our current progress as that $\plholant{\neq_2}{f}$ is \numP-hard if $cz = 1$.

    We now assume $cz \neq 1$. 
    Set $t = -1$ in $\Tilde{f}$, we get the signature  $\Tilde{f}_{-1}$ with the signature matrix 
    $M(\tilde{f}_{-1})=
    \left[\begin{smallmatrix} 
    0 & 0 & 0 & 1-cz\\ 
    0 & -wy & b^2-w^2 & 0 \\ 
    0 & -y^2 & -wy & 0 \\ 
    cz-1 & 0 & 0 & 0\end{smallmatrix}\right]$. 
    Since $cz \neq 1$, its support size is either 5 or 6.
    Then we have $\Tilde{f}_{-1} \notin \mathscr{A}$ and $\Tilde{f}_{-1} \notin \mathscr{P}$.
    A direct computation shows $\Tilde{f}_{-1} \notin \widehat{\mathscr{M}}$. 
    Also, by Lemma~\ref{lm:is matchgate}, $\Tilde{f}_{-1} \in \mathscr{M}$ if and only if $by = \pm(1-cz)$. 
    Therefore, by Theorem~\ref{thm:planar six-vertex}, $\plholant{\neq_2}{f}$ is \numP-hard unless $by = \pm (1-cz)$. Note that if $by = 1-cz$, then $f \in \mathscr{M}$ and thus $\plholant{\neq_2}{f}$ can be computed in polynomial time.

    It remains to consider the case $by = cz -1$. 
    If we do the binary modification to $f^\pi$ on variable $x_4$ to get signature $f^\pi_t$ on RHS whose signature matrix is $M(f^\pi_t) = \left[\begin{smallmatrix} 
    1 & 0 & 0 & yt \\ 
    0 & z & 0 & 0 \\ 
    0 & w & ct & 0 \\ 
    b & 0 & 0 & t\end{smallmatrix}\right]$. 
    Connect $f^\pi_t$ with $f$ via $(\neq_2)^{\otimes2}$ we get a signature $\hat{f}$ on RHS whose signature matrix is $M(\hat{f}) = M(f^\pi_t)NM(f) = \left[\begin{smallmatrix} 
    y(t+1) & 0 & 0 & 1+byt\\ 
    0 & wz & z^2 & 0 \\ 
    0 & c^2t+w^2 & wz & 0 \\ 
    t+by & 0 & 0 & b(t+1)\end{smallmatrix}\right]$. 
    If we set $t = -1$, then we get $M(\hat{f}_{-1}) = \left[\begin{smallmatrix} 
    0 & 0 & 0 & 1-by\\ 
    0 & wz & z^2 & 0 \\ 
    0 & w^2-c^2 & wz & 0 \\ 
    by-1 & 0 & 0 & 0\end{smallmatrix}\right]$ on RHS. 
    If $by\neq 1$, then the support of $\hat{f}_{-1}$ has size 5 or 6. Therefore $\hat{f}_{-1} \notin \mathscr{A}$ and $\hat{f}_{-1} \notin \mathscr{P}$. 
    A direct calculation shows $\hat{f}_{-1} \notin \widehat{\mathscr{M}}$. 
    Also, by Lemma~\ref{lm:is matchgate}, $\hat{f}_{-1} \in \mathscr{M}$ if and only if $(1-by)^2=c^2z^2$. 
    Plug into $by = cz -1$ will yield $by=0$ which is a contradiction. 
    Therefore, when $cz \neq 1$, by Theorem~\ref{thm:planar six-vertex} we know $\plholant{\neq_2}{f}$ is \numP-hard unless $by=1$. 
    If $by=1$, combined with $by=cz-1$ we have $cz = 2$. 
    Moreover, if $by=1$ and $w^2=c^2$, then $M(\hat{f}_{-1}) = \left[\begin{smallmatrix} 
    0 & 0 & 0 & 0\\ 
    0 & wz & z^2 & 0 \\ 
    0 & 0 & wz & 0 \\ 
    0 & 0 & 0 & 0\end{smallmatrix}\right]$ whose support size is 3. Thus it is not in $\mathscr{A}$ nor in $\mathscr{P}$. By Lemma~\ref{lm:is matchgate} and Lemma~\ref{lm: is matchgate hat}, it is not in $\mathscr{M}$ nor in $\widehat{\mathscr{M}}$. Therefore in this case by Theorem~\ref{thm:planar six-vertex} the problem $\plholant{\neq_2}{f}$ is \numP-hard.
    We summarize the current progress: $\plholant{\neq_2}{f}$ is \numP-hard unless $f\in \mathscr{M}$ or $(by=1) \wedge (cz=2)\wedge(w^2\neq c^2)$.
    
    We now set $t = -\frac{w^2}{c^2}$ in $\hat{f}$ (note that since $w^2 \neq c^2$, we have $t \neq -1$) and rotate it by $\frac{\pi}{2}$.
    Plug into $by=1$ and apply Lemma~\ref{lm:let a=x}, the resulting signature matrix is $\left[\begin{smallmatrix}
        1 & 0 & 0 & \frac{wz}{1+t}\\
        0 & 1 & 0 & 0\\
        0 & \frac{z^2}{1+t} & 1 & 0\\
        \frac{wz}{1+z} & 0 & 0 & 1
    \end{smallmatrix}\right]$. 
    This signature is of the form $\left[\begin{smallmatrix} 1 & 0 & 0 & b' \\ 0 & c' & 0 & 0 \\ 0 & w' & z' & 0 \\ y' & 0 & 0 & 1\end{smallmatrix}\right]$ with $c'z' = 1 $ and therefore it is \numP-hard by the previous argument.
    This completes our proof.
\end{proof}

\subsection{${\sf N}=2$: two zero entries}\label{sec: N=2}

For this case, there are two zero entries within $\{(b,y),(c,z),(d,w)\}$. 
If the two zero entries are in different pairs, by the rotational symmetry it suffices to consider the following three cases: 
(i) $M(f) = \left[\begin{smallmatrix} 1 & 0 & 0 & b \\ 0 & c & d & 0 \\ 0 & w & 0 & 0 \\ 0 & 0 & 0 & 1\end{smallmatrix}\right]$ with $bcdw \neq0$, 
(ii) $M(f) = \left[\begin{smallmatrix} 1 & 0 & 0 & b \\ 0 & c & d & 0 \\ 0 & 0 & z & 0 \\ 0 & 0 & 0 & 1\end{smallmatrix}\right]$ with $bcdz \neq 0$, or 
(iii) $M(f) = \left[\begin{smallmatrix} 1 & 0 & 0 & b \\ 0 & c & 0 & 0 \\ 0 & w & z & 0 \\ 0 & 0 & 0 & 1\end{smallmatrix}\right]$ with $bczw \neq 0$. 
If the two zero entries are in a single pair, by the rotation symmetry it suffices to consider the following two cases: (iv) $M(f) = \left[\begin{smallmatrix} 1 & 0 & 0 & 0 \\ 0 & c & d & 0 \\ 0 & w & z & 0 \\ 0 & 0 & 0 & 1\end{smallmatrix}\right]$ with $cdzw \neq 0$
or (v) $M(f) = \left[\begin{smallmatrix} 1 & 0 & 0 & b \\ 0 & c & 0 & 0 \\ 0 & 0 & z & 0 \\ y & 0 & 0 & 1\end{smallmatrix}\right]$ with $bcyz \neq 0$.

\subsubsection{Two zero entries in different pairs}
The following lemma is proved in~\cite{CaiF17}.
\begin{lemma}\label{lm:general eight vertex model N=2 differet pairs}
    If $M(f) = \left[\begin{smallmatrix} 1 & 0 & 0 & b \\ 0 & c & d & 0 \\ 0 & w & 0 & 0 \\ 0 & 0 & 0 & 1\end{smallmatrix}\right]$ with $bcdw \neq0$, then \holant{\neq_2}{f} is \numP-hard.
\end{lemma}

We now show that the only additional tractable class is given by matchgates.
\begin{lemma}\label{lm: 1bcdw1}
    If $M(f) = \left[\begin{smallmatrix} 1 & 0 & 0 & b \\ 0 & c & d & 0 \\ 0 & w & 0 & 0 \\ 0 & 0 & 0 & 1\end{smallmatrix}\right]$ with $bcdw \neq0$, then \plholant{\neq_2}{f} is \numP-hard unless $f\in \mathscr{M}$, or more explicitly $dw=-1$, in which cases \plholant{\neq_2}{f} is computable in polynomial time.
\end{lemma}
\begin{proof}
    Rotate $f$ by $\pi$, we have $f^\pi$ whose signature matrix is $M(f^\pi) = \left[\begin{smallmatrix}
        1 & 0 & 0 & 0\\
        0 & 0 & d & 0\\
        0 & w & c & 0\\
        b & 0 & 0 & 1
    \end{smallmatrix} \right]$. Connect $f$ and $f^\pi$ by $N$, we have a signature on RHS whose signature matrix is 
    $$
     M(f) N M(f^\pi) = \left[\begin{smallmatrix}
        2b & 0 & 0 & 1\\
        0 & cw & c^2+d^2 & 0\\
        0 & w^2 & cw & 0\\
        1 & 0 & 0 & 0
    \end{smallmatrix} \right].
    $$
    By Lemma~\ref{lm:let a=x}, we then have $\plholant{\neq_2}{\Tilde{f}} \leq^p_T \plholant{\neq_2}{f}$ where the signature matrix of $\Tilde{f}$ is
    $
     M(\Tilde{f}) = \left[\begin{smallmatrix}
        0 & 0 & 0 & 1\\
        0 & cw & c^2+d^2 & 0\\
        0 & w^2 & cw & 0\\
        1 & 0 & 0 & 0
    \end{smallmatrix} \right].
    $ Since $bcdw \neq 0$, the support size of $\Tilde{f}$ is either 5 or 6, and thus $\Tilde{f}$ does not have affine support. 
    By Theorem~\ref{thm:planar six-vertex}, $\plholant{\neq_2}{\Tilde{f}}$ is \numP-hard and thus $\plholant{\neq_2}{f}$ is \numP-hard unless $\Tilde{f} \in \widehat{\mathscr{M}}$ or $\Tilde{f} \in \mathscr{M}$. 
    If $\Tilde{f}\in \widehat{\mathscr{M}}$, then we have $\hat{f}=H^{\otimes4}\tilde{f}\in \M$, where $\hat{f}$ has the signature matrix
    $$
    M(\hat{f}) =  H^{\otimes 2} M(\Tilde{f}) H^{\otimes 2} = 
    \frac{1}{4}
    \left[\begin{smallmatrix}
        2 + d^2 + (c+w)^2 & d^2 + c^2 - w^2 & -d^2 -c^2 + w^2 & 2-d^2-(c+w)^2\\
        -d^2-c^2+w^2 & -2-d^2 -(c-w)^2 &  -2+d^2+(c-w)^2 & d^2 + c^2 - w^2 \\
        d^2+c^2-w^2 & -2+d^2+(c-w)^2 & -2-d^2-(c-w)^2 & -d^2-c^2+w^2 \\
        2-d^2-(c+w)^2 & -d^2 -c^2 + w^2 & d^2+c^2-w^2 & 2+d^2+(c+w)^2
    \end{smallmatrix} \right].
    $$
    If $\hat{f}$ satisfies the Even Parity Condition, then $$(2 + d^2 + (c+w)^2)^2-(2-d^2-(c+w)^2)^2 = (-2-d^2 -(c-w)^2)^2-(-2+d^2+(c-w)^2)^2$$ which implies $cw=0$, a contradiction. 
    If $\hat{f}$ satisfies the Odd Parity Condition, then 
    $$
    2 + d^2 + (c+w)^2 = 2-d^2 -(c+w)^2 = 0
    $$ 
    which implies $4=0$, a contradiction. 
    Thus $\Tilde{f} \notin \widehat{\mathscr{M}}$. 
    If $\Tilde{f} \in \mathscr{M}$, then by Lemma~\ref{lm:is matchgate} we have $dw = \pm 1$. 
    If $dw = -1$, by Lemma~\ref{lm:is matchgate} again we have $f \in \mathscr{M}$. From now on we assume $dw=1$.
    Then 
    $M(f) = \left[\begin{smallmatrix} 1 & 0 & 0 & b \\
    0 & c & d & 0 \\ 
    0 & \frac{1}{d} & 0 & 0 \\ 
    0 & 0 & 0 &1
    \end{smallmatrix}\right].$
    
    By Corollary~\ref{cor:get any binary when exactly three nonzero}, we can realize $(0,1,u,0)^T$ for any $u \in \mathbb{C}$. 
    If we first do a binary modification of $(0,1,\frac{c}{b},0)^T$ on variable $x_4$ of $f$ and then do a binary modification of $(0,1,\frac{cd^2}{b},0)^T$ on variable $x_1$ of $f$, we get the signature $\left[\begin{smallmatrix}
        1 & 0 & 0 & c\\
        0 & c & \frac{cd}{b} & 0\\
        0 & \frac{cd}{b} & 0 & 0\\
        0 & 0 & 0 & \frac{c^2d^2}{b^2}
    \end{smallmatrix}\right]$ on RHS. 
    By Lemma~\ref{lm:let a=x} and then normalize by dividing by $\frac{cd}{b}$, we have $\plholant{\neq_2}{g} \leq^p_T \plholant{\neq_2}{f}$ where $M(g) = \left[\begin{smallmatrix}
        1 & 0 & 0 & k\\
        0 & k & 1 & 0\\
        0 & 1 & 0 & 0\\
        0 & 0 & 0 & 1
    \end{smallmatrix} \right]$ and $k = \frac{b}{d}\neq 0$. 
    We now finish the proof by showing $\plholant{\neq_2}{g}$ is \numP-hard.

    By Lemma~\ref{lm: interpolate S for N=2}, we have $\plholant{\neq_2}{g,\mathcal{S}}\le_T^p \plholant{\neq_2}{g}$.
    By Corollary~\ref{cor:get any binary when exactly three nonzero} again, we can do binary modification of $(0,1,\frac{1}{k},0)^T$ to variable $x_3$ of $g$ and then do binary modifications of $(0,1,0,0)^T$ to variables $x_1$ and $x_4$ of $g$ to get a signature $\left[\begin{smallmatrix}
        1 & 0 & 0 & 0\\
        0 & 1 & 0 & 0\\
        0 & 0 & 0 & 0\\
        0 & 0 & 0 & 0
    \end{smallmatrix}\right] = \left[\begin{smallmatrix}
        1 & 0\\
        0 & 0
    \end{smallmatrix}\right] \otimes \left[ \begin{smallmatrix}
        1 & 0\\
        0 & 1
    \end{smallmatrix}\right]$ on RHS. 
    Connecting all its variables with $\neq_2$, we have $N\left[\begin{smallmatrix}
        1 & 0 & 0 & 0\\
        0 & 1 & 0 & 0\\
        0 & 0 & 0 & 0\\
        0 & 0 & 0 & 0
    \end{smallmatrix}\right]N = \left[\begin{smallmatrix}
        0 & 0 & 0 & 0\\
        0 & 0 & 0 & 0\\
        0 & 0 & 1 & 0\\
        0 & 0 & 0 & 1
    \end{smallmatrix}\right] = \left[\begin{smallmatrix}
        0 & 0\\
        0 & 1
    \end{smallmatrix}\right] \otimes \left[ \begin{smallmatrix}
        1 & 0\\
        0 & 1
    \end{smallmatrix}\right]$ on LHS. 
    This can be seen as  $h=\Delta_1^{\otimes2}\otimes(=_2)$, where $\Delta_1$ is a unary signature which pins its input to be 1.
    Consider $h^{\otimes2}=\Delta_1^{\otimes4}\otimes(=_2)^{\otimes2}$, and connect $\Delta_1^{\otimes4}$ with the 4 variables of $g$.
    This contributes a global factor $g_{1111}=1$ to the Holant value. 
    Thus, we get $(=_2)^{\otimes2}$ on LHS, i.e., $\plholant{\neq_2,(=_2)^{\otimes2}}{g,\mathcal{S}}\le \plholant{\neq_2}{g,\mathcal{S}}$.
    Thus, we have $\plholant{\neq_2,(=_2)^{\otimes2}}{g,\mathcal{S}}\le_T^p\plholant{\neq_2}{g}$.
    Then by Lemma~\ref{lm: swap gate}, we have $\holant{\neq_2}{g}\le_T^p\plholant{\neq_2}{g}$.
    Then $\plholant{\neq_2}{g}$ is \numP-hard by Lemma~\ref{lm:general eight vertex model N=2 differet pairs}.
    It follows that $\plholant{\neq_2}{f}$ is \numP-hard.
\end{proof}

\begin{lemma}\label{lm: 1bcdz1}
    If $M(f) = \left[\begin{smallmatrix} 1 & 0 & 0 & b \\ 0 & c & d & 0 \\ 0 & 0 & z & 0 \\ 0 & 0 & 0 & 1\end{smallmatrix}\right]$ with $bcdz \neq0$, then \plholant{\neq_2}{f} is \numP-hard unless $f\in \mathscr{M}$, or more explicitly $cz=1$, in which cases \plholant{\neq_2}{f} is computable in polynomial time.
\end{lemma}
\begin{proof}
    Rotate $f$ by $\frac{\pi}{2}$, $\pi$ and $\frac{3\pi}{2}$, we get $f^{\frac{\pi}{2}}$, $f^{\pi}$ and $f^{\frac{3\pi}{2}}$ with signature matrices $M(f^{\frac{\pi}{2}}) = \left[\begin{smallmatrix}
        1 & 0 & 0 & z\\
        0 & b & 0 & 0\\
        0 & d & 0 & 0\\
        c & 0 & 0 & 1
    \end{smallmatrix}\right]$, $M(f^{\pi}) = \left[\begin{smallmatrix}
        1 & 0 & 0 & 0\\
        0 & z & d & 0\\
        0 & 0 & c & 0\\
        b & 0 & 0 & 1
    \end{smallmatrix} \right]$, 
    and $M(f^{\frac{3\pi}{2}}) = \left[\begin{smallmatrix}
        1 & 0 & 0 & c\\
        0 & 0 & 0 & 0\\
        0 & d & b & 0\\
        z & 0 & 0 & 1
    \end{smallmatrix}\right]$ respectively. 
    Connecting $f$ and $f^\pi$ by two $\neq_2$, 
    we get a signature with the signature matrix $M(f)NM(f^\pi) = \left[\begin{smallmatrix}
        2b & 0 & 0 & 1\\
        0 & dz & c^2+d^2 & 0\\
        0 & z^2 & dz & 0 \\
        1 & 0 & 0 & 0
    \end{smallmatrix}\right]$. Be Lemma~\ref{lm:let a=x}, we get $\plholant{\neq_2}{g} \leq^p_T \plholant{\neq_2}{f}$ where $M(g) = \left[\begin{smallmatrix}
        0 & 0 & 0 & 1\\
        0 & dz & c^2+d^2 & 0\\
        0 & z^2 & dz & 0 \\
        1 & 0 & 0 & 0
    \end{smallmatrix} \right]$. 
    Note that the support size of $g$ is 5 or 6, and thus $g$ does not have affine support. By Theorem~\ref{thm:planar six-vertex}, $\plholant{\neq_2}{g}$ is \numP-hard and thus $\plholant{\neq_2}{f}$ is \numP-hard unless $g \in \widehat{\mathscr{M}}$ or $g \in \mathscr{M}$. 
    A direct calculation similar to the proof of Lemma~\ref{lm: 1bcdw1} shows that $g\notin \widehat{\mathscr{M}}$. Also, $g \in \mathscr{M}$ iff $cz=\pm1$ by Lemma~\ref{lm:is matchgate}. 
    If $cz=1$, then by Lemma~\ref{lm:is matchgate} again we know $f \in \mathscr{M}$ and thus $\plholant{\neq_2}{f}$ is tractable. 
    If $cz=-1$, connecting $f^{\frac{\pi}{2}}$ and $f^{\frac{3\pi}{2}}$ by two $\neq_2$,
    we get a signature $h$ with the signature matrix $M(h)=M(f^{\frac{\pi}{2}})NM(f^{\frac{3\pi}{2}}) = \left[\begin{smallmatrix}
        2c & 0 & 0 & 1+cz\\
        0 & bd & b^2 & 0\\
        0 & d^2 & bd & 0\\
        1+cz & 0 & 0 & 2z
    \end{smallmatrix}\right] =  \left[\begin{smallmatrix}
        2c & 0 & 0 & 0\\
        0 & bd & b^2 & 0\\
        0 & d^2 & bd & 0\\
        0 & 0 & 0 & 2z
    \end{smallmatrix}\right]$. By Lemma~\ref{lm: outer full inner degenerate}, $\plholant{\neq_2}{h}$ is \numP-hard and thus $\plholant{\neq_2}{f}$ is \numP-hard unless $b=-d$. Assume $b=-d$.
    Connecting $f^{\frac{3\pi}{2}}$ and $f^{\frac{\pi}{2}}$ by two $(\neq_2)$, we get a signature $h_1$ with the signature matrix $M(h_1)=M(f^{\frac{3\pi}{2}})NM(f^{\frac{\pi}{2}}) = \left[\begin{smallmatrix}
        2c & 0 & 0 & 1+cz\\
        0 & 0 & 0 & 0\\
        0 & b^2+d^2 & 0 & 0\\
        1+cz & 0 & 0 & 2z
    \end{smallmatrix}\right] = \left[\begin{smallmatrix}
        2c & 0 & 0 & 0\\
        0 & 0 & 0 & 0\\
        0 & 2b^2 & 0 & 0\\
        0 & 0 & 0 & 2z
    \end{smallmatrix}\right].$ By Lemma~\ref{lm: 1dw1}, we know $\plholant{\neq_2}{h_1}$ is \numP-hard and thus $\plholant{\neq_2}{f}$ is \numP-hard. This finishes the proof.
\end{proof}

Through exactly the same procedure in the proof of Lemma~\ref{lm: 1bcdz1}, one can prove the following.

 
\begin{lemma}\label{lm: 1bcwz1}
    If $M(f) = \left[\begin{smallmatrix} 1 & 0 & 0 & b \\ 0 & c & 0 & 0 \\ 0 & w & z & 0 \\ 0 & 0 & 0 & 1\end{smallmatrix}\right]$ with $bcwz \neq0$, then \plholant{\neq_2}{f} is \numP-hard unless $f\in \mathscr{M}$, or more explicitly $cz=1$, in which cases \plholant{\neq_2}{f} is computable in polynomial time.
\end{lemma}

\subsubsection{Two zero entries in a single pair}
We proceed to consider the case when $M(f) = \left[\begin{smallmatrix} 1 & 0 & 0 & 0 \\ 0 & c & d & 0 \\ 0 & w & z & 0 \\ 0 & 0 & 0 & 1\end{smallmatrix}\right]$ with $cdzw \neq 0$.
\begin{lemma}\label{lm:1cdwz1}
    If $M(f) = \left[\begin{smallmatrix} 1 & 0 & 0 & 0 \\ 0 & c & d & 0 \\ 0 & w & z & 0 \\ 0 & 0 & 0 & 1\end{smallmatrix}\right]$ with $cdzw \neq 0$, then \plholant{\neq_2}{f} is \numP-hard unless $f\in\mathscr{M}$, or more explicitly $cz-dw = 1$, in which case \plholant{\neq_2}{f} is computable in polynomial time.
\end{lemma}
\begin{proof}
    We first suppose $\det\left[\begin{smallmatrix}c & d \\ w & z\end{smallmatrix}\right]=0$. If $c+d \neq 0$ or $z+d \neq 0$, then \plholant{\neq_2}{f} is \numP-hard by Lemma~\ref{lm: outer full inner degenerate}. Otherwise, $c=z=-d \neq 0$ and since $\operatorname{det}\left[\begin{smallmatrix}c & d \\ w & z\end{smallmatrix}\right]=0$, we have $d=w$. So $\left[\begin{smallmatrix}c & d \\ w & z\end{smallmatrix}\right]=\left[\begin{smallmatrix}
        c & -c \\ -c & c
    \end{smallmatrix}\right]$. By connecting two copies of $f$ using $\neq 2$ we get a signature $f_{1}$ with the signature matrix
    $
    M\left(f_{1}\right)=M(f) N M(f)=\left[\begin{smallmatrix}
    0 & 0 & 0 & 1 \\
    0 & -2 c^{2} & 2 c^{2} & 0 \\
    0 & 2 c^{2} & -2 c^{2} & 0 \\
    1 & 0 & 0 & 0
    \end{smallmatrix}\right] .
    $
    The support of $f_{1}$ is not an affine subspace. By Theorem~\ref{thm:planar six-vertex}, $\plholant{\neq_2}{f_1}$ is \numP-hard and thus $\plholant{\neq_2}{f}$ is \numP-hard unless $f_1 \in \mathscr{M}$ or $\widehat{\mathscr{M}}$. By Lemma~\ref{lm:is matchgate}, $f\notin \mathscr{M}$. 
    Also, note that
    $H^{\otimes 2} M(f) H^{\otimes 2} = \left[\begin{smallmatrix}
        2 & 0 & 0 & 2\\
        0 & -2-8c^2 & -2+8c^2 & 0\\
        0 & -2+8c^2 & -2-8c^2 & 0\\
        2 & 0 & 0 & 2
    \end{smallmatrix}\right] $. 
    By Lemma~\ref{lm:is matchgate}, $f\in \widehat{\M}$ iff $c=0$, a contradiction. 
    In summary, if $\det \left[\begin{smallmatrix}
        c & d\\
        w & z
    \end{smallmatrix}\right]=0$, then $\plholant{\neq_2}{f}$ is \numP-hard.

    We now assume the inner matrix has full rank. By Lemma~\ref{lm:get one binary when not exists epsilon}, either it is the case that $(c=\epsilon z)\wedge (d = \epsilon w)$ for some $\epsilon = \pm 1$, or we have $(0,1,t,0)^T$ on RHS where $t \neq \pm 1$. Assume now it is the latter case. 
    Then by Lemma~\ref{lm: get any binary from a binary not fifth root} and Corollary~\ref{cor: get (0,1,0,0) from a binary third or fourth root}, we have $(0,1,0,0)^T$ on RHS. 
    By doing a binary modification on variables $x_1$ and $x_4$ or $f$ using $(0,1,0,0)^T$, we get the signature $g$ on RHS whose signature matrix is $M(g) = \left[\begin{smallmatrix}
        1 & 0 & 0 & 0\\
        0 & c & 0 & 0\\
        0 & 0 & 0 & 0\\
        0 & 0 & 0 & 0
    \end{smallmatrix}\right] = \left[\begin{smallmatrix}
        1 & 0\\
        0 & 0
    \end{smallmatrix}\right] \otimes \left[\begin{smallmatrix}
        1 & 0\\
        0 & c
    \end{smallmatrix}\right]$.
    This can be seen as $\Delta_0^{\otimes2}\otimes(1,0,0,c)^T$, where $\Delta_0$ is the unary signature which pins its input to be 0.
    Note that $g^{\otimes2}=\Delta_0^{\otimes4}\otimes(1,0,0,c)^T\otimes(1,0,0,c)^T$.
    Connect $\Delta_0^{\otimes4}$ with one copy of $f$ using four $\neq_2$. 
    This contributes a global factor $f_{1111}=1$ to the Holant value.
    Thus, we realize the signature  $h=(1,0,0,c)^T\otimes(1,0,0,c)^T$ on RHS, whose signature matrix is 
    $M(h)=\left[\begin{smallmatrix}
        1 & 0 & 0 & 0\\
        0 & c & 0 & 0\\
        0 & 0 & c & 0\\
        0 & 0 & 0 & c^2
    \end{smallmatrix}\right]$.
    By connecting $f$ and $h$ using $\neq_2$, we get the signature $\Tilde{f}$ on RHS whose signature matrix is 
    $M(\Tilde{f}) 
    =M(f)NM(h)= \left[\begin{smallmatrix}
        0 & 0 & 0 & c^2\\
        0 & cd & c^2 & 0\\
        0 & cz & cw & 0\\
        1 & 0 & 0 & 0
    \end{smallmatrix}\right]$. The support of $\Tilde{f}$ is not an affine subspace, and thus by Theorem~\ref{thm:planar six-vertex} $\plholant{\neq_2}{\Tilde{f}}$ is \numP-hard unless $\Tilde{f} \in \mathscr{M}$ or $\widehat{\mathscr{M}}$. If $\Tilde{f} \in \mathscr{M}$, then by Lemma~\ref{lm:is matchgate} we have $dw - cz = -1$ which implies $f \in \mathscr{M}$. 
    If $\Tilde{f} \in \widehat{\mathscr{M}}$, then $\Tilde{f}_1\in \M$, where 
    $$
    M(\Tilde{f}_1)=
    H^{\otimes 2} M(\Tilde{f}) H^{\otimes 2} = \frac{1}{4}\left[\begin{smallmatrix}
        1 + 2c^2 + cd+ cw + cz & 1-cd+cw-cz & 1-2c^2+cd-cw+cz & 1-cd-cw-cz\\
        -1-cd+cw+cz & -1-2c^2+cd+cw-cz & -1-cd-cw+cz & -1+2c^2+cd-cw-cz\\
        -1+2c^2+cd-cw-cz & -1-cd-cw+cz & -1-2c^2+cd+cw-cz & -1-cd+cw+cz\\
        1-cd-cw-cz & 1-2c^2+cd-cw+cz & 1-cd+cw-cz & 1+2c^2+cd+cw+cz
    \end{smallmatrix}\right].
    $$
    If the $\Tilde{f}_1$ satisfies the Even Parity Condition, then $1-cd+cw-cz =-1-cd+cw+cz=-1+2c^2+cd-cw-cz=0$, 
    which implies $d=w,cz=1$ and $c^2=1$. 
    A direct computation shows that $\det M_{\text{In}}(\Tilde{f}_1)-\det M_{\text{Out}}(\Tilde{f}_1)=-\frac{1}{2}c(1+c^2)(d+w)\neq 0$, contradicting Lemma~\ref{lm:is matchgate}.
    If $\Tilde{f}_1$ satisfies the Odd Parity Condition, then $1+2c^2+cd+cw+cz = -1-2c^2+cd+cw-cz = -1-cd-cw+cz = 0$, 
    which implies $d=-w,cz=1$ and $c^2=-1$. 
    A direct computation shows that $(\Tilde{f}_1)_{0010} (\Tilde{f}_1)_{1101}-(\Tilde{f}_1)_{0001} (\Tilde{f}_1)_{1110} - (\Tilde{f}_1)_{0100} (\Tilde{f}_1)_{1011} + (\Tilde{f}_1)_{0111} (\Tilde{f}_1)_{1000}=\frac{1}{2}c(-1+c^2)(d-w)\neq 0$, contradicting Lemma~\ref{lm:is matchgate}.
    
    Therefore, it suffices to consider the case $(c=\epsilon z) \wedge (d = \epsilon w)$ for some $\epsilon=\pm 1$ from now on.

    We first consider the case when $\epsilon=1$, i.e.\ $M(f) = \left[\begin{smallmatrix}
        1 & 0 & 0 & 0\\
        0 & c & d & 0\\
        0 & d & c & 0\\
        0 & 0 & 0 & 1
    \end{smallmatrix}\right]$. This case is dealt by the following three claims.
    \paragraph{Claim A} If $f\notin \mathscr{M}$, then $\plholant{\neq_2}{f}$ is \numP-hard if $c^2+d^2 \neq 0$ and $c^2-d^2 \neq -1$.
    
    Connect two copies of $f$ with $N$ we get a signature $f_2$ whose signature matrix given by $M(f_2)= M(f) N M(f) = \left[\begin{smallmatrix}
        0 & 0 & 0 & 1\\
        0 & 2cd & c^2+d^2 & 0\\
        0 & c^2+d^2 & 2cd & 0\\
        1 & 0 & 0 & 0
    \end{smallmatrix}\right]$. Assume for now that $c^2+d^2 \neq 0$. Then the support of $f_2$ is not an affine subspace. By Theorem~\ref{thm:planar six-vertex}, $\plholant{\neq_2}{f_2}$ is \numP-hard and thus $\plholant{\neq_2}{f}$ is \numP-hard unless $f_2 \in \widehat{\mathscr{M}}$ or $\mathscr{M}$. If $f_2 \in \widehat{\mathscr{M}}$, then 
    $$H^{\otimes 2} M(f) H^{\otimes 2} = \left[\begin{smallmatrix}
        2+2c^2+4cd+2d^2 & 0 & 0 & 2-2c^2-4cd-2d^2\\
        0 & -2-2c^2+4cd-2d^2 & -2+2c^2-4cd+2d^2 & 0\\
        0 & -2+2c^2-4cd+2d^2 & -2-2c^2+4cd-2d^2 & 0\\
        2-2c^2-4cd-2d^2 & 0 & 0 & 2+2c^2+4cd+2d^2
    \end{smallmatrix}\right] \in \mathscr{M}.$$ By Lemma~\ref{lm:is matchgate}, this will imply $cd=0$, a contradiction. If $f_2\in \mathscr{M}$, then by Lemma~\ref{lm:is matchgate} this will imply $c^2-d^2 = \pm1$. If $c^2-d^2 = 1$, then $f\in\mathscr{M}$. This finishes the proof of Claim A.
    
    \paragraph{Claim B} If $f\notin \mathscr{M}$, then $\plholant{\neq_2}{f}$ is \numP-hard if $c^2+d^2 \neq 0$.

    By Claim A, we can assume $c^2-d^2 = -1$. Rotate $f$ by $\frac{\pi}{2}$ to get the signature $f^{\frac{\pi}{2}}$ whose signature matrix is $M(f^{\frac{\pi}{2}}) = \left[\begin{smallmatrix}
        1 & 0 & 0 & c\\
        0 & 0 & d & 0\\
        0 & d & 0 & 0\\
        c & 0 & 0 & 1
    \end{smallmatrix}\right]$. Connect two copies of $f^{\frac{\pi}{2}}$ with $N$ we have $M(f^{\frac{\pi}{2}})NM(f^{\frac{\pi}{2}}) = \left[\begin{smallmatrix}
        2c & 0 & 0 & 1+c^2\\
        0 & 0 & d^2 & 0\\
        0 & d^2 & 0 & 0\\
        1+c^2 & 0 & 0 & 2c
    \end{smallmatrix}\right]$. Rotate it by $\frac{\pi}{2}$ and normalize by $2c$, we get the signature $f_3$ on RHS whose signature matrix is $M(f_3) = \left[\begin{smallmatrix}
        1 & 0 & 0 & 0\\
        0 & \frac{1+c^2}{2c} & \frac{d^2}{2c} & 0\\
        0 & \frac{d^2}{2c} & \frac{1+c^2}{2c} & 0\\
        0 & 0 & 0 & 1
    \end{smallmatrix}\right] = \left[\begin{smallmatrix}
        1 & 0 & 0 & 0\\
        0 & \frac{d^2}{2c} & \frac{d^2}{2c} & 0\\
        0 & \frac{d^2}{2c} & \frac{d^2}{2c} & 0\\
        0 & 0 & 0 & 1
    \end{smallmatrix}\right]$. The signature matrix of $f_3$ has the same form as that of $f$. Therefore by Claim A ($d\neq 0$), we know $\plholant{\neq_2}{f_3}$ is \numP-hard and thus $\plholant{\neq_2}{f}$ is \numP-hard.
    
    \paragraph{Claim C} If $f\notin \mathscr{M}$, then $\plholant{\neq_2}{f}$ is \numP-hard.

    By Claim B, it suffices to consider the case when $c^2+d^2 = 0 \Rightarrow d = \pm \mathfrak{i}c$. If $256c^{16} \neq 1$, then $2cd$ is not a power of $\sqrt{\mathfrak{i}}$, as $(2cd)^8 = (\pm 2\mathfrak{i}c^2)^8 = 256c^{16}\neq 1$. 
    We then have $M(f_2) = \left[\begin{smallmatrix}
        0 & 0 & 0 & 1\\
        0 & 2cd & 0 & 0\\
        0 & 0 & 2cd & 0\\
        1 & 0 & 0 & 0
    \end{smallmatrix}\right]$ is \numP-hard by Theorem~\ref{thm:planar six-vertex}, Lemma~\ref{lm:is product type?}, the definition of $\mathscr{A}$, Lemma~\ref{lm:is matchgate}, and a direct calculation for $H^{\otimes 2} M(f) H^{\otimes 2}$.
    If $256c^{16}=1$, then $|c| = \frac{\sqrt{2}}{2}$. On the other hand, $f_3 = \left[\begin{smallmatrix}
        1 & 0 & 0 & 0\\
        0 & \frac{1+c^2}{2c} & \frac{d^2}{2c} & 0\\
        0 & \frac{d^2}{2c} & \frac{1+c^2}{2c} & 0\\
        0 & 0 & 0 & 1
    \end{smallmatrix}\right]$ (note that here we don't have the condition $c^2-d^2 = -1$). By Claim B, $\plholant{\neq_2}{f_3}$ is \numP-hard and thus $\plholant{\neq_2}{f}$ is \numP-hard unless $(\frac{1+c^2}{2c})^2+(\frac{d^2}{2c})^2=0 \Rightarrow c^4+c^2+1 = 0$ since $c^2+d^2 = 0$. But then by quadratic formula $c^2 = \frac{-1\pm \sqrt{3}\mathfrak{i}}{2} \Rightarrow |c| = 1$, contradicting $|c| = \frac{\sqrt{2}}{2}$. This finishes the proof of Claim C.

    It remains to consider the case when $\epsilon = -1$, i.e.\ $M(f) = \left[\begin{smallmatrix}
        1 & 0 & 0 & 0\\
        0 & c & d & 0\\
        0 & -d & -c & 0\\
        0 & 0 & 0 & 1
    \end{smallmatrix}\right]$. In this case, we can get the binary $(0,1,-1,0)^T$ on RHS and do the binary modification on variable $x_2$ to get $\left[\begin{smallmatrix}
        1 & 0 & 0 & 0\\
        0 & c & d & 0\\
        0 & d & c & 0\\
        0 & 0 & 0 & -1
    \end{smallmatrix}\right]$, whose complexity is equivalent to $\left[\begin{smallmatrix}
        1 & 0 & 0 & 0\\
        0 & -\mathfrak{i}c & -\mathfrak{i}d & 0\\
        0 & -\mathfrak{i}d & -\mathfrak{i}c & 0\\
        0 & 0 & 0 & 1
    \end{smallmatrix}\right]$, which is of the form when $\epsilon = 1$. Therefore we can apply Claim C, and thus $\plholant{\neq_2}{f}$ is \numP-hard unless $-c^2+d^2 = 1$, in which case $f \in \mathscr{M}$.

    Our proof is now complete.
\end{proof}

\begin{lemma}\label{lm:1bycz1}
    If $M(f) = \left[\begin{smallmatrix} 1 & 0 & 0 & b \\ 0 & c & 0 & 0 \\ 0 & 0 & z & 0 \\ y & 0 & 0 & 1\end{smallmatrix}\right]$ with $bcyz \neq 0$, then \plholant{\neq_2}{f} is \numP-hard unless $f\in\mathscr{M}$, or more explicitly $cz+by = 1$, in which cases \plholant{\neq_2}{f} is computable in polynomial time.
\end{lemma}
\begin{proof}
    Rotate $f$ by $\frac{\pi}{2}$ and $\frac{3\pi}{2}$, we have signatures $f^{\frac{\pi}{2}}$ and $f^{\frac{3\pi}{2}}$ on RHS whose signature matrices are given by $M(f^{\frac{\pi}{2}}) = \left[\begin{smallmatrix} 1 & 0 & 0 & z \\ 0 & b & 0 & 0 \\ 0 & 0 & y & 0 \\ c & 0 & 0 & 1\end{smallmatrix}\right]$ and $M(f^{\frac{3\pi}{2}}) = \left[\begin{smallmatrix} 1 & 0 & 0 & c \\ 0 & y & 0 & 0 \\ 0 & 0 & b & 0 \\ z & 0 & 0 & 1\end{smallmatrix}\right]$ respectively. 
    Connect $f^{\frac{\pi}{2}}$ and $f^{\frac{3\pi}{2}}$ via $(\neq_2)^{\otimes 2}$, we have the signature $\Tilde{f}$ whose signature matrix is $M(\Tilde{f}) = M(f^{\frac{\pi}{2}})NM(f^{\frac{3\pi}{2}})=
    \left[\begin{smallmatrix} 2z & 0 & 0 & 1+cz \\ 0 & 0 & b^2 & 0 \\ 0 & y^2 & 0 & 0 \\ 1+cz & 0 & 0 & 2c\end{smallmatrix}\right]$. We divide our discussions depending on whether $cz= -1$. 
    
    If $cz \neq -1$, then $1+cz \neq 0$. By Lemma~\ref{lm:let a=x} and Lemma~\ref{lm: 1bcdw1} (after rotation by $\frac{\pi}{2}$), we know $\plholant{\neq_2}{\Tilde{f}}$ is \numP-hard and thus $\plholant{\neq_2}{f}$ is \numP-hard unless $4cz = (1+cz)^2-b^2y^2$, which implies $by = 1-cz$ or $by = cz-1$. If $by = 1-cz$, then by Lemma~\ref{lm:is matchgate} $f\in \mathscr{M}$ and the problem $\plholant{\neq_2}{f}$ is thus tractable. Therefore when $cz \neq -1$, it remains to consider the case $by = cz-1$. 
    If we rotate $f$ by $\pi$, we have signature $f^{\pi}$ on RHS whose signature matrix is $M(f^{\pi}) = \left[\begin{smallmatrix} 1 & 0 & 0 & y \\ 0 & z & 0 & 0 \\ 0 & 0 & c & 0 \\ b & 0 & 0 & 1\end{smallmatrix}\right]$. 
    Connect $f$ and $f^{\pi}$ via $(\neq_2)^{\otimes 2}$, we have signature $\hat{f}$ on RHS whose signature matrix is $M(\hat{f}) = \left[\begin{smallmatrix} 2b & 0 & 0 & 1+by \\ 0 & 0 & c^2 & 0 \\ 0 & z^2 & 0 & 0 \\ 1+by & 0 & 0 & 2y\end{smallmatrix}\right]$. Plug into $by+1 = cz$, we have $M(\hat{f}) = \left[\begin{smallmatrix} 2b & 0 & 0 & cz \\ 0 & 0 & c^2 & 0 \\ 0 & z^2 & 0 & 0 \\ cz & 0 & 0 & 2y\end{smallmatrix}\right]$. By Lemma~\ref{lm:let a=x} and Lemma~\ref{lm: 1bcdw1} (after rotation by $\frac{\pi}{2}$), we know $\plholant{\neq_2}{\hat{f}}$ is \numP-hard and therefore \plholant{\neq_2}{f} is \numP-hard.

    If $cz = -1$, then $M(\Tilde{f}) = \left[\begin{smallmatrix} 2z & 0 & 0 & 0 \\ 0 & 0 & b^2 & 0 \\ 0 & y^2 & 0 & 0 \\ 0 & 0 & 0 & 2c\end{smallmatrix}\right]$. By Lemma~\ref{lm:let a=x} and Lemma~\ref{lm: 1dw1}, we know $\plholant{\neq_2}{\Tilde{f}}$ is \numP-hard and thus \plholant{\neq_2}{f} is \numP-hard unless $b^2y^2 = \pm 4cz = \mp 4$. In all cases, $1+by \neq 0$ since $|by|=2$. Therefore by Lemma~\ref{lm:let a=x} and Lemma~\ref{lm: 1bcdw1} (after rotation by $\frac{\pi}{2}$),  we know $\plholant{\neq_2}{\hat{f}}$ is \numP-hard and therefore \plholant{\neq_2}{f} is \numP-hard unless $4by = (1+by)^2-c^2z^2$, which implies $(1-by)^2 = 1$. Since $by \neq 0$, we have $by=2$, and thus $cz+by = 1$ and $f\in \mathscr{M}$. This finishes our proof.
\end{proof}

This finishes our discussion when ${\sf N}=2$.


\subsection{${\sf N=3}:$ three zero entries}\label{subsec: N=3}

There are exactly three nonzero entries within $\{(b,y),(c,z),(d,w)\}$.
If two of the three nonzero entries are in the same pair, by rotation symmetry, it suffices to consider the following four cases:
(i) $M(f) = \left[\begin{smallmatrix} 1 & 0 & 0 & 0 \\ 0 & c & d & 0 \\ 0 & w & 0 & 0 \\ 0 & 0 & 0 & 1\end{smallmatrix}\right]$ with $cdw \neq0$,
(ii) $M(f) = \left[\begin{smallmatrix} 1 & 0 & 0 & 0 \\ 0 & c & 0 & 0 \\ 0 & w & z & 0 \\ 0 & 0 & 0 & 1\end{smallmatrix}\right]$ with $cwz \neq0$, 
(iii) $M(f) = \left[\begin{smallmatrix} 1 & 0 & 0 & 0 \\ 0 & c & d & 0 \\ 0 & 0 & z & 0 \\ 0 & 0 & 0 & 1\end{smallmatrix}\right]$ with $cdz \neq0$,
(iv) $M(f) = \left[\begin{smallmatrix} 1 & 0 & 0 & b \\ 0 & c & 0 & 0 \\ 0 & 0 & z & 0 \\ 0 & 0 & 0 & 1\end{smallmatrix}\right]$ with $bcz \neq0$.
If the three nonzero entries are in different pairs, by rotational symmetry, it suffices to consider the following two cases:
(v) $M(f) = \left[\begin{smallmatrix} 1 & 0 & 0 & b \\ 0 & c & d & 0 \\ 0 & 0 & 0 & 0 \\ 0 & 0 & 0 & 1\end{smallmatrix}\right]$ with $bcd \neq0$,
(vi) $M(f) = \left[\begin{smallmatrix} 1 & 0 & 0 & 0 \\ 0 & c & d & 0 \\ 0 & 0 & 0 & 0 \\ y & 0 & 0 & 1\end{smallmatrix}\right]$ with $cdy \neq0$.

The following lemma is proved in~\cite{CaiF17}.
\begin{lemma}\label{lm:general eight vertex model N=3}
    If $M(f) =  \left[\begin{smallmatrix} 1 & 0 & 0 & 0 \\ 0 & c & d & 0 \\ 0 & w & 0 & 0 \\ 0 & 0 & 0 & 1\end{smallmatrix}\right]$ with $cdw \neq 0$, then \holant{\neq_2}{f} is \numP-hard.
\end{lemma}

We now show that the only additional tractable class is given by matchgates.


\begin{lemma}\label{lm: 1cdw1}
    If $M(f) = \left[\begin{smallmatrix} 1 & 0 & 0 & 0 \\ 0 & c & d & 0 \\ 0 & w & 0 & 0 \\ 0 & 0 & 0 & 1\end{smallmatrix}\right]$ with $cdw \neq 0$, then \plholant{\neq_2}{f} is \numP-hard unless $f\in \mathscr{M}$, or more explicitly $dw=-1$, in which cases \plholant{\neq_2}{f} is computable in polynomial time.
\end{lemma}
\begin{proof}
    Rotate $f$ by $\pi$ we have the signature $f^{\pi}$ on RHS with signature matrix $M(f^{\pi}) = \left[\begin{smallmatrix} 1 & 0 & 0 & 0 \\ 
    0 & 0 & d & 0 \\
    0 & w & c & 0 \\ 
    0 & 0 & 0 & 1
    \end{smallmatrix}\right]$. Connect $f$ and $f^{\pi}$ via $(\neq_2)^{\otimes 2}$, we have the signature $\Tilde{f}$ on RHS whose signature matrix is $M(\Tilde{f}) = M(f)NM(f^{\pi}) = \left[\begin{smallmatrix} 0 & 0 & 0 & 1 \\ 
    0 & cw & c^2+d^2 & 0 \\ 
    0 & w^2 & cw & 0 \\ 
    1 & 0 & 0 & 0
    \end{smallmatrix}\right]$.
    The support size of $\Tilde{f}$ is either $5$ or $6$, so $\Tilde{f}$ doesn't have affine support.
    By Theorem~\ref{thm:planar six-vertex}, \plholant{\neq_2}{\Tilde{f}} is \numP-hard and thus \plholant{\neq_2}{f} is \numP-hard unless $\Tilde{f} \in \widehat{\mathscr{M}}$ or $\Tilde{f} \in \mathscr{M}$.
    If $\Tilde{f} \in \widehat{\mathscr{M}}$, then we have $\title{f}_1\in \M$, where $\Tilde{f}_1$ has the signature matrix
    $$
    M(\Tilde{f}_1)=
    H^{\otimes 2} M(\Tilde{f}) H^{\otimes 2} = \frac{1}{4}\left[\begin{smallmatrix} 
    2+c^2+d^2+2cw+w^2 & c^2+d^2-w^2 & -c^2-d^2+w^2 & 2-c^2-d^2-2cw-w^2 \\ 
    -c^2-d^2+w^2 & -2-c^2-d^2+2cw-w^2 & -2+c^2+d^2-2cw+w^2 & c^2+d^2-w^2 \\ 
    c^2+d^2-w^2 & -2+c^2+d^2-2cw+w^2 & -2-c^2-d^2+2cw-w^2 & -c^2-d^2+w^2 \\
    2-c^2-d^2-2cw-w^2 & -c^2-d^2+w^2 & c^2+d^2-w^2 & 2+c^2+d^2+2cw+w^2\end{smallmatrix}\right].
    $$
    If $\Tilde{f}_1$ satisfies the Odd Parity Condition, then $2+c^2+d^2+2cw+w^2=2-c^2-d^2-2cw-w^2= 0$, which implies $4=0$, contradiction. 
    Thus, it satisfies the Even Parity Condition, 
    but then the determinant condition reduces to $cw=0$, a contradiction again. Thus we conclude that $\Tilde{f} \notin \widehat{\mathscr{M}}$. It remains to consider the case when $\Tilde{f} \in \mathscr{M}$. By Lemma~\ref{lm:is matchgate}, it implies that $d^2w^2 = 1$, or equivalently $dw = \pm 1$. If $dw = -1$, then $f \in \mathscr{M}$ by Lemma~\ref{lm:is matchgate} and thus it is tractable. We will from now on assume $dw = 1$ in this proof.

    Now 
    $M(f) = \left[\begin{smallmatrix} 1 & 0 & 0 & 0 \\ 
    0 & c & d & 0\\ 
    0 & \frac{1}{d} & 0 & 0\\ 
    0 & 0 & 0 & 1\end{smallmatrix}\right].$
    Rotate $f$ by $\frac{\pi}{2}$ to get the signature $f^{\frac{\pi}{2}}$ on RHS whose signature matrix is 
    $M(f^{\frac{\pi}{2}}) = \left[\begin{smallmatrix} 1 & 0 & 0 & 0 \\ 
    0 & 0 & \frac{1}{d} & 0 \\ 
    0 & d & 0 & 0 \\ 
    c & 0 & 0 & 1\end{smallmatrix}\right]$. 
    Connect its variables $x_3$ and $x_4$ using $(\neq_2)$, we get the binary signature $(0,1,d^2,0)^T$ on RHS. Do binary modification using $(0,1,d^2,0)^T$ to the variable $x_4$, we get $\left[\begin{smallmatrix}1 & 0 & 0 & 0 \\ 
    0 & 0 & d & 0 \\ 
    0 & d & 0 & 0 \\ 
    c & 0 & 0 & d^2
    \end{smallmatrix}\right]$ on RHS. By Lemma~\ref{lm:let a=x}, it suffices to consider the complexity of $\plholant{\neq_2}{\hat{f}}$ where $M(\hat{f}) = \left[\begin{smallmatrix} d & 0 & 0 & 0 \\ 
    0 & 0 & d & 0 \\ 
    0 & d & 0 & 0 \\ 
    c & 0 & 0 & d
    \end{smallmatrix}\right]$. Since dividing by a constant does not change the complexity, we divide the signature matrix of $\hat{f}$ by $d$ but still call it $\hat{f}$. Now $M(\hat{f}) = \left[\begin{smallmatrix} 1 & 0 & 0 & 0 \\ 
    0 & 0 & 1 & 0 \\ 
    0 & 1 & 0 & 0 \\ 
    k & 0 & 0 & 1
    \end{smallmatrix}\right]$ where $k = \frac{c}{d}$. Rotate $\hat{f}$ by $\frac{3\pi}{2}$ we have the signature $\hat{f}^{\frac{3\pi}{2}}$ whose signature matrix is $M(\hat{f}^{\frac{3\pi}{2}}) = \left[\begin{smallmatrix} 1 & 0 & 0 & 0 \\ 0 & k & 1 & 0 \\ 0 & 1 & 0 & 0 \\ 0 & 0 & 0 & 1\end{smallmatrix}\right]$. 
    Connect $2s+1$ many $\hat{f}^{\frac{3\pi}{2}}$ with $2s$ many $(\neq_2) \otimes (\neq_2)$, we have the signatures on RHS with signature matrices $\left[\begin{smallmatrix} 1 & 0 & 0 & 0 \\ 0 & (2s+1)k & 1 & 0 \\ 0 & 1 & 0 & 0 \\ 0 & 0 & 0 & 1\end{smallmatrix}\right]$. 
    Therefore we can interpolate the signature $\mathcal{S}$. 
    Moreover, by Corollary~\ref{cor:get any binary when exactly three nonzero}, we have the binary signatures $(0,1,\frac{1}{k},0)^T$ and $(0,1,0,0)^T$. 
    By doing binary modification with $(0,1,\frac{1}{k},0)^T$ on variable $x_3$ of $\hat{f}^{\frac{3\pi}{2}}$ and then with $(0,1,0,0)^T$ on variables $x_1$ and $x_4$ of $\hat{f}^{\frac{3\pi}{2}}$, we get the signature $\left[\begin{smallmatrix} 1 & 0 & 0 & 0 \\ 0 & 1 & 0 & 0 \\ 0 & 0 & 0 & 0 \\ 0 & 0 & 0 & 0\end{smallmatrix}\right]$ on RHS. 
    By the same argument in Lemma~\ref{lm: 1bcdw1}, we can realize $(=_2)^{\otimes 2}$ on LHS.
    Then by Lemma~\ref{lm: swap gate} and Lemma~\ref{lm:general eight vertex model N=3}, we have $\plholant{\neq_2}{\hat{f}} \equiv_T \holant{\neq_2}{\hat{f}}$ is \numP-hard when $dw=1$. This completes the proof.
\end{proof}

\begin{lemma}\label{lm: 1cwz1}
    $M(f) = \left[\begin{smallmatrix} 1 & 0 & 0 & 0 \\ 0 & c & 0 & 0 \\ 0 & w & z & 0 \\ 0 & 0 & 0 & 1\end{smallmatrix}\right]$ with $cwz \neq0$, then \plholant{\neq_2}{f} is \numP-hard unless $f\in\mathscr{M}$, or more explicitly $cz=1$, in which cases \plholant{\neq_2}{f} is computable in polynomial time.
\end{lemma}
\begin{proof}
    Rotate $f$ by $\frac{\pi}{2}$, $\pi$ and $\frac{3\pi}{2}$ to get signatures $f^{\frac{\pi}{2}}$, $f^{\pi}$ and $f^{\frac{3\pi}{2}}$ respectively, whose signature matrices are $M(f^{\frac{\pi}{2}}) = \left[\begin{smallmatrix} 
    1 & 0 & 0 & z \\ 
    0 & 0 & w & 0 \\ 
    0 & 0 & 0 & 0 \\ 
    c & 0 & 0 & 1\end{smallmatrix}\right]$, 
    $M(f^{\pi}) = \left[\begin{smallmatrix} 
    1 & 0 & 0 & 0 \\ 
    0 & z & 0 & 0 \\ 
    0 & w & c & 0 \\ 
    0 & 0 & 0 & 1\end{smallmatrix}\right]$ and 
    $M(f^{\frac{3\pi}{2}}) = \left[\begin{smallmatrix} 
    1 & 0 & 0 & c \\ 
    0 & 0 & w & 0 \\ 
    0 & 0 & 0 & 0 \\ 
    z & 0 & 0 & 1\end{smallmatrix}\right]$. 
    Connect $f$ and $f^{\pi}$ via $(\neq_2)^{\otimes2}$ we get a signature $\Tilde{f}$ whose signature matrix is $M(\Tilde{f} )=M(f)NM(f^\pi) = 
    \left[\begin{smallmatrix} 
    0 & 0 & 0 & 1 \\ 
    0 & cw & c^2 & 0 \\ 
    0 & w^2+z^2 & cw & 0 \\ 
    1 & 0 & 0 & 0\end{smallmatrix}\right]$. 
    Since $cwz \neq 0$, the support size of $\Tilde{f}$ is either 5 or 6, and thus $\Tilde{f}$ does not have affine support. 
    By Theorem~\ref{thm:planar six-vertex}, the problem \plholant{\neq_2}{\Tilde{f}} is \numP-hard and thus the problem \plholant{\neq_2}{f} is \numP-hard unless $\Tilde{f} \in \mathscr{M}$ or $\Tilde{f} \in \widehat{\mathscr{M}}$. By Lemma~\ref{lm:is matchgate}, $\Tilde{f} \in \mathscr{M}$ if and only if $cz = \pm 1$. 
    If $cz = 1$, then $f \in \mathscr{M}$ and therefore $\plholant{\neq_2}{f}$ can be computed in polynomial time. 
    Also, $\Tilde{f} \in \widehat{\mathscr{M}}$ if and only if $H^{\otimes 4}\Tilde{f} \in \mathscr{M}$, which simplifies to $cw=0$, a contradiction. 
    Therefore, we can assume $cw = -1$. 
    Now connect $f^{\frac{\pi}{2}}$ and $f^{\frac{3\pi}{2}}$ via $(\neq_2)^{\otimes2}$ to get a signature $g$ whose signature matrix is $M(g)=M(f^{\frac{\pi}{2}})NM(f^{\frac{3\pi}{2}}) =  
    \left[\begin{smallmatrix} 
    2z & 0 & 0 & 1+cz \\ 
    0 & 0 & w^2 & 0 \\ 
    0 & 0 & 0 & 0 \\ 
    1+cz & 0 & 0 & 2c\end{smallmatrix}\right] = 
    \left[\begin{smallmatrix} 
    2z & 0 & 0 & 0 \\ 
    0 & 0 & w^2 & 0 \\ 
    0 & 0 & 0 & 0 \\ 
    0 & 0 & 0 & 2c\end{smallmatrix}\right]$. 
    By Lemma~\ref{lm:let a=x} and Lemma~\ref{lm: 1dw1}, $\plholant{\neq_2}{g}$ is \numP-hard, and thus $\plholant{\neq_2}{f}$ is \numP-hard.
    This finishes our proof.
\end{proof}

\begin{lemma}\label{lm: 1cdz1}
    If $M(f) = \left[\begin{smallmatrix} 1 & 0 & 0 & 0 \\ 0 & c & d & 0 \\ 0 & 0 & z & 0 \\ 0 & 0 & 0 & 1\end{smallmatrix}\right]$ with $cdz \neq0$, then \plholant{\neq_2}{f} is \numP-hard unless $f\in\mathscr{M}$, or more explicitly $cz=1$, in which cases \plholant{\neq_2}{f} is computable in polynomial time.
\end{lemma}
\begin{proof}
    The proof follows exactly as the proof of Lemma~\ref{lm: 1cwz1}.
\end{proof}

\begin{lemma}\label{lm: 1bcz1}
    If $M(f) = \left[\begin{smallmatrix} 1 & 0 & 0 & b \\ 0 & c & 0 & 0 \\ 0 & 0 & z & 0 \\ 0 & 0 & 0 & 1\end{smallmatrix}\right]$ with $bcz \neq0$, then \plholant{\neq_2}{f} is \numP-hard unless $f\in\mathscr{M}$, or more explicitly $cz=1$, in which cases \plholant{\neq_2}{f} is computable in polynomial time.
\end{lemma}
\begin{proof}
    Rotate $f$ by $\frac{\pi}{2}$ and $\frac{3\pi}{2}$ to get signatures $f^{\frac{\pi}{2}}$ and $f^{\frac{3\pi}{2}}$ respectively, whose signature matrices are $M(f^{\frac{\pi}{2}}) = \left[\begin{smallmatrix} 
    1 & 0 & 0 & z \\ 
    0 & b & 0 & 0 \\ 
    0 & 0 & 0 & 0 \\ 
    c & 0 & 0 & 1\end{smallmatrix}\right]$ and 
    $M(f^{\frac{3\pi}{2}}) = \left[\begin{smallmatrix} 
    1 & 0 & 0 & c \\ 
    0 & 0 & 0 & 0 \\ 
    0 & 0 & b & 0 \\ 
    z & 0 & 0 & 1\end{smallmatrix}\right]$. 
    Connect $f^{\frac{\pi}{2}}$ with $f^{\frac{3\pi}{2}}$ via $(\neq_2)^{\otimes2}$ to get a signature $\Tilde{f}$ whose signature matrix is $M(\Tilde{f})=M(f^{\frac{\pi}{2}})NM(f^{\frac{3\pi}{2}}) =  
    \left[\begin{smallmatrix} 
    2z & 0 & 0 & 1+cz \\ 
    0 & 0 & b^2 & 0 \\ 
    0 & 0 & 0 & 0 \\ 
    1+cz & 0 & 0 & 2c
    \end{smallmatrix}\right]$ which after rotation by $\frac{\pi}{2}$ amounts to the signature 
    $
    M(\Tilde{f}^{\frac{\pi}{2}})=
    \left[\begin{smallmatrix} 
    2z & 0 & 0 & 0 \\ 
    0 & 1+cz & 0 & 0 \\ 
    0 & b^2 & 1+cz & 0 \\ 
    0 & 0 & 0 & 2c\end{smallmatrix}\right]$. 
    If $1+cz=0$, then by Lemma~\ref{lm:let a=x} and Lemma~\ref{lm:case d=w}, $\plholant{\neq_2}{\Tilde{f}}$ is \numP-hard, and thus $\plholant{\neq_2}{f}$ is \numP-hard.
    If $1+cz\neq 0$,
    by Lemma~\ref{lm:let a=x} and Lemma~\ref{lm: 1cwz1}, $\plholant{\neq_2}{\Tilde{f}}$ is \numP-hard, and thus $\plholant{\neq_2}{f}$ is \numP-hard, unless $(1+cz)^2=4cz \Leftrightarrow cz =1$. This finishes the proof.
\end{proof}

Now we handle the cases
(v) $M(f) = \left[\begin{smallmatrix} 1 & 0 & 0 & b \\ 0 & c & d & 0 \\ 0 & 0 & 0 & 0 \\ 0 & 0 & 0 & 1\end{smallmatrix}\right]$ with $bcd \neq0$,
and
(vi) $M(f) = \left[\begin{smallmatrix} 1 & 0 & 0 & 0 \\ 0 & c & d & 0 \\ 0 & 0 & 0 & 0 \\ y & 0 & 0 & 1\end{smallmatrix}\right]$ with $cdy \neq0$
in the following lemma.

\begin{lemma}\label{lm: 1bcdy1, by=0}
    If $M(f) = \left[\begin{smallmatrix} 1 & 0 & 0 & b \\ 0 & c & d & 0 \\ 0 & 0 & 0 & 0 \\ y & 0 & 0 & 1\end{smallmatrix}\right]$ with $cd \neq0$ and $by=0$, then \plholant{\neq_2}{f} is \numP-hard.
\end{lemma}
\begin{proof}
Rotate $f$ by $\pi$ to get the signature $f^\pi$ with the signature matrix $M(f^\pi) = \left[\begin{smallmatrix} 
    1 & 0 & 0 & y \\ 
    0 & 0 & d & 0 \\ 
    0 & 0 & c & 0 \\ 
    b & 0 & 0 & 1\end{smallmatrix}\right].$ 
    Connect $f^\pi$ and $f$ via $(\neq_2)^{\otimes2}$, we get the signature $g$ with the signature matrix
    $M(g) = \left[\begin{smallmatrix} 
    2b & 0 & 0 & 1+by \\ 
    0 & cd & d^2 & 0 \\ 
    0 & c^2 & cd & 0 \\ 
    1+by & 0 & 0 & 2y
    \end{smallmatrix}\right]
    =
    \left[\begin{smallmatrix} 
    2b & 0 & 0 & 1 \\ 
    0 & cd & d^2 & 0 \\ 
    0 & c^2 & cd & 0 \\ 
    1 & 0 & 0 & 2y
    \end{smallmatrix}\right].$
    By Lemma~\ref{lm:let a=x}, we have $\plholant{\neq_2}{g}\equiv_T\plholant{\neq_2}{\Tilde{g}}$, where $\Tilde{g}$ has the signature matrix
    $
    M(\tilde{g})=
    \left[\begin{smallmatrix} 
    0 & 0 & 0 & 1 \\ 
    0 & cd & d^2 & 0 \\ 
    0 & c^2 & cd & 0 \\ 
    1 & 0 & 0 & 0
    \end{smallmatrix}\right].$
    Since the support size of $\tilde{g}$ is 6, it doesn't have affine support.
    Then by Theorem~\ref{thm:planar six-vertex}, $\plholant{\neq_2}{\tilde{g}}$ is \numP-hard, and thus $\plholant{\neq_2}{f}$ is \numP-hard, unless $\tilde{g}\in \M$ or $\tilde{g}\in \widehat{\M}$.
    Note that $\tilde{g}\notin \M$ by Lemma~\ref{lm:is matchgate}.
    If $\tilde{g}\in \widehat{\M}$, then $H^{\otimes4}\tilde{g}\in \M$, which reduces to $cd=0$, contradiction.
\end{proof}
This finishes our discussion when ${\sf N}=3$.

\subsection{${\sf N}=4$: four zero entries} \label{subsec: N=4}

If ${\sf N}=4$ and there is at most one zero pair among $\{(b,y),(c,z),(d,w)\}$, by rotation symmetry, it suffices to consider the following three cases:
(i) $M(f) = \left[\begin{smallmatrix} 
1 & 0 & 0 & 0 \\ 
0 & c & d & 0 \\ 
0 & 0 & 0 & 0 \\ 
0 & 0 & 0 & 1
\end{smallmatrix}\right]$ with $cd \neq0$, 
(ii) $M(f) = \left[\begin{smallmatrix} 
1 & 0 & 0 & b \\ 
0 & 0 & d & 0 \\ 
0 & 0 & 0 & 0 \\ 
0 & 0 & 0 & 1
\end{smallmatrix}\right]$ with $bd \neq0$, 
(iii) $M(f) = \left[\begin{smallmatrix} 
1 & 0 & 0 & b \\ 
0 & c & 0 & 0 \\ 
0 & 0 & 0 & 0 \\ 
0 & 0 & 0 & 1
\end{smallmatrix}\right]$ with $bc \neq0$.
Case (i) can be subsumed by Lemma~\ref{lm: 1bcdy1, by=0}.
In the following we handle Case (ii) and Case (iii).

\begin{lemma}\label{lm:1bd1}
    If $M(f) = \left[\begin{smallmatrix}
        1 & 0 & 0 & b\\
        0 & 0 & d & 0\\
        0 & 0 & 0 & 0\\
        0 & 0 & 0 & 1
    \end{smallmatrix} \right]$ with $bd \neq 0$, then \plholant{\neq_2}{f} is \numP-hard.
\end{lemma}
\begin{proof}
    Rotate $f$ by $\pi$ to get signature $f^{\pi}$ with signature matrix $M(f^{\pi}) = \left[\begin{smallmatrix}
        1 & 0 & 0 & 0\\
        0 & 0 & d & 0\\
        0 & 0 & 0 & 0\\
        b & 0 & 0 & 1
    \end{smallmatrix} \right]$. 
    Connect $f$ and $f^{\pi}$ via $(\neq_2)^{\otimes2}$, we get the signature $\Tilde{f}$ with the signature matrix $M(\Tilde{f}) = \left[\begin{smallmatrix}
        2b & 0 & 0 & 1\\
        0 & 0 & d^2 & 0\\
        0 & 0 & 0 & 0\\
        1 & 0 & 0 & 0
    \end{smallmatrix} \right]$. 
    By Lemma~\ref{lm:let a=x} and Theorem~\ref{thm:planar six-vertex}, the problem \plholant{\neq_2}{\Tilde{f}} is \numP-hard and therefore the problem \plholant{\neq_2}{f} is \numP-hard unless the signature $\left[\begin{smallmatrix}
        0 & 0 & 0 & 1\\
        0 & 0 & d^2 & 0\\
        0 & 0 & 0 & 0\\
        1 & 0 & 0 & 0
    \end{smallmatrix} \right]$ is one of the tractable cases. Since $d\neq0$, its support is not affine and thus it is not of product or affine type. 
    By Lemma~\ref{lm:is matchgate} and Lemma~\ref{lm: is matchgate hat3}, we know it is not matchgate transformable. 
    This finishes the proof. 
\end{proof}

\begin{lemma}\label{lm: 1bc1}
    If $M(f) = \left[\begin{smallmatrix} 1 & 0 & 0 & b \\ 0 & c & 0 & 0 \\ 0 & 0 & 0 & 0 \\ 0 & 0 & 0 & 1\end{smallmatrix}\right]$ with $bc \neq0$, then \plholant{\neq_2}{f} is \numP-hard.
\end{lemma}
\begin{proof}
    The proof follows exactly as the proof of Lemma~\ref{lm:1bd1}.
\end{proof}

\subsection{Put things together}
We summarize Case II where there is at least one zero entry and at most one zero pair. 

\begin{theorem}\label{thm: at least one zero at most one zero pair}
    Let $f$ be a 4-ary signature with the signature matrix
$M(f)=\left[\begin{smallmatrix}
1 & 0 & 0 & b\\
0 & c & d & 0\\
0 & w & z & 0\\
y & 0 & 0 & 1
\end{smallmatrix}\right]$. 
If there is at least one zero entry in $\{(b,y),(c,z),(d,w)\}$ and at most one pair of $(b, y)$, $(c, z)$ and $(d, w)$ is a zero pair, then $\plholant{\neq_2}{f}$ is \numP-hard unless $f\in \mathscr{M}$.
\end{theorem}

Combining Lemma~\ref{lm: 1cz1} and Lemma~\ref{lm: 1dw1}, we have the following theorem which handles Case I and Case II.

\begin{theorem}\label{thm: at least one zero}
Let $f$ be a 4-ary signature with the signature matrix
$M(f)=\left[\begin{smallmatrix}
1 & 0 & 0 & b\\
0 & c & d & 0\\
0 & w & z & 0\\
y & 0 & 0 & 1
\end{smallmatrix}\right]$.
If there is at least one zero entry in $\{(b,y),(c,z),(d,w)\}$, then $\plholant{\neq_2}{f}$ is \numP-hard unless $f\in \mathscr{P}$, or $f$ is $\mathscr{A}$-transformable, or $f$ is $\M$-transformable, in which cases $\plholant{\neq_2}{f}$ is tractable.
\end{theorem}

The following lemma shows that in a particular case where there are exactly three nonzeros in the inner matrix, the only tractable class is matchgate.
This lemma will be useful in Section~\ref{sec: dichotomy for all nonzero}.
\begin{lemma}\label{lm: three nonzeros in inner matrix}
Let $f$ be a 4-ary signature with the signature matrix
$M(f)=\left[\begin{smallmatrix}
a & 0 & 0 & b\\
0 & 0 & d & 0\\
0 & w & z & 0\\
y & 0 & 0 & x
\end{smallmatrix}\right]$
with $dwz\neq 0$, then either $f\in \M$, or
$\plholant{\neq_2}{f}$ is \numP-hard.
\end{lemma}

\begin{proof}
    If $ax=0$, then by Lemma~\ref{lm:let a=x}, $\plholant{\neq_2}{f}\equiv_T \plholant{\neq_2}{f'}$, where $f'$ has the signature matrix
    $M(f')=\left[\begin{smallmatrix}
    0 & 0 & 0 & b\\
    0 & 0 & d & 0\\
    0 & w & z & 0\\
    y & 0 & 0 & 0
    \end{smallmatrix}\right].$
    If at most one of $b$ and $y$ is 0, then by Theorem 5.2 in \cite{CaiFS21}, $\plholant{\neq_2}{f'}$ is \numP-hard unless $f'\in \M \Leftrightarrow by=dw \Leftrightarrow f\in \M$.
    The lemma follows.
    If $b=y=0$, then the support size of $f'$ is $3$.
    Thus, $f'\notin \mathscr{P}$ and $f'\notin \mathscr{A}$.
    By Theorem~\ref{thm:planar six-vertex}, $\plholant{\neq_2}{f'}$ is \numP-hard unless $f'\in \M$ or $f'\in \widehat{\M}$.
    A direct computation shows that $f'\notin \widehat{\M}$.
    Thus, $\plholant{\neq_2}{f'}$ is \numP-hard unless $f'\in \M$, which is equivalent to $f\in \M$ as $ax=0$.
    The lemma follows.

    If $ax\neq 0$, normalizing $a=x=1$ by Lemma~\ref{lm:let a=x} and appeal to Lemma~\ref{lm: 1bcdw1}, Lemma~\ref{lm: 1cdw1} and Lemma~\ref{lm: 1bcdwy1} (possibly with some rotation) the lemma follows.
\end{proof}

\section{Case III: Symmetric Condition}\label{sec: three or opposite pairs}
Let $f$ be a 4-ary signature with the signature matrix $M(f) = \left[\begin{smallmatrix} 
1 & 0 & 0 & b \\ 
0 & c & d & 0 \\ 
0 & w & z & 0 \\ 
y & 0 & 0 & 1 \end{smallmatrix}\right]$, where $byczdw\neq 0$. 
In this section, we obtain the complexity classification of $\plholant{\neq_2}{f}$ where $(b,y),(c,z),(d,w)$ are three equal (Theorem~\ref{thm: three equal pairs}) or opposite pairs (Lemma~\ref{lm: symmetric disequality dichotomy}).  
This is the most intricate setting, where we employ a combinatorial transformation to the {\sc Even Coloring} model to identify new tractable classes, and apply conformal lattice interpolation to establish \numP-hardness.

We first consider the case $(b=y)\land(c=z)\land(d=w)$. 
In Section~\ref{subsec: transformation from even coloring}, we introduce the transformation to {\sc Even Coloring} (Lemma~\ref{lm: even coloring connection in general}) and provide a brief example illustrating how it can be used to establish \numP-hardness (Lemma~\ref{lm: ars b=-1 or c=-1}). 
This transformation will also be used later to identify a newly discovered tractable type in Lemma~\ref{lm: b=pm c, d=pm 1}.

In Section~\ref{subsec: apply conformal lattice interpolation}, we define a lattice $L_{(u,v,w)}$ specified by the multiplicative relations among the normalized eigenvalues of $M(f)$ after a holographic transformation.
We apply conformal lattice interpolation to interpolate a signature $g$ from $f$ (Lemma~\ref{lm:interpolation xyz}).
We proceed with a rank-based analysis in Section~\ref{subsec: rank<=1} and Section~\ref{subsec: rank2}, corresponding to the cases where the rank of $L_{(u,v,w)}$ is at most 1 or exactly 2, respectively. The rank-3 case is addressed in the final dichotomy for the three equal pairs case (Lemma~\ref{lm: symmetric equality dichotomy}).

If the signature $g$ we interpolate in Lemma~\ref{lm:interpolation xyz} happens to be the disequality crossover $\mathcal{S'}$ (Lemma~\ref{lm:interpolate disequality crossover}), then by Lemma~\ref{lm:disequality crossover}, the problem $\plholant{\neq_2}{f}$ becomes equivalent to its non-planar counterpart $\holant{\neq_2}{f}$.
Then we are done by the previous result:

\begin{lemma}\cite{CaiF17}\label{lm: eight-vertex ars equality}
    Let $f$ be a signature with the signature matrix
    $M(f)=\left[\begin{smallmatrix}
    1 & 0 & 0 & b\\
    0 & c & d & 0\\
    0 & d & c & 0\\
    b & 0 & 0 & 1\\
    \end{smallmatrix}\right]$ with $bcd\neq 0$,
    then $\operatorname{Holant}(\neq_2|f)$ is
    $\#\operatorname{P}$-hard, or $f$ is
    $\mathscr{P}$-transformable, or $\mathscr{A}$-transformable,
    or $\mathscr{L}$-transformable.
\end{lemma}

If the signature $g$ happens to define a \numP-hard planar six-vertex model (condition~\ref{cd:six vertex model} in Lemma~\ref{lm:interpolation something hard}), or a \numP-hard planar eight-vertex model in Case I (condition~\ref{cd:y=-1 two 0 pairs} in Lemma~\ref{lm:interpolation something hard}) or Case II (conditions~\ref{cd:y=-1 not matchgate},\ref{cd:x=z not matchgate},\ref{cd:x+z=0 not matchgate} in Lemma~\ref{lm:interpolation something hard}), the \numP-hardness of $\plholant{\neq_2}{g}$ then implies the \numP-hardness of $\plholant{\neq_2}{f}$.

However, there are several tricky exceptional cases (\ref{case: p_2=0,p_1=1},~\ref{case: p_2=0,p_1=2},~\ref{case: p_2=1},~\ref{case: p_2=-1} in Lemma~\ref{lm: rank2 c neq pm bd}) in the rank-2 case.
In these cases, no easy arguments suffice, and we have to appeal to the hardness of non-singular redundant signatures (condition~\ref{cd: redundent} in Lemma~\ref{lm:interpolation something hard}).
Moreover, $\M$-transformable signatures bury deeply in these cases.
To establish \numP-hardness, we employ geometric properties of polynomials, as captured in Lemma~\ref{lm: roots on unit circle} and Corollary~\ref{cor: roots on unit circle}.

If the conformal lattice interpolation does not yield \numP-hardness or the disequality crossover $\mathcal{S'}$, even after rotating $f$, we show that $b,c$ and $d$ must satisfy certain relations.
This either leads to $f$ is $\M$-transformable, or we can appeal to the combinatorial transformation to {\sc Even Coloring} plus a holographic transformation (Lemma~\ref{lm:even coloring transformation}), which transforms the problem to Case I or to the planar six-vertex model (Lemma~\ref{lm: b=pm c, d=pm 1}).
It's in the final case (via combinatorial transformation to the planar six-vertex model) where the new tractable type $(b=-c=\mathfrak{i}\tan(\frac{\beta\pi}{8}))\land (d=\pm1)$ is found.

 The above discussion culminates at Section~\ref{subsec: Dichotomy for three equal pairs} (Lemma~\ref{lm: symmetric equality dichotomy}, Theorem~\ref{thm: three equal pairs}), which gives a complexity classification on $\plholant{\neq_2}{f}$, where $(b,y),(c,z),(d,w)$ are three equal pairs.
 Finally, in Section~\ref{subsec: Dichotomy for three opposite pairs}, we handle the case where $(b,y),(c,z),(d,w)$ are three opposite pairs by applying holographic transformation and binary modification to transform $f$ into the form corresponding to the three equal pairs case (Lemma~\ref{lm: symmetric disequality dichotomy}).

\subsection{Combinatorial transformation from the {\sc Even Coloring} model}\label{subsec: transformation from even coloring}

The (planar) {\sc Even Coloring} problem is the problem $\PlHolant(g)$, equivalently $\plholant{=_2}{g}$, where $g$ is a $2n$-ary signature supported on even Hamming weight.
In other words, given a $2n$-regular planar graph $G$, a valid configuration for the {\sc Even Coloring} problem assigns each edge either {\it green} or {\it red}, such that every vertex has an even number of incident green edges.
This assignment is called an {\it even coloring} of $G$.
We denote by $\mathscr{C}(G)$  the set of all even colorings of $G$.
The goal is to compute the sum over $\mathscr{C}(G)$ of the product of local weights.

The following lemma generalizes a result previously established for the case $n=2$ in \cite{CaiL20}. 
Using the same underlying idea, we obtain a proof that extends to arbitrary $n$ in a simpler and more direct way.

\begin{lemma}\label{lm: even coloring connection in general}
Let $f$ and $g$ both be $2n$-ary signatures supported on even Hamming weight.
If $f(x_1,x_2,\ldots,x_{2n-1},x_{2n})=g(x_1,\overline{x_2},\ldots,x_{2n-1},\overline{x_{2n}})=g(\overline{x_1},x_2,\ldots,\overline{x_{2n-1}},x_{2n})$ ($g$ with even or odd index input bits flipped) for all $x_1,x_2,\ldots,x_{2n}\in \{0,1\}$, then $\plholant{\neq_2}{f}\equiv_p^T\plholant{=_2}{g}$.
\end{lemma}

\begin{proof}
Given a signature grid $\Omega = (G, \pi)$ of $\plholant{\neq_2}{f}$, we define a signature grid $\Omega'$ of $\plholant{=_2}{g}$ on the same underlying graph $G$, where every labeling $\neq_2$ is replaced by $=_2$, and $f$ is replaced by $g$.
We prove $\plholant{\Omega;\neq_2}{f}=\plholant{\Omega';=_2}{g}$ by showing a one-to-one correspondence between terms in the two sums of product.

Notice that every vertex labeled by $\neq_2$ in $G$ has degree 2.
If we remove such a vertex and merge its two incident edges, we obtain a $2n$-regular graph $G_1$. 
A valid $0$–$1$ edge assignment of $\Omega$ then corresponds to an even orientation of $G_1$, where $0$ denotes an incoming edge and $1$ denotes an outgoing edge.
Also, a valid $0$–$1$ edge assignment of $\Omega'$ corresponds to an even coloring of $G_1$, where $0$ denotes green and $1$ denotes red.

Without loss of generality, we may assume that $G$ is connected, as well as $G_1$. 
It is well known that a connected planar graph is Eulerian if and only if its dual is bipartite~\cite{WELSH1969375}. 
Since $G_1$ is planar and $2n$-regular, its dual $G_1^*$ is bipartite. 
Hence, we can color the faces of $G_1$ with two colors, say {\it black} and {\it white}, so that any two adjacent faces (i.e., faces sharing an edge) have different colors. 
We assume that the outer face of $G_1$ is colored white. 
We then orient the edges along each black face counterclockwise. 
This yields an Eulerian orientation of $G_1$, where the in-degree equals the out-degree at every vertex, as shown in Figure~\ref{fig: BW_orientation}. 
We refer to this as the {\it canonical orientation} of $G_1$, 
and to the corresponding $0$-$1$ edge assignment of $\Omega$ as the {\it canonical assignment}, 
denoted by $\tau_1$ and $\tau$, respectively.

\begin{figure}[h!]
\centering

\includegraphics[width=0.44\linewidth]{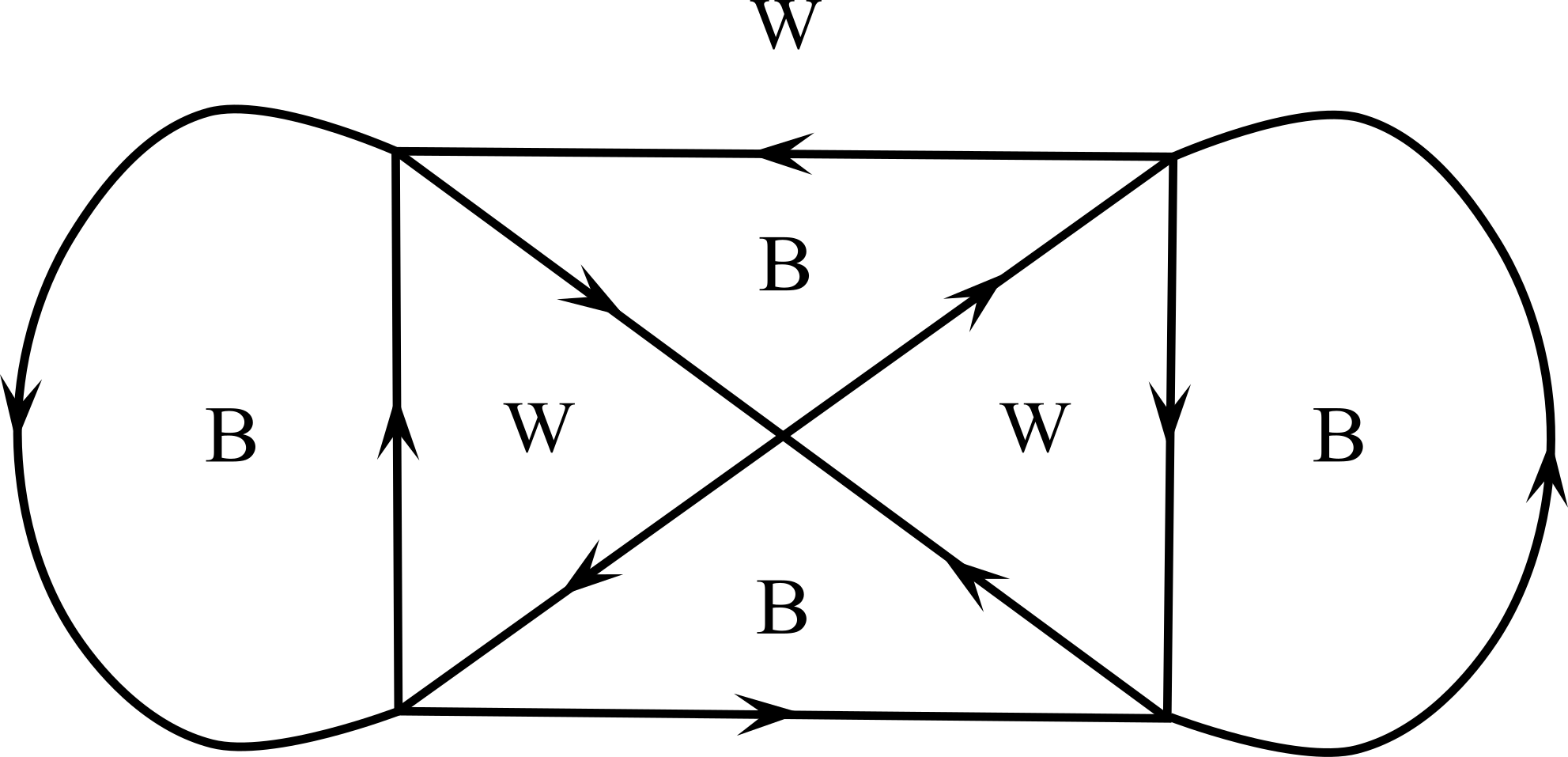}

\caption{A proper 2-coloring of the faces of a planar 4-regular graph $G_1$ and its canonical orientation $\tau_1$. 
All black faces are orientated counterclockwise.}
\label{fig: BW_orientation}
\end{figure}

For any even orientation $\sigma_1\in \mathscr{O}(G_1)$, define an even coloring $\sigma'_1\in \mathscr{C}(G_1)$ by taking ``XOR" between $\sigma_1$ and the canonical orientation $\tau_1$:
$\sigma'_1(e)=green$ if $\sigma_1(e) = \tau_1(e)$, and $\sigma'_1(e)=red$ otherwise.
This defines a bijection between
$\mathscr{O}(G_1)$ and $\mathscr{C}(G_1)$.
Note that $\sigma_1$ and $\sigma'_1$ correspond to two valid $0$-$1$ edge assignments $\sigma$ and $\sigma'$ of $\Omega$ and $\Omega'$, with the relation $\sigma'(e)=\sigma(e)\oplus\tau(e)$ for every edge $e$ of $G$.
This defines a bijection between valid $0$-$1$ edge assignments of $\Omega$ and $\Omega'$.

Next, we show that $\sigma$ and $\sigma'$ have the same evaluation in $\plholant{\Omega;\neq_2}{f}$ and $\plholant{\Omega';=_2}{g}$.
Suppose in $\Omega$, a signature $f$ labeled at an arbitrary vertex $v$ evaluates to $f(x_1,x_2,\ldots,x_{2n-1},x_{2n})=a$ under $\sigma$. 
There are two possible local configurations at vertex $v$ with respect to the black–white face coloring around it. 
Namely, the faces can be colored W–B-$\ldots$-W-B or B–W–$\ldots$-B-W in counterclockwise order, starting from the face to the left of the leading edge variable of $f$.
Since we orient all black faces counterclockwise, the canonical assignment $\tau_1$ around $v$ is $10\ldots10$ or $01\ldots01$.
So $\sigma'|_{E(v)}=\sigma|_{E(v)}\oplus \tau|_{E(v)}=x_1x_2\ldots x_{2n-1}x_{2n}\oplus01\ldots01=x_1\overline{x_2}\ldots,x_{2n-1}\overline{x_{2n}}$ or $\sigma'|_{E(v)}=\overline{x_1}x_2\ldots\overline{x_{2n-1}}x_{2n}$.
Thus, in $\Omega'$, the signature $g$ labeled at $v$ evaluates to $g(x_1,\overline{x_2},\ldots,x_{2n-1},\overline{x_{2n}})=a$ or $g(\overline{x_1},x_2,\ldots,\overline{x_{2n-1}},x_{2n})=a$.
Therefore, we prove that the corresponding assignments $\sigma$ and $\sigma'$ has the same evaluation in $\plholant{\Omega;\neq_2}{f}$ and $\plholant{\Omega';=_2}{g}.$
The lemma follows.
\end{proof}

\begin{corollary}\cite{CaiL20}\label{cor: combinatorial transformation from even coloring}
    $\plholant{\neq_2}{\left[\begin{smallmatrix} a& 0 & 0 & b \\ 
    0 & c & d & 0 \\ 
    0 & d & c & 0 \\ 
    b & 0 & 0 & a \end{smallmatrix}\right]}
    \equiv_T^p
    \plholant{=_2}{\left[\begin{smallmatrix} d & 0 & 0 & c \\ 
    0 & b & a & 0 \\ 
    0 & a & b & 0 \\ 
    c & 0 & 0 & d \end{smallmatrix}\right]}$.
\end{corollary}

Note that a holographic transformation by $Z$ or $HZ$ transforms $=_2$ into $\neq_2$.
Combined with Corollary~\ref{cor: combinatorial transformation from even coloring}, we have the following transformation within the planar eight-vertex model.

\begin{lemma}\cite{CaiL20}\label{lm:even coloring transformation}
    Let $M_{\rm Z}^{\rm PL} = \frac{1}{2}\left[\begin{smallmatrix}
        1 & -1 & 1 & -1\\
        -1 & 1 & 1 & -1\\
        1 & 1 & 1 & 1\\
        1 & 1 & -1 & -1
    \end{smallmatrix}\right]$, $M_{\rm HZ}^{\rm PL} = \frac{1}{2}\left[\begin{smallmatrix}
        1 & -1 & 1 & 1\\
        -1 & 1 & 1 & 1\\
        1 & 1 & 1 & -1\\
        1 & 1 & -1 & 1
    \end{smallmatrix}\right].$
    For any signature  $f$  with
    $M(f) = \left[\begin{smallmatrix} 
    a & 0 & 0 & b \\ 
    0 & c & d & 0 \\ 
    0 & d & c & 0 \\ 
    b & 0 & 0 & a \end{smallmatrix}\right]$,
    define $M(f') = \left[\begin{smallmatrix} a' & 0 & 0 & b' \\
    0 & c' & d' & 0 \\ 
    0 & d' & c' & 0 \\
    b' & 0 & 0 & a' \end{smallmatrix}\right]$ with $\left[\begin{smallmatrix}b' \\ c' \\ d' \\ a'\end{smallmatrix}\right]=
    M \left[\begin{smallmatrix}b \\ c \\ d \\ a\end{smallmatrix}\right]$, where $M\in \{M_{\rm Z}^{\rm PL},M_{\rm HZ}^{\rm PL}\}.$
    Then $\plholant{\neq_2}{f} \equiv_T^p \plholant{\neq_2}{f'}$.
\end{lemma}

\begin{remark}
    In the following we will only apply Lemma~\ref{lm:even coloring transformation} to obtain \numP-hardness result and to discover new tractable type.
    However, the general Lemma~\ref{lm: even coloring connection in general} is of independent interest. The combined paradigm of combinatorial transformation using Lemma~\ref{lm: even coloring connection in general} and holographic transformation via $Z$ or $HZ$ may offer a general framework for discovering new tractable cases in counting even orientations on planar graphs.
    If, after applying these two transformations, the resulting problem corresponds to a tractable case of counting Eulerian orientations on planar graphs, then we also obtain a tractable case for counting even orientations on planar graphs.
    However, few tractable cases are currently known for planar counting of Eulerian orientations (apart from affine, product-type, matchgate signatures). 
    Thus, this direction remains open for future investigation.
\end{remark}

In the following, we directly apply Lemma~\ref{lm:even coloring transformation} and invoke the \numP-hardness of non-singular redundant signatures to handle the case where $(b = -1)\lor(c = -1)$, which is exactly the setting where the conformal lattice interpolation in the next subsection does not apply.

\begin{lemma}\label{lm: ars b=-1 or c=-1}
    Let $M(f)=
    \left[\begin{smallmatrix} 
    1 & 0 & 0 & b \\ 
    0 & c & d & 0 \\
    0 & d & c & 0 \\ 
    b & 0 & 0& 1
    \end{smallmatrix}\right]$.
    If $b=-1$ or $c=-1$, then $\plholant{\neq_2}{f}$ is \numP-hard unless $f\in \M$, in which case $\plholant{\neq_2}{f}$ is tractable.

\end{lemma}

\begin{proof}
    We only need to prove the case $b=-1$, as a $\pi/2$ rotation $f^{\frac{\pi}{2}}$
    exchanges $b$ and $c$, and $f\in \M \Leftrightarrow f^{\frac{\pi}{2}}\in \M$ by Lemma~\ref{lm:is matchgate}.
    
    We have 
    $M_{\rm Z}^{\rm PL}\left[\begin{smallmatrix} 
    -1 \\ c \\ d \\ 1
    \end{smallmatrix}\right]=
    \left[\begin{smallmatrix} 
    -2-c+d \\ c+d \\ c+d \\ -2+c-d
    \end{smallmatrix}\right]. 
  $
 Then $\plholant{\neq_2}{f}\equiv_T^p\plholant{\neq_2}{f'}$ by Lemma~\ref{lm:even coloring transformation}, where $f'$ is a signature with the signature matrix
  $ M(f')=\left[\begin{smallmatrix} 
    -2+c-d & 0 & 0 & -2-c+d \\ 
    0 & c+d & c+d & 0 \\
    0 & c+d & c+d & 0 \\ 
    -2-c+d & 0 & 0& -2+c-d
    \end{smallmatrix}\right].$
    By Theorem~\ref{thm:redundant}, if 
    $ \mathrm{det}\left[\begin{smallmatrix} 
    -2+c-d & 0  & -2-c+d \\ 
    0 & c+d& 0 \\
    -2-c+d & 0& -2+c-d
    \end{smallmatrix}\right]\neq0$, 
    i.e., $(c+d)(c-d)\neq 0$, then  $\plholant{\neq_2}{f'}$ is \numP-hard, which implies that $\plholant{\neq_2}{f}$ is \numP-hard.
    If $(c+d)(c-d)=0$, then $1-b^2=c^2-d^2=0$. If follows that $f\in \M$ by Lemma~\ref{lm:is matchgate}.
\end{proof}

\subsection{Apply conformal lattice interpolation}\label{subsec: apply conformal lattice interpolation}

After a $Z$-holographic transformation to the $\PlHolant{(\tilde{f})}$ setting, $M(\tilde{f})$ can be diagonalized with eigenvalues $c+d,-1+b,-c+d$ and $1+b$.
The case $b=-1$ has been handled in Lemma~\ref{lm: ars b=-1 or c=-1}.
Thus, we can normalize by $1+b$ and the partition function $\plholant{\Omega;\neq_2}{f}$ can be expressed as a multivariate polynomial of those normalized eigenvalues $(u,v,w)=(\frac{c+d}{1+b},\frac{-1+b}{1+b},\frac{-c+d}{1+b})$.
Thus, by conformal lattice interpolation, if a signature $M(g)$ has a finer lattice defined by its eigenvalues (after $Z$-transformation and normalizing), we can compute $\PlHolant_{\Omega}$ for $\plholant{\neq_2}{g}$ by computing it for $\plholant{\neq_2}{f}$:

\begin{lemma}\label{lm:interpolation xyz}
    Let $f$ be a signature with the signature matrix $M(f) = \left[\begin{smallmatrix} 1 & 0 & 0 & b \\ 0 & c & d & 0 \\ 0 & d & c & 0 \\ b & 0 & 0 & 1 \end{smallmatrix}\right]$ with $bcd \neq 0,b\neq -1$. Let $(u,v,w)=(\frac{c+d}{1+b},\frac{-1+b}{1+b},\frac{-c+d}{1+b})$, and $(x,y,z)\in \C^3$. 
    If the respective lattices corresponding to the two tuples satisfy $L_{(u,v,w)}\subseteq L_{(x,y,z)}$, then for any signature set $\mathcal{F}$ containing $f$, we have $\plholant{\neq_2}{\mathcal{F}\cup g}\le_T^p\plholant{\neq_2}{\mathcal{F}}$, where
    $M(g) = \left[\begin{smallmatrix} 
    1-y & 0 & 0 & 1+y \\ 
    0 & x-z & x+z & 0 \\ 
    0 & x+z & x-z & 0 \\ 
    1+y & 0 & 0 & 1-y \end{smallmatrix}\right]$.
\end{lemma}

\begin{proof}
     By a holographic transformation with $Z =\frac{1}{\sqrt{2}} \left[\begin{smallmatrix} 
     1 & 1 \\ 
     \mathfrak{i} & -\mathfrak{i}
     \end{smallmatrix}\right]$, 
     we have $(\neq_2) (Z^{-1})^{\otimes 2} = (=_2)$, and 
     therefore, by Valiant's Holant Theorem~\ref{thm: holographic transformation}, $\plholant{\neq_2}{\mathcal{F}}\equiv_T^p\plholant{=_2}{\widetilde{\mathcal{F}}}$,
     where $\widetilde{\mathcal{F}} = Z {\mathcal{F}}$ denotes the set of transformed signatures.
     For the given $f \in \mathcal{F}$, the  signature matrix of
     $\tilde{f} = Z^{\otimes 4}f  \in \widetilde{\mathcal{F}}$ is
     $$M({\tilde{f}}) = Z^{\otimes 2} M(f)(Z^{T})^{\otimes 2} = \frac{1}{2}\left[\begin{smallmatrix}
        1+b+c+d & 0 & 0 & -1-b+c+d\\
        0 & -1+b-c+d & -1+b+c-d & 0\\
        0 & -1+b+c-d & -1+b-c+d & 0\\
        -1-b+c+d & 0 & 0 & 1+b+c+d
    \end{smallmatrix}\right].$$ 
    Let $P= \left[\begin{smallmatrix}
        1 & 0 & 0 & 1\\
        0 & 1 & 1 & 0\\
        0 & 1 & -1 & 0\\
        1 & 0 & 0 & -1
    \end{smallmatrix} \right]$. Then $P^2 = 2I_4$ where $I_4$ denotes the identity matrix of size 4. Also, we have $$M({\tilde{f}})  = \frac{1}{2}P\left[\begin{smallmatrix}
        c+d & 0 & 0 & 0\\
        0 & -1+b & 0 & 0\\
        0 & 0 & -c+d & 0\\
        0 & 0 & 0 & 1+b
    \end{smallmatrix} \right]P
    =\frac{1+b}{2}P\left[\begin{smallmatrix}
        u & 0 & 0 & 0\\
        0 & v & 0 & 0\\
        0 & 0 & w & 0\\
        0 & 0 & 0 & 1
    \end{smallmatrix} \right]P,$$
    and 
    $$M({\tilde{g}}) =Z^{\otimes 2}M(g)(Z^T)^{\otimes2} =\left[\begin{smallmatrix}
        1+x & 0 & 0 & -1+x\\
        0 & y+z & y-z & 0\\
        0 & y-z & y+z & 0\\
        -1+x & 0 & 0 & 1+x
    \end{smallmatrix} \right]
    =P\left[\begin{smallmatrix}
        x & 0 & 0 & 0\\
        0 & y & 0 & 0\\
        0 & 0 & z & 0\\
        0 & 0 & 0 & 1
    \end{smallmatrix} \right]P.$$

    To prove the lemma we only need to   show $\plholant{=_2}{\widetilde{\mathcal{F}}\cup\tilde{g}}\le_T^p\plholant{=_2}{\widetilde{\mathcal{F}}}$, where $\tilde{g}=Z^{\otimes4}g$. 
    We construct a series of gadgets $\tilde{f}_s$ by a chain of $s$ many copies of $\tilde{f}$ linked by the double {\sc Equality}.
    Clearly $\tilde{f}_{s}$ has the signature matrix 
$M(\tilde{f}_{s})=(\frac{1+b}{2})^s(P D P)^s=\frac{(1+b)^s}{2}PD^sP,
$
where $D=\mathrm{diag}(u,v,w,1)$, and $D^s=\mathrm{diag}(u^s,v^s,w^s,1)$.

The matrix $M(\tilde{f_s})$ has a good form for polynomial interpolation. 
Suppose $\tilde{g}$ appears $m$ times in an instance $\Omega$ of $\plholant{=_2}{{\widetilde{\mathcal{F}}}\cup\tilde{g}}$.
For every $s\ge 1$, we replace each appearance of $\tilde{g}$ by a copy of
the gadget $\tilde{f}_{s}$ to get an instance $\Omega_{s}$ of $\plholant{=_2}{\widetilde{\mathcal{F}}}$.
We can treat each of the $m$ appearances of $\tilde{f_s}$ in $\Omega_s$ as a new gadget composed of three signatures in sequence $P$, $D^s$ and $P$ (ignoring the constant factor $(1+b)^s/2$), and denote this new instance by $\Omega_s'$.
We divide $\Omega_{s}'$ into two parts A and B.
Part A consists of $m$ signatures $D^s$ and the signature of part A is  $(D^s)^{\otimes{m}}$. 
We can write $(D^s)^{\otimes{m}}$ as a column vector indexed by $\{0, 1\}^{4m}$.
Part B consists of the rest of $\Omega_{s}'$, including $2m$ occurrences of $P$, and the signature of part B is  a tensor expressed as a row vector $R$. 
Then, the Holant value of $\Omega_{s}'$ is the dot product $\langle R, (D^s)^{\otimes{m}}\rangle$, which is a summation  over all $\{0, 1\}$-assignments for  the $4m$ edges connecting the two parts. 
We can stratify all $\{0, 1\}$-assignments of these $4m$ bits having a nonzero evaluation of a term in Pl-Holant$_{\Omega_{s}'}$ into the following categories:
\begin{itemize}
\item
There are $j_1$ many copies of $D^s$ receiving inputs $0000$ 
\item
There are $j_2$ many copies of $D^s$ receiving inputs $0101$;
\item
There are $j_3$ many copies of $D^s$ receiving inputs $1010$;
\end{itemize}
where $j_1+j_2+j_3\le m$ (and $m-(j_1+j_2+j_3)$ many copies of $D^s$ receiving inputs $1111$).

For any assignment in the category with parameter $(j_1,j_2,j_3)$, the evaluation of
$(D^s)^{\otimes{m}}$ is clearly $\frac{(1+b)^{sm}}{2^m}(u^{j_1}v^{j_2}w^{j_3})^s$. 
Let $c_{j_1j_2j_3}$ be the summation of values of  part B over all assignments in the category $(j_1,j_2,j_3)$.
Note that $c_{j_1j_2j_3}$ is independent from the value of $s$ since we view the gadget $\tilde{f}_{s}$ as a block. 
Then, we rewrite the dot product summation and get 
$$\PlHolant_{\Omega_{s}}=\PlHolant_{\Omega_{s}'}=\langle R, (D^s)^{\otimes{m}}\rangle=
\sum\limits_{\substack{j_1,j_2,j_3 \geq 0 \\ j_1+j_2+j_3\leq m}} \frac{(1+b)^{sm}}{2^m}(u^{j_1}v^{j_2}w^{j_3})^s c_{j_1,j_2,j_3}.$$
Under this stratification, the Holant value of $\Omega$
is
$$\PlHolant_{\Omega} =
\sum\limits_{\substack{j_1,j_2,j_3 \geq 0 \\ j_1+j_2+j_3\leq m}} x^{j_1}y^{j_2}z^{j_3} c_{j_1,j_2,j_3}.$$
Since $L_{(u,v,w)}\subseteq L_{(x,y,z)}$, we can compute $\PlHolant_{\Omega}$ in polynomial time given $\PlHolant_{\Omega_{s}}$ for $1\le s \le \binom{m+3}{3}$ by Lemma~\ref{lm:general interpolation}.
This proves that $\plholant{=_2}{\widetilde{\mathcal{F}}\cup\tilde{g}}\le_T^p\plholant{=_2}{\widetilde{\mathcal{F}}}$.
\end{proof}

\begin{lemma}\label{lm:interpolate disequality crossover}
    Let $M(f) = \left[\begin{smallmatrix} 1 & 0 & 0 & b \\ 0 & c & d & 0 \\ 0 & d & c & 0 \\ b & 0 & 0 & 1 \end{smallmatrix}\right]$ with $bcd \neq 0,b\neq -1$. 
    Let $(u,v,w)=(\frac{c+d}{1+b},\frac{-1+b}{1+b},\frac{-c+d}{1+b})$.
    If $L_{(u,v,w)}\subseteq L_{(1,1,-1)}$, then  $\plholant{\neq_2}{\mathcal{F}}\equiv_T^p \holant{\neq_2}{\mathcal{F}}$ for any signature set $\mathcal{F}$ containing $f$.
\end{lemma}

\begin{proof}
    We only need to show $\holant{\neq_2}{\mathcal{F}}\le_T^p\plholant{\neq_2}{\mathcal{F}}$; the other direction is trivial. 
    By Lemma~\ref{lm:interpolation xyz}, we can interpolate signature $g$ with signature matrix $
    M(g) = \left[\begin{smallmatrix} 
    0 & 0 & 0 & 2 \\ 
    0 & 2 & 0 & 0 \\ 
    0 & 0 & 2 & 0 \\ 
    2 & 0 & 0 & 0 \end{smallmatrix}\right]$, i.e., $\plholant{\neq_2}{\mathcal{F}\cup g}\le_T^p\plholant{\neq_2}{\mathcal{F}}.$ 
    Note that $g$ is the disequality crossover $\mathcal{S'}$ up to a scalar. 
    Then by Lemma~\ref{lm:disequality crossover}, we have $\holant{\neq_2}{\mathcal{F}}\le_T^p \plholant{\neq_2}{\mathcal{F}\cup g}$.
    The lemma follows.
\end{proof}

\begin{lemma}\label{lm:interpolation something hard}
    Let $M(f) = \left[\begin{smallmatrix} 1 & 0 & 0 & b \\ 0 & c & d & 0 \\ 0 & d & c & 0 \\ b & 0 & 0 & 1 \end{smallmatrix}\right]$ with $bcd \neq 0,b\neq -1$. 
    Let $u=\frac{c+d}{1+b},v=\frac{-1+b}{1+b},w=\frac{-c+d}{1+b}$, and $x,y,z\in \C$.
    If $L_{(u,v,w)}\subseteq L_{(x,y,z)}$ and $(x,y,z)$ satisfies any one of the following conditions, then $\plholant{\neq_2}{f}$ is \numP-hard.
    \begin{enumerate}[label=\textup{(\roman*)}]
        \item \label{cd:six vertex model}
        $(y=1)\land(x\neq \pm z)\land(xz\neq 1);$ 
        
        \item \label{cd:y=-1 not matchgate}
        $(y=-1)\land(x\neq \pm z)\land(xz\neq -1);$

        \item \label{cd:y=-1 two 0 pairs}
        $(y=-1)\land (x=z)\land (x\neq 0,\pm1,\pm i)$;
        \item \label{cd:x=z not matchgate} 
        $(x=z)\land(y\neq \pm 1)\land(x\neq 0)\land(y\neq x^2);$
        
        \item \label{cd:x+z=0 not matchgate}
        $(x+z=0)\land(y\neq \pm 1)\land(x\neq z)\land(y\neq -x^2);$
        
        \item \label{cd: redundent}$(1+y=x+z)\land(y\neq-1)\land (x\neq 1)\land(z\neq1).$
    \end{enumerate}
\end{lemma}

\begin{proof}
     By Lemma~\ref{lm:interpolation xyz}, we can interpolate signature $g$ with the signature matrix 
    $$
    M(g)=
    \left[\begin{smallmatrix} 
    1-y & 0 & 0 & 1+y \\ 
    0 & x-z & x+z & 0 \\ 
    0 & x+z & x-z & 0 \\ 
    1+y & 0 & 0 & 1-y
    \end{smallmatrix}\right]
    $$ and $\plholant{\neq_2}{g}\le_T^p\plholant{\neq_2}{f}.$
    
    Suppose $(x,y,z)$ satisfies condition~\ref{cd:six vertex model}, then $
    M(g)=
    \left[\begin{smallmatrix} 
    0 & 0 & 0 & 2 \\ 
    0 & x-z & x+z & 0 \\ 
    0 & x+z & x-z & 0 \\ 
    2 & 0 & 0 & 0
    \end{smallmatrix}\right].
    $
    Since $x\pm z\neq 0$, $\plholant{\neq_2}{g}$ is \numP-hard unless $g\in \M$ or $\widehat{\M}$ by Theorem~\ref{thm:planar six-vertex}.
    If $g\in \M$, we have $\det M_{\rm{Out}}(g) = \det M_{\rm{In}}(g)$ by Lemma~\ref{lm:is matchgate}, i.e., $(x-z)^2-(x+z)^2=-4$, contradicting $xz\neq 1$.
    If $g\in \widehat{\M}$, we have $g'\in \M$, where $g'$ has the signature matrix
    $
    M(g')=H^{\otimes2}M(g)H^{\otimes2}=
    \left[\begin{smallmatrix} 
    1+x & 0 & 0 & 1-x \\ 
    0 & -1-z & -1+z & 0 \\ 
    0 & -1+z & -1-z & 0 \\ 
    1-x & 0 & 0 & 1+x
    \end{smallmatrix}\right].
    $
    By Lemma~\ref{lm:is matchgate} we have $\det M_{\rm{Out}}(g') = \det M_{\rm{In}}(g')$, which implies $x=z$, contradicting condition~\ref{cd:six vertex model}.
    Thus, $g\notin \M \cup \widehat{\M}$, and $\plholant{\neq_2}{g}$ is \numP-hard.
    Then $\plholant{\neq_2}{f}$ is \numP-hard.

    Suppose $(x,y,z)$ satisfies condition~\ref{cd:y=-1 not matchgate}, then
    $
    M(g)=
    \left[\begin{smallmatrix} 
    2 & 0 & 0 & 0 \\ 
    0 & x-z & x+z & 0 \\ 
    0 & x+z & x-z & 0 \\ 
    0 & 0 & 0 & 2
    \end{smallmatrix}\right].
    $
    Since $x\pm z\neq 0$, $\plholant{\neq_2}{g}$ is \numP-hard unless $g\in \M$ by Lemma~\ref{lm:1cdwz1}.
    If $g\in \M$, we have $(x-z)^2-(x+z)^2=4$ by Lemma~\ref{lm:is matchgate}, contradicting $xz\neq -1$.
    Thus, $\plholant{\neq_2}{g}$ is \numP-hard and so $\plholant{\neq_2}{f}$ is \numP-hard.

     Suppose $(x,y,z)$ satisfies condition~\ref{cd:y=-1 two 0 pairs}, then
    $
    M(g)=
    2\left[\begin{smallmatrix} 
    1 & 0 & 0 & 0 \\ 
    0 & 0 & x & 0 \\ 
    0 & x & 0 & 0 \\ 
    0 & 0 & 0 & 1
    \end{smallmatrix}\right].
    $
    Since $x\neq 0,\pm1,\pm i$, $\plholant{\neq_2}{g}$ is \numP-hard by Lemma~\ref{lm:case d=w}.

    Suppose $(x,y,z)$ satisfies condition~\ref{cd:x=z not matchgate}, then
    $
    M(g)=
    \left[\begin{smallmatrix} 
    1-y & 0 & 0 & 1+y \\ 
    0 & 0 & 2x & 0 \\ 
    0 & 2x & 0 & 0 \\ 
    1+y & 0 & 0 & 1-y
    \end{smallmatrix}\right].
    $ 
    Since $y\neq\pm1$ and $x+z=2x\neq0$, $\plholant{\neq_2}{g}$ is \numP-hard unless $g\in \M$ by Lemma~\ref{lm:1cdwz1} (after a rotation).
    If $g\in \M$, we have $-4x^2=(1-y)^2-(1+y)^2$ by Lemma~\ref{lm:is matchgate}, contradicting $y\neq x^2$.
    Thus, $\plholant{\neq_2}{g}$ is \numP-hard and so $\plholant{\neq_2}{f}$ is \numP-hard.

    Suppose $(x,y,z)$ satisfies condition~\ref{cd:x+z=0 not matchgate}, then
    $
    M(g)=
    \left[\begin{smallmatrix} 
    1-y & 0 & 0 & 1+y \\ 
    0 & 2x & 0 & 0 \\ 
    0 & 0 & 2x & 0 \\ 
    1+y & 0 & 0 & 1-y
    \end{smallmatrix}\right].
    $ 
    Since $y\neq\pm1$ and $x-z=2x\neq0$, $\plholant{\neq_2}{g}$ is \numP-hard unless $g\in \M$ by Lemma~\ref{lm:1bycz1}.
    If $g\in \M$, we have $4x^2=(1-y)^2-(1+y)^2$ by Lemma~\ref{lm:is matchgate}, contradicting $y\neq -x^2$.
    Thus, $\plholant{\neq_2}{g}$ is \numP-hard and so $\plholant{\neq_2}{f}$ is \numP-hard.

     Suppose $(x,y,z)$ satisfies condition~\ref{cd: redundent}, then
    $
    M(g)=
    \left[\begin{smallmatrix} 
    1-y & 0 & 0 & 1+y \\ 
    0 & x-z & 1+y & 0 \\ 
    0 & 1+y & x-z & 0 \\ 
    1+y & 0 & 0 & 1-y
    \end{smallmatrix}\right].
    $
    Note that
    $
    M(g^{\frac{\pi}{2}})=
    \left[\begin{smallmatrix} 
    1-y & 0 & 0 & x-z \\ 
    0 & 1+y & 1+y & 0 \\ 
    0 & 1+y & 1+y & 0 \\ 
    x-z & 0 & 0 & 1-y
    \end{smallmatrix}\right]
    $
    is redundant signature.
    By Theorem~\ref{thm:redundant}, $\plholant{\neq_2}{g}\equiv_T^p \plholant{\neq_2}{g^{\frac{\pi}{2}}}$ is \numP-hard unless $\det
    \left[\begin{smallmatrix} 
    1-y &  0 & x-z \\ 
    0   & 1+y & 0 \\ 
    x-z & 0  & 1-y
    \end{smallmatrix}\right]
    =(1+y)((1-y)^2-(x-z)^2)=0$.
    If $(1+y)((1-y)^2-(x-z)^2)=0$, then $(1-y)^2=(x-z)^2$ since $y\neq -1$.
    Then $(1+y)^2-(1-y)^2=(x+z)^2-(x-z)^2$ since $1+y=x+z$. It follows that $y=xz$.
    Plug $y=xz$ into $1+y=x+z$, we have $(x-1)(z-1)=0$, contradicting $x\neq 1$ and $z\neq 1$.
    Thus, $\plholant{\neq_2}{g}$ is \numP-hard and so $\plholant{\neq_2}{f}$ is \numP-hard.
\end{proof}

\subsection{Lattice rank at most 1}\label{subsec: rank<=1}

\begin{lemma}\label{lm: rank<=1}
      Let $f$ be a signature with the signature matrix 
      $M(f) = \left[\begin{smallmatrix} 1 & 0 & 0 & b \\ 0 & c & d & 0 \\ 0 & d & c & 0 \\ b & 0 & 0 & 1 \end{smallmatrix}\right],$ with $bcd \neq 0,b\neq -1$. 
    Let $(u,v,w)=(\frac{c+d}{1+b},\frac{-1+b}{1+b},\frac{-c+d}{1+b})$. If $\mathrm{rank}(L_{(u,v,w)})\leq 1 $, then one of the following holds:
    \begin{itemize}
        \item 
        $\plholant{\neq_2}{\mathcal{F}}\equiv_T^p \holant{\neq_2}{\mathcal{F}}$ for any signature set $\mathcal{F}$ containing $f$;
        \item $\plholant{\neq_2}{f}$ is \numP-hard;
        \item $f$ is $\M$-transformable.
    \end{itemize}
\end{lemma}
\begin{proof}
    Suppose $\mathrm{rank}(L_{(u,v,w)})=0$, i.e.,  $L_{(u,v,w)}= \{(0,0,0)\}$.
    We can choose any $(x,y,z)\in \C^3$ that satisfies condition~\ref{cd:six vertex model} in Lemma~\ref{lm:interpolation something hard}, and clearly $L_{(u,v,w)}\subseteq L_{(x,y,z)}$.
    Then by Lemma~\ref{lm:interpolation something hard}, $\plholant{\neq_2}{f}$ is \numP-hard.

    Now suppose $\mathrm{rank}(L_{(u,v,w)})=1$ and $L_{(u,v,w)}=\braket{(p,q,r)}$,
    the lattice with a basis vector $(p,q,r)$, where $p,q,r\in \Z$. 
    Without loss of generality, we may assume $p\ge 0$.
    Then for $x,y,z\in \C$, $L_{(u,v,w)}\subseteq L_{(x,y,z)}$ is equivalent to $x^py^qz^r=1$. 
    In the following, corresponding to different values of $(p,q,r)$, we do polynomial interpolation and show that either  $\plholant{\neq_2}{f}\equiv_T^p \holant{\neq_2}{f}$, or $\plholant{\neq_2}{f}$ is \numP-hard, or $(p,q,r)\in\{(1,1,-1),(1,-1,-1),(1,-1,1)\}$.

    If $r$ is even, let $(x,y,z)=(1,1,-1)$.
    Then $x^py^qz^r=1$, thus $L_{(u,v,w)}\subseteq L_{(x,y,z)}$.
    By Lemma~\ref{lm:interpolate disequality crossover}, $\plholant{\neq_2}{\mathcal{F}}\equiv_T^p \holant{\neq_2}{\mathcal{F}}$ for any signature set $\mathcal{F}$ containing $f$.
    
    If $p\pm r\neq 0$, let $(x,y,z)=(t^r,1,t^{-p})$, where $t\in \C$ is to be determined.
    Then $x^py^qz^r=1$.
    Note that $x=\pm z$ is equivalent to $t^{r+p}=\pm 1$, and $xz=1$ is equivalent to $t^{r-p}=1$.
    Since $p\pm r\neq 0$, there are finitely many $t$ satisfying $t^{r+p}=\pm 1$ or $t^{r-p}=1$. 
    Select $t\in \C$ such that $(t^{r+p}\neq \pm1)\land (t^{r-p}\neq 1)$, then $(x\neq \pm z)\land (xz\neq 1)$.
    Thus, $(x,y,z)$ satisfies condition~\ref{cd:six vertex model} in Lemma~\ref{lm:interpolation something hard}, and $\plholant{\neq_2}{f}$ is \numP-hard.

    If $p=r\ge 3$, let $(x,y,z)=(t\zeta_p,1,t^{-1})$, where $t\in \C$ is to be determined.
    Then $x^py^qz^r=1$, and $xz=\zeta_p\neq 1$.
    Note that $x\pm z=0$ is equivalent to $t^2=\pm \zeta_p^{-1}$.
    Select $t$ such that $t^2\neq \pm \zeta_p^{-1}$, then $x\pm z\neq 0$.
    Thus, $(x,y,z)$ satisfies condition~\ref{cd:six vertex model} in Lemma~\ref{lm:interpolation something hard}, and $\plholant{\neq_2}{f}$ is \numP-hard.

    If $p=-r\ge 3$, let $(x,y,z)=(t\zeta_p,1,t)$, where $t\in \C$ is to be determined.
    Then $x^py^qz^r=1$.
    Note that $x\pm z=0$ is equivalent to $t=0$.
    Select $t\neq 0$, then $x\pm z\neq 0$.
    Thus, $(x,y,z)$ satisfies condition~\ref{cd:six vertex model} in Lemma~\ref{lm:interpolation something hard}, and $\plholant{\neq_2}{f}$ is \numP-hard.

    For the remaining cases we may assume $p=\pm r=1$. 
    Next we claim $|q|=1$, otherwise $\plholant{\neq_2}{f}$ is \numP-hard.

    If $|q|\ge 3$, let $(x,y,z)=(1,\zeta_{|q|},1)$.
    Then $x^py^qz^r=1$.
    Note that $(x,y,z)$ satisfies condition~\ref{cd:x=z not matchgate} in Lemma~\ref{lm:interpolation something hard}.
    Thus, $\plholant{\neq_2}{f}$ is \numP-hard.

    If $|q|=2$, let $y=-1$. 
    If $p=r=1$, let $x=t$ and $z=\frac{1}{t}$, where $t\in\C$ is to be determined.
    Then $x^py^qz^r=1$.
    Note that $x\pm z=0$ is equivalent to $t^2=\pm 1$. 
    Select $t\in \C$ such that $t^2\neq \pm 1$, then $(x,y,z)$ satisfies condition~\ref{cd:y=-1 not matchgate} in Lemma~\ref{lm:interpolation something hard}.
    Thus, $\plholant{\neq_2}{f}$ is \numP-hard.
    If $p=-r=1$, let $x=z=2$.
    Then $x^py^qz^r=1$.
    By Lemma~\ref{lm:interpolation xyz}, we can interpolate $g$ with the signature matrix
   $M(g) = \left[\begin{smallmatrix} 
    1-y & 0 & 0 & 1+y \\ 
    0 & x-z & x+z & 0 \\ 
    0 & x+z & x-z & 0 \\ 
    1+y & 0 & 0 & 1-y \end{smallmatrix}\right]
    =2\left[\begin{smallmatrix} 
    1 & 0 & 0 & 0 \\ 
    0 & 0 & 2 & 0 \\ 
    0 & 2 & 0 & 0 \\ 
    0 & 0 & 0 & 1 \end{smallmatrix}\right].$
    By Lemma~\ref{lm:case d=w}, $\plholant{\neq_2}{g}$ is \numP-hard, so $\plholant{\neq_2}{f}$ is \numP-hard.
    
    Now assume $p=\pm r=1,q=\pm 1$. 
    If $p=q=r=1$, let $x=t,y=t^{-2},z=t$, where $t\in \C$ is to be determined.
    Then $x^py^qz^r=1$.
    Note that $x+z\neq 0$ if $t\neq 0$, and $y=x^2$ is equivalent to $t^4=1$.
    Select $t\in \C$ such that $t\neq 0$ and $t^4\neq 1$, then $(x,y,z)$ satisfies condition~\ref{cd:x=z not matchgate} in Lemma~\ref{lm:interpolation something hard}.
    Thus, $\plholant{\neq_2}{f}$ is \numP-hard.

    So far, we conclude that either  $\plholant{\neq_2}{f}\equiv_T^p \holant{\neq_2}{f}$, or $\plholant{\neq_2}{f}$ is \numP-hard, or $(p,q,r)\in\{(1,1,-1),(1,-1,-1),(1,-1,1)\}$.
    If $(p,q,r)=(1,1,-1)$, we have $uvw^{-1}=1$, which simplifies to $d=bc$.
    In this case, $f$ is $\M$-transformable by Lemma~\ref{lm:symmetric equality_matchgate-transformable}.
    If $(p,q,r)=(1,-1,-1)$, we have $uv^{-1}w^{-1}=1$, which simplifies to $d=-bc$.
    Again $f$ is $\M$-transformable by Lemma~\ref{lm:symmetric equality_matchgate-transformable}.
    If $(p,q,r)=(1,-1,1)$, we have $uv^{-1}w=1$, which simplifies to $1-b^2=c^2-d^2$.
    In this case, $f\in \M$ by Lemma~\ref{lm:is matchgate}.
\end{proof}

\begin{corollary}\label{cor:lattice contains certain elements then f is matchgate}
   Let $f$ be a signature with the signature matrix
   $M(f) = \left[\begin{smallmatrix} 
   1 & 0 & 0 & b \\ 
   0 & c & d & 0 \\ 
   0 & d & c & 0 \\ 
   b & 0 & 0 & 1 \end{smallmatrix}\right],$ with $bcd \neq 0,b\neq -1$. 
    Let $u=\frac{c+d}{1+b},v=\frac{-1+b}{1+b},w=\frac{-c+d}{1+b}$. 
    If $(1,1,-1)\in L_{(u,v,w)}$ or $(1,-1,-1)\in L_{(u,v,w)}$, then $f$ is $\M$-transformable.
    If $(1,-1,1)\in L_{(u,v,w)}$, then $f\in \M$.
\end{corollary}

\subsection{Lattice rank 2}\label{subsec: rank2}
In some cases when $\mathrm{rank}(L_{(u,w,v)})=2$, none of conditions~\ref{cd:six vertex model}, \ref{cd:y=-1 not matchgate}, \ref{cd:y=-1 two 0 pairs}, \ref{cd:x=z not matchgate}, \ref{cd:x+z=0 not matchgate} in Lemma~\ref{lm:interpolation something hard} is helpful to prove \numP-hardness of the problem $\plholant{\neq_2}{f}$.
Thus, we have to appeal to the most intricate condition~\ref{cd: redundent} in Lemma~\ref{lm:interpolation something hard}.
Condition~\ref{cd: redundent} requires us to find a ``good'' solution $(x,y,z)$ to $1+y=x+z$, where ``good'' means $x,z\neq 1$ and $y\neq -1$.
In these most intricate cases (\ref{case: p_2=0,p_1=1},~\ref{case: p_2=0,p_1=2},~\ref{case: p_2=1},~\ref{case: p_2=-1} in Lemma~\ref{lm: rank2 c neq pm bd}), the condition $L_{(u,v,w)}\subseteq L_{(x,y,z)}$ implies that $x,y,z$ can be expressed in monomials of some $t\in \C$, with norm 1 coefficients. 
Then we are asked to analyze the distribution of zeros of the Laurent polynomial $P(t):=x+z-y-1$ in terms of $t$.
If we are able to find a zero $t$ of $P(t)$ that lies outside the unit circle, then clearly $x,z\neq 1$ and $y\neq -1$.
It turns out that the following lemma is very powerful in the analysis.

\begin{lemma}\label{lm: roots on unit circle}
    Let $p(z)=\sum\limits_{k=0}^na_kz^k$ be a polynomial with complex coefficients, where $n\in \N, a_n\neq 0$. 
    If all the roots of $p(z)$ lie on the unit circle, then there exists $\theta\in \R$, such that $a_k=e^{\mathfrak{i}\theta}\overline{a_{n-k}}$ for $0\le k\le n$.
\end{lemma}
\begin{proof}
    By the Fundamental Theorem of Algebra, $p(z)$ has $n$ roots  $z_1,z_2,...,z_n$ 
    counting multiplicity.
    Then $p(z)=a_n\prod\limits_{k=1}^n(z-z_k)$, and $a_0=(-1)^na_n\prod\limits_{k=1}^nz_k$.
    By assumption we have $|z_k|=1$ for $1\le k \le n$.
    From $\prod\limits_{k=1}^nz_k=(-1)^n\frac{a_0}{a_n}$ we have $|a_0|=|a_n|\neq 0$.
    Consider the polynomial $p^*(z)=z^n\overline{p(\frac{1}{\overline{z}})}$.
    We have
    \begin{equation*}
        \begin{aligned}
            p^*(z) =&z^n\overline{a_n\prod\limits_{k=1}^n(\frac{1}{\overline{z}}-z_k)}
            =\overline{a_n}z^n\prod\limits_{k=1}^n(\frac{1}{z}-\overline{z_k})
            =\overline{a_n}z^n\prod\limits_{k=1}^n(\frac{1}{z}-\frac{1}{z_k})
            =\overline{a_n}z^n\prod\limits_{k=1}^n(\frac{z_k-z}{z_kz})\\
            =&\frac{\overline{a_n}(-1)^n}{\prod\limits_{k=1}^n{z_k}}
            \prod\limits_{k=1}^n({z-z_k})
            =\frac{\overline{a_n}a_n}{a_0}\prod\limits_{k=0}^n(z-z_k)
            =\frac{\overline{a_n}}{a_0}p(z).
        \end{aligned}
    \end{equation*}
    Since $|a_0|=|a_n|\neq 0$, there exists $\theta\in \R$ such that $\frac{\overline{a_n}}{a_0}=e^{-\mathfrak{i}\theta}$.
    Thus,
    \begin{equation*}
        \begin{aligned}
            p(z)
            =e^{\mathfrak{i}\theta}p^*(z)
            =e^{\mathfrak{i}\theta}z^n\overline{p(\frac{1}{\overline{z}})}
            =e^{\mathfrak{i}\theta}z^n\overline{\sum\limits_{k=0}^na_k(\frac{1}{\overline{z}})^k}
            =e^{\mathfrak{i}\theta}\sum\limits_{k=0}^n \overline{a_k}z^{n-k}
            =\sum\limits_{k=0}^n e^{\mathfrak{i}\theta}\overline{a_{n-k}}z^{k}.
        \end{aligned}
    \end{equation*}
    Therefore, $a_k=e^{\mathfrak{i}\theta}\overline{a_{n-k}}$ for $0\le k\le n$.
\end{proof}

\begin{corollary}\label{cor: roots on unit circle}
    Let $p(z)=\sum\limits_{k=-r}^na_kz^k$ be a Laurent polynomial with complex coefficients, where $n,r\in \N, a_{-r}\neq 0,a_n\neq 0$. 
    If all the zeros of $p(z)$ lie on the unit circle, then there exists $\theta\in \R$, such that $a_k=e^{\mathfrak{i}\theta}\overline{a_{n-r-k}}$ for $-r\le k\le n$.
\end{corollary}

\begin{proof}
    Let $\tilde{p}(z)=z^rp(z)=\sum\limits_{k=-r}^n a_kz^{k+r}= \sum\limits_{j=0}^{n+r}a_{j-r}z^j$.
    Since $a_{-r}\neq 0$, we have $0$ is not a root of $\tilde{p}(z)$.
    Thus, every root of $\tilde{p}(z)$ is a zero of $p(z)$.
    By assumption, all  roots of $\tilde{p}(z)$ lie on the unit circle.
    Then by Lemma~\ref{lm: roots on unit circle}, there exists $\theta\in \R$ such that $a_{j-r}=e^{\mathfrak{i}\theta}\overline{a_{n-j}}$, for $0\le j\le n+r$. 
    The corollary follows.
\end{proof}

\begin{lemma}\label{lm:c=bd}
    Let $f$ be a signature with  signature matrix 
      $M(f) = \left[\begin{smallmatrix} 1 & 0 & 0 & b \\ 0 & c & d & 0 \\ 0 & d & c & 0 \\ b & 0 & 0 & 1 \end{smallmatrix}\right],$ with $bcd \neq 0,b\neq -1$. 
    Let $(u,v,w)=(\frac{c+d}{1+b},\frac{-1+b}{1+b},\frac{-c+d}{1+b})$.
    If $(2,2,-2)\in L_{(u,v,w)}$, then $d=bc$ or $c=bd$.
    If $(2,-2,-2)\in L_{(u,v,w)}$, then $d=-bc$ or $c=-bd$.
\end{lemma}
\begin{proof}
    If $(2,2,-2)\in L_{(u,v,w)}$, then $u^2v^2w^{-2}=1$.
    A direct computation shows that $(d-bc)(c-bd)=0$.
    It's similar for the case $(2,-2,-2)\in L_{(u,v,w)}$.
\end{proof}

\begin{lemma}\label{lm: rank2 c neq pm bd}

Let $f$ be a signature with the signature matrix 
$M(f) = \left[\begin{smallmatrix} 1 & 0 & 0 & b \\ 0 & c & d & 0 \\ 0 & d & c & 0 \\ b & 0 & 0 & 1 \end{smallmatrix}\right],$ with $bcd \neq 0,b\neq -1,$ and $c\neq \pm bd$. 
Let $(u,v,w)=(\frac{c+d}{1+b},\frac{-1+b}{1+b},\frac{-c+d}{1+b})$.
If $\mathrm{rank}(L_{(u,v,w)})=2 $, then one of the following holds:
    \begin{itemize}
        \item $\plholant{\neq_2}{\mathcal{F}}\equiv_T^p \holant{\neq_2}{\mathcal{F}}$ for any signature set $\mathcal{F}$ containing $f$;
        \item $\plholant{\neq_2}{f}$ is \numP-hard;
        \item 
        $f$ is $\M$-transformable.
    \end{itemize}
\end{lemma}

We remark that in our proof, we need to 
set $x,y,z$ to some special complex numbers (roots of unity) to carry out interpolations (even if we only intend to prove a dichotomy
for real-valued signatures).
This is one of the technical reasons to prove dichotomy for complex signatures.

\begin{proof}
    Let $\{(p_1,-q_1,r_1),(p_2,-q_2,r_2)\}$ be a lattice basis for $L_{(u,v,w)}$,
    where $p_i,q_i,r_i\in \Z$ for $i=1,2$.
    Applying the Euclidean algorithm, we may assume $r_1=0$, i.e., $L_{(u,v,w)}=\braket{(p_1,-q_1,0),(p_2,-q_2,r_2)}$.
    Furthermore, we may assume $p_1\ge 0,r_2\ge 0$; and if $p_1>0$, then $|p_2|\le \frac{1}{2}p_1$.
    Note that for $x,y,z\in \C$, $L_{(u,v,w)}\subseteq L_{(x,y,z)}$ is equivalent to 
    \begin{equation}\label{eq:lattice containing condition}
    \left\{
    \begin{aligned}
        &x^{p_1} = y^{q_1}, \\
        &x^{p_2} z^{r_2} = y^{q_2}.
    \end{aligned}
    \right.
    \end{equation}

    In the following, we first show that either $\plholant{\neq_2}{\mathcal{F}}\equiv_T^p \holant{\neq_2}{\mathcal{F}}$ for any signature set $\mathcal{F}$ containing $f$, or $\plholant{\neq_2}{f}$ is \numP-hard for ``generic'' choices of $p_1,q_1,p_2,q_2,r_2$, by selecting proper $x,y,z\in \C$ satisfying Equation~\ref{eq:lattice containing condition} and appealing to Lemma~\ref{lm:interpolate disequality crossover} or  Lemma~\ref{lm:interpolation something hard}.
    Then we deal with the ``exceptional'' choices of $p_1,q_1,p_2,q_2,r_2$.

    We first divide cases in terms of  $r_2$, and show that either $\plholant{\neq_2}{f}\equiv_T^p \holant{\neq_2}{f}$ or $\plholant{\neq_2}{f}$ is \numP-hard unless $r_2=1$.
    
\begin{itemize}
    \item If $r_2$ is even, let $(x,y,z)=(1,1,-1)$, which satisfies Equation~\ref{eq:lattice containing condition}.
    Then by Lemma~\ref{lm:interpolate disequality crossover}, $\plholant{\neq_2}{\mathcal{F}}\equiv_T^p \holant{\neq_2}{\mathcal{F}}$ for any signature set $\mathcal{F}$ containing $f$.
    \item 
     If $r_2\ge 3$, let $(x,y,z)=(1,1,\zeta_{r_2})$. 
    Then $(x,y,z)$ satisfies Equation~\ref{eq:lattice containing condition} and condition~\ref{cd:six vertex model} in Lemma~\ref{lm:interpolation something hard}.
    Thus, $\plholant{\neq_2}{f}$ is \numP-hard.
\end{itemize}
    Now we may assume $r_2=1$.
    Then Equation~\ref{eq:lattice containing condition} becomes
    \begin{equation}\label{eq:lattice containing condition r_2=1}
    \left\{
    \begin{aligned}
        &x^{p_1} = y^{q_1}, \\
        &x^{p_2}  z = y^{q_2}.
    \end{aligned}
    \right.
    \end{equation}
    
    Next we divide cases in terms of  $p_2$.

    \begin{itemize}
        \item 
   
    Assume $|p_2|\ge2$.
    
    \begin{itemize}
        \item If $p_1=0$, let $(x,y,z)=(t,1,t^{-p_2})$, where $t\in \C$ is to be determined.
    Then $(x,y,z)$ satisfies Equation~\ref{eq:lattice containing condition r_2=1}.
    Note that $x=\pm z$ is equivalent to $t^{1+p_2}=\pm1$, and $xz=1$ is equivalent to $t^{1-p_2}=1$.
    Since $|p_2|\ge 2$, there are finitely many $t$ such that $t^{1+p_2}=\pm1$ or $t^{1-p_2}=1$.
    Select $t\in \C$ such that $t^{1+p_2}\neq \pm1$ and $t^{1-p_2}\neq 1$, then $(x,y,z)$ satisfies condition~\ref{cd:six vertex model} in Lemma~\ref{lm:interpolation something hard}.
    Thus, $\plholant{\neq_2}{f}$ is \numP-hard.
    \item 
    If $p_1>0$, then $p_1\ge 4$ by $|p_2|\le \frac{1}{2}p_1$ and $|p_2|\ge 2$.
    Let $(x,y,z)=(\zeta_{p_1},1,\zeta_{p_1}^{-p_2})$, then $(x,y,z)$ satisfies Equation~\ref{eq:lattice containing condition r_2=1}.
    Note that if $x=z$, then $\zeta_{p_1}^{1+p_2}=1$, which implies $p_1\mid 1+p_2$.
    Since $|p_2|\ge 2$, we have $1+p_2\neq 0$ and so $p_1\le |1+p_2|$. This contradicts $|p_2|\le \frac12p_1$ and $p_1\ge 4$.
    Thus, $x\neq z$.
    Similarly, we can check $xz\neq 1$.
    If $x+z\neq 0$, then $(x,y,z)$ satisfies condition~\ref{cd:six vertex model} in Lemma~\ref{lm:interpolation something hard} and so $\plholant{\neq_2}{f}$ in \numP-hard.
    If $x+z=0$, we have $\zeta_{p_1}^{1+p_2}=-1$, which implies $1+p_2=(k+\frac{1}{2})p_1$ for some $k\in \Z$.
    Since $|p_2|\le \frac{1}{2}p_1$, we have $k=0$ and $1+p_2=\frac{1}{2}p_1$.
    Combined with $|p_2|\ge 2$ and $p_1\ge 4$, we have $p_1\ge 6$.
    In this case, $(x,y,z)=(\zeta_{p_1},1,-\zeta_{p_1})$.
    By Lemma~\ref{lm:interpolation xyz}, we have $\plholant{\neq_2}{f,g}\le_T^p\plholant{\neq_2}{f}$, where $g$ has the signature matrix
    $$
    M(g)=
    \left[\begin{smallmatrix} 
    1-y & 0 & 0 & 1+y \\ 
    0 & x-z & x+z & 0 \\ 
    0 & x+z & x-z & 0 \\ 
    1+y & 0 & 0 & 1-y
    \end{smallmatrix}\right]
    =2
    \left[\begin{smallmatrix} 
    0 & 0 & 0 & 1 \\ 
    0 & \zeta_{p_1} & 0 & 0 \\ 
    0 & 0 & \zeta_{p_1} & 0 \\ 
    1 & 0 & 0 & 0
    \end{smallmatrix}\right]
    $$ 
    By $p_1\ge6$ and Theorem~\ref{thm:planar six-vertex}, $\plholant{\neq_2}{g}$ is \numP-hard unless $p_1=8$.
    When $p_1=8$, we have $p_2=3$ by $1+p_2=\frac{1}{2}p_1$.
    Thus, $\plholant{\neq_2}{f}$ is \numP-hard unless $p_1=8,p_2=3$.
    \end{itemize}
    
    \item
    Assume $p_2=0$.
    \begin{itemize}
        \item If $p_1=0$, let $(x,y,z)=(t,1,1)$, where $t\in \C$ is to be determined. 
        Then $(x,y,z)$ satisfies Equation~\ref{eq:lattice containing condition r_2=1}. 
        Select $t\neq \pm1$, then $(x,y,z)$ satisfies condition~\ref{cd:six vertex model} in Lemma~\ref{lm:interpolation something hard} and so $\plholant{\neq_2}{f}$ is \numP-hard.
        \item If $p_1\ge 3$, let $(x,y,z)=(\zeta_{p_1},1,1)$, which satisfies Equation~\ref{eq:lattice containing condition r_2=1} and condition~\ref{cd:six vertex model} in Lemma~\ref{lm:interpolation something hard}.
        Thus, $\plholant{\neq_2}{f}$ is \numP-hard. 
    \end{itemize}
 \end{itemize}

    To summarize, either $\plholant{\neq_2}{f}\equiv_T^p\holant{\neq_2}{f}$ or $\plholant{\neq_2}{f}$ is \numP-hard, unless $r_2=1$ and $p_1,p_2$ belong to the following cases:
    \begin{enumerate}
        \item $p_2=3,p_1=8$;
        \item $p_2=0,p_1=1$;
        \item $p_2=0,p_1=2$;
        \item $p_2=1$;
        \item $p_2=-1$.
    \end{enumerate}

    \vspace{.1in}
    \noindent
    In the following we deal with these five exceptional cases to complete the
    proof of Lemma~\ref{lm: rank2 c neq pm bd}.

    \begin{enumerate}
        \item $p_2=3,p_1=8.$\label{case: p_2=3,p_1=8}
        
        In this case, $\{(8,-q_1,0),(3,-q_2,1)\}$ is a lattice basis for $L_{(u,v,w)}$.
        Now for $x,y,z\in \C$, $L_{(u,v,w)}\subseteq L_{(x,y,z)}$ is equivalent to 
    \begin{equation}\label{eq:lattice containing condition p1=8,p2=3}
    \left\{
    \begin{aligned}
        &x^8 = y^{q_1}, \\
        &x^3z = y^{q_2}.
    \end{aligned}
    \right.
    \end{equation}
        \begin{itemize}
            \item 
            If $q_1,q_2$ are both odd, let $(x,y,z)=(\zeta_{16},-1,\zeta_{16}^5)$, which satisfies Equation~\ref{eq:lattice containing condition p1=8,p2=3} and condition~\ref{cd:y=-1 not matchgate} in Lemma~\ref{lm:interpolation something hard}.
            Thus, $\plholant{\neq_2}{f}$ is \numP-hard.
            
            \item 
            If $q_1$ is odd and $q_2$ is even, let $(x,y,z)=(\zeta_{16},-1,\zeta_{16}^{13})$, by the same reason above, $\plholant{\neq_2}{f}$ is \numP-hard.
            
            \item 
            If $q_1$ is even and $q_2$ is odd, let $(x,y,z)=(\zeta_8,-1,\zeta_8)$, which satisfies Equation~\ref{eq:lattice containing condition p1=8,p2=3} and condition~\ref{cd:y=-1 two 0 pairs} in Lemma~\ref{lm:interpolation something hard}.
            Thus, $\plholant{\neq_2}{f}$ is \numP-hard.

            \item 
            Assume $q_1,q_2$ are both even.
            Note that $(2,2q_2-q_1,-2)\in L_{(u,v,w)}$.
            We claim $2q_2-q_1\neq 0$.
            Otherwise $u^2w^{-2}=1$, which simplifies to $cd=0$, contradicting the assumption of this Lemma.
            If $|2q_2-q_1|= 2$,
            then $(2,2,-2)\in L_{(u,v,w)}$ or $(2,-2,-2)\in L_{(u,v,w)}$.
            By Lemma~\ref{lm:c=bd}, either $d=\pm bc$ and $f$ is $\M$-transformable by Lemma~\ref{lm:symmetric equality_matchgate-transformable}, or $c=\pm bd$, contradicting the assumption of this Lemma.
            Now assume $|2q_2-q_1|\ge 4$.
            Let $y=\zeta_{|2q_2-q_1|}$, then $y\neq \pm1$.
            Choose an arbitray $x$ such that $x^4=-y^{q_2}$ and set $z=-x$.
            Then $x^8=y^{2q_2}=y^{q_1}$ and $x^3z=-x^4=y^{q_2}$, i.e.,  $(x,y,z)$ satisfies Equation~\ref{eq:lattice containing condition p1=8,p2=3}.
            If $y=-x^2$, we replace $x$ by $\mathfrak{i}x$, and $z$ by $\mathfrak{i}z$;
            note that after the replacement $(x,y,z)$
            still satisfies Equation~\ref{eq:lattice containing condition p1=8,p2=3}.
            Then $y\neq -x^2$ and so $(x,y,z)$ satisfies condition~\ref{cd:x+z=0 not matchgate} in Lemma~\ref{lm:interpolation something hard}.
            Thus, $\plholant{\neq_2}{f}$ is \numP-hard.
        \end{itemize}
        
        \item $p_2=0,p_1=1.$\label{case: p_2=0,p_1=1}
        
         In this case, $\{(1,-q_1,0),(0,-q_2,1)\}$ is a lattice basis for $L_{(u,v,w)}$.
         Note that $(1,q_2-q_1,-1)\in L_{(u,v,w)}$ and $(1,-(q_1+q_2),1)\in L_{(u,v,w)}$.
         We claim $q_2\neq q_1$.
         Otherwise, $(1,0,-1)\in L_{(u,v,w)}$ and so $uw^{-1}=1$, which simplifies to $c=0$, contradicting the assumption of this Lemma.
         If $|q_2-q_1|=1$, then $(1,1,-1)\in L_{(u,v,w)}$ or $(1,-1,-1)\in L_{(u,v,w)}$, which implies $f\in \M$ by Corollary~\ref{cor:lattice contains certain elements then f is matchgate}.
         If $q_1+q_2=1$, then $(1,-1,1)\in L_{(u,v,w)}$ and so $f$ is $\M$-transformable by Corollary~\ref{cor:lattice contains certain elements then f is matchgate}.
         In the following of this case, we assume $|q_2-q_1|\ge 2$ and $q_1+q_2\neq 1$.
         
        \begin{claim}\label{claim: exist good root p2=0 p1=1}
            Assume $|q_2-q_1|\ge2$ and $q_1+q_2\neq 1$.
            If $q_1,q_2\notin\{0,1\}$, then there exists $t\in \C,|t|\neq 1$ such that $t$ is a zero of $P(t):=t^{q_2}+t^{q_1}-t-1$. ($P(\cdot)$ is a Laurent polynomial. If 
            $q_1<0$ or $q_2 < 0$, then $t=0$ is a pole, but \emph{not} a zero of $P(\cdot)$.)
        \end{claim}
        \begin{claimproof}{\ref{claim: exist good root p2=0 p1=1}}
            Suppose by contradiction that all zeros of $P(\cdot)$ lie on the unit circle.
            By assumption, we have $t^{q_1},t^{q_2},-t,-1$ are distinct monomials.
            Then by Corollary~\ref{cor: roots on unit circle}, these four monomials must be grouped in pairs 
            such that the sum of the degrees in each pair is the same.
            In other words, the sum of some two elements in ${q_1, q_2, 1, 0}$ equals the sum of the other two.
            This implies $|q_1-q_2|=1$ or $q_1+q_2=1$, contradiction.
        \end{claimproof} 

        Now we handle the case where $p_2=0,p_1=1$ using Claim~\ref{claim: exist good root p2=0 p1=1}.
        As $p_2=0,p_1=1$, for $x,y,z\in \C$, $L_{(u,v,w)}\subseteq L_{(x,y,z)}$ is equivalent to 
        \begin{equation}\label{eq:lattice containing condition p2=0,p1=1}
        \left\{
        \begin{aligned}
            &x = y^{q_1}, \\
            &z = y^{q_2}.
        \end{aligned}
        \right.
        \end{equation}

        \begin{itemize}
            \item 
            If $q_1,q_2\notin \{0,1\}$,
            by Claim~\ref{claim: exist good root p2=0 p1=1}, let $t$ be a root of $P(t)=t^{q_2}+t^{q_1}-t-1$ where $|t|\neq 1$. 
            Let $(x,y,z)=(t^{q_1},t,t^{q_2})$, then $(x,y,z)$ satisfies Equation~\ref{eq:lattice containing condition p2=0,p1=1} and $1+y=x+z$.
            Also, $(y\neq -1)\land(x\neq 1)\land (z\neq 1)$, since $|t|\neq 1$.
            Thus, $(x,y,z)$ satisfies condition~\ref{cd: redundent} in Lemma~\ref{lm:interpolation something hard} and so $\plholant{\neq_2}{f}$ is \numP-hard.

            \item 
            If $q_1=0$, then $|q_2|\ge 2$ by $|q_2-q_1|\ge2$.
            Let $(x,y,z)=(1,\zeta_{2|q_2|},-1)$, then $(x,y,z)$ satisfies Equation~\ref{eq:lattice containing condition p2=0,p1=1} and condition~\ref{cd:x+z=0 not matchgate} in Lemma~\ref{lm:interpolation something hard}.
            Thus, $\plholant{\neq_2}{f}$ is \numP-hard.

            \item 
            If $q_1=1$, then $|q_2-1|\ge 2$. 
            Let $(x,y,z)=(\zeta_{2|q_2-1|},\zeta_{2|q_2-1|},-\zeta_{2|q_2-1|})$, then $(x,y,z)$ satisfies Equation~\ref{eq:lattice containing condition p2=0,p1=1} and condition~\ref{cd:x+z=0 not matchgate} in Lemma~\ref{lm:interpolation something hard}.
            Thus, $\plholant{\neq_2}{f}$ is \numP-hard.

            \item 
            If $q_2=0$ or $q_2=1$, by the symmetry of $q_1,q_2$, it's the same as the case where $q_1=0$ or $q_2=1$ and thus $\plholant{\neq_2}{f}$ is \numP-hard.
        \end{itemize}
        The above covers all cases of $q_1, q_2$.
        The case  $p_2=0,p_1=1$ has been proved.
        
        \item $p_2=0,p_1=2.$\label{case: p_2=0,p_1=2}
        
         In this case, $\{(2,-q_1,0),(0,-q_2,1)\}$ is a lattice basis for $L_{(u,v,w)}$.
         Note that $(2,2q_2-q_1,-2)\in L_{(u,v,w)}$.
         We claim $q_1\neq 2q_2$.
         Otherwise, $(2,0,-2)\in L_{(u,v,w)}$ and so $u^2w^{-2}=1$, which simplifies to $cd=0$, a contradiction to the assumption of the lemma.
         
        \begin{claim}\label{claim: exist good root p2=0 p1=2}
            Assume $q_1\neq 2q_2$. 
            If  $q_1\neq 0$ and $q_2\neq 0,1$,
            then there exists $t\in \C,|t|\neq 1$ such that 
            $t$ is a zero of $P(t):=t^{2q_2}-t^{q_1}-t^2-1$.
        \end{claim}
        \begin{claimproof}{\ref{claim: exist good root p2=0 p1=2}}
            First consider the case where $q_1=2$.
            Then $P(t)=t^{2q_2}-2t^{2}-1$, where $t^{2q_2},-2t^{2},-1$ are different 
            monomial terms.
            If all 
            zeros of $P(t)$ lie on the unit circle, then there exists $\theta\in \R$ such that $1=e^{\mathfrak{i}\theta}\cdot(-1)$ and $-2=e^{\mathfrak{i}\theta}\cdot(-2)$ by Corollary~\ref{cor: roots on unit circle}; 
            however  such a $\theta$ clearly doesn't exist, a contradiction.
            Thus, there exists $t\in \C,|t|\neq 1$ such that $P(t)=0$. 
            
            If $q_1\neq 2$, then $t^{2q_2},-t^{q_1},-t^2,-1$ are different monomial terms. 
            Suppose all  zeros of $P(t)$ lie on the unit circle.
            By Corollary~\ref{cor: roots on unit circle}, these monomials must be grouped in pairs such that the coefficients in each pair differ by a global constant $e^{\mathfrak{i}\theta}$ up to conjugation.
            But this is impossible, as the coefficients are $1,-1,-1,-1$.
        \end{claimproof} 
        
        Now we handle the case where $p_2=0,p_1=2$ using Claim~\ref{claim: exist good root p2=0 p1=2}. 
        As $p_2=0,p_1=2$, for $x,y,z\in \C$, $L_{(u,v,w)}\subseteq L_{(x,y,z)}$ is equivalent to 
        \begin{equation}\label{eq:lattice containing condition p2=0,p1=2}
        \left\{
        \begin{aligned}
            &x^2 = y^{q_1}, \\
            &z = y^{q_2}.
        \end{aligned}
        \right.
        \end{equation}

        \begin{itemize}
            \item 
            If $q_1\neq 0$ and $q_2\neq 0,1$,
            by Claim~\ref{claim: exist good root p2=0 p1=2}, let $t$ be a root of $P(t)=t^{2q_2}-t^{q_1}-t^2-1$ where $|t|\neq 1$. 
            Let $(x,y,z)=(-t^{q_1},t^2,t^{2q_2})$, then $(x,y,z)$ satisfies Equation~\ref{eq:lattice containing condition p2=0,p1=2} and $1+y=x+z$.
            Also, $(y\neq -1)\land(x\neq 1)\land (z\neq 1)$, since $|t|\neq 1$.
            Thus, $(x,y,z)$ satisfies condition~\ref{cd: redundent} in Lemma~\ref{lm:interpolation something hard} and so $\plholant{\neq_2}{f}$ is \numP-hard.

            \item 
            If $q_1=0$, then $q_2\neq 0$ by $q_1\neq2q_2$.
            \begin{itemize}
            \item If $|q_2|\ge 2$, let $(x,y,z)=(1,\zeta_{2|q_2|},-1)$, which satisfies Equation~\ref{eq:lattice containing condition p2=0,p1=2} and condition~\ref{cd:x+z=0 not matchgate} in Lemma~\ref{lm:interpolation something hard}.
            Thus, $\plholant{\neq_2}{f}$ is \numP-hard.
            \item If $q_2=1$, then $\{(2,0,0),(0,-1,1)\}$ is a lattice basis for $L_{(u,v,w)}$.
            Then $(u=-1)\land (v=w)$, which simplifies to $(c=-b)\land(d=-1)$, contradicting the assumption that $c\neq bd$.
            \item If $q_2=-1$, then $=\{(2,0,0),(0,1,1)\}$ is a lattice basis for $L_{(u,v,w)}$.
            Note that $(2,-2,-2)\in L_{(u,v,w)}$.
            By Lemma~\ref{lm:c=bd}, either $d=-bc$ and so $f$ is $\M$-transformable by Lemma~\ref{lm:symmetric equality_matchgate-transformable}, or $c=-bd$,
            contradicting the assumption of this lemma.
            \end{itemize}

            \item 
            If $q_2=0$, then $q_1\neq 0$ by $q_1\neq2q_2$.
            \begin{itemize}
            \item If $|q_1|>2$, let $(x,y,z)=(-1,\zeta_{|q_1|},1)$, which satisfies Equation~\ref{eq:lattice containing condition p2=0,p1=2} and condition~\ref{cd:x+z=0 not matchgate} in Lemma~\ref{lm:interpolation something hard}.
            Thus, $\plholant{\neq_2}{f}$ is \numP-hard.

            \item If $q_1=\pm 1$, let $(x,y,z)=(\mathfrak{i},-1,1)$, which satisfies Equation~\ref{eq:lattice containing condition p2=0,p1=2} and condition~\ref{cd:y=-1 not matchgate} in Lemma~\ref{lm:interpolation something hard}.
            Thus, $\plholant{\neq_2}{f}$ is \numP-hard.

            \item If $q_1=2$, then $\{(2,-2,0),(0,0,1)\}$ is a lattice basis for $L_{(u,v,w)}$.
            Then $(uv^{-1}=-1)\land (w=1)$, which simplifies to $(c=-b)\land(d=1)$, contradicting the assumption that $c\neq -bd$.

            \item If $q_1=-2$, then $\{(2,2,0),(0,0,1)\}$ is a lattice basis for $L_{(u,v,w)}$.
            Note that $(2,2,-2)\in L_{(u,v,w)}$.
            By Lemma~\ref{lm:c=bd}, either $d=bc$ and so $f$ is $\M$-transformable by Lemma~\ref{lm:symmetric equality_matchgate-transformable}, or $c=bd$,
            contradicting the assumption of this lemma.
            \end{itemize}

            \item 
            If $q_2=1$, then $q_1\neq 2$ by $q_1\neq2q_2$.
        Also, we may assume $q_1\neq 0$, since the case where $q_1=0$ has been discussed above.
        \begin{itemize}
            \item If $q_1\ge 5$ or $q_1\le -1$, then $|q_1-2|\ge 3$.
            Let $(x,y,z)=(-\zeta_{|q_1-2|},\zeta_{|q_1-2|},\zeta_{|q_1-2|})$, which satisfies Equation~\ref{eq:lattice containing condition p2=0,p1=2} and condition~\ref{cd:x+z=0 not matchgate} in Lemma~\ref{lm:interpolation something hard}.
            Thus, $\plholant{\neq_2}{f}$ is \numP-hard.

            \item If $q_1=1$ or $q_1=3$, let $(x,y,z)=(\mathfrak{i},-1,-1)$, which satisfies Equation~\ref{eq:lattice containing condition p2=0,p1=2} and condition~\ref{cd:y=-1 not matchgate} in Lemma~\ref{lm:interpolation something hard}.
            Thus, $\plholant{\neq_2}{f}$ is \numP-hard.

            \item If $q_1=4$, then $\{(2,-4,0),(0,-1,1)\}$ is a lattice basis for $L_{(u,v,w)}$.
            Note that $(2,-2,-2)\in L_{(u,v,w)}$.
            By Lemma~\ref{lm:c=bd}, either $d=-bc$ and so $f$ is $\M$-transformable by Lemma~\ref{lm:symmetric equality_matchgate-transformable}, or $c=-bd$,
            contradicting the assumption of this lemma.
            \end{itemize}
        \end{itemize}
        
        The above covers all cases of $q_1,q_2$.
        The case $p_2=0,p_1=2$ has been proved.

        \item $p_2=1$. \label{case: p_2=1}
        
        In this case, $\{(p_1,-q_1,0),(1,-q_2,1)\}$ is a lattice basis for $L_{(u,v,w)}$.
        We may assume $p_1\ge0$. 
        If $q_2=1$, then $(1,-1,1)\in L_{(u,v,w)}$ and so $f\in \M$ by Corollary~\ref{cor:lattice contains certain elements then f is matchgate}.
        In the following, we assume $q_2\neq 1$.

        \begin{claim}\label{claim: exist good root p2=1}
            Assume $q_2\neq 1$. If $p_1\ge3$, then there exists $t\in \C,|t|\neq 1$ such that $P(t):=\omega t^{q_1}+\omega^{-1}t^{p_1q_2-q_1}-t^{p_1}-1=0$, where $\omega=\zeta_{p_1}$.
        \end{claim}
        \begin{claimproof}{\ref{claim: exist good root p2=1}}
        Suppose by contradiction that all zeros of $P(t)$ lie on the unit circle.
        First, assume $q_1,p_1q_2-q_1,p_1,0$ are all distinct.
        By Corollary~\ref{cor: roots on unit circle}, the monomial terms $\omega t^{q_1},\omega^{-1}t^{p_1q_2-q_1},-t^{p_1},-1$ must be grouped in pairs such that the coefficients in each pair differ by a global constant $e^{\mathfrak{i}\theta}$ up to conjugation.
        The other groupings being clearly impossible, 
        the only possible way of grouping is $\{(\omega t^{q_1},\omega^{-1}t^{p_1q_2-q_1}),(-t^{p_1},-1)\}$.
        Also, by Corollary~\ref{cor: roots on unit circle}, we have $q_1+(p_1q_2-q_1)=p_1$, contradicting $q_2\neq 1$. 
        Thus, there exists $t\in \C,|t|\neq 1$ such that $P(t)=0$.

        Next, we deal with the case where $q_1,p_1q_2-q_1,p_1,0$ are not all distinct.
        \begin{itemize}
            \item 
            If $q_1=p_1q_2-q_1$, then $P(t)=(\omega+\omega^{-1})t^{q_1}-t^{p_1}-1$.
            If $q_1,p_1,0$ are all distinct, by Corollary~\ref{cor: roots on unit circle}, as $|\omega+\omega^{-1}| \ne 1$, $t^{q_1}$ cannot be grouped
            to the other terms,  we have $p_1=2q_1$.
            Combined with $q_1=p_1q_2-q_1$
          and $p_1\ge 3$, 
            we have $q_2=1$, contradiction.
            If $q_1,p_1,0$ are not all distinct, as $p_1\ge 3$, it could only be that
            $q_1 =p_1$ or $q_1 =0$.
            If $q_1=p_1$, then $P(t)=(\omega+\omega^{-1}-1)t^{p_1}-1$.
            Since $|\omega+\omega^{-1}-1| = |2 \cos \frac{2 \pi}{p_1} - 1| < 1$, as
            $p_1 \ge 3$, the roots of $P(t)$ lie outside the unit circle, 
            contradiction.
            If $q_1=0$, then $P(t)=-t^{p_1}+\omega+\omega^{-1}-1$.
            Similarly, we get a contradiction.
           
            \item 
            Now assume $q_1\neq p_1q_2-q_1$. 
            If $q_1=p_1$, then $P(t)=(\omega-1)t^{p_1}$+$\omega^{-1}t^{p_1(q_2-1)}-1$, and $q_2\neq 2$ by $q_1\neq p_1q_2-q_1$.
            Also $p_1\ge 3$ and $q_2\neq 1$ by assumptions of the claim.
            Thus, $(\omega-1)t^{p_1}$, $\omega^{-1}t^{p_1(q_2-1)}, -1$ are monomial terms of distinct degrees.
            By Corollary~\ref{cor: roots on unit circle}, we have $(w^{-1}=e^{\mathfrak{i}
            \theta}\cdot(-1))\land(\omega-1=e^{\mathfrak{i}\theta}\overline{\omega-1})$, which implies $w^2=1$, contradiction.
            
            \item 
            Now assume $q_1\neq p_1q_2-q_1$ and $q_1\neq p_1$.
            If $p_1q_2-q_1=p_1$, then $P(t)=\omega t^{q_1}+(\omega^{-1}-1)t^{p_1}-1$, and $q_1\neq 0$ by $q_2\neq 1$.
            Thus, $\omega t^{q_1},(\omega^{-1}-1)t^{p_1},-1$ are monomial terms of distinct degrees.
            By Corollary~\ref{cor: roots on unit circle}, we have $(\omega=e^{\mathfrak{i}
\theta}\cdot(-1))\land(\omega^{-1}-1=e^{\mathfrak{i}\theta}\overline{\omega^{-1}-1})$, which implies $\omega^2=1$, contradiction.

            \item 
            Now assume $q_1,p_1q_2-q_1,p_1$ are distinct.
            If $q_1=0$, then $P(t)=\omega^{-1}t^{p_1q_2-q_1}-t^{p_1}+\omega-1$.
            By a similar argument above, we get a contradiction.
            It's similar for the case where $p_1q_2-q_1=0$.
        \end{itemize}  
        \end{claimproof}
        %
       
        Now we handle the case $p_2=1$ using Claim~\ref{claim: exist good root p2=1}.
        As $p_2=1$, for $x,y,z\in \C$, $L_{(u,v,w)}\subseteq L_{(x,y,z)}$ is equivalent to 
        \begin{equation}\label{eq:lattice containing condition p2=1}
        \left\{
        \begin{aligned}
            &x^{p_1} = y^{q_1}, \\
            &xz = y^{q_2}.
        \end{aligned}
        \right.
        \end{equation}

        \begin{itemize}
        \item 
        If $p_1\ge 3$, by Claim~\ref{claim: exist good root p2=1}, let $t$ be a root of $P(t)=\omega t^{q_1}+\omega^{-1}t^{p_1q_2-q_1}-t^{p_1}-1$, where $\omega=\zeta_{p_1}$ and $|t|\neq 1$.
        Let $(x,y,z)=(\omega t^{q_1},t^{p_1},\omega^{-1}t^{p_1q_2-q_1})$, which satisfies Equation~\ref{eq:lattice containing condition p2=1}.
        Also, $(x,y,z)$ satisfies condition~\ref{cd: redundent} in Lemma~\ref{lm:interpolation something hard} by $|t|\neq 1$.
        Thus, $\plholant{\neq_2}{f}$ is \numP-hard.

        \item 
        If $p_1=0$, then $\{(0,-q_1,0),(1,-q_2,1)\}$ is a lattice basis for $L_{(u,v,w)}$.
        We may assume $q_1\ge 3$.
        Otherwise, $v^2=1$ and so $b=0$.
        We may further assume $|q_2|\le \frac{1}{2} q_1$.
        Let $y=\zeta_{q_1}$ and $x=-z=\sqrt{-\zeta_{q_1}^{q_2}}$, then $(x,y,z)$ satisfies Equation~\ref{eq:lattice containing condition p2=1}.
        Note that $y\neq \pm1$ by $q_1\ge 3$.
        Also, $y=-x^2$ implies $y=y^{q_2}$, and so $q_1\mid |q_2-1|$.
        Since $q_2\neq 1$, we have $|q_2-1|\ge q_1$. 
        This contradicts $|q_2|\le \frac{1}{2}q_1$ and $q_1\ge 3$.
        Thus, $y\neq -x^2$ and $(x,y,z)$ satisfies condition~\ref{cd:x+z=0 not matchgate} in Lemma~\ref{lm:interpolation something hard}.
        Thus, $\plholant{\neq_2}{f}$ is \numP-hard.

        \item 
        If $p_1=1$, then $\{(1,-q_1,0),(0,q_1-q_2,1)\}$ is a lattice basis for $L_{(u,v,w)}$ and it's reduced to the case where $p_2=0,p_1=1$.
        
        \item 
        If $p_1=2$, then $\{(2,-q_1,0),(1,-q_2,1)\}$ is a lattice basis for $L_{(u,v,w)}$.
        Note that $(1,q_2-q_1,-1)\in L_{(u,v,w)}$. 
        We claim $q_1\neq q_2$, otherwise $u=w$, contradicting $c\neq 0$.
        If $|q_2-q_1|=1$, then $f\in \M$ by Corollary~\ref{cor:lattice contains certain elements then f is matchgate}.
        In the following we assume $|q_2-q_1|\ge 2$.  We will finish the
        proof of the sub-case $p_1=2$ after proving Claim~\ref{claim: exist good root p2=1, p1=2}.
        \begin{claim}\label{claim: exist good root p2=1, p1=2}
            Assume $|q_2-q_1|\ge 2$ and $q_2\neq 1$. 
            There exists $t\in \C,|t|\neq 1$ such that $P(t):=t^{q_1}+t^{2q_2-q_1}+t^2+1=0$.
        \end{claim}
        \begin{claimproof}{\ref{claim: exist good root p2=1, p1=2}}
        Suppose by contradiction that all the roots of $P(t)$ lie on the unit circle.
        First, assume $q_1,2q_2-q_1,2,0$ are all distinct.
        By Corollary~\ref{cor: roots on unit circle}, either $q_1+(2q_2-q_1)=2+0$, contradicting $q_2\neq 1$, or $q_1+2=(2q_2-q_1)+0$, contradicting $|q_2-q_1|\ge 2$, or $q_1+0=(2q_2-q_1)+2$, contradicting $|q_2-q_1|\ge 2$.

        Next, we deal with the case where $q_1,2q_2-q_1,2,0$ are not all distinct.
        Note that $q_1\neq 2q_2-q_1$ by $|q_2-q_1|\ge 2$.
        \begin{itemize}
            \item
            If $q_1=2$, then $P(t)=t^{2q_2-2}+2t^2+1$.
            By $q_2\neq q_1$ and $q_2\neq 1$, $t^{2q_2-2},2t^2,1$ are monomial terms of distinct degrees.
            By Corollary~\ref{cor: roots on unit circle}, we have $(2q_2-2)+0=2+2$.
            Then $q_2=3$, contradicting $|q_2-q_1|\ge 2$.

            \item 
            Now assume $q_1\neq 2$.
            If $2q_2-q_1=2$, then $P(t)=t^{q_1}+2t^2+1$.
            Note that $q_1\neq 0$, otherwise $q_2=1$.
            Then $t^{q_1},2t^2,1$ are monomials of distinct degrees.
            By Corollary~\ref{cor: roots on unit circle}, we have $q_1=4$, then $q_2=3$, contradicting $|q_2-q_1|\ge 2$.

            \item 
            Now assume $q_1,2q_2-q_1,2$ are all distinct.
            If $q_1=0$, then $P(t)=t^{2q_2}+t^2+2$, where $t^{2q_2},t^2,2$ are monomial terms of distinct degrees.
            By Corollary~\ref{cor: roots on unit circle}, we have $2q_2+2=0$.
            Then $q_2=-1$, contradicting $|q_2-q_1|\ge 2$.
            It's similar if $2q_2-q_1=0$.
        \end{itemize}
        \end{claimproof}

        Now we handle the sub-case where $p_1=2$.
        By Claim~\ref{claim: exist good root p2=1, p1=2}, let $t$ be a root of $P(t)=t^{q_1}+t^{2q_2-q_1}+t^2+1=0$, where $|t|\neq 1$.
        Let $(x,y,z)=(-t^{q_1},t^2,-t^{2q_2-q_1})$, which satisfies Equation~\ref{eq:lattice containing condition p2=1} and condition~\ref{cd: redundent} in Lemma~\ref{lm:interpolation something hard}.
        Thus, $\plholant{\neq_2}{f}$ is \numP-hard.
        \end{itemize}

        The above covers all cases of $p_1,q_1,q_2$.
        The case $p_2=1$ has been proved.
        
        \item $p_2=-1$.\label{case: p_2=-1}
        
        In this case, $\{(p_1,-q_1,0),(-1,-q_2,1)\}$ is a lattice basis for $L_{(u,v,w)}$.
        We may assume $p_1\ge0$.
        We claim $q_2\neq 0$. 
        Otherwise, $(-1,0,1)\in L_{(u,v,w)}$ and $u^{-1}w=1$, contradicting $c=0$.
        If $q_2=\pm1$, then $(1,1,-1)\in L_{(u,v,w)}$ or $(1,-1,-1)\in L_{(u,v,w)}$ and so $f$ is $\M$-transformable by Corollary~\ref{cor:lattice contains certain elements then f is matchgate}.
        In the following, we may assume $|q_2|>1$.
        
        \begin{claim}\label{claim: exist good root p_2=-1}
            Assume $|q_2|>1$. 
            If $p_1\ge 3$, then there exists $t\in \C,|t|\neq 1$ such that $P(t):=\omega t^{q_1}+\omega t^{p_1q_2+q_1}-t^{p_1}-1=0$, where $\omega=\zeta_{p_1}$.
        \end{claim}
        
        \begin{claimproof}{\ref{claim: exist good root p_2=-1}}
        Suppose by contradiction that all the roots of $P(t)$ lie on the unit circle.
        First, assume $q_1,p_1q_2+q_1,p_1,0$ are all distinct.
        By Corollary~\ref{cor: roots on unit circle}, the monomials $\omega t^{q_1},\omega t^{p_1q_2+q_1},-t^{p_1},-1$ must be grouped in pairs such that the coefficients in each pair differ by a global constant $e^{\mathfrak{i}\theta}$ up to conjugation.
        Since $p_1\ge 3$, there are two possible ways of grouping: $\{(w t^{q_1},-t^{p_1}),(\omega t^{p_1q_2+q_1},-1)\}$ or $\{(w t^{q_1},-1),(\omega t^{p_1q_2+q_1},-t^{p_1})\}$.
        Also, by Corollary~\ref{cor: roots on unit circle}, either $q_1+p_1=p_1q_2+q_1$  or $q_1=(p_1q_2+q_1)+p_1$, both contradicting $|q_2|>1$.
        Thus, there exists $t\in \C,|t|\neq 1$ such that $P(t)=0$.
        
        Next, we deal with the case where $q_1,p_1q_2+q_1,p_1,0$ are not all distinct.
        Note that $q_1\neq p_1q_2+q_1$ since $p_1q_2\neq 0$.
        \begin{itemize}
            \item 
            If $q_1=p_1$, then $P(t)=\omega^{p_1(q_2+1)}+(\omega-1)t^{p_1}-1$, where $\omega^{p_1(q_2+1)},(\omega-1)t^{p_1},-1$ are monomials of distinct degrees. 
            By Corollary~\ref{cor: roots on unit circle}, we have $p_1(q_2+1)=2p_1$, contradicting $q_1\neq -1$.

            \item 
            Now assume $q_1\neq p_1$.
            If $p_1q_2+q_1=p_1$, then $q_1\neq 0$, otherwise $q_2=1$, contradiction.
            Then $P(t)=\omega t^{q_1}+(\omega-1)t^{p_1}-1$, where $\omega t^{q_1},(\omega-1)t^{p_1},-1$ are monomials of distinct degrees.
            By Corollary~\ref{cor: roots on unit circle}, we have $q_1=2p_1$.
            This implies $q_2=-1$ combined with $p_1q_2+q_1=p_1$, contradiction.

            \item 
            Now assume $q_1,p_1q_2+q_1,p_1$ are all distinct.
            If $q_1=0$, then $P(t)=\omega t^{p_1q_2}-t^{p_1}+\omega-1$, where $\omega t^{p_1q_2},-t^{p_1},\omega-1$ are monomials of distinct degrees.
            By Corollary~\ref{cor: roots on unit circle}, we have $p_1q_2+p_1=0$, contradicting $q_2\neq -1$.
            If $p_1q_2+q_1=0$, then $P(t)=\omega t^{q_1}-t^{p_1}+\omega-1$, where $\omega t^{q_1},-t^{p_1},\omega-1$ are monomials of distinct degrees.
            By Corollary~\ref{cor: roots on unit circle}, we have $q_1+p_1=0$.
            This implies $q_2=1$ combined with $p_1q_2+q_1=0$, contradiction.
            
        \end{itemize}
       
        \end{claimproof}
        Now we handle the case where $p_2=-1$ using Claim~\ref{claim: exist good root p_2=-1}.
        As $p_2=-1$, for $x,y,z\in \C$, $L_{(u,v,w)}\subseteq L_{(x,y,z)}$ is equivalent to 
        \begin{equation}\label{eq:lattice containing condition p2=-1}
        \left\{
        \begin{aligned}
            &x^{p_1} = y^{q_1}, \\
            &x^{-1}z = y^{q_2}.
        \end{aligned}
        \right.
        \end{equation}
        
        \begin{itemize}
            \item If $p_1\ge 3$, by Claim~\ref{claim: exist good root p_2=-1}, let $t$ be a root of $P(t)=\omega t^{q_1}+\omega t^{p_1q_2+q_1}-t^{p_1}-1$, where $\omega=\zeta_{p_1}$ and $|t|\neq 1$.
            Let $(x,y,z)=(\omega t^{q_1},t^{p_1},\omega t^{p_1q_2+q_1})$, which satisfies Equation~\ref{eq:lattice containing condition p2=-1}.
            Also, $(x,y,z)$ satisfies condition~\ref{cd: redundent} in Lemma~\ref{lm:interpolation something hard} by $|t|\neq 1$.
            Thus, $\plholant{\neq_2}{f}$ is \numP-hard.
            
            \item If $p_1=0$, then $\{(0,-q_1,0),(-1,-q_2,1)\}$ is a lattice basis for $L_{(u,v,w)}$.
            Same as the case where $p_2=1,p_1=0$, we may assume $q_1\ge 3$ and $|q_2|\le \frac{1}{2} q_1$.
            Let $y=\zeta_{q_1}$, we claim $1+y^{q_2}\ne 0$.
            Otherwise, $q_2=(k+\frac{1}{2})q_1$ for some $k\in \Z$.
            Since $|q_2|\le \frac{1}{2}q_1$, we have $k=0$ or $k=-1$, and so $q_2=\frac{1}{2}q_1$ or $q_2=-\frac{1}{2}q_1$.
            Then $(-2,0,2)\in L_{(u,v,w)}$, since $(-2,\pm 2q_1-q_2,2)\in L_{(u,v,w)}$.
            This implies $u^{-2}w^{2}=1$, contradicting $cd=0$.
            Since $1+y^{q_2}\neq 0$, we let $x=\frac{1+y}{1+y^{q_2}}$ and $z=\frac{(1+y)y^{q_2}}{1+y^{q_2}}$, then $(x,y,z)$ satisfies Equation~\ref{eq:lattice containing condition p2=-1}.
            Also, $1+y=x+z$ and $y\neq \pm 1$.
            Note that $x=1$ implies $y^{q_2-1}=1$, and $q_1\mid |q_2-1|$.
            Since $q_2\neq 1$, we have $|q_2-1|>q_1$. 
            This contradicts $|q_2|\le \frac{1}{2}q_1$ and $q_1\ge 3$. 
            Thus, $x\neq 1$.
            Similarly, we can check that $z\neq 1$.
            Then $(x,y,z)$ satisfies condition~\ref{cd: redundent} in Lemma~\ref{lm:interpolation something hard}.
            Thus, $\plholant{\neq_2}{f}$ is \numP-hard.

            \item If $p_1=1$, then $\{(1,-q_1,0),(0,-q_1-q_2,1)\}$ is a lattice basis for $L_{(u,v,w)}$ and it's reduced to the case where $p_2=0,p_1=1$.

            \item If $p_1=2$, then $\{(1,-q_1,0), (1,-q_1-q_2,1)\}$ is a lattice basis for $L_{(u,v,w)}$ and it's reduced to the case where $p_2=1$.
        \end{itemize}
    The above covers all cases of $p_1,q_1,q_2$.
    The case $p_2=-1$ has been proved.
    \end{enumerate}
    So far all the exceptional cases of $p_1,p_2$ has been handled.
    This completes the proof of Lemma~\ref{lm: rank2 c neq pm bd}.
\end{proof}

In the following (Lemma~\ref{lm: rank=2 dichotomy}), we are going to obtain the complexity dichotomy in the case where $\mathrm{rank}(L_{(u,v,w)})=2$.
Note that in Lemma~\ref{lm: rank=2 dichotomy} we no longer assume $c\neq \pm bd$.
First, we need a lemma to handle the special case where $(b=\pm c)\land (d=\pm1)$.
It is in this special case that the new tractable type is found:

\begin{lemma} \label{lm: b=pm c, d=pm 1}
Let $f$ be a signature with the signature matrix 
    $
    M(f)=
    \left[\begin{smallmatrix} 
    1 & 0 & 0 & b \\ 
    0 & c & d & 0 \\ 
    0 & d & c & 0 \\ 
    b & 0 & 0 & 1
    \end{smallmatrix}\right]
    $.
    Assume $d=\pm 1$, then 
    \begin{itemize}
        \item 
        If $b=c$, then $\plholant{\neq_2}{f}$ is \numP-hard unless $b=0,\pm1, \pm \mathfrak{i}$, in which cases $\plholant{\neq_2}{f}$ is tractable.

        \item 
        If $b=-c$, then $\plholant{\neq_2}{f}$ is \numP-hard unless $b=0,\pm1,\pm \mathfrak{i},\mathfrak{i}\tan(\frac{\beta\pi}{8})$,  where $\beta$ is odd, in which cases $\plholant{\neq_2}{f}$ is tractable.
    \end{itemize}
\end{lemma}
\begin{proof}
    In the following proof, all tractability and hardness results are obtained by applying the combinatorial transformation to {\sc Even Coloring} plus a holographic transformation (Lemma~\ref{lm:even coloring transformation}). 
    The transformed problem can then be handled either by Case I (Lemma~\ref{lm:case d=w}) or by the planar six-vertex model (Theorem~\ref{thm:planar six-vertex}).
    
    If $(b=c)\land(d=1)$, then 
    $
    M(f)=
    \left[\begin{smallmatrix} 
    1 & 0 & 0 & b \\ 
    0 & b & 1 & 0 \\ 
    0 & 1 & b & 0 \\ 
    b & 0 & 0 & 1
    \end{smallmatrix}\right].
    $
    Note that $
        M_{\rm Z}^{\rm PL}
        \left[\begin{smallmatrix} 
        b \\ b\\ 1\\ 1
        \end{smallmatrix}\right]
        =
        \left[\begin{smallmatrix} 
        0 \\ 0\\ 1+b\\ -1+b
        \end{smallmatrix}\right]$.
        By Lemma~\ref{lm:even coloring transformation}, $\plholant{\neq_2}{f}\equiv_T^p\plholant{\neq_2}{f'}$, where $f'$ has the signature matrix
        $
        M(f')=
        \left[\begin{smallmatrix} 
        -1+b & 0 & 0 & 0 \\ 
        0 & 0 & 1+b & 0 \\ 
        0 & 1+b & 0 & 0 \\ 
        0 & 0 & 0 & -1+b
        \end{smallmatrix}\right].
        $
        If $b=1$, then $f\in \M$ by Lemma~\ref{lm:is matchgate}.
        If $b\neq 1$, normalize $f'$ by $-1+b$, we have $\plholant{\neq_2}{f}\equiv_T^p\plholant{\neq_2}{f''}$, where $f''$ has the signature matrix
        $
        M(f'')=
        \left[\begin{smallmatrix} 
        1 & 0 & 0 & 0 \\ 
        0 & 0 & \frac{1+b}{-1+b} & 0 \\ 
        0 & \frac{1+b}{-1+b} & 0 & 0 \\ 
        0 & 0 & 0 & 1
        \end{smallmatrix}\right].
        $
        By Lemma~\ref{lm:case d=w}, $\plholant{\neq_2}{f''}$ is \numP-hard unless $\frac{1+b}{-1+b}=0,\pm1,\pm \mathfrak{i}$ , in which cases $\plholant{\neq_2}{f''}$ is tractable.
        Thus, $\plholant{\neq_2}{f}$ is \numP-hard unless $b=0,\pm1, \pm \mathfrak{i}$, in which cases $\plholant{\neq_2}{f}$ is tractable.

        If $(b=c)\land(d=-1)$, then
        $
        M(f)=
        \left[\begin{smallmatrix} 
        1 & 0 & 0 & b \\ 
        0 & b & -1 & 0 \\ 
        0 & -1 & b & 0 \\ 
        b & 0 & 0 & 1
        \end{smallmatrix}\right].
        $
        Note that 
        $
        M_{\rm HZ}^{\rm PL}
        \left[\begin{smallmatrix} 
        b \\ b\\ -1\\ 1
        \end{smallmatrix}\right]
        =
        \left[\begin{smallmatrix} 
        0 \\ 0\\ -1+b\\ 1+b
        \end{smallmatrix}\right]$.
        By Lemma~\ref{lm:even coloring transformation}, $\plholant{\neq_2}{f}\equiv_T^p\plholant{\neq_2}{f'}$, where $f'$ has the signature matrix
        $
        M(f')=
        \left[\begin{smallmatrix} 
        1+b & 0 & 0 & 0 \\ 
        0 & 0 & -1+b & 0 \\ 
        0 & -1+b & 0 & 0 \\ 
        0 & 0 & 0 & 1+b
        \end{smallmatrix}\right].
        $
        Similar to the argument in the case where $(b=c)\land (d=1)$, $\plholant{\neq_2}{f}$ is \numP-hard unless $b=0,\pm1,\pm \mathfrak{i}$, in which cases $\plholant{\neq_2}{f}$ is tractable.

        If $(b=-c)\land(d=1)$, then
        $
        M(f)=
        \left[\begin{smallmatrix} 
        1 & 0 & 0 & b \\ 
        0 & -b & 1 & 0 \\ 
        0 & 1 & -b & 0 \\ 
        b & 0 & 0 & 1
        \end{smallmatrix}\right].
        $
        Note that 
        $
        M_{\rm HZ}^{\rm PL}
        \left[\begin{smallmatrix} 
        b \\ -b\\ 1\\ 1
        \end{smallmatrix}\right]
        =
        \left[\begin{smallmatrix} 
        1+b \\ 1-b\\ 0\\ 0
        \end{smallmatrix}\right]$.
        By Lemma~\ref{lm:even coloring transformation}, $\plholant{\neq_2}{f}\equiv_T^p\plholant{\neq_2}{f'}$, where $f'$ has the signature matrix
        $
        M(f')=
        \left[\begin{smallmatrix} 
        0 & 0 & 0 & 1+b \\ 
        0 & 1-b & 0 & 0 \\ 
        0 & 0 & 1-b & 0 \\ 
        1+b & 0 & 0 & 0
        \end{smallmatrix}\right].
        $
        If $b=-1$, then $f\in \M$ by Lemma~\ref{lm:is matchgate}.
        If $b\neq -1$, normalize $f'$ by $1+b$, we have $\plholant{\neq_2}{f}\equiv_T^p\plholant{\neq_2}{f''}$, where $f''$ has the signature matrix
        $
        M(f'')=
        \left[\begin{smallmatrix} 
        0 & 0 & 0 & 1 \\ 
        0 & \frac{1-b}{1+b} & 0 & 0 \\ 
        0 & 0 & \frac{1-b}{1+b} & 0 \\ 
        1 & 0 & 0 & 0
        \end{smallmatrix}\right].
        $
        By Theorem~\ref{thm:planar six-vertex}, $\plholant{\neq_2}{f''}$ is \numP-hard unless $\frac{1-b}{1+b}=0$ or $\sqrt{\mathfrak{i}}^{-\beta}$, where $\beta\in \Z$.
        Thus, $\plholant{\neq_2}{f}$ is \numP-hard unless $b=0,\pm1,\pm \mathfrak{i},\frac{\sqrt{\mathfrak{i}}^\beta-1}{\sqrt{\mathfrak{i}}^\beta+1}=\mathfrak{i}\tan(\frac{\beta\pi}{8})$,  where $\beta$ is odd, in which cases $\plholant{\neq_2}{f}$ is tractable.

        If $(b=-c)\land(d=-1)$, then
        $
        M(f)=
        \left[\begin{smallmatrix} 
        1 & 0 & 0 & b \\ 
        0 & -b & -1 & 0 \\ 
        0 & -1 & -b & 0 \\ 
        b & 0 & 0 & 1
        \end{smallmatrix}\right].
        $
        Note that 
        $
        M_{\rm Z}^{\rm PL}
        \left[\begin{smallmatrix} 
        b \\ -b\\ -1\\ 1
        \end{smallmatrix}\right]
        =
        \left[\begin{smallmatrix} 
        -1+b \\ -1-b\\ 0\\ 0
        \end{smallmatrix}\right]$.
        By Lemma~\ref{lm:even coloring transformation}, $\plholant{\neq_2}{f}\equiv_T^p\plholant{\neq_2}{f'}$, where $f'$ has the signature matrix
        $
        M(f')=
        \left[\begin{smallmatrix} 
        0 & 0 & 0 & -1+b \\ 
        0 & -1-b & 0 & 0 \\ 
        0 & 0 & -1-b & 0 \\ 
        -1+b & 0 & 0 & 0
        \end{smallmatrix}\right].
        $
        Similar to the argument in the case where $(b=-c)\land(d=1)$, $\plholant{\neq_2}{f}$ is \numP-hard unless $b=0,\pm1,\pm \mathfrak{i},\frac{\sqrt{\mathfrak{i}}^\beta-1}{\sqrt{\mathfrak{i}}^{\beta}+1}=\mathfrak{i}\tan(\frac{\beta\pi}{8})$, where $\beta$ is an odd integer, in which cases $\plholant{\neq_2}{f}$ is tractable.
\end{proof}

\begin{lemma}\label{lm: rank=2 dichotomy}
    Let $f$ be a signature with the signature matrix 
    $
    M(f)=
    \left[\begin{smallmatrix} 
    1 & 0 & 0 & b \\ 
    0 & c & d & 0 \\ 
    0 & d & c & 0 \\ 
    b & 0 & 0 & 1
    \end{smallmatrix}\right]
    $, with $bcd\neq 0,b\neq -1,c\neq -1$.
    Let $(u,v,w)=(\frac{c+d}{1+b},\frac{-1+b}{1+b},\frac{-c+d}{1+b})$.
    If $\mathrm{rank}(L_{(u,v,w)})=2$,
    then one of the following holds:
    \begin{itemize}
        \item $\plholant{\neq_2}{f}\equiv_T^p \holant{\neq_2}{f}$;
        \item $\plholant{\neq_2}{f}$ is \numP-hard;
        \item 
        $f$ is $\M$-transformable.
    \end{itemize}
\end{lemma}

\begin{proof}
    If $c\neq \pm bd$, then by Lemma~\ref{lm: rank2 c neq pm bd} we are done.
    
    In the following we assume $c=\pm bd$.
    Consider the rotated signature $f^{\frac{\pi}{2}}$ with the signature matrix
    $
    M(f^{\frac{\pi}{2}})=
    \left[\begin{smallmatrix} 
    1 & 0 & 0 & c \\ 
    0 & b & d & 0 \\ 
    0 & d & b & 0 \\ 
    c & 0 & 0 & 1
    \end{smallmatrix}\right].
    $
    Let $u^*=\frac{b+d}{1+c},v^*=\frac{-1+c}{1+c},w^*=\frac{-b+d}{1+c}$.
    \begin{itemize}
        \item 
        If $\mathrm{rank}(L_{(u^*,v^*,w^*)})\le 1$, then by Lemma~\ref{lm: rank<=1}, either $\plholant{\neq_2}{f^{\frac{\pi}{2}}}\equiv_T^p \holant{\neq_2}{f^{\frac{\pi}{2}}}$, which is equivalent to $\plholant{\neq_2}{f}\equiv_T^p \holant{\neq_2}{f}$; 
        or $\plholant{\neq_2}{f^{\frac{\pi}{2}}}$ is \numP-hard, which is equivalent to  $\plholant{\neq_2}{f}$ is \numP-hard; 
        or $f^{\frac{\pi}{2}}$ is $\M$-transformable, which is equivalent to that $f$ is $\M$-transformable by Lemma~\ref{lm:symmetric equality_matchgate-transformable}.

        \item 
        If $\mathrm{rank}(L_{(u^*,v^*,w^*)})=2$ and $b\neq \pm cd$, apply Lemma~\ref{lm: rank2 c neq pm bd} to $f^{\frac{\pi}{2}}$ then we are done.
        Now we assume that $\mathrm{rank}(L_{(u^*,v^*,w^*)})=2$ and $b= \pm cd$.
        \begin{itemize}
        \item If $(c=bd)\land (b=cd)$, then $(b=c)\land (d=1)$ or $(b=-c)\land (d=-1)$.

        \begin{itemize}
        \item 
        If it was the first case, then by Lemma~\ref{lm: b=pm c, d=pm 1}, $\plholant{\neq_2}{f}$ is \numP-hard unless $b=0,\pm1,\pm\mathfrak{i}$.
        If $b=0,\pm \mathfrak{i}$, one can check that $u,v,w$ are all roots of unity.
        Then by Lemma~\ref{lm: lattice full rank equiv roots of unity}, $\mathrm{rank}(L_{(u,v,w)})=3$, contradiction.
        Note that $b=-1$ is forbidden by the assumption of this lemma.
        If $b=1$, then $(u,v,w)=(1,0,0)$ and so $\mathrm{rank}(L_{(u,v,w)})=1$, contradiction.
        
        \item 
        If it was the second case, 
        then by Lemma~\ref{lm: b=pm c, d=pm 1}, $\plholant{\neq_2}{f}$ is \numP-hard unless $b=0,\pm1,\pm\mathfrak{i},\mathfrak{i}\tan(\frac{\beta\pi}{8})=\frac{\sqrt{\mathfrak{i}}^\beta-1}{\sqrt{\mathfrak{i}}^\beta+1}$, where $\beta$ is an odd integer.
        If $b=0,\pm \mathfrak{i},\frac{\sqrt{\mathfrak{i}}^\beta-1}{\sqrt{\mathfrak{i}}^\beta+1}$, where $\beta$ is odd, one can check that $u,v,w$ are all roots of unity.
        Then by Lemma~\ref{lm: lattice full rank equiv roots of unity}, $\mathrm{rank}(L_{(u,v,w)})=3$, contradiction.
        Note that $b=-1$ is forbidden by the assumption of this lemma.
        If $b=1$, then $(u,v,w)=(-1,0,0)$ and so $\mathrm{rank}(L_{(u,v,w)})=1$, contradiction.
        
        \end{itemize}

        \item If $(c=bd)\land(b=-cd)$, then $(b=\mathfrak{i}c)\land(d=-\mathfrak{i})$ or $(b=-\mathfrak{i}c)\land(d=\mathfrak{i})$.
        So $1-b^2=c^2-d^2$ and $f\in \M$ by Lemma~\ref{lm:is matchgate}.
        
        \item If $(c=-bd)\land (b=cd)$, then $(b=\mathfrak{i}c)\land(d=\mathfrak{i})$ or $(b=-\mathfrak{i}c)\land(d=-\mathfrak{i})$.
        So $1-b^2=c^2-d^2$ and $f\in \M$ by Lemma~\ref{lm:is matchgate}.
        
        \item If $(c=-bd)\land (b=-cd)$, then $(b=c)\land(d=-1)$ or $(b=-c)\land (d=1)$.
        Similar to the argument in the case where $(c=bd)\land(b=cd)$, $\plholant{\neq_2}{f}$ is \numP-hard.
        \end{itemize}
        
        \item 
        If $\mathrm{rank}(L_{(u^*,v^*,w^*)})=3$, then $u^*,v^*,w^*$ are all roots of unity by Lemma~\ref{lm: lattice full rank equiv roots of unity}.
        As $|v^*|=|\frac{-1+c}{1+c}|=1$, we have $c\in \R\mathfrak{i}$.
        Also, $|u^*|=|w^*|$ implies $|b+d|=|-b+d|$, so $b$ and $d$ are orthogonal in the complex plane.
        Combined with $c=\pm bd$ and $c\in \R\mathfrak{i}$, we have either $(b\in \R)\land(d\in \R\mathfrak{i})$, or $(b\in \R\mathfrak{i})\land(d\in \R)$.
        By $|u^*|=|\frac{b+d}{1+c}|=|\frac{b+d}{1\pm bd}|=1$, we have $|b|^2+|d|^2=1+|b|^2|d|^2$, so $|b|=1$ or $|d|=1$.
        
        First, assume $(b\in \R)\land(d\in \R\mathfrak{i})$.
        \begin{itemize}
            \item 
            If $|b|=1$, then  $b=\pm1$ by $b\in \R$.
            Combined with $c=\pm bd$, we have $c=\pm d$ and $1-b^2=c^2-d^2=0$, so $f\in \M$ by Lemma~\ref{lm:is matchgate}.

            \item 
            If $|d|=1$, then $d=\pm i$ by $d\in \R \mathfrak{i}$.
            Combined with $c=\pm bd$, we have $c=\pm \mathfrak{i}b$ and $1-b^2=c^2-d^2$, so $f\in \M$ by Lemma~\ref{lm:is matchgate}.
        \end{itemize}

        Second, assume $(b\in \R\mathfrak{i})\land(d\in \R)$.
        \begin{itemize}
            \item 
            If $|b|=1$, then $d=\pm\frac{c}{b}=\pm\overline{b}c=\mp bc$ by $c=\pm bd$ and $b\in \R\mathfrak{i}$.
            So $f$ is $\M$-transformable by Lemma~\ref{lm:symmetric equality_matchgate-transformable}.

            \item If $|d|=1$, then $(b=\pm c)\land(d=\pm 1)$ by $d\in \R$ and $c=\pm bd$.
            By the same argument as the case where $\mathrm{rank}(L_{(u^*,v^*,w^*)})=2$, $\plholant{\neq_2}{f}$ is \numP-hard.
        \end{itemize}
    \end{itemize}
\end{proof}

\subsection{Dichotomy for three equal pairs}\label{subsec: Dichotomy for three equal pairs}

\begin{lemma}\label{lm: symmetric equality dichotomy}
    Let $f$ be a signature with the signature matrix 
    $
    M(f)=
    \left[\begin{smallmatrix} 
    1 & 0 & 0 & b \\ 
    0 & c & d & 0 \\ 
    0 & d & c & 0 \\ 
    b & 0 & 0 & 1
    \end{smallmatrix}\right]
    $, with $bcd\neq 0$.
    Then one of the following holds:
    \begin{itemize}
        \item $\plholant{\neq_2}{\mathcal{F}}\equiv_T \holant{\neq_2}{\mathcal{F}}$ for any signature set $\mathcal{F}$ containing $f$;
        \item $\plholant{\neq_2}{f}$ is \numP-hard;
        \item 
        $\plholant{\neq_2}{f}$ is tractable. 
        More specifically,
        \begin{enumerate}
            \item $f$ is $\M$-transformable;
            \item $d=\pm1, b=c=0$;
            \item $d=\pm1$, $b=-c=\mathfrak{i}\tan(\frac{\beta\pi}{8})$, where $\beta\in \Z$ is odd.
        \end{enumerate}
    \end{itemize}
\end{lemma}

\begin{proof}
    If $b=-1$ or $c=-1$, by Lemma~\ref{lm: ars b=-1 or c=-1} we are done.
    Now we assume $b\neq -1$ and $c\neq -1$.
    Let $(u,v,w)=(\frac{c+d}{1+b},\frac{-1+b}{1+b},\frac{-c+d}{1+b})$ and $(u^*,v^*,w^*)=(\frac{b+d}{1+c},\frac{-1+c}{1+c},\frac{-b+d}{1+c})$.
    If $\mathrm{rank}(L_{(u,v,w)})\le2$, then by Lemma~\ref{lm: rank<=1} and Lemma~\ref{lm: rank=2 dichotomy} we are done.
    If $\mathrm{rank}(L_{(u^*,v^*,w^*)})\le2$, apply Lemma~\ref{lm: rank<=1} and Lemma~\ref{lm: rank=2 dichotomy} to $f^{\frac{\pi}{2}}$ then we are done. 
    Now assume $\mathrm{rank}(L_{(u,v,w)})=\mathrm{rank}(L_{(u^*,v^*,w^*)})=3$.
    By Lemma~\ref{lm: lattice full rank equiv roots of unity}, we have $u,v,w,u^*,v^*,w^*$ are all roots of unity.
    Then $|v|=|v^*|=1$, which implies that $b,c\in \R\mathfrak{i}$.
    Also, $|u|=|w|$ implies that $|c+d|=|-c+d|$, so $d\in \R$.
    So $|u|^2=\frac{|c|^2+|d|^2}{1+|b|^2}=1$, and $|u^*|^2=\frac{|b|^2+|d|^2}{1+|c|^2}=1$.
    This implies $(b=\pm c)\land (d=\pm1)$.
    If $b=c$, then by Lemma~\ref{lm: b=pm c, d=pm 1}, $\plholant{\neq_2}{f}$ is \numP-hard unless $b=0,\pm1,\pm \mathfrak{i}$.
    Note that the case where $b=\pm 1$ can be subsumed to $f\in \M$ by Lemma~\ref{lm:is matchgate}, and $b=\pm \mathfrak{i}$ can be subsumed to that $f$ is $\M$-transformable by Lemma~\ref{lm:symmetric equality_matchgate-transformable}.
    If $b=-c$, then by Lemma~\ref{lm: b=pm c, d=pm 1}, $\plholant{\neq_2}{f}$ is \numP-hard unless $b=0,\pm1,\pm \mathfrak{i},\mathfrak{i}\tan(\frac{\beta\pi}{8})$,  where $\beta$ is odd.
    Also, the case where $b=\pm 1$ or $b=\pm \mathfrak{i}$ can be subsumed to that $f\in \M$ or $f$ is $\M$-transformable.
    The lemma follows.
\end{proof}

    \begin{theorem}\label{thm: three equal pairs}
    Let $f$ be a signature with the signature matrix 
    $
    M(f)=
    \left[\begin{smallmatrix} 
    1 & 0 & 0 & b \\ 
    0 & c & d & 0 \\ 
    0 & d & c & 0 \\ 
    b & 0 & 0 & 1
    \end{smallmatrix}\right]
    $, with $bcd\neq 0$.
    Then $\plholant{\neq_2}{f}$ is \numP-hard unless in the following cases where $\plholant{\neq_2}{f}$ is tractable:
    \begin{enumerate}
        \item  
        $f$ is $\mathscr{A}$-transformable, or $\mathscr{P}$-transformable, or $\mathscr{L}$-transformable;
        
        \item 
        $f$ is $\M$-transformable;

        \item 
        $d=\pm 1,b=-c=\mathfrak{i}\tan(\frac{\beta\pi}{8})$, where $\beta\in \Z$ is odd.
    \end{enumerate}
    \end{theorem}
    \begin{proof}
        Note that $f\in \mathscr{A}$ if $d=\pm1,b=c=0$.
        Then by Lemma~\ref{lm: symmetric equality dichotomy} and Lemma~\ref{lm: eight-vertex ars equality}, the theorem follows.
    \end{proof}

\subsection{Dichotomy for three opposite pairs}\label{subsec: Dichotomy for three opposite pairs}

The following lemma is proved in~\cite{CaiF17}.

\begin{lemma}\label{lm: eight vertex ars-disequality}
    Let $f$ be a signature with the signature matrix
    $M(f)=\left[\begin{smallmatrix}
    1 & 0 & 0 & b\\
    0 & c & d & 0\\
    0 & -d & -c & 0\\
    -b & 0 & 0 & 1\\
    \end{smallmatrix}\right]$ with $bcd\neq 0$,
    then $\operatorname{Holant}(\neq_2|f)$ is
    $\#\operatorname{P}$-hard, or $f$ is
     $\mathscr{P}$-transformable, or $\mathscr{A}$-transformable,
     or $\mathscr{L}$-transformable.
\end{lemma}

Now we generalize Lemma~\ref{lm: eight vertex ars-disequality} to the planar setting.

\begin{lemma}\label{lm: symmetric disequality dichotomy}
    Let $f$ be a signature with the signature matrix
    $M(f)=\left[\begin{smallmatrix}
    1 & 0 & 0 & b\\
    0 & c & d & 0\\
    0 & -d & -c & 0\\
    -b & 0 & 0 & 1\\
    \end{smallmatrix}\right]$ with $bcd\neq 0$,
    Then $\plholant{\neq_2}{f}$ is \numP-hard unless $f$ is $\mathscr{P}$-transformable, or $\mathscr{A}$-transformable, or $\mathscr{L}$-transformable, or $\M$-transformable.
\end{lemma}

\begin{proof}
    By the holographic transformation using
    $\left[\begin{smallmatrix}
    1 & 0\\
    0 & \sqrt{\mathfrak{i}}\\
    \end{smallmatrix}\right]$, we have
    \[
    \plholant{\neq_2}{f}\equiv_T^p\plholant{\neq_2}{\tilde{f}},
    \]
    where $M(\tilde{f})=\left[\begin{smallmatrix}
     1 & 0 & 0 & \mathfrak{i} b\\
     0 & \mathfrak{i} c & \mathfrak{i} d & 0\\
     0 & -\mathfrak{i} d & -\mathfrak{i} c & 0\\
     -\mathfrak{i} b & 0 & 0 & -1\\
    \end{smallmatrix}\right]$.
    By connecting the variables $x_3, x_4$ (resp. $x_1,x_2$) of ${f}$ using $(\neq_2)$, we get the binary signature
    $(c+d)(0,1, -1, 0)^T$ (resp. $(c-d)(0,1, -1, 0)^T$).
    Since $cd\neq 0$, there is at least one nonzero in $\{c+d, c-d\}$.
    Thus, we can realize $(0,1,-1,0)^T$ after normalizing by a nonzero scalar.
    By doing a binary modification to the variable $x_1$ of $\tilde{f}$ using $(0, 1, -1, 0)^T$,
    we get a signature $f'$ with the signature matrix
    $M(f')=
    \left[\begin{smallmatrix}
     1 & 0 & 0 & b'\\
     0 & c' & d' & 0\\
     0 & d' & c' & 0\\
     b' & 0 & 0 & 1\\
    \end{smallmatrix}\right],
    $
    where $b'=\mathfrak{i}b,c'=\mathfrak{i}c,d'=\mathfrak{i}d$.
    
    
    By Lemma~\ref{lm: symmetric equality dichotomy}, we have one of the following holds:
    \begin{itemize}
        \item 
        $\plholant{\neq_2}{\mathcal{F}}\equiv_T \holant{\neq_2}{\mathcal{F}}$ for any signature set $\mathcal{F}$ containing $f$.
        In this case, let $\mathcal{F}=\{f,f'\}$, then we have $\holant{\neq_2}{f}\le_T^p \holant{\neq_2}{\mathcal{F}}\equiv_T\plholant{\neq_2}{\mathcal{F}}\le \plholant{\neq_2}{f}$.
        So $\plholant{\neq_2}{f}\equiv_T\holant{\neq_2}{f}$.
        By Lemma~\ref{lm: eight vertex ars-disequality}, we are done.
        
        \item 
        $\plholant{\neq_2}{f'}$ is \numP-hard.
        This implies that $\plholant{\neq_2}{f}$ is \numP-hard.
        \item 
        $\plholant{\neq_2}{f'}$ is tractable. 
        More specifically,
        \begin{enumerate}
            \item $f'$ is $\M$-transformable.
            
            Note that $f'$ is $\M$-transformable $\Leftrightarrow d'=\pm b'c' \Leftrightarrow d=\pm \mathfrak{i}bc \Leftrightarrow f$ is $\M$-transformable by Lemma~\ref{lm:symmetric equality_matchgate-transformable} and Lemma~\ref{lm:symmetric disequality_matchgate-transformable}. 
            
            \item $d'=\pm1, b'=c'=0$.

            In this case, $d=\pm \mathfrak{i}, b=c=0$. This implies $f\in \mathscr{A}$.
            
            \item $d'=\pm1$, $b'=-c'=\mathfrak{i}\tan(\frac{\beta\pi}{8})$, where $\beta\in \Z$ is odd.

            In this case, $d=\pm \mathfrak{i}, b=-c=\tan(\frac{\beta\pi}{8})$.
            Connecting three copies of $f$ with the middle one rotated by $\pi$, we get a signature $f_3$ with the signature matrix
            $$
            M(f_3)=M(f)NM(f^\pi)NM(f)=
            \begin{bNiceMatrix}
             1 & 0 & 0 & b_3\\
             0 & c_3 & -d & 0\\
             0 & d & -c_3 & 0\\
             -b_3 & 0 & 0 & 1\\
            \end{bNiceMatrix},
            $$
            where $b_3=c_3=\frac{b(3-b^2)}{1-3 b^2}$.
            A direct computation shows that $|b_3|\neq 1$ for any odd integer $\beta$.

            By the same argument above, we can realize a signature $f_3$ with the signature matrix
            $M(f_3')=
            \left[\begin{smallmatrix}
             1 & 0 & 0 & \mathfrak{i}b_3\\
             0 & \mathfrak{i}c_3 & -\mathfrak{i}d & 0\\
             0 & -\mathfrak{i}d & \mathfrak{i}c_3 & 0\\
             \mathfrak{i}b_3 & 0 & 0 & 1\\
            \end{smallmatrix}\right].
            $
            Note that $f_3'$ has the arrow reversal symmetry property.
            If $f_3'$ satisfies the first two conditions in Lemma~\ref{lm: symmetric equality dichotomy}, then by the same argument above we are done.
            Now we check that $f_3'$ satisfies none of the third conditions in Lemma~\ref{lm: symmetric equality dichotomy}.
            Note that $\det M_{\rm{Out}}(f_3')-\det M_{\rm{In}}(f_3')=\det M_{\rm{Out}}(f'^{\frac{\pi}{2}}_3)-\det M_{\rm{In}}(f'^{\frac{\pi}{2}}_3)=1+ 2 b_3^2 -d^2=2(1+b_3'^2)\neq 0$, since $|b_3|\neq 1$.
            So $f'_3\notin \M$ and $f'^{\frac{\pi}{2}}_3\notin \M$ by Lemma~\ref{lm:is matchgate}.
            Also, $-\mathfrak{i}d=\pm(\mathfrak{i}b_3)\cdot (\mathfrak{i}c_3)$ implies $|b_3|=1$, contradiction.
            So $-\mathfrak{i}d\neq \pm(\mathfrak{i}b_3)\cdot (\mathfrak{i}c_3)$ and
            $f_3'$ is not $\M$-transformable by Lemma~\ref{lm:symmetric equality_matchgate-transformable}.
            Moreover, $(\mathfrak{i}b_3,\mathfrak{i}c_3)$ is not the case that $\mathfrak{i}b_3=\mathfrak{i}c_3=0$ or $\mathfrak{i}b_3=-\mathfrak{i}c_3=\mathfrak{i}\tan(\frac{\beta\pi}{8})$.
            Thus, $f_3'$ satisfies none of the third conditions in Lemma~\ref{lm: symmetric equality dichotomy}.
            The lemma follows.
        \end{enumerate}
    \end{itemize}    
\end{proof}

\section{Case III: Asymmetric Condition}\label{sec: dichotomy for all nonzero}
Let $f$ be a 4-ary signature with the signature matrix
$M(f)=\left[\begin{smallmatrix}
a & 0 & 0 & b\\
0 & c & d & 0\\
0 & w & z & 0\\
y & 0 & 0 & a\\
\end{smallmatrix}\right]$.
with $abcdyzw\neq 0$.
After Section~\ref{sec: three or opposite pairs}, we may assume it's not the case that $b=\epsilon y,c=\epsilon z$ and $d=\epsilon w$ for some $\epsilon=\pm1$.
Under this assumption, we appeal to M\"obius transformation to generate binary signatures to do binary modification.
The M\"obius transformation is well-defined only when the inner matrix of $M(f)$ has full rank. In Section~\ref{subsec: degenerate inner matrices}, we first address the exceptional case $by = cz = dw$, where the inner matrix remains degenerate even after any rotation of $f$.
Then in Section~\ref{subsec: achieving dichotomy by binary modification}, we are free to apply binary modifications, thanks to Lemma~\ref{lm: get any binary from a binary not fifth root} and Corollary~\ref{cor: get (0,1,0,0) from a binary third or fourth root}. 
This allows us to transform $M(f)$ into a form with zero entries, reducing to previously handled cases and thereby establishing the dichotomy.

\subsection{Degenerate inner matrices}\label{subsec: degenerate inner matrices}

Now we consider the case that $by=cz=dw$.
This is precisely when the inner matrix of the  signature matrix
is always degenerate under all possible rotations of the variables.
Note that to use  M\"{o}bius transformations to generate binary signatures
using Lemma~\ref{lm: get any binary from a binary not fifth root} or Corollary~\ref{cor: get (0,1,0,0) from a binary third or fourth root},
we need a signature matrix having a  full-rank inner matrix.
So we have to treat this case separately.

\begin{lemma}\label{lm: degenerate inner matrices}
Let $f$ be a 4-ary signature with the signature matrix
$M(f)=\left[\begin{smallmatrix}
a & 0 & 0 & b\\
0 & c & d & 0\\
0 & w & z & 0\\
y & 0 & 0 & a\\
\end{smallmatrix}\right]$ satisfying
 $abcdyzw\neq 0$ and $by=cz=dw$, then either $f\in \M$, or $\plholant{\neq_2}{f}$ is \numP-hard, or $f$ satisfies arrow reversal equality, and thus the complexity of $\plholant{\neq_2}{f}$ is characterized by Theorem~\ref{thm: three equal pairs}.
\end{lemma}
\begin{proof}
By normalizing $a=1$, $M(f)=\left[\begin{smallmatrix}
1 & 0 & 0 & b\\
0 & c & d & 0\\
0 & w & z & 0\\
y & 0 & 0 & 1\\
\end{smallmatrix}\right]$.
If $by=1$, then $\det M_{\rm{Out}}(f)=\det M_{\rm{In}}(f)$, so $f\in \M$ by Lemma~\ref{lm:is matchgate}.
In the following we assume that $by\neq 1$.
If $c+d\neq 0$ or $z+d\neq 0$,
then $\plholant{\neq_2}{f}$ is $\#\operatorname{P}$-hard
by Lemma~\ref{lm: outer full inner degenerate}.
Otherwise, $c=z=-d$, $w = \frac{cz}{d}=d$, and  we have
$M(f)=\left[\begin{smallmatrix}
1 & 0 & 0 & b\\
0 & -d & d & 0\\
0 & d & -d & 0\\
y & 0 & 0 & 1\\
\end{smallmatrix}\right]$. Here from $by=cz$ we have $by=d^2$.
Now we consider
$M(f^{\frac{\pi}{2}}) =
\left[\begin{smallmatrix}
1 & 0 & 0 & -d\\
0 & b & d & 0\\
0 & d & y & 0\\
-d & 0 & 0 & 1\\
\end{smallmatrix}\right]$.
By $(-d)^2 = by \neq 1$ we can repeat the argument above
to show that either  $\plholant{\neq_2}{f}$ is \numP-hard,
or $b=y =-d$.
In the latter case, 
$M(f)=\left[\begin{smallmatrix}
1 & 0 & 0 & -d\\
0 & -d & d & 0\\
0 & d & -d & 0\\
-d & 0 & 0 & 1\\
\end{smallmatrix}\right]$.
Then $f$ satisfies arrow reversal equality, and the complexity of $\plholant{\neq_2}{f}$ is characterized by Theorem~\ref{thm: three equal pairs}.
\end{proof}

\subsection{Achieving dichotomy by binary modification}\label{subsec: achieving dichotomy by binary modification}
After Section~\ref{subsec: degenerate inner matrices}, we assume that the inner matrix $\left[\begin{smallmatrix}
 c & d \\
 w & z \\
\end{smallmatrix}\right]$ has full rank.
By Lemma~\ref{lm:get one binary when not exists epsilon}, we can realize a binary signature $g=(0,1,t,0)$ for some $t\in \C$ and $t\neq \pm 1$.
By Lemma~\ref{lm: get any binary from a binary not fifth root}, it $t^i$ are all distinct for $1\le i \le 5$, we can realize $(0,1,u,0)^T$ for any $u\in\C$.
Then, by carefully selecting binary signatures of the form $(0, 1, u, 0)^T$ and applying binary modifications to $f$ using these signatures, we can transform $M(f)$ into a form with zero entries, reducing to previously handled cases and thereby establishing the dichotomy.

\begin{lemma}\label{lm: with-any-binary}
Let $g=(0,1,t,0)^T$ be a binary signature where $t^i$ are distinct for $1\leq i\leq 5$, and $f$ be a 4-ary signature with the signature matrix
$M(f)=\left[\begin{smallmatrix}
a & 0 & 0 & b\\
0 & c & d & 0\\
0 & w & z & 0\\
y & 0 & 0 & a\\
\end{smallmatrix}\right]$ with $abcdyzw\neq 0$, where
$\left[\begin{smallmatrix}
 c & d \\
 w & z \\
\end{smallmatrix}\right]$ has full rank, then either $\plholant{\neq_2}{f,g}$ is \numP-hard or $f\in \M$.
\end{lemma}

\begin{proof}
By normalizing $a=1$, $M(f)=\left[\begin{smallmatrix}
1 & 0 & 0 & b\\
0 & c & d & 0\\
0 & w & z & 0\\
y & 0 & 0 & 1\\
\end{smallmatrix}\right]$.
Let $\mathcal{B}=\{g_1, g_2, g_3\}$ be a set of three binary signatures $g_i=(0, 1, t_i, 0)^T$, for some $t_i \in \mathbb{C}$ to be determined. 
By Lemma \ref{lm: get any binary from a binary not fifth root}, we have $\plholant{\neq_2}{\{f\}\cup \mathcal{B}}\leqslant \plholant{\neq_2}{f, g}.$  We will show that $\plholant{\neq_2}{\{f\}\cup \mathcal{B}}$ is \#P-hard and it follows that $\plholant{\neq_2}{f, g}$ is \#P-hard, unless $f\in \M$.

Rotate $f$ by $\pi$ we get the signature $f^{\pi}$ with the signature matrix 
$M(f^{\pi})=\left[ \begin{smallmatrix}
1 & 0 & 0 & y \\
0 & z & d & 0\\
0 & w & c & 0\\
b & 0 & 0 & 1\\
\end{smallmatrix} \right]$.
By doing a binary modification to the variable $x_1$ of $f^{\pi}$ using $g_1$,
we get the signature $f^{\pi}_{t_1}$ with the signature matrix
$M(f^\pi_{t_1})=\left[ \begin{smallmatrix}
1 & 0 & 0 & y \\
0 & z & d & 0\\
0 & t_1w & t_1c & 0\\
t_1b & 0 & 0 & t_1\\
\end{smallmatrix} \right]$. 
Note that $\det M_{\text{In}}(f^{\pi}_{t_1})=t_1\det M_{\text{In}}(f)$, and $\det M_{\text{Out}}(f^{\pi}_{t_1})=t_1\det M_{\text{Out}}(f)$.
Connecting variables $x_4, x_3$ of $f$ with variables $x_1$, $x_2$ of $f^\pi_{t_1}$ both using $(\neq_2)$, we get a signature $f_1$ with the signature matrix
$$M(f_1)=\left[ \begin{matrix}
a_1 & 0 & 0 & b_1 \\
0 & c_1 & d_1 & 0\\
0 & w_1 & z_1 & 0\\
y_1 & 0 & 0 & x_1\\
\end{matrix} \right]=M(f)NM(f^\pi_{t_1})=
\left[ \begin{matrix}
bt_1+b & 0 & 0 & t_1+by \\
0 & cwt_1+dz & c^2t_1+d^2 & 0\\
0 & w^2t_1+z^2 & cwt_1+dz & 0\\
byt_1+1 & 0 & 0 & yt_1+y\\
\end{matrix} \right].$$
We first show there is a $t_1\neq 0$ such that $c_1d_1w_1\neq0$ and $dwc_1z_1-czd_1w_1\neq 0$.
Consider the quadratic polynomial
$$
p(t)=dw(cwt+dz)(cwt+dz)-cz(c^2t+d_2)(w_2t+z_2).
$$
We have $p(t_1)=dwc_1z_1-czd_1w_1$.
Notice that the coefficient of the quadratic term in $p(t)$ is $c^2w^2(dw -cz)\neq 0$, since $cw\neq 0$ and
$\left[\begin{smallmatrix}
 c & d \\
 w & z \\
\end{smallmatrix}\right]$ has full rank.
That is, $p(t)$ has degree 2, and hence it has at most two roots.
Also, we have the following implications by the definitions of $c_1,d_1,w_1:c_1=0\Rightarrow t_1=-\frac{dz}{cw}; d_1=0\Rightarrow t_1=-\frac{d^2}{c^2};w_1=0\Rightarrow t_1=-\frac{z^2}{w^2}.$
Therefore, we can choose a $t_1\notin\{0,-\frac{dz}{cw},-\frac{d^2}{c^2},-\frac{z^2}{w^2}\}$  and $t_1$ is not a root of $p(t)$.
Then we have $t_1\neq 0,c_1d_1w_1\neq 0$ and $dwc_1z_1-czd_1w_1\neq 0$.

By doing a binary modification to the variable $x_1$ of $f$ using $g_2$,
we get the signature $f_{t_2}$ with the signature matrix
$M(f_{t_2})=\left[ 
\begin{smallmatrix}
1 & 0 & 0 & b \\
0 & c & d & 0\\
0 & t_2w & t_2z & 0\\
t_2y & 0 & 0 & t_2\\
\end{smallmatrix} \right]$. 
Note that $\det M_{\text{In}}(f_{t_2})=t_2\det M_{\text{In}}(f)$, and
$\det M_{\text{Out}}(f_{t_2})=t_2\det M_{\text{Out}}(f).$ 
Connecting variables $x_4, x_3$ of $f_1$ with variables $x_1$, $x_2$ of $f_{t_2}$ both using $(\neq_2)$, we get a signature $f_2$ with  the signature matrix
\begin{equation*}
\begin{aligned}
M(f_2)
=&\left[ \begin{matrix}
a_2 & 0 & 0 & b_2 \\
0 & c_2 & d_2 & 0\\
0 & w_2 & z_2 & 0\\
y_2 & 0 & 0 & x_2\\
\end{matrix} \right]
=M(f_1)NM(f_{t_2})\\
=&\left[ \begin{matrix}
ya_1t_2+b_1 & 0 & 0 &  a_1t_2+bb_1\\
0 & wc_1t_2+cd_1 & zc_1t_2+dd_1  & 0\\
0 & ww_1t_2+cz_1 & zw_1t_2+dz_1 & 0\\
yy_1t_2+x_1  & 0 & 0 & y_1t_2+bx_1\\
\end{matrix} \right].
\end{aligned}
\end{equation*}
Since $wc_1\neq 0$ and $cd_1\neq 0$, we can let $t_2=-\frac{cd_1}{wc_1}$ and $t_2\neq 0$.
Then $c_2=wc_1t_2+cd_1=0$.
Since $dwc_1z_1-czd_1w_1\neq 0$, we have $z_2=zw_1t_2+dz_1=\frac{1}{wc_1}(dwc_1z_1-czd_1w_1)\neq 0$.
Note that 
\begin{equation*}
\begin{aligned}
\det M_{\text{In}}(f_2)
&=\det M_{\text{In}}(f_1)\cdot (-1)\cdot \det M_{\text{In}}(f_{t_2})\\
&=\det M_{\text{In}}(f)\cdot(-1)\cdot \det M_{\text{In}}(f^\pi_{t_1})\cdot(-1)\cdot \det M_{\text{In}}(f_{t_2})\\
&=t_1t_2\det M_{\text{In}}(f)^3\\
&\neq 0.
\end{aligned}
\end{equation*}
Also, we have $\det M_{\text{In}}(f_2)=c_2z_2-d_2w_2=-d_2w_2$.
So $d_2w_2\neq 0$.
Therefore, $M(f_2)$ has the form $
\left[ \begin{smallmatrix}
a_2 & 0 & 0 & b_2 \\
0 & 0 & d_2 & 0\\
0 & w_2 & z_2 & 0\\
y_2 & 0 & 0 & x_2\\
\end{smallmatrix}\right],
$
where $d_2w_2z_2\neq 0$.
By Lemma~\ref{lm: three nonzeros in inner matrix}, $\plholant{\neq_2}{f}$ is \numP-hard unless $f_2\in \M$, i.e., $\det M_{\text{In}}(f_2)=\det M_{\text{Out}}(f_2)$.
Since $\det M_{\text{In}}(f_2)=t_1t_2\det M_{\text{In}}(f)^3$ and similarly $\det M_{\text{Out}}(f_2)=t_1t_2\det M_{\text{Out}}(f)^3$, where $t_1t_2\neq 0$, we have $\det M_{\text{In}}(f)^3=\det M_{\text{Out}}(f)^3$.

In the following, we construct a signature $f_4$ by connecting 7 copies of $f$, possibly with some rotation and binary modification, such that $\plholant{\neq_2}{f_4}$ is \numP-hard unless $\det M_{\text{In}}(f_4)=\det M_{\text{Out}}(f_4)$, which would imply $\det M_{\text{In}}(f)^7=\det M_{\text{Out}}(f)^7$.
Combined with $\det M_{\text{In}}(f)^3=\det M_{\text{Out}}(f)^3$, we would have $\det M_{\text{In}}(f)=\det M_{\text{Out}}(f)$, then $f\in \M$ by Lemma~\ref{lm:is matchgate}.

By doing a binary modification to the variable $x_1$ of $f_1$ using $g_3$,
we get the signature $f_{1t_3}$ with the signature matrix
$M(f_{1t_3})=\left[ \begin{smallmatrix}
a_1 & 0 & 0 & b_1 \\
0 & c_1 & d_1 & 0\\
0 & t_3w_1 & t_3z_1 & 0\\
t_3b_1 & 0 & 0 & t_3x_1\\
\end{smallmatrix} \right]$. 
Note that $\det M_{\text{In}}(f_{1t_3})=t_3\det M_{\text{In}}(f_1)$ and $\det M_{\text{Out}}(f_{1t_3})=t_3\det M_{\text{Out}}(f_1)$.
Connecting variables $x_4, x_3$ of $f_1$ with variables $x_1$, $x_2$ of $f_{1t_3}$ both using $(\neq_2)$, we get a signature $f_3$ with  the signature matrix

\begin{equation*}
\begin{aligned}
M(f_3)
=&\left[ \begin{matrix}
a_3 & 0 & 0 & b_3 \\
0 & c_3 & d_3 & 0\\
0 & w_3 & z_3 & 0\\
y_3 & 0 & 0 & x_3\\
\end{matrix} \right]
=M(f_1)NM(f_{1t_3})\\
=&
\left[ \begin{matrix}
a_1y_1t_3+a_1b_1 & 0 & 0 & a_1x_1t_3+b_1^2  \\
0 & c_1(w_1t_3+d_1) & c_1z_1t_3+d_1^2 & 0\\
0 & w_1^2t_3+c_1z_1 & z_1(w_1t_3+d_1) & 0\\
y_1^2t_3+a_1x_1 & 0 & 0 & x_1y_1t_3+b_1x_1 \\
\end{matrix} \right].
\end{aligned}
\end{equation*}
Since $d_1w_1\neq 0$, we can let $t_3=-\frac{d_1}{w_1}$ and $t_3\neq 0$.
Then $c_3=z_3=0$.
Notice that 
\begin{equation*}
\begin{aligned}
\det M_{\text{In}}(f_3)&=\det M_{\text{In}}(f_1)\cdot(-1)\cdot \det M_{\text{In}}(f_{1t_3})\\
&=-\det M_{\text{In}}(f_1)\cdot t_3\det M_{\text{In}}(f_{1})\\
&=-t_3 
\left[ 
\det M_{\text{In}}(f)\cdot(-1)\cdot\det M_{\text{In}}(f^\pi_{t_1})
\right]^2\\
&=-t_3t_1^2\det M_{\text{In}}(f)^4  \\ 
&\neq 0.\\
\end{aligned}
\end{equation*}
Also, we have $\det M_{\text{In}}(f_3)=c_3z_3-d_3w_3=-d_3w_3$.
Therefore, $M(f_3)$ has the form
$\left[ \begin{smallmatrix}
a_3 & 0 & 0 & b_3 \\
0 & 0 & d_3 & 0\\
0 & w_3 & 0 & 0\\
y_3 & 0 & 0 & x_3\\
\end{smallmatrix} \right],
$
where $d_3w_3\neq 0$.

Finally, connecting variables $x_4, x_3$ of $f_2$ with variables $x_1$, $x_2$ of $f_{3}$ both using $(\neq_2)$, we get a signature $f_4$ with  the signature matrix
$$M(f_4)=\left[ \begin{matrix}
a_4 & 0 & 0 & b_4 \\
0 & c_4 & d_4 & 0\\
0 & w_4 & z_4 & 0\\
y_4 & 0 & 0 & x_4\\
\end{matrix} \right]=M(f_2)NM(f_3)=
\left[ \begin{matrix}
a_3b_2 + a_2 y_3 & 0 & 0 & b_2 b_3 + a_2 x_3 \\
0 & 0 & d_2 d_3 & 0\\
0 & w_2 w_3 & z_2 d_3 & 0\\
a_3 x_2 + y_2 y_3 & 0 & 0 & b_3 x_2 + x_3 y_2\\
\end{matrix} \right],$$
where $d_2d_3\neq 0,w_2w_3\neq 0,z_2d_3\neq 0$.
By Lemma~\ref{lm: three nonzeros in inner matrix}, $\plholant{\neq_2}{f_4}$ is \numP-hard unless $f_4\in \M$, i.e. $\det M_{\text{In}}(f_4)=\det M_{\text{Out}}(f_4)$.
Notice that 
\begin{equation*}
\begin{aligned}
\det M_{\text{In}}(f_4)&=\det M_{\text{In}}(f_2)\cdot(-1)\cdot \det M_{\text{In}}(f_{3})\\
&=t_1t_2\det M_{\text{In}}(f)^3\cdot t_3t_1^2\det M_{\text{In}}(f)^4\\
&=t_3t_2t_1^3\det M_{\text{In}}(f)^7.
\end{aligned}
\end{equation*}
Similarly, $\det M_{\text{Out}}(f_4)=t_3t_2t_1^3\det M_{\text{Out}}(f)^7$, where $t_3t_2t_1\neq 0$.
Thus, we have $\det M_{\text{In}}(f)^7=\det M_{\text{Out}}(f)^7$.

We conclude that $\plholant{\neq_2}{f,g}$ is \numP-hard, unless $\det M_{\text{In}}(f)^3=\det M_{\text{Out}}(f)^3\neq 0$ and 
$\det M_{\text{In}}(f)^7=\det M_{\text{Out}}(f)^7\neq 0$.
This implies $\det M_{\text{In}}(f)^4=\det M_{\text{Out}}(f)^4$, and further $\det M_{\text{In}}(f)=\det M_{\text{Out}}(f)$, i.e., $f\in \M$.
\end{proof}

\begin{lemma}\label{thm: with (0,1,0,0)}
Let $f$ be a 4-ary signature with the signature matrix
$M(f)=\left[\begin{smallmatrix}
a & 0 & 0 & b\\
0 & c & d & 0\\
0 & w & z & 0\\
y & 0 & 0 & a\\
\end{smallmatrix}\right]$ where $abcdyzw\neq 0$ and
$\left[\begin{smallmatrix}
 c & d \\
 w & z \\
\end{smallmatrix}\right]$ has full rank, then either $\plholant{\neq_2}{f, (0, 1, 0, 0)^T}$ is \numP-hard, or $f\in \M$, or $f$ is  $\mathscr{A}$-transformable.
\end{lemma}

The following proof is similar to that of Theorem 8.2 in~\cite{CaiF17}, with the key difference that we cannot perform arbitrary pinning using $(1, 0)^T \otimes (0, 1)^T$, as we are restricted to planar gadget constructions. Additionally, matchgate signatures are embedded within the proof.

\begin{proof}
By a normalization in $f$, we may assume that $d=1$.

Note that we have $(0, 1, 0, 0)^T=(1, 0)^T\otimes (0, 1)^T$. 
Thus we can set $x_i=0$ and $x_j=1$ for any two adjacent variables $x_i, x_j$ of $f$.
By setting $x_1=0$ and $x_2=1$, $x_3=0$ and $x_4=1$, $x_1=0$ and $x_4=1$, $x_3=0$ and $x_2=1$ respectively using $(0, 1, 0, 0)^T$,
we get the binary signatures
$(0, 1, c, 0)^T$,
$(0, 1, z, 0)^T$,
$(0, 1, b, 0)^T$, and
$(0, 1, y, 0)^T$,
i.e., for any $u\in\{c,z,b,y\}$, we can get $(0, 1, u, 0)^T$.
Thus if there exists $u\in\{c,z,b, y\}$ such that $u^i$ are distinct for $1\leq i\leq 5$, then
$\plholant{\neq_2}{f, (0, 1, 0, 0)^T}$ is \numP-hard by Lemma~\ref{lm: with-any-binary}.
So we may assume  all of $\{c,z,b, y\}$ are in $\{\pm 1, \pm \mathfrak{i}, \omega, \omega^2\}$,
where $\omega = e^{2 \pi {\mathfrak{i}}/3}$.

By Lemma~\ref{lm:inverse binary}, we have $(0, 1, c^{-1}, 0)^T$
and $(0, 1, z^{-1}, 0)^T$.
By binary modifications using $(0, 1, z^{-1}, 0)^T$
and $(0, 1, c^{-1}, 0)^T$ to the variables
$x_1, x_3$ of $f$ respectively, we get a signature $f_1$ with the signature matrix
$M(f_1)=\left[\begin{smallmatrix}
a & 0 & 0 & \frac{b}{c}\\
0 & 1 & 1 & 0\\
0 & \frac{w}{cz} & 1 & 0\\
\frac{y}{z} & 0 & 0 & \frac{a}{cz}\\
\end{smallmatrix}\right]$.
Since 
$\left[\begin{smallmatrix}
 c & 1\\
 w & z \\
\end{smallmatrix}\right]$ has full rank, we have $\frac{w}{cz}\neq 1$.
By doing a loop to variables $x_1,x_2$ of $f_1$ using $\neq_2$, we get the binary signature 
$h=M(f)(\neq_2)=(0, 2, 1+\frac{w}{cz}, 0)^T$.
We may assume that $|1+\frac{w}{cz}|=2$ or $\frac{w}{cz}=-1$.
Otherwise, $\plholant{\neq_2}{f,(0,2,1+\frac{w}{cz},0)^T}$ is \numP-hard by Lemma~\ref{lm: with-any-binary}. 
Thus $\plholant{\neq_2}{f}$ is \numP-hard.

We claim that $\frac{w}{cz}=-1$ or we are done.
By pinning $x_4=0$ and $x_3=1$ using $(0, 1, 0, 0)^T$, we get the binary signature $(0, 1, \frac{w}{cz}, 0)^T$.
If $(\frac{w}{cz})^i$ are distinct for $1\leq i\leq 5$, then
the problem
$\plholant{\neq_2}{f, (0, 1, \frac{w}{cz}, 0)^T)}$ is \numP-hard by Lemma~\ref{lm: with-any-binary}.
Thus $\plholant{\neq_2}{f, (0, 1, 0, 0)^T}$ is \numP-hard.
So we may assume that $\frac{w}{cz}\in\{-1, \pm\mathfrak{i}, \omega, \omega^2\}$.
If $\frac{w}{cz}=\pm\mathfrak{i},\omega,\omega^2$, then $|1+\frac{w}{cz}|\neq 2$.
So we may assume that $(0, 1, \frac{w}{cz}, 0)^T=(0, 1, -1, 0)^T$.

Suppose there exists  $u\in\{c,z,b, y\}$ such that $u=\omega$ or $\omega^2$.
By linking $(0, 1, u, 0)^T$ and $(0, 1, -1, 0)^T$ using $\neq_2$, we get
the binary signature $(0, 1, -u, 0)^T$.
Note that $(-u)^i$ are distinct for $1\leq i\leq 6$, therefore
$\operatorname{Holant}(\neq_2|f, (0, 1, -u, 0)^T)$ is \numP-hard by Lemma~\ref{lm: with-any-binary}.
It follows that  $\plholant{\neq_2}{f, (0, 1, 0, 0)^T}$ is \numP-hard.

Now we can assume
$\{c,z,b, y\}\subseteq\{1, -1, \mathfrak{i}, -\mathfrak{i}\}$.
Then by $\frac{w}{cz}=-1$, we have $w^4=1$.
Thus $\{c,z,b,y,w\}\subseteq\{1, -1, \mathfrak{i}, -\mathfrak{i}\}$.
  There exist
  $j, k,  \ell, m, n\in\{0, 1, 2, 3\}$, such that
\[b=\mathfrak{i}^{j},~~~ y=\mathfrak{i}^{k},~~~ w=\mathfrak{i}^{\ell}, ~~~
  c=\mathfrak{i}^{m}, ~~~ z=\mathfrak{i}^{n}.\]
The signature matrices of $f$ and $f_1$ are respectively
$$M(f)=\left[\begin{smallmatrix}
a & 0 & 0 & \mathfrak{i}^{j}\\
0 &  \mathfrak{i}^{m} &1 & 0\\
0 & \mathfrak{i}^{\ell} & \mathfrak{i}^{n} & 0\\
\mathfrak{i}^{k}& 0 & 0 & a
\end{smallmatrix}\right],~~~~\mbox{and}~~~~
M(f_1)=\left[\begin{smallmatrix}
a & 0 & 0 & \mathfrak{i}^{j-m}\\
0 & 1 & 1 & 0\\
0 & -1 & 1 & 0\\
\mathfrak{i}^{k-n}& 0 & 0 & a\mathfrak{i}^{-m-n}
\end{smallmatrix}\right].$$
We have $\ell-m-n\equiv 2 \pmod 4$ by $\frac{w}{cz}=\mathfrak{i}^{\ell-m-n}=-1$.

Note that 
$M(f^{\frac{\pi}{2}})=
\left[\begin{smallmatrix}
a & 0 & 0 & \mathfrak{i}^{n}\\
0 & \mathfrak{i}^{j} & \mathfrak{i}^{\ell} & 0\\
0 & 1 & \mathfrak{i}^{k} & 0\\
\mathfrak{i}^{m}& 0 & 0 & a
\end{smallmatrix}\right]$.
By doing a loop to variables $x_3,x_4$ of $f^{\frac{\pi}{2}}$ using $(0, 1, \mathfrak{i}^{-k}, 0)^T$, we get the binary signature
\[
M(f^{\frac{\pi}{2}})N(0, \mathfrak{i}^{-k}, 1, 0)^T
=(0, \mathfrak{i}^{j}(1+\mathfrak{i}^{\ell-j-k}),2 , 0)^T.
\]
If $\ell-j-k\equiv 1 \pmod 2$, then $|\mathfrak{i}^{j}(1+\mathfrak{i}^{\ell-j-k})|
= \sqrt{2}$ and
Holant$(\neq_2|f, (0, 2, \mathfrak{i}^{j}(1+\mathfrak{i}^{\ell-j-k}), 0)^T)$ is $\#$P-hard by Lemma~\ref{lm: with-any-binary}.
Thus $\plholant{\neq_2}{f,(0,1,0,0)^T}$ is \numP-hard.
Hence we may assume $\ell-j-k\equiv 0 \pmod 2$.
By $\ell-m-n\equiv 2 \pmod 4$, we have
\[ j+k + m+n \equiv 0 \pmod 2.\]

By connecting two copies of $f_1$, with the second copy rotated by $\pi$, we get a signature $f_2$ with the signature matrix
\[M(f_2)=M(f_1)NM(f_1^{\pi})=
\left[\begin{smallmatrix}
2 \mathfrak{i}^{j-m} a& 0 & 0 & \mathfrak{i}^{-m-n} (a^2+\mathfrak{i}^{j+k})\\
0 & 0 & 2 & 0\\
0 & 2 & 0 & 0\\
\mathfrak{i}^{-m-n}(a^2+\mathfrak{i}^{j+k}) & 0 & 0 & 2 \mathfrak{i}^{k-m-2n}a
\end{smallmatrix}\right].
\]
Notice that 
\begin{equation}\label{eq: with (0,1,0,0)}
\begin{aligned}
    &\det M_{\text{In}}(f_2)=-\det M_{\text{In}}(f_1)^2=-\frac{1}{c^2z^2}\det M_{\text{In}}(f)^2 \\
    &\det M_{\text{Out}}(f_2)=-\det M_{\text{Out}}(f_1)^2=-\frac{1}{c^2z^2}\det M_{\text{Out}}(f)^2
\end{aligned}
\end{equation}
Also, note that 
\[M(f^{\frac{\pi}{2}}_2)=M(f_1)NM(f_1^{\pi})=
\left[\begin{smallmatrix}
2 \mathfrak{i}^{j-m} a& 0 & 0 & 0\\
0 & \mathfrak{i}^{-m-n} (a^2+\mathfrak{i}^{j+k}) & 2 & 0\\
0 & 2 & \mathfrak{i}^{-m-n}(a^2+\mathfrak{i}^{j+k}) & 0\\
0 & 0 & 0 & 2 \mathfrak{i}^{k-m-2n}a
\end{smallmatrix}\right].
\]
By setting $x_3=0$ and $x_4=1$ using $(0,1,0,0)^T$, we get a binary signature $(0,2,\mathfrak{i}^{-m-n}(a^2+\mathfrak{i}^{j+k}),0)^T$.
Then either $a^2+\mathfrak{i}^{j+k}=0$ or $|a^2+\mathfrak{i}^{j+k}|=2$, otherwise $\plholant{\neq_2}{f,(0,2,\mathfrak{i}^{-m-n}(a^2+\mathfrak{i}^{j+k}),0)^T}$ is \numP-hard by Lemma~\ref{lm: with-any-binary}, and thus $\plholant{\neq_2}{f}$ is \numP-hard.

We first assume $a^2+\mathfrak{i}^{j+k}\neq 0$.
By Lemma~\ref{lm:1cdwz1}, either $f_2\in \M$, or $\plholant{\neq_2}{f_2}$ is \numP-hard, and thus $\plholant{\neq_2}{f}$ is \numP-hard.
If $f_2\in \M$, then $\det M_{\text{In}}(f)^2=\det M_{\text{Out}}(f)^2$ by Lemma~\ref{lm:is matchgate} and Equation~\ref{eq: with (0,1,0,0)}.
If $\det M_{\text{In}}(f)=\det M_{\text{Out}}(f)$, then $f\in \M$.
Otherwise, $\det M_{\text{In}}(f)=-\det M_{\text{Out}}(f)$.
Then $a^2=\mathfrak{i}^{j+k}+\mathfrak{i}^\ell-\mathfrak{i}^{m+n}$.
As $\ell-m-n\equiv 2\pmod 4$, $a^2=\mathfrak{i}^{j+k}+2\mathfrak{i}^\ell$. 
If $\ell-j-k\equiv 0 \pmod 4$, then $a^2=3\mathfrak{i}^{j+k}$, contradicting $|a^2+\mathfrak{i}^{j+k}|=2$.
If $\ell-j-k\equiv 2 \pmod4$, then $a^2=-\mathfrak{i}^{j+k}$, contradicting the assumption that $a^2+\mathfrak{i}^{j+k}\neq 0$.

Now we assume $a^2+\mathfrak{i}^{j+k}=0$.
Then $
M(f)=\left[\begin{smallmatrix}
\mathfrak{i}^{\frac{j+k}{2}+\epsilon} & 0 & 0 & \mathfrak{i}^{j}\\
0 &  \mathfrak{i}^{m} &1 & 0\\
0 & \mathfrak{i}^{\ell} & \mathfrak{i}^{n} & 0\\
\mathfrak{i}^{k}& 0 & 0 & \mathfrak{i}^{\frac{j+k}{2}+\epsilon}
\end{smallmatrix}\right]
$, where $\epsilon=\pm 1$.
Let $r=j+m$.
We apply a holographic transformation
defined by
$\left[\begin{smallmatrix}
 1 & 0\\
 0 & \gamma
\end{smallmatrix}\right]$, where $\gamma^2=\mathfrak{i}^{\frac{j+k}{2}+r+\epsilon}$, then we get
$
\plholant{\neq_2}{f}\equiv_{T}\plholant{\neq_2}{\hat{f}},
$
where $\hat{f}=\left[\begin{smallmatrix}
 1 & 0\\
 0 & \gamma
\end{smallmatrix}\right]^{\otimes 4}f$, 
and its signature matrix is
$M(\hat{f})=\mathfrak{i}^{\frac{j+k}{2}+\epsilon}\left[\begin{smallmatrix}
1 & 0 & 0 & \mathfrak{i}^{j+r}\\
0 &  \mathfrak{i}^{m+r} &\mathfrak{i}^{r} & 0\\
0  & \mathfrak{i}^{\ell+r} & \mathfrak{i}^{n+r}& 0\\
\mathfrak{i}^{k+r}& 0 & 0 & -\mathfrak{i}^{j+k+2r}
\end{smallmatrix}\right]$.
By $\ell \equiv m+n+ 2 \pmod 4$, we have
\[M(\hat{f})=\mathfrak{i}^{\frac{j+k}{2}+\epsilon}\left[\begin{smallmatrix}
1 & 0 & 0 & \mathfrak{i}^{j+r}\\
0 & \mathfrak{i}^{r} & \mathfrak{i}^{m+r} & 0\\
0 & \mathfrak{i}^{n+r} & -\mathfrak{i}^{m+n+r} & 0\\
\mathfrak{i}^{k+r}& 0 & 0 & -\mathfrak{i}^{j+k+2r}
\end{smallmatrix}\right].\]
This function is an affine function; indeed
 let
\[Q(x_1, x_2, x_3)=(k-n-r)x_1x_2+(2j+2)x_1x_3+(2j)x_2x_3+(n+r)x_1+rx_2+(j+r)x_3,\]
then
 $\hat{f}(x_1, x_2, x_4, x_3)=\mathfrak{i}^{\frac{j+k}{2}+\epsilon}\cdot
\mathfrak{i}^{Q(x_1, x_2, x_3)}$ on the support
of $\hat{f}$: $x_1 + x_2 + x_3 + x_4 \equiv 0 \pmod 2$.
 Moreover, $k-n-r \equiv j + k + n + m \equiv 0 \pmod 2$,
all cross terms have even coefficients.
Thus $\hat{f}\in\mathscr{A}$.
\end{proof}

\begin{theorem}\label{thm: all nonzero entries, not three equal or opposite pairs}
    Let $f$ be a 4-ary signature with the signature matrix
    $M(f)=\left[\begin{smallmatrix}
    a & 0 & 0 & b\\
    0 & c & d & 0\\
    0 & w & z & 0\\
    y & 0 & 0 & a\\
    \end{smallmatrix}\right]$
    with $abcdyzw\neq 0$.
    If it is not the case where $b=\epsilon y,c=\epsilon z$, and $d=\epsilon w$ for some $\epsilon=\pm 1$,
    then either $f\in \M$, or $f$ is $\mathscr{A}$-transformable, or $\plholant{\neq_2}{f}$ is $\numP$-hard.
\end{theorem}
\begin{proof}
If $by=cz=dw$, then by Lemma~\ref{lm: degenerate inner matrices} $f$ satisfies arrow reversal equality, contradiction.
Otherwise, by the rotational symmetry of the outer pairs
$(b, y), (c, z)$,
we can assume that $\left[\begin{smallmatrix}
c & d\\
w & z
\end{smallmatrix}\right]$ has full rank.

By Lemma~\ref{lm:get one binary when not exists epsilon} we can realize $(0, 1, t, 0)^T$ for $t
\neq  \pm 1$.
If $t=0$, we have $(0, 1, 0, 0)^T$, then we are done by Lemma~\ref{thm: with (0,1,0,0)}.
Otherwise, we have $(0, 1, t, 0)^T$, where $t\neq 0$ and $t \not = \pm  1$.
If $t^i$ are distinct for $1\leq i\leq 5$, then we are done by Lemma~\ref{lm: with-any-binary}.
Otherwise, $t$ is an $n$-th primitive root of unity for $n=3$ or 4.
Then we have $(0, 1, 0, 0)^T$ by
Corollary~\ref{cor: get (0,1,0,0) from a binary third or fourth root},
and we are done by Lemma~\ref{thm: with (0,1,0,0)} once again.    
\end{proof}

Finally we are ready to prove the main theorem: Theorem~\ref{thm: main theorem}.

\section{Proof of Main Theorem}\label{sec: proof of main theorem}
If $ax=0$, by Lemma~\ref{lm:let a=x}, we have $\plholant{\neq_2}{f}\equiv_T \plholant{\neq_2}{f'}$, where $f'$ has the signature matrix
$M(f')=\left[\begin{smallmatrix}
0 & 0 & 0 & b\\
0 & c & d & 0\\
0 & w & z & 0\\
y & 0 & 0 & 0
\end{smallmatrix}\right]$.
This corresponds to the planar six-vertex model, whose complexity is fully characterized by Theorem~\ref{thm:planar six-vertex}.

Now we assume $ax\neq 0$.
By a holographic transformation using $T=\left[\begin{smallmatrix}
    1 & 0\\
    0 & \sqrt[4]{\frac{a}{x}}
\end{smallmatrix}\right]$,
$f$ is transformed to the signature $\tilde{f}$ with the signature matrix
$$M(\tilde{f})=
a\left[\begin{matrix}
1 & 0 & 0 & \frac{b}{\sqrt{ax}}\\
0 & \frac{c}{\sqrt{ax}} & \frac{d}{\sqrt{ax}} & 0\\
0 & \frac{w}{\sqrt{ax}} & \frac{z}{\sqrt{ax}} & 0\\
\frac{z}{\sqrt{ax}} & 0 & 0 & 1
\end{matrix}\right],$$
and $\plholant{\neq_2}{f}\equiv_T^p\plholant{\neq_2}{\tilde{f}}$ by Theorem~\ref{thm: holographic transformation}.
If there is at least one zero in $\{b, y, c, z, d, w\}$, then by Theorem~\ref{thm: at least one zero}, either $\plholant{\neq_2}{\tilde{f}}$ is \numP-hard and thus $\plholant{\neq_2}{f}$ is \numP-hard; 
or $\tilde{f}$ is in $\mathscr{P}$, or $\mathscr{A}$-transformable, or $\mathscr{M}$-transformable, which implies that $f$ is $\mathscr{P}$-transformable, or $\mathscr{A}$-transformable, or $\mathscr{M}$-transformable, respectively.

Below we assume there are no zeros in  $\{b, y, c, z, d, w\}$.

If $\{(b,y),(c,z),(d,w)\}$ are three equal pairs, then by Theorem~\ref{thm: three equal pairs}, either $\plholant{\neq_2}{\tilde{f}}$ is \numP-hard and thus $\plholant{\neq_2}{f}$ is \numP-hard; or $\tilde{f}$ is $\mathscr{P}$-transformable, or $\mathscr{A}$-transformable, or $\mathscr{L}$-transformable, or $\mathscr{M}$-transformable; which implies that $f$ is $\mathscr{P}$-transformable, or $\mathscr{A}$-transformable, or $\mathscr{L}$-transformable, or $\mathscr{M}$-transformable, repectively;
or $(\frac{d}{\sqrt{ax}}=\frac{w}{\sqrt{ax}}=\pm1)\land (\frac{b}{\sqrt{ax}}=-\frac{c}{\sqrt{ax}}=\mathfrak{i}\tan(\frac{\beta\pi}{8}))$, which implies that $(d=w=\pm \sqrt{ax})\land (b=-c=\mathfrak{i}\tan(\frac{\beta\pi}{8})\sqrt{ax})$, for some odd integer $\beta$.
Thus, the theorem holds in the case of three equal pairs.
Similarly, if $\{(b,y),(c,z),(d,w)\}$ are three opposite pairs, the theorem follows from Lemma~\ref{lm: symmetric disequality dichotomy}.

If it's not the case that $\{(b,y),(c,z),(d,w)\}$ are three equal or opposite pairs, the theorem follows from Theorem~\ref{thm: all nonzero entries, not three equal or opposite pairs}.
\qed

\appendix

\section{Interpolation via Möbius transformation}\label{subsec: mobius}
We first show that under the assumption it's not the case that $b=\epsilon y,c=\epsilon z$ and $d=\epsilon w$ for some $\epsilon=\pm1$, we can realize a binary signature $g = (0,1,t,0)^T$ with $t \ne \pm1$, which serves as a “trigger” for generating polynomially many binary signatures using M\"{o}bius transformation techniques (Lemma~\ref{lm: get any binary from a binary not fifth root}, Corollary~\ref{cor: get (0,1,0,0) from a binary third or fourth root}).
\begin{lemma}\label{lm:get one binary when not exists epsilon}
Let $f$ be a 4-ary signature with the signature matrix $M(f) = \left[\begin{smallmatrix}a & 0 & 0 & b \\ 0 & c & d & 0 \\ 0 & w & z & 0 \\ y & 0 & 0 & a \end{smallmatrix}\right]$. 
If it is not the case that $b=\epsilon y,c= \epsilon z$ and $d = \epsilon w$ where $\epsilon = \pm 1$, then
$$
\plholant{\neq_2}{(0,1,t,0),f} \leq^p_T \plholant{\neq_2}{f},
$$
for some $t\in \mathbb{C}$ and $t \neq \pm 1$.
\end{lemma}

\begin{proof}
We first assume it is not the case that $c=\epsilon z$ and $d=\epsilon w$.
\begin{itemize}
\item 
First deal with the case $\det A=0$,
where $A =
\left[\begin{smallmatrix}
c & d\\
w & z
\end{smallmatrix}\right]$.
Since it is not the case that $c= \epsilon z$ and $d = \epsilon w$ where $\epsilon = \pm 1$, the four entries of $A$ cannot be all 0.
By the rotational symmetry of $f$, we may assume $c \not =0$ or $d\neq 0$.

We first consider the case $c\neq 0$.
Normalizing $c =1$ we have
$A
=
\left[\begin{smallmatrix}
1 & d\\
w & dw
\end{smallmatrix}\right]$.
By a loop on $x_1, x_2$ of $f$ via $(\neq_2)$,
and a loop on $x_3, x_4$ of $f$ via $(\neq_2)$,
we get both $(1+w)(0, 1, d, 0)^T$ and $(1+d)(0, 1, w, 0)^T$.
If $w = \pm 1$ and $d = \pm 1$,
then we have $1=\epsilon\cdot dw$ and $d=\epsilon \cdot w$ for $\epsilon=dw=\pm1$, contradiction.
So we have
$w  \not = \pm 1$  or $d \not = \pm 1$.
Assume $w  \not = \pm 1$; the other case is symmetric.
Since we have $(1+d) \cdot (0, 1, w, 0)^T$,
if $d  \not = - 1$, then we have $(0, 1, w, 0)^T$ with $w\neq \pm1$
satisfying the requirement.
Suppose $d = -1$, then from $(1+w) \cdot (0, 1, d, 0)^T$,
we get $(0, 1, -1, 0)^T$.
By a loop on $x_3, x_4$ of $f$ via $(0, 1, -1, 0)^T$,
we get $(1-d)  \cdot (0, 1, w, 0)^T = 2 (0, 1, w, 0)^T$ with $w\neq \pm1$
satisfying the requirement.

Now we consider the case $d\neq 0$.
Normalizing $d=1$ we have 
$A=
\left[\begin{smallmatrix}
c & 1\\
cz & z
\end{smallmatrix}\right]$.
By a loop on $x_1, x_2$ of $f$ via $(\neq_2)$,
and a loop on $x_3, x_4$ of $f$ via $(\neq_2)$,
we get both $(1+z)(0, c, 1, 0)^T$ and $(1+c)(0, 1, z, 0)^T$.
If $c = \pm 1$ and $z = \pm 1$,
then we have $c=\epsilon\cdot z$ and $1=\epsilon \cdot cz$ for $\epsilon=cz=\pm1$, contradiction.
So we have
$c  \not = \pm 1$  or $z \not = \pm 1$.
Assume $c \not = \pm 1$; the other case is symmetric.
Since we have $(1+z) \cdot (0, c, 1, 0)^T$,
if $z  \not = - 1$, then we have $(0, c, 1, 0)^T$, and also $(0,1,c,0)^T$ after a rotation, with $c\neq 1$
satisfying the requirement.
Suppose $z = -1$, then from $(1+c) \cdot (0, 1, z, 0)^T$,
we get $(0, 1, -1, 0)^T$.
By a loop on $x_1, x_2$ of $f$ via $(0, 1, -1, 0)^T$,
we get $(1-z)  \cdot (0, c, 1, 0)^T = 2 (0, c, 1, 0)^T$
satisfying the requirement.

\item 

In the following we assume
 $\det A \not =0$.
By connecting $x_3$ and $x_4$ of $f$ via $(\neq_2)$
we get $(0, c+d, w+z, 0)^T$.
If  $w+z =0$, then $c+d \not = 0$ as $\det A\neq 0$.
So we have $(0,1,0,0)^T$ with $t=0$ satisfying the requirement. 
Otherwise we can write
 $(0, c+d, w+z, 0)^T = (w+z) \cdot (0, \frac{c+d}{w+z}, 1, 0)^T$.
If $\frac{c+d}{w+z} \not = \pm 1$ then we are done.
Hence we may assume $\frac{c+d}{w+z} = \pm 1$.
Similarly we may assume $\frac{c+w}{d+z} = \pm 1$,
by connecting $x_1$ and $x_2$  of $f$ via $(\neq_2)$
and getting $(0, c+w, d+z, 0)$.

If either $\frac{c+d}{w+z} = -1$ or $\frac{c+w}{d+z} = -1$, then we have
\begin{equation}\label{sum-cdwz-is-zero}
c+d + w +z  =  0,
\end{equation}
and we realize the signature $(0,1,-1,0)^T$.
By connecting $x_3$ and $x_4$ of $f$  via $(0, 1, -1, 0)^T$
we get the binary signature $(0, c-d, w-z, 0)^T$.
If one of $c-d$ or $w-z=0$, then the other one is nonzero
as $\det A \not = 0$,
and we get $(0,1,0,0)^T$ satisfying the requirement.
So we may assume $c-d \not = 0$ and $w-z \not = 0$,
and we have $(0,1, \frac{w-z}{c-d}, 0)^T$.
If $\frac{w-z}{c-d} \not = \pm 1$, then we are done.
If  $\frac{w-z}{c-d} = 1$, then $c-d = w-z$.
Combined with (\ref{sum-cdwz-is-zero}), we get $c= -z$ and $d= -w$, contradicting the assumption.
If  $\frac{w-z}{c-d} = -1$, then $c-d = z-w$.
Combined with (\ref{sum-cdwz-is-zero}), we get $c= -w$ and $d = -z$.
This contradicts $\det A \not = 0$.

Thus we may assume
$\frac{c+d}{w+z} = \frac{c+w}{d+z}=1$.
Then from $c+d = w+z$ and $c+w = d+z$,  we get $c= z$ and $d= w$, contradicting the assumption.
\end{itemize}
If it is the case that $c=\epsilon z$ and $d=\epsilon w$ for some $\epsilon=\pm1$, we consider $f^{\frac{\pi}{2}}$ with the signature matrix
$M(f^{\frac{\pi}{2}})=
\left[\begin{smallmatrix}
a & 0 & 0 & z\\
0 & b & w & 0\\
0 & d & y & 0\\
c & 0 & 0 & a
\end{smallmatrix}\right]$.
We have $b\neq \epsilon y$ by the assumption of this lemma.
Then it is not the case that $b=\epsilon y$ and $w=\epsilon d$ where $\epsilon=\pm1$.
The same argument as above applies.
\end{proof}

The following lemma follows from standard argument using polynomial interpolation.

\begin{lemma}\cite{CaiFS21}\label{lm: mobius: any binary}
Let $\mathcal{F}$ be an arbitrary signature set.
Suppose there is a planar gadget construction in $\plholant{\neq_2}{\mathcal{F}}$ that can produce a sequence of polynomially many distinct binary signatures of the form $g = (0,1,t,0)^T$ of polynomial size description, then for any $u\in \mathbb{C}$,
$$
\plholant{\neq_2}{\mathcal{F}, (0,1,u,0)^T} \leq^p_T \plholant{\neq_2}{\mathcal{F}}. 
$$
\end{lemma}

\begin{corollary}\label{cor: mobius: with (0,1,t,0) not of unity}
The conclusion of Lemma~\ref{lm: mobius: any binary} holds if $\mathcal{F}$ contains some binary signature $g=(0,1,t,0)^T$, where $t \neq 0$ and is not a root of unity.
\end{corollary}
\begin{proof}
Connecting $s$ copies of $g$ using $\neq_2$, we get the binary signature $g_s$ with signature matrix
$
M(g_s)=\left[\begin{smallmatrix} 0 & 1 \\ t & 0 \end{smallmatrix}\right] (\left[\begin{smallmatrix} 0 & 1 \\ 1 & 0 \end{smallmatrix}\right]
\left[\begin{smallmatrix} 0 & 1 \\ t & 0 \end{smallmatrix}\right])^{s-1} = 
\left[\begin{smallmatrix} 0 & 1 \\ t^{s} & 0 \end{smallmatrix}\right]$.
Since $t$ is not root of unity, these gadgets $g_s$ are all distinct.
\end{proof}

The following lemma employs M\"obius transformation to show that, given a binary “trigger” $g = (0,1,t,0)^T$ where the powers $t^i$ are all distinct for $1 \le i \le 5$, we can realize any binary signature of the form $(0,1,u,0)^T$.
This is useful in Lemma~\ref{lm: with-any-binary}.
The proof is a nontrivial generalization of Lemma 4.1 in~\cite{CaiF17}, as we are restricted to planar gadget constructions. 
Specifically, we encounter a M\"obius transformation defined by the matrix $\left[\begin{smallmatrix} 
1 & \lambda \\ 
\lambda & -1 \end{smallmatrix}\right]$, where $\lambda \in \R \mathfrak{i}$, which has projective order 2 (see the case $\lambda+\overline{\lambda}=e^{2\mathfrak{i} \theta}\lambda+\bar{\lambda}=0$ in the proof).
Thus, we cannot use it to generate polynomially many binary signatures in this setting. Instead, we perform a limit analysis and invoke facts about cyclotomic extensions to construct a binary signature $(0,1,u,0)^T$ where $u$ is not a root of unity. 
We then conclude the proof by applying Corollary~\ref{cor: mobius: with (0,1,t,0) not of unity}.

\begin{lemma}\label{lm: get any binary from a binary not fifth root}
Let $g = (0,1,t,0)^T$ be a binary signature where $t^i$ are distinct for $1 \leq i \leq 5$, and $f$ be a signature with the signature matrix $M(f) = \left[\begin{smallmatrix} a & 0 & 0 & b \\ 0 & c & d & 0 \\ 0 & w & z & 0 \\ y & 0 & 0 & a \end{smallmatrix} \right]$ with $acdzw \neq 0$, where $\left[\begin{smallmatrix} c & d \\ w & z\end{smallmatrix}\right]$ has full rank, then for any $u \in \mathbb{C}$, $\plholant{\neq_2}{f,(0,1,u,0)^T} \leq^p_{T} \plholant{\neq_2}{f,g}$.
\end{lemma}

\begin{proof}

By connecting two copies of $g$ using $(\neq_2)$, we get a signature $g_2$ with the signature matrix 
$
M(g_2)=
\left[\begin{smallmatrix}
 0& 1\\
 t& 0\\
\end{smallmatrix}\right]
\left[\begin{smallmatrix}
 0& 1\\
 1& 0\\
\end{smallmatrix}\right]
\left[\begin{smallmatrix}
 0& 1\\
 t& 0\\
\end{smallmatrix}\right]=
\left[\begin{smallmatrix}
 0& 1\\
 t^2& 0\\
\end{smallmatrix}\right]
.$
Similarly, we can construct
at least five distinct  $g_i=(0, 1, t^i, 0)^T$ for $1\leq i\leq 5$.


By connecting the variables $x_3$ and $x_4$ of the signature $f$ with the variables $x_1$ and $x_2$ of $g_i$ using $\neq_2$
for $1 \leq i \leq 5$ respectively,  we get binary signatures
$$h_i=M(f)Ng_i=
\left[\begin{smallmatrix}
a & 0& 0& b\\
0 & c& d& 0\\
0 & w& z& 0\\
y & 0 & 0 & a
\end{smallmatrix}\right]
\left[\begin{smallmatrix}
0\\
t^i\\
1\\
0\\
\end{smallmatrix}\right]
=\left[\begin{smallmatrix}
0\\
ct^i+d\\
wt^i+z\\
0\\
\end{smallmatrix}\right].
$$
Let $\varphi(\mathfrak z)=\dfrac{c\mathfrak z+d}{w\mathfrak z+z}$.
Since $\det\left[\begin{smallmatrix}
c & d\\
w & z
\end{smallmatrix}\right]\neq 0$,
$\varphi(\mathfrak z)$ is a M\"{o}bius transformation  of the  extended  complex  plane $\widehat{\mathbb{C}}$.
We rewrite $h_i$ as $(wt^i+z)(0, \varphi(t^i), 1, 0)^T$, with the understanding that if $wt^i+z=0$, then $ \varphi(t^i)=\infty$, and we define
$(wt^i+z)(0, \varphi(t^i), 1, 0)^T$ to be $(0, ct^i+d, 0, 0)^T$.
Having $(0, \varphi(t^i), 1, 0)^T$ is equivalent to
having $(0, 1, \varphi(t^i), 0)^T$.
If there is a $t^i$ such that $\varphi(t^i)\neq 0, \infty$ or a root of unity for $1\leq i\leq 5$,
 then by
 Corollary~\ref{cor: mobius: with (0,1,t,0) not of unity}
 $\plholant{\neq_2}{f, (0, 1, u, 0)^T}\leq^p_T\plholant{\neq_2}{f, (0, 1, \varphi(t^i), 0)^T}$,
for all $u \in \mathbb{C}$.
 Otherwise, $\varphi(t^i)$ is $0, \infty$ or a root of unity for $1\leq i\leq 5$. Since
 $\varphi(\mathfrak z)$ is a bijection of $\widehat{\mathbb{C}}$, there is at most one $t^i$ such that $\varphi(t^i)=0$ and at most one $t^i$ such that $\varphi(t^i)=\infty$.
That means, there are at least three $t^i$ such that $|\varphi(t^i)|=1$.
 Since a M\"{o}bius transformation is determined by any 3 distinct points, mapping 3 distinct points from $S^1$ to $S^1$ implies that
 this $\varphi(\mathfrak z)$ maps $S^1$ homeomorphically onto $S^1$.
 So in fact,  $|\varphi(t^i)|=1$, for all $1\leq i\leq 5$.


A M\"{o}bius transformation mapping 3 distinct points from $S^1$ to $S^1$ has a special form
$\mathcal{M}(\lambda, e^{\mathfrak{i}\theta})$:
${\mathfrak z} \mapsto e^{\mathfrak{i}\theta}\dfrac{\mathfrak z+\lambda}{1+\bar{\lambda}\mathfrak z}$,
where $|\lambda|\neq 1$.
By normalization in $f$, we may assume $z=1$, since $z \not =0$.
 Comparing coefficients with $\varphi(\mathfrak z)$ we have
 $c= e^{\mathfrak{i}\theta}$,
$d = e^{\mathfrak{i}\theta} \lambda$ and
$w = \bar{\lambda}$.
Since $w \not = 0$ we have $\lambda \not = 0$.
Thus, $$M(f)=\left[\begin{matrix}
a & 0 & 0 & b\\
0 & e^{\mathfrak{i}\theta} & e^{\mathfrak{i}\theta}\lambda & 0\\
0 & \bar{\lambda} & 1 & 0\\
y & 0 & 0 & a\\
\end{matrix}\right].$$
Rotate $f$ by $\pi$ we have
$M(f^{\pi}) = \left[\begin{smallmatrix}
a & 0 & 0 & y \\
0 & 1 & e^{\mathfrak{i}\theta}\lambda & 0 \\
0 & \overline{\lambda} & e^{\mathfrak{i}\theta} & 0 \\
b & 0 & 0 & a
\end{smallmatrix}\right].$ 
By connecting $f$ and $f^{\pi}$ using $(\neq_2)^{\otimes 2}$,
we get a signature $f_1$ with the signature matrix
 $$
 M(f_1)=M(f) N M(f^\pi) = \left[ \begin{matrix} * & 0 & 0 & * \\ 0 & e^{\mathfrak{i} \theta}(\lambda+\bar{\lambda}) & e^{2\mathfrak{i}\theta }(\lambda^2+1) & 0 \\ 0 & {\overline{\lambda}}^2+1 & e^{\mathfrak{i} \theta}(\lambda+\overline{\lambda}) & 0 \\ * & 0 & 0 & * \end{matrix} \right].
 $$
\begin{itemize}
    \item 
    If $\lambda+\overline{\lambda}\neq 0$,
we can normalize $M(f_1)$ by the nonzero scalar $e^{\mathfrak{i}\theta}(\lambda + \overline{\lambda})$,
and get
$M(f_1)=\left[\begin{smallmatrix}
* & 0 & 0 & *\\
0 & 1 & \delta & 0\\
0 & \bar{\delta} & 1 & 0\\
* & 0 & 0 & *\\
\end{smallmatrix}\right]$,
where $\delta=e^{\mathfrak{i}\theta}\frac{\lambda^2+1}{\lambda+\overline{\lambda}}$.
Since $|\lambda|\neq 1$, we have $\delta\neq 0$.
Note that
$\left[\begin{smallmatrix}
1 & \delta\\
\bar{\delta} & 1
\end{smallmatrix}\right]
=M_{\rm{In}}(f_1)=M_{\rm{In}}(f)M_{\rm{In}}(N)M_{\rm{In}}(f^\pi)
$ is
the product of three non-singular $2 \times 2$ matrices, thus
it is also non-singular.
So $|\delta| \not =1$.
The two eigenvalues of
$\left[\begin{smallmatrix}
1 &  \delta \\
\bar{\delta} & 1
\end{smallmatrix}\right]$
 are $1 + |\delta|$ and $1- |\delta|$, both are real and must be nonzero.

Obviously $|{1+ |\delta|}| \neq  |{1- |\delta|}|$,
since $|\delta| \neq 0$.
This implies that
 there are no integer $n> 0$ and complex number $\mu$ such that
$\left[\begin{smallmatrix}
1 &  {\delta} \\
  \bar{\delta} & 1
\end{smallmatrix}\right]^n=\mu I$,
i.e., $\left[\begin{smallmatrix}
1 &  \delta \\
\bar{\delta} & 1
\end{smallmatrix}\right]$ has infinite projective order.
Note that
$\left[\begin{smallmatrix}
1 &  \delta \\
\bar{\delta} & 1
\end{smallmatrix}\right]$
defines a M\"{o}bius transformation $\psi({\mathfrak z})$
of the form
$\mathcal{M}(\lambda, e^{{\mathfrak{i}}\theta})$ with $\lambda = \delta$
and $\theta =0$:
${\mathfrak z} \mapsto \psi({\mathfrak z})=
\dfrac{{\mathfrak z}+\delta}{1+\bar{\delta}{\mathfrak z}}$,
mapping  $S^1$ to $S^1$.
Since $\left[\begin{smallmatrix}
1 &  \delta \\
\bar{\delta} & 1
\end{smallmatrix}\right]$ has infinite projective order,
the  M\"{o}bius transformation $\psi({\mathfrak z})$
defines an infinite group.
%

We can connect the binary signature
$g_i(x_1, x_2)$ via $N$ to $f_1$. This gives us
the binary signatures  (for $1 \leq i \leq 5$)
$$h^{(1)}_i=M(f_1)Ng_i=
\left[\begin{smallmatrix}
* & 0& 0& *\\
0 & 1 &  \delta & 0\\
0 & \bar{\delta}& 1 & 0\\
* & 0& 0& *\\
\end{smallmatrix}\right]
\left[\begin{smallmatrix}
0\\
t^i\\
1\\
0\\
\end{smallmatrix}\right]
=C_{(i, 0)}\left[\begin{smallmatrix}
0\\
\psi(t^i)\\
1\\
0\\
\end{smallmatrix}\right],
$$
where
$\psi({\mathfrak z})$ is the
 M\"{o}bius transformation
defined by
 the matrix
$\left[\begin{smallmatrix}
1 &  {\delta} \\
\bar{\delta}  & 1
\end{smallmatrix}\right]$,
and $C_{(i, 0)} = 1 + \bar{\delta} t^i$.
Since $\psi$ maps $S^1$ to $S^1$, and $|t^i| =1$,
clearly $C_{(i, 0)} \not =0$, and
$\psi(t^i) \in S^1$.

Now we can use $h^{(1)}_i(x_2, x_1)
= (0, 1, \psi(t^i), 0)^{T}$ in place of
$g_i(x_1, x_2) = (0, 1, t^i, 0)^{T}$
and repeat this construction. Then we get
\[
h^{(2)}_i=
\left[\begin{smallmatrix}
* & 0& 0& *\\
0 & 1 &  \delta & 0\\
0 & \bar{\delta}& 1 & 0\\
* & 0& 0& *\\
\end{smallmatrix}\right]
\left[\begin{smallmatrix}
0\\
\psi(t^i)\\
1\\
0\\
\end{smallmatrix}\right]
=C_{(i, 1)}\left[\begin{smallmatrix}
0\\
\psi^2(t^i)\\
1\\
0\\
\end{smallmatrix}\right],
\]
where $\psi^2$ is the composition $\psi \circ \psi$,
corresponding to
$\left[\begin{smallmatrix}
1 &  {\delta} \\
\bar{\delta}  & 1
\end{smallmatrix}\right]^2$,
and $C_{(i, 1)} = 1 + \bar{\delta} \psi(t^i) \not =0$.

We can iterate this process and get polynomially many
 $h^{(k)}_i = (0, \psi^k(t^i), 1, 0)^{T}$
for $1 \le i \le 5$ and $k \ge 1$.

If for each $i\in\{1, 2, 3\}$, there is some $n_i > 0$ such that
$\psi^{n_i}(t^i)=t^i$,  then $\psi^{n_0}(t^i)=t^i$,
 for $n_0=n_1n_2n_3 >0$,
and all $1\leq i\leq 3$, i.e., the M\"{o}bius transformation $\psi^{n_0}$
 fixes three distinct
complex numbers $t, t^2, t^3$.
So the M\"{o}bius transformation is the identity map, i.e.,
$\psi^{n_0}(\mathfrak z)=\mathfrak z$ for all $\mathfrak z\in\mathbb{C}$.
This implies that $\psi(\mathfrak z)$ defines a group of finite order,
a contradiction.
Therefore, there is an $i \in \{1,2,3\}$
such that $\psi^{n}(t^i) \not = t^i$ for all  $n\in\mathbb{N}$.
This implies that
 $(1, \psi^{n}(t^i))$ are all distinct for $n\in\mathbb{N}$,
since $\psi$ maps $S^1$  1-1 onto $S^1$.
Then we can generate polynomially many distinct binary signatures
of the form $(0, 1, \psi^{n}(t^i), 0)^T$. By Lemma~\ref{lm: mobius: any binary}, for any $u\in\mathbb{C}$ we have
$\plholant{\neq_2}{f, (0, 1, u, 0)^T}\leq^p_T \plholant{\neq_2}{f, g}.$
\item
If $\lambda+\overline{\lambda}=0$ (thus $\lambda \in \mathfrak{i}\R$)  but $e^{2\mathfrak{i} \theta}\lambda+\bar{\lambda}\neq 0$, by connecting $f^\pi$ and $f$ using $(\neq_2)^{\otimes2}$, we get a signature $f_2$ with the signature matrix
 $$
 M(f_2)=M(f^\pi) N M(f) = \left[ \begin{matrix} 
 * & 0 & 0 & * \\ 
 0 & e^{2\mathfrak{i} \theta}\lambda+\bar{\lambda} & 1+e^{2\mathfrak{i}\theta}\lambda^2 & 0 \\ 
 0 & e^{2\mathfrak{i}\theta}+{\overline{\lambda}}^2 & e^{2\mathfrak{i} \theta}\lambda+\overline{\lambda} & 0 \\ 
 * & 0 & 0 & * \end{matrix} \right].
 $$

Normalize $M(f_2)$ by the nonzero scalar $e^{2\mathfrak{i}\theta} \lambda + \overline{\lambda}$,
we get
$M(f_2)=\left[\begin{smallmatrix}
* & 0 & 0 & *\\
0 & 1 & \delta & 0\\
0 & \bar{\delta} & 1 & 0\\
* & 0 & 0 & *\\
\end{smallmatrix}\right]$,
where $\delta=\frac{1+e^{2\mathfrak{i}\theta}\lambda^2}{e^{2\mathfrak{i} \theta}\lambda+\overline{\lambda}}$.  (We have used the fact that $\lambda \in \mathfrak{i}\R$.)
Again the inner matrix 
$\left[\begin{smallmatrix}
1 &  {\delta} \\
\bar{\delta}  & 1
\end{smallmatrix}\right]$
has infinite projective order and defines a M\"obius transformation from $S^1$ to $S^1$.
The same argument as above applies.

\item 
If $\lambda+\overline{\lambda}=e^{2\mathfrak{i} \theta}\lambda+\bar{\lambda}=0$,
then $e^{\mathfrak{i}\theta}=\pm 1$ as $\lambda\neq 0$.
If $e^{\mathfrak{i}\theta}=1$, then 
$M_{\rm In}(f)=\left[\begin{smallmatrix}
1 & \lambda\\
\overline{\lambda} & 1
\end{smallmatrix}\right]
$
has infinite projective order and defines a M\"obius transformation from $S^1$ to $S^1$.
The same argument as above applies.

In the following we assume $e^{\mathfrak{i}\theta}=-1$. This is a tricky case,
and we will use some fact about  cyclotomic fields from algebraic number theory. 
After a normalization by $-1$, and note that $\overline{\lambda}=-\lambda$, the signature matrix of $f$ can be written as
$$M(f)=\left[\begin{matrix}
* & 0 & 0 & *\\
0 & 1 & \lambda & 0\\
0 & \lambda & -1 & 0\\
* & 0 & 0 & *\\
\end{matrix}\right]
=\left[\begin{matrix}
* & 0 & 0 & *\\
0 & 1 & r\mathfrak{i} & 0\\
0 & r\mathfrak{i} & -1 & 0\\
* & 0 & 0 & *\\
\end{matrix}\right],$$
where $r\in \R$ and $r\neq0,\pm1$ by $\lambda\neq 0$ and $|\lambda|\neq 1$.
By connecting variables $x_1$ and $x_2$ of $f$ using $\neq_2$, we realize the binary signature $g_0=(0,1+r\mathfrak{i},r\mathfrak{i}-1,0)^T$.
Connecting two copies of $f$ using $(\neq_2)^{\otimes2}$, we get the signature $f_3$ with the signature matrix 
$$
 M(f_3)=M(f) N M(f) = \left[ \begin{matrix} 
 * & 0 & 0 & * \\ 
 0 & 2r\mathfrak{i} & -r^2-1 & 0 \\ 
 0 & -r^2-1 & -2r\mathfrak{i} & 0 \\ 
 * & 0 & 0 & * \end{matrix} \right]
=2r\mathfrak{i}\left[ \begin{matrix} 
 * & 0 & 0 & * \\ 
 0 & 1 & \frac{r^2+1}{2r}\mathfrak{i} & 0 \\ 
 0 & \frac{r^2+1}{2r}\mathfrak{i} & -1 & 0 \\ 
 * & 0 & 0 & * \end{matrix} \right].
$$
Define the map $$T(r):=\frac{r^2+1}{2r}, \text{ for } r\in\R,r\neq 0,\pm1.$$
Note that if $r\neq 0,\pm1$, then $T(r)\neq 0,\pm1$.
By connecting variables $x_1$ and $x_2$ of $f_3$ using $\neq_2$, we realize the binary signature $g_1=(0,1+T(r)\mathfrak{i},T(r)\mathfrak{i}-1,0)^T$.
If we replace $f$ with $f_3$ and also $r$ with $T(r)$, we can further realize $g_2=(0,1+T^{(2)}(r)\mathfrak{i},T^{(2)}(r)\mathfrak{i}-1)^T$.
Repeating the process, we can realize $g_n=(0,1+T^{(n)}(r)\mathfrak{i},T^{(n)}(r)\mathfrak{i}-1)^T$, for arbitrary fixed $n\in \N$.
Let $u_n:=\frac{1+T^{(n)}(r)\mathfrak{i}}{T^{(n)}(r)\mathfrak{i}-1}$ and $v_n:=-u_n=\frac{1+T^{(n)}(r)\mathfrak{i}}{1-T^{(n)}(r)\mathfrak{i}}$.
If there exists $n\in \N$ such that $u_n$ is not a root of unity, 
then we have the binary $g_n=(0,1,u_n,0)$ and by Corollary~\ref{cor: mobius: with (0,1,t,0) not of unity} we are done.
Suppose by contradiction that for every $n\in \N$, $u_n$ is a root of unity, and thus $v_n$ is a root of unity.
Define the sequence $\{r_n\}_{n\in \N}$ by $r_0=r$ and $r_{n+1}=T(r_n)$. 
Then for every $n\in \N$, $v_n=\frac{1+r_n\mathfrak{i}}{1-r_n\mathfrak{i}}$ is a root of unity.
Suppose $v_0=\frac{1+r\mathfrak{i}}{1-r\mathfrak{i}}$ is a $p$-th root of unity.
It's clear that $\{v_n\}_{n\in \N}\subseteq \Q[\zeta_p,\mathfrak{i}]$, which is a  fixed finite extension of $\Q$ (as $p$ is fixed). 
\begin{itemize}
    \item Assume $r>0$.
    Since $r\neq 1$ we have $r_1=T(r)=\frac{r^2+1}{2r}>1$.
    Given $r_n>1$, we have $r_{n+1}=T(r_n)=\frac{r_n^2+1}{2r_n}>1$.
    Also, $r_{n+1}-r_n=\frac{1-r_n^2}{2r_n}<0$.
    So $\{r_n\}_{n\ge1}$ is strictly
    monotonically decreasing, thus
    $\lim\limits_{n\to\infty}r_n$ exists and is bounded below by 1
    (thus the limit is nonzero), and by taking the limit in 
    $r_{n+1}=\frac{r_n^2+1}{2r_n}$, we get
     $\lim\limits_{n\to\infty}r_n=1$.
    So $\theta_n:=\arg(v_n)=2\arctan(r_n)$ is strictly decreasing and $\lim\limits_{n\to \infty}\theta_n=\frac{\pi}{2}$.
    For any $\epsilon>0$, there exists $N\in \N$ such that $\theta_N\in (\frac{\pi}{2},\frac{\pi}{2}+\epsilon)$.
    By assumption, $v_N$ is a root of unity of (minimal) order $m$.
    Then $m\theta_N=2k \pi$ for some $k\in \N$.
    This implies $$\frac{\theta_N}{\pi}-\frac{1}{2}=\frac{4k-m}{2m}.$$
    Since $\theta_N>\frac{\pi}{2}$, we have $4k-m\ge1$.
    Thus, 
    $$
    \frac{1}{2m}\le \frac{4k-m}{2m}=\frac{\theta_N}{\pi}-\frac{1}{2}<\frac{\epsilon}{\pi}.
    $$
    So $m>\frac{\pi}{2\epsilon}$. 
    When $\epsilon$ is arbitrarily small, $m$ is arbitrarily large, and $\deg_{\Q}(v_N)=\varphi(m)>\frac{m}{\log\log m}$ can be arbitrarily large, 
    where $\varphi(\cdot)$ is Euler's totient function, which is
    the  degree of the field extension over $\Q$, for the cyclotomic field $\Q_m$,
    contradicting $v_N\in \Q[\zeta_p,\mathfrak{i}]$.

    \hspace{0.3cm}
    \item 
    Assume $r<0$.
    Similarly we can show that $\{r_n\}_{n\ge1}$ is increasing and $\lim\limits_{n\to \infty}r_n=-1$. 
    Note that $\overline{v_n}=\frac{1-r_n\mathfrak{i}}{1+r_n\mathfrak{i}}$, where $\{-r_n\}_{n\ge 1}$ is decreasing and $\lim\limits_{n\to \infty}(-r_n)=1$.
    Then the same argument as in the case $r>0$ applies.
\end{itemize}
\end{itemize}
\end{proof}

\begin{corollary}\label{cor: get (0,1,0,0) from a binary third or fourth root}
Let $g = (0,1,t,0)^T$ be a binary signature where $t$ is an $n$-th primitive root of unity for $n = 3$ or 4, and let $f$ be a signature with the signature matrix $M(f) = \left[\begin{smallmatrix} a & 0 & 0 & b \\ 0 & c & d & 0 \\ 0 & w & z & 0 \\ y & 0 & 0 & a \end{smallmatrix} \right]$ with $acdzw \neq 0$, where $\left[\begin{smallmatrix} c & d \\ w & z\end{smallmatrix}\right]$ has full rank, then $\plholant{\neq_2}{f,(0,1,0,0)^T} \leq^p_{T} \plholant{\neq_2}{f,g}$.
\end{corollary}

\begin{proof}
Suppose we have the binary signature  $(0, 1, t, 0)^T$
where $t$ is
an $n$-th primitive root of unity, for $n=3$ or $n=4$.
Then $(0, 1, t^i, 0)^T$ are distinct for $i=1, 2, 3$.
In the proof of Lemma~\ref{lm: get any binary from a binary not fifth root},
if one of $\varphi(t^i)=0$ or $\infty$, then we have $(0, 1, 0, 0)^T$
or $(0, 0, 1, 0)^T$, which means we have both.
If none of $\varphi(t^i)=0$ or $\infty$ for  $i=1, 2, 3$,
then the proof of  Lemma~\ref{lm: get any binary from a binary not fifth root}
shows that we have $(0, 1, u, 0)^T$ for
any $u\in\mathbb{C}$. Thus we have $(0, 1, 0, 0)^T$ as well.
The corollary follows.
\end{proof}

\begin{corollary}\label{cor:get any binary when exactly three nonzero}
Let $f$ be a 4-ary signature with the signature matrix $M(f) = \left[\begin{smallmatrix} a & 0 & 0 & b \\ 0 & c & d & 0 \\ 0 & w & z & 0 \\ y & 0 & 0 & a \end{smallmatrix} \right]$. If there are exactly three nonzero entries in the inner matrix $\left[\begin{smallmatrix} c& d \\ w & z \end{smallmatrix}\right]$ then for any $u \in \mathbb{C}$,
$$
\plholant{\neq_2}{f,(0,1,u,0)^T} \leq_T^p \plholant{\neq_2}{f}.
$$
\end{corollary}

\begin{proof}
The proof of Corollary 4.4 in~\cite{CaiF17} preserves planarity. 
By the rotation symmetry, we can assume that $(z= 0) \wedge (cdw \neq0)$ or $(w=0)\wedge (cdz \neq 0)$, where the former case is handled in the proof there and the latter case can be done similarly.
\end{proof}

\bibliographystyle{plain}
\bibliography{reference.bib}

@article{CaiF17,
title = {Complexity classification of the eight-vertex model},
journal = {Information and Computation},
volume = {293},
pages = {105064},
year = {2023},
issn = {0890-5401},
doi = {https://doi.org/10.1016/j.ic.2023.105064},
url = {https://www.sciencedirect.com/science/article/pii/S0890540123000676},
author = {Jin-Yi Cai and Zhiguo Fu},
}

@inproceedings{CaiFS21,
  author    = {Jin{-}Yi Cai and
               Zhiguo Fu and
               Shuai Shao},
  title     = {New Planar {P}-time Computable Six-Vertex Models and a Complete Complexity
               Classification},
  booktitle = {{SODA}},
  pages     = {1535--1547},
  publisher = {{SIAM}},
  year      = {2021}
}

@article{CaiLX14,
  author    = {Jin{-}Yi Cai and
               Pinyan Lu and
               Mingji Xia},
  title     = {The complexity of complex weighted Boolean {\#}{CSP}},
  journal   = {J. Comput. Syst. Sci.},
  volume    = {80},
  number    = {1},
  pages     = {217--236},
  year      = {2014}
}

@article{CaiL20,
  author    = {Jin{-}Yi Cai and
               Tianyu Liu},
  title     = {{FPRAS} via {MCMC} where it mixes torpidly (and very little effort)},
  journal   = {CoRR},
  volume    = {abs/2010.05425},
  year      = {2020}
}

@article{CaiFX18,
  author    = {Jin{-}Yi Cai and
               Zhiguo Fu and
               Mingji Xia},
  title     = {Complexity classification of the six-vertex model},
  journal   = {Inf. Comput.},
  volume    = {259},
  number    = {Part},
  pages     = {130--141},
  year      = {2018}
}

@article{k-regular-spin,

author    = {Peng Yang and
               Yuan Huang and
               Zhiguo Fu},
title = {A complexity trichotomy for k-regular asymmetric spin systems with complex edge functions},
journal = {Theoretical Computer Science},
volume = {1020},
pages = {114835},
year = {2024},
}

@article{caiguowilliams13,
  title={A complete dichotomy rises from the capture of vanishing signatures},
  author={Cai, Jin-Yi and Guo, Heng and Williams, Tyson},
  journal={SIAM Journal on Computing},
  volume={45},
  number={5},
  pages={1671--1728},
  year={2016},
  publisher={SIAM}
}

@article{val02b,
  title={Expressiveness of matchgates},
  author={Valiant, Leslie G.},
  journal={Theoretical Computer Science},
  volume={289},
  number={1},
  pages={457--471},
  year={2002},
  publisher={Elsevier}
}

@book{jcbook,
  title={Complexity Dichotomies for Counting Problems: Volume 1, Boolean Domain},
  author={Cai, Jin-Yi and Chen, Xi},
  year={2017},
  publisher={Cambridge University Press}
}

@article{Val08,
  title={Holographic algorithms},
  author={Valiant, Leslie G.},
  journal={SIAM Journal on Computing},
  volume={37},
  number={5},
  pages={1565--1594},
  year={2008},
  publisher={SIAM}
}

@book{ahlfors,
  title={Complex analysis: an introduction to the theory of analytic functions of one complex variable},
  author={Ahlfors, Lars},
  year={1979},
  publisher={McGraw-Hill}
}

@article{bulatov,
  title={The complexity of the counting constraint satisfaction problem},
  author={Bulatov, Andrei},
  journal={Journal of the ACM (JACM)},
  volume={60},
  number={5},
  pages={1--41},
  year={2013},
  publisher={ACM New York, NY, USA}
}

@article{caichen,
  title={{Complexity of counting CSP with complex weights}},
  author={Cai, Jin-Yi and Chen, Xi},
  journal={Journal of the ACM (JACM)},
  volume={64},
  number={3},
  pages={1--39},
  year={2017},
  publisher={ACM New York, NY, USA}
}

@article{caifu16,
  title={Holographic algorithm with matchgates is universal for planar \#{CSP} over boolean domain},
  author={Cai, Jin-Yi and Fu, Zhiguo},
  journal={SIAM Journal on Computing},
  volume={special issue of STOC 17},
  pages={50--151},
  year={2022},
  publisher={SIAM}
}

@article{Cai-Lu-Xia,
  title={Computational complexity of {H}olant problems},
  author={Cai, Jin-Yi and Lu, Pinyan and Xia, Mingji},
  journal={SIAM Journal on Computing},
  volume={40},
  number={4},
  pages={1101--1132},
  year={2011},
  publisher={SIAM}
}

@article{dyer-richerby,
  title={An effective dichotomy for the counting constraint satisfaction problem},
  author={Dyer, Martin and Richerby, David},
  journal={SIAM Journal on Computing},
  volume={42},
  number={3},
  pages={1245--1274},
  year={2013},
  publisher={SIAM}
}

@article{goldberg-jerrum-patterson,
  title={The computational complexity of two-state spin systems},
  author={Goldberg, Leslie Ann and Jerrum, Mark and Paterson, Mike},
  journal={Random Structures \& Algorithms},
  volume={23},
  number={2},
  pages={133--154},
  year={2003},
  publisher={Wiley Online Library}
}

@article{goldberg-jerrum-potts,
  title={Approximating the partition function of the ferromagnetic {Potts} model},
  author={Goldberg, Leslie Ann and Jerrum, Mark},
  journal={Journal of the ACM (JACM)},
  volume={59},
  number={5},
  pages={1--31},
  year={2012},
  publisher={ACM New York, NY, USA}
}

@article{jerrum-sinclair,
  title={{Polynomial-time approximation algorithms for the Ising model}},
  author={Jerrum, Mark and Sinclair, Alistair},
  journal={SIAM Journal on Computing},
  volume={22},
  number={5},
  pages={1087--1116},
  year={1993},
  publisher={SIAM}
}

@incollection{Kasteleyn1961,
  title={The statistics of dimers on a lattice},
  author={Kasteleyn, Pieter W.},
  booktitle={Classic Papers in Combinatorics},
  pages={281--298},
  year={2009},
  publisher={Springer}
}

@article{Kasteleyn1967,
  title={Graph theory and crystal physics},
  author={Kasteleyn, Pieter W.},
  journal={Graph theory and theoretical physics},
  pages={43--110},
  year={1967},
  publisher={Academic Press}
}

@inproceedings{lu-ying-li,
  title={Correlation decay up to uniqueness in spin systems},
  author={Li, Liang and Lu, Pinyan and Yin, Yitong},
  booktitle={Proceedings of the twenty-fourth annual ACM-SIAM symposium on Discrete algorithms},
  pages={67--84},
  year={2013},
  organization={SIAM}
}

@article{Pauling,
  title={The structure and entropy of ice and of other crystals with some randomness of atomic arrangement},
  author={Pauling, Linus},
  journal={Journal of the American Chemical Society},
  volume={57},
  number={12},
  pages={2680--2684},
  year={1935},
  publisher={ACS Publications}
}

@article{TF61,
  title={Dimer problem in statistical mechanics-an exact result},
  author={Temperley, Harold N. V. and Fisher, Michael E.},
  journal={Philosophical Magazine},
  volume={6},
  number={68},
  pages={1061--1063},
  year={1961},
  publisher={Taylor \& Francis}
}

@article{val02a,
  title={Quantum circuits that can be simulated classically in polynomial time},
  author={Valiant, Leslie G.},
  journal={SIAM Journal on Computing},
  volume={31},
  number={4},
  pages={1229--1254},
  year={2002},
  publisher={SIAM}
}

@article{bulatov2012csp,
  title={The complexity of weighted and unweighted \#{CSP}},
  author={Bulatov, Andrei and Dyer, Martin and Goldberg, Leslie Ann and Jalsenius, Markus and Jerrum, Mark and Richerby, David},
  journal={Journal of Computer and System Sciences},
  volume={78},
  number={2},
  pages={681--688},
  year={2012},
  publisher={Elsevier}
}

@article{cai-chen-lu-nonnegative-csp,
  title={Nonnegative weighted \#{CSP}: An effective complexity dichotomy},
  author={Cai, Jin-Yi and Chen, Xi and Lu, Pinyan},
  journal={SIAM Journal on Computing},
  volume={45},
  number={6},
  pages={2177--2198},
  year={2016},
  publisher={SIAM}
}

@inproceedings{real-holant-focs2021,
  title={A Dichotomy for Real {Boolean} {Holant} Problems},
  author={Shao, Shuai and Cai, Jin-Yi},
  booktitle={Proceedings of the 61st IEEE Annual Symposium on Foundations of Computer Science (FOCS)},
   pages= {1091--1102},
  year={2020}
}

@article{cai2015holant,
  title={{FKT is not universal} — {A planar Holant dichotomy for symmetric constraints}},
  author={Cai, Jin-Yi and Fu, Zhiguo and Guo, Heng and Williams, Tyson},
  journal={Theory of Computing Systems},
  volume={66},
  number={1},
  pages={143--308},
  year={2022},
  publisher={Springer}
}

@inproceedings{de2011quantum,
  title={Quantum algorithms for classical lattice models},
  author={De\hspace{1ex}las\hspace{1ex}Cuevas, Gemma and D{\"u}r, Wolfgang and Van\hspace{1ex}den\hspace{1ex}Nest, Maarten and MartinDelgado, Miguel A},
  journal={New Journal of Physics},
  volume={13},
  number={9},
  pages={093021},
  year={2011},
  publisher={IOP Publishing}
}

@article{sutherland1970two,
  title={Two-Dimensional Hydrogen Bonded Crystals without the Ice Rule},
  author={Sutherland, Bill},
  journal={Journal of Mathematical Physics},
  volume={11},
  number={11},
  pages={3183--3186},
  year={1970},
  publisher={American Institute of Physics}
}

@article{fan1970general,
  title={General lattice model of phase transitions},
  author={Fan, Chungpeng and Wu, F Yu},
  journal={Physical Review B},
  volume={2},
  number={3},
  pages={723},
  year={1970},
  publisher={APS}
}

@article{WELSH1969375,
title = "Euler and bipartite matroids",
journal = "Journal of Combinatorial Theory",
volume = "6",
number = "4",
pages = "375 - 377",
year = "1969",
issn = "0021-9800",
doi = "https://doi.org/10.1016/S0021-9800(69)80033-5",
url = "http://www.sciencedirect.com/science/article/pii/S0021980069800335",
author = "D.J.A. Welsh",
}

@article{Valiant79,
  title={The Complexity of Enumeration and Reliability Problems},
  author={Valiant, Leslie G.},
  journal={SIAM Journal on Computing},
  volume={8},
  number={3},
  pages={410--421},
  year={1979},
  publisher={SIAM}
}

@article{baxter1971eightvertex,
  author       = {Baxter, R. J.},
  title        = {Eight‑Vertex Model in Lattice Statistics},
  journal      = {Physical Review Letters},
  volume       = {26},
  number       = {14},
  pages        = {832--834},
  month        = apr,
  year         = {1971},
  doi          = {10.1103/PhysRevLett.26.832}
}

\end{document}